\documentclass[12pt]{article}
\usepackage{amsmath} 
\usepackage{amsthm}
\usepackage{amstext}
\usepackage{amssymb}
\usepackage{enumerate}
\usepackage{array}
\usepackage{color}
\usepackage{bbm}
\usepackage{natbib}
\usepackage{graphicx}
\usepackage{caption}
\usepackage{booktabs}
\usepackage[title]{appendix}
\usepackage{url} 
\usepackage{fancyhdr}
\usepackage{authblk}
\usepackage{xr}
\usepackage{makecell,longtable}
\usepackage{titletoc}
\allowdisplaybreaks

\renewcommand{\cite}{\citep}
\newcommand{\blind}{1}
\newtheorem{remark}{Remark}
\newtheorem{assumption}{Assumption}
\newtheorem{theorem}{Theorem}
\newtheorem{lemma}{Lemma}
\newtheorem{example}{Example}
\newtheorem{corollary}{Corollary}

\usepackage{algorithm, algorithmicx, algpseudocode}

\floatname{algorithm}{Algorithm}

\usepackage{xr}
\makeatletter
\newcommand*{\addFileDependency}[1]{
	\typeout{(#1)}
	\@addtofilelist{#1}
	\IfFileExists{#1}{}{\typeout{No file #1.}}
}
\makeatother

\begin{document}

\def\spacingset#1{\renewcommand{\baselinestretch}%
{#1}\small\normalsize} \spacingset{1}


\if1\blind
{
  \title{\bf A unified theory of conditional coverage in conformal prediction with applications}
  \author{
    Yinjie Min [nk.yjmin@gmail.com]\\
    School of Statistics and Data Science, Nankai University\\
    Liuhua Peng [liuhua.peng@unimelb.edu.au]\\ 
    School of Mathematics \& Statistics, the University of Melbourne\\
    and \\
    Changliang Zou [nk.chlzou@gmail.com]\\
    School of Statistics and Data Science, Nankai University}
  \maketitle
} \fi

\if0\blind
{
  \bigskip
  \bigskip
  \bigskip
  \begin{center}
    {\LARGE\bf A unified theory of conditional coverage in conformal prediction with applications}
\end{center}
  \medskip
} \fi

\bigskip
\begin{abstract}
Conformal prediction provides prediction sets with finite-sample marginal coverage, but many applications require coverage guarantees that adapt to individual test points, a subpopulation, or a structural component of the data. Existing methods targeting conditional coverage are largely analyzed case by case, leaving limited general theory for understanding where conditional miscoverage comes from, how different procedures should be compared, and how such guarantees can be extended beyond i.i.d.~data.
We address these gaps through a unified framework and theory for conformal methods targeting conditional coverage. Our central contribution is a non-asymptotic decomposition of conditional miscoverage into three interpretable components: score-estimation error, finite-sample calibration error, and intrinsic conditional-mismatch error. This decomposition clarifies the mechanisms behind asymptotic conditional validity and places existing methods within a common analytical lens. Building on this framework, we derive principled guidance for conditional-coverage-oriented model selection, and develop localized methods with asymptotic conditional guarantees under covariate shift. Finally, we extend the framework to structured data, with concrete applications to graph-structured and hierarchical settings. Numerical experiments corroborate the theory and demonstrate the effectiveness of the proposed procedures.
\end{abstract}

\noindent%
{\it Keywords:}  Asymptotic conditional validity; Conditional distribution; Finite-sample calibration; Model selection; Non-asymptotic error bound

\vfill

\newpage
\spacingset{1.9} 

\section{Introduction}

Conformal prediction, introduced by \cite{vovk2005algorithmic}, provides distribution-free and model-agnostic prediction sets with finite-sample coverage guarantees. 
Let $(X_i,Y_i),~i=1,\ldots,n+1$ be a sequence of independent and identically distributed
(i.i.d.) data, where $X_i \in \mathcal{X} \subseteq \mathbb{R}^d$ is a feature vector and $Y_i\in \mathcal{Y} \subseteq \mathbb{R}$ is the response.
The first $n$ samples serve as calibration data, while $(X_{n+1},Y_{n+1})$ denotes a test point with $Y_{n+1}$ unobserved.
Conformal prediction constructs a prediction set $\widehat{C}(X_{n+1})$ for the test point that satisfies the following finite-sample marginal coverage guarantee
\begin{gather*}
    \mathrm{Pr}\bigl(Y_{n+1}\in\widehat{C}(X_{n+1})\bigr)\geq 1-\alpha\,,
\end{gather*}
for a specified level $1-\alpha\in(0,1)$. Split conformal prediction (SCP), one of the most common variants, uses a pre-trained conformity score function $v:\mathcal{X}\times\mathcal{Y}\to\mathbb{R}$.
For example, $v(x,y)=|y-\hat{\mu}(x)|$ with a pre-trained prediction model $\hat{\mu}(\cdot):\mathcal{X}\to\mathcal{Y}$.
SCP then yields the prediction set $\widehat{C}_{\rm SCP}(X_{n+1})=\left\{y:v(X_{n+1},y)\leq Q\left(1-\alpha;(n+1)^{-1}\left(\sum_{i=1}^{n}\delta_{v(X_i,Y_i)}+\delta_{\infty}\right)\right)\right\}$ that guarantees marginal coverage,
where $Q(1-\alpha;\cdot)$ denotes the $(1-\alpha)$-quantile of the distribution in the second argument, and $\delta_a$ is the point mass at $a$. This guarantee holds in finite samples and requires no assumptions beyond exchangeability \cite{vovk2005algorithmic}.

Marginal validity, however, is a population-level guarantee that can mask severe local failures. A procedure may achieve correct marginal coverage while systematically undercovering in certain regions of the feature space. This motivates conditional targets, such as test-conditional coverage,
\begin{gather*}
    \mathrm{Pr}\bigl(Y_{n+1}\in\widehat{C}(X_{n+1})\mid X_{n+1}=x\bigr)\,,
\end{gather*}
as well as other notions such as subgroup-conditional coverage in batch data \cite{jung2023batch} and coverage conditional on specific communities in graph data, which we refer to as community-conditional coverage (Section~\ref{sec: community conditional}). 
However, exact finite-sample distribution-free conditional coverage is generally impossible, particularly for the test-conditional target \cite{lei2013,foygel2021limits}.
The practical challenge is therefore to construct conformal prediction sets that {\it retain marginal validity} while {\it achieving meaningful approximate or asymptotic conditional coverage}.

For example, in the test-conditional setting for SCP, if the conditional cumulative distribution function (c.d.f.) of $v(X_{n+1},Y_{n+1})$ given $X_{n+1}=x$, denoted by $F_{v\mid X}(\cdot\mid x)$, were known, then the theoretical prediction set
\begin{gather}\label{eq:theo_CP}
    \widehat{C}_{\rm Theo}(X_{n+1})=\left\{y:v(X_{n+1},y)\leq Q\bigl(1-\alpha;F_{v\mid X}(\cdot\mid X_{n+1})\bigr)\right\}
\end{gather}
would achieve the desired test-conditional coverage property. 
In finite samples, however, this theoretical conditional quantile $Q\bigl(1-\alpha;F_{v\mid X}(\cdot\mid X_{n+1})\bigr)$ is not directly accessible. Existing conformal methods targeting conditional coverage can be broadly understood as approximating this quantile through two methodological routes.

The first route improves conditional performance by modifying the conformity score while retaining standard conformal calibration. Examples include the locally weighted residual score of \cite{lei2018distribution}, which rescales residuals by an estimate of the conditional mean absolute deviation to accommodate heteroscedasticity. Conformalized quantile regression \cite[CQR]{romano2019conformalized} constructs scores from fitted conditional quantiles, while distributional conformal prediction \cite[DCP]{Chernozhukov2021distribution} transforms a base score using an estimated conditional distribution. In these methods, conditional information enters primarily through the score rather than the calibration threshold. By contrast, a second route keeps a pre-trained or base score and focuses on approximating the theoretical conditional quantile in \eqref{eq:theo_CP} through different calibration strategies.
Localized conformal prediction \cite[LCP]{guan2023localized} uses local weights around $X_{n+1}$, Randomized LCP \cite[RLCP]{hore2025conformal} stabilizes this approximation via randomized local weights, and  Conditional Calibration \cite[CC]{gibbs2025conformal} estimates the conditional quantile directly through functional quantile regression. Thus, in this route, conditional information is incorporated mainly through the calibration threshold.


Despite recent progress, several gaps remain. 
First, existing methods for conditional coverage are largely developed and analyzed separately. As a result, there is no unified framework that explicitly characterizes the components driving the gap between the achieved conditional coverage and the target level $1-\alpha$, i.e., the {\it conditional miscoverage}, and therefore no clear account of when and why asymptotic conditional validity holds.  Second, conditional coverage remains underdeveloped for data with complex structure, such as graph-structured prediction \cite{lunde2025conformal}, where calibration exploits structural symmetry rather than sample exchangeability \cite{dobriban2025symmpi}. Third, when conditional coverage is the primary target, existing model selection approaches \cite{braun2025conditional, zhou2026conformal} recast the problem as a supervised classification task on hold-out data, leaving open whether the construction mechanism of conformal procedures can itself be used to guide selection without a separate hold-out set.


To bridge these gaps, we develop a unified theory of conditional coverage for conformal prediction. The theory applies across the main existing approaches and enables a range of applications and extensions.
The contributions of this paper are as follows.

1. We introduce a unified framework for conformal methods targeting conditional coverage through weighted conformal quantiles, expressing a broad class of existing procedures within this formulation. We then derive a non-asymptotic bound for conditional miscoverage, as well as an averaged version under weaker conditions. Together, these results clarify how existing procedures achieve asymptotic conditional validity and identify the key quantities that govern their error rates.

2. The unified theory provides methodological guidance when conditional coverage is the primary target. In particular, it offers a principled basis for conditional-coverage-oriented model selection, motivates a generalization of RLCP through a broader auxiliary reweighting scheme, and supports extensions to localized conformal $p$-values.
Leveraging the weighted framework, our theory adapts naturally to localization methods under covariate shift, enabling new procedures with asymptotic conditional coverage.

3. We extend the framework beyond i.i.d. data through a symmetry-based formulation inspired by SymmPI \cite{dobriban2025symmpi}. A weighted version handling distributional shifts under group actions is developed in the Supplementary Material. Concrete applications are developed for graph-structured data, targeting community-conditional coverage, and for hierarchical data.

\paragraph*{Related Work}


There is a substantial literature on conformal prediction targeting conditional coverage, encompassing score-based methods such as locally weighted residual scores \cite{lei2018distribution}, CQR \cite{romano2019conformalized}, DCP \cite{Chernozhukov2021distribution}, and GLCP \cite{min2025personalized}, as well as calibration-based methods such as LCP \cite{guan2023localized}, RLCP \cite{hore2025conformal}, and CC \cite{gibbs2025conformal}.  Recent unifying perspectives have clarified related aspects of the conformal landscape: \cite{angelopoulos2024theoretical}  studied theoretical foundations and asymptotic behavior of conformal prediction; \cite{barber2026unifying} offered a partial-information conditioning viewpoint; and  \cite{dobriban2025symmpi} extended marginal validity theory to structured data. In contrast, our framework is specifically directed at conditional coverage, provides a non-asymptotic decomposition that applies across a broad class of existing methods, and yields new procedures and extensions not covered by prior work.

For model selection, existing work mostly focused on  efficiency while preserving marginal coverage \cite{yang2025selection,liang2024conformal}. Recent works such as \cite{braun2025conditional} and \cite{zhou2026conformal} study conditional-coverage-aware selection by framing coverage evaluation as a supervised learning problem on hold-out data. Our selection criterion is complementary: it is derived directly from the conditional-miscoverage decomposition, targets the construction mechanism of conformal procedures, and does not require a separate hold-out set.

The rest of the paper is organized as follows. Section~\ref{sec: general framework} introduces the unified framework. Section~\ref{sec: Unified Theory of Conditional Coverage} develops the main theory, including the non-asymptotic error bound for conditional miscoverage and averaged conditional miscoverage. Section~\ref{sec:application} presents a wide range of applications and extensions. Numerical experiments are reported in Section~\ref{sec: experimental simulation}, and concluding remarks are in Section~\ref{sec: conclusion}. Further technical details, theoretical proofs, and additional applications and experiments are deferred to the Supplementary Material.

\section{A Unified Framework}\label{sec: general framework}

In this section, we introduce a unified framework for conformal methods targeting conditional coverage via weighted conformal quantiles, encompassing a broad class of existing approaches.
Let $Z=\{(X_1,Y_1),\ldots,(X_n,Y_n),(X_{n+1},Y_{n+1})\}$ denote the collection of calibration and test data.
We consider a conformity score of the form, $s(x,y;z)$, to indicate its dependence on the data $z$.  We first state the following assumptions on the data distribution and the score construction, which are standard in the conformal literature.

\begin{assumption}\label{ass: data_distribution_and_score_invariance}
    (i) The calibration pairs $\{(X_i,Y_i)\}_{i\in[n]}$ and the test pair $(X_{n+1},Y_{n+1})$ are i.i.d.~from $P_{X}\times P_{Y\mid X}$; (ii) For $z=\{(x_1,y_1),\ldots,(x_n,y_n),(x_{n+1},y_{n+1})\}$, $s(x,y;z)$ is permutation invariant in $\{(x_i,y_i)\}_{i\in[n+1]}$.
\end{assumption}

Given a trial response value $y$, denote $Z^y=\{(X_1,Y_1),\ldots,(X_n,Y_n),(X_{n+1},y)\}$. We introduce the following conformal prediction set
\begin{align}
    \widehat{C}(X_{n+1}) = \left\{ y:\ s(X_{n+1},y;Z^y)\leq q(Z^y;\alpha)\right\}\,,\label{eq: general Cvec}
\end{align}
where $$q(Z^y;\alpha) = Q\left( 1-\alpha; \left\{\sum_{i=1}^{n+1}w(X_i)\right\}^{-1}\left\{\sum_{i=1}^{n}w(X_i)\delta_{s(X_i,Y_i;Z^y)}+w(X_{n+1})\delta_{\infty}\right\} \right)$$ is a ``weighted'' conformal quantile, and $w(x)$ can be either a fixed or a random function. 

The form of $\widehat{C}(X_{n+1})$ is introduced primarily to unify conformal methods targeting conditional coverage, though it is closely related to weighted conformal prediction, which was originally proposed to address covariate shift \cite[WCP]{tibshirani2019conformal}. 
For score-based methods, such as CQR and split DCP, $s(X_{n+1},y;Z^y)$ takes the form of $v(X_{n+1},y)$ and $w(\cdot)\equiv 1$, where $v(x,y)$ is a pre-trained score that does not depend on $Z^y$. For calibration-based methods, the score and weight functions vary case by case. As an example, the localized conformal prediction without marginal-coverage adjustment, referred to as BaseLCP in \cite{hore2025conformal}, is equivalent to using the score $s(X_{n+1},y;Z^y)=v(X_{n+1},y)$ and the weight function $w(x)=K(x,X_{n+1};h)$, where $K(\cdot,\cdot;h)$ is a kernel function with bandwidth $h$. In contrast, the standard LCP method proposed by \cite{guan2023localized} can be written in the form of \eqref{eq: general Cvec} with $s(x,y;Z)=\{\sum_{j=1}^{n+1}K(x,X_j;h)\mathbbm{1}(v(X_j,Y_j)\leq v(x,y))\}/\{\sum_{j=1}^{n+1}K(x,X_j;h)\}$ and $w(\cdot)\equiv 1$, while RLCP \cite{hore2025conformal} uses $s(X_{n+1},y;Z^y)=v(X_{n+1},y)$ and $w(x)=K(x,\widetilde X;h)$ with $\widetilde{X}$ drawn from a distribution with density proportional to $K(\cdot,X_{n+1};h)$ conditional on $X_{n+1}$.



The marginal coverage of $\widehat{C}(X_{n+1})$ is established in the next theorem.
\begin{theorem}\label{theo: marginal general}
    Suppose Assumption~\ref{ass: data_distribution_and_score_invariance} holds. For $X\sim P_X$, define $\mathcal{W}_X=\{w_0(\cdot)\geq 0:\mathbb{E}\{w_0(X)\mid w_0\}=1 \text{ and }\mathbb{E}\{w_0(x)\}=1 \text{ for all } x\in\mathcal{X}\}$, where $\mathbb{E}\{w_0(X)\mid w_0\}$ is the expectation taken over $X\sim P_X$, and $\mathbb{E}\{w_0(x)\}$ is the expectation taken over the possible randomness of $w_0$ for a fixed $x\in\mathcal{X}$. 
    Then
    \begin{align*} 
        1-\alpha-\inf_{w_0\in\mathcal{W}_X}\mathbb{E}\left\{ d_Z(w,w_0) \right\}&\leq \mathrm{Pr}\left( Y_{n+1}\in \widehat{C}(X_{n+1}) \right)\\
        &\leq 1-\alpha+\inf_{w_0\in\mathcal{W}_X}\mathbb{E}\left\{ d_Z(w,w_0) \right\}+\mathbb{E}\left\{ \widehat{w}_{\max}(Z) \right\},
    \end{align*}
    where $d_Z(w,w_0)=\sum_{j=1}^{n+1}\bigl|\bigl\{ \sum_{i=1}^{n+1}w(X_i) \bigr\}^{-1}w(X_j)-\bigl\{ \sum_{i=1}^{n+1}w_0(X_i) \bigr\}^{-1}w_0(X_j)\bigr|$, 
    $\widehat{w}_{\max}(Z)=\sup_{c}\bigl\{ \sum_{j=1}^{n+1}w(X_j) \bigr\}^{-1}\bigl\{ \sum_{i=1}^{n+1}w(X_i)\mathbbm{1}(s(X_i,Y_i;Z)=c) \bigr\}$.
\end{theorem}
   
The class $\mathcal W_X$ collects normalized weighting schemes whose average normalized mass assigned to each fixed covariate value is one. Such weights preserve the label symmetry required for marginal conformal validity. The term $d_Z(w,w_0)$ measures how far the actual normalized weights are from an admissible marginally valid weighting scheme.
The terms $\mathbb{E}\left\{ \widehat{w}_{\max}(Z) \right\}$ and
$\inf_{w_0\in\mathcal{W}_X}\mathbb{E}\{d_Z(w,w_0)\}$ jointly quantify the marginal coverage gap. The former depends mainly on the dispersion of the weights $\{w(X_i)\}_{i\in[n+1]}$. When $w(\cdot) \equiv 1$ and there are no ties among the scores, $\widehat{w}_{\max}(Z)=(n+1)^{-1}$. To determine whether the gap induced by $\inf_{w_0\in\mathcal{W}_X}\mathbb{E}\{d_Z(w,w_0)\}$ is zero, it entails examining whether $\sup_{x\in\mathcal{X}}|\mathbb{E}\{\bar{w}(x)\}-1|=0$, where $\bar{w}(x)=w(x)/\mathbb{E}\{w(X)\mid w\}$ for $X\sim P_X$.  
When $w(\cdot)\equiv 1$, we have $\sup_{x\in\mathcal{X}}|\mathbb{E}\{\bar{w}(x)\}-1|=0$. If, in addition, there are no ties, the marginal coverage reduces to that of the classical full-conformal prediction set, i.e., $[1-\alpha,\;1-\alpha+1/(n+1)]$.
For BaseLCP, $w(x)=K(x,X_{n+1};h)$ and $\bar{w}(x)=w(x)/K_X(X_{n+1})$, where $K_X(x)=\mathbb{E}\{K(x,X;h)\}$. It can be shown that $\mathbb{E}\{\bar{w}(x)\}$ is not identically equal to $1$ when $K_X(X_{n+1})$ is not constant. Thus, $\inf_{w_0\in\mathcal{W}_X}\mathbb{E}\{d_Z(w,w_0)\}>0$ for BaseLCP, leading to a nonzero marginal coverage gap. 
In RLCP, $w(x)=K(x,\widetilde{X};h)$ and $\bar{w}(x)=w(x)/K_X(\widetilde{X})$. A direct integration yields $\mathbb{E}\{\bar{w}(x)\}=1$ for all $x\in\mathcal{X}$. Hence, RLCP guarantees marginal coverage. A detailed discussion of marginal validity can be found in Section~\ref{sec: discussion of marginal} of the Supplementary Material.

Under this unified framework, the conditional coverage is governed by the relationship between the target theoretical conditional quantile and the ``weighted'' conformal quantile $q(Z^y;\alpha)$. Writing $\widehat{C}(X_{n+1})$ in the form of \eqref{eq: general Cvec} is therefore key to a generic analysis of conditional validity across different approaches. 
 
\section{Unified Theory of Conditional Coverage}\label{sec: Unified Theory of Conditional Coverage}
Under the unified framework introduced in Section~\ref{sec: general framework}, achieving conditional coverage is governed by two design choices: the score and the weighting scheme $w(\cdot)$. 
These perspectives raise a fundamental question: how do different design choices affect conditional coverage, and when can asymptotic conditional validity be achieved?
To answer this question, we develop a unified theory that quantifies conditional miscoverage through an error decomposition applicable across different conformal methods.

Marginal validity is typically a baseline requirement in conformal prediction. Consequently, we focus on weighting schemes that preserve marginal coverage, as characterized in Section~\ref{sec: general framework}. The conditional-miscoverage theory developed in this section, however, applies to a general weighting function $w(\cdot)$ and is not restricted to marginally valid weights.

We consider a unified formulation of the conditioning event.
Let $\mathcal{T}$ be an index set and let $\phi:\mathcal{T}\to\mathcal{A}$ map each $t\in\mathcal{T}$ to a measurable event $\phi(t)\in\mathcal{A}$. The conditional coverage of interest is $$\mathrm{Pr}\bigl(Y_{n+1}\in\widehat{C}(X_{n+1})\mid \phi(t)\bigr).$$ 
For example, when test-conditional coverage is considered, $\phi(t)=\{X_{n+1}=t\}$ with $\mathcal{T}=\mathcal{X}$. 
For group-conditional coverage, one may take $\mathcal{T}=\{0,1\}$ and $\phi(t)=\{\mathbbm{1}(X_{n+1}\in B)=t\}$ for some fixed $B\subset\mathcal{X}$.
We associate the conditioning event $\phi(T)$ with a variable $T=T(X_{n+1})\in\mathcal T$, which is independent of the calibration data $\{(X_i,Y_i)\}_{i\in[n]}$. In the test-conditional setting, one takes $T=X_{n+1}$, whereas in the group-conditional setting, $T=\mathbbm{1}(X_{n+1}\in B)$.
Let $P_T$ be the law of $T(X)$ for $X\sim P_X$. 
Given $w(\cdot)$, define $P_{T,w}$ to be the law of $T(X_{w})$, where $X_{w}\sim P_{X,w}$ and $P_{X,w}$ has density proportional to $w(x)dP_X(x)$.
Let $r_{w}=dP_{T,w}/dP_{T}$ whenever $P_{T,w}\ll P_{T}$. In the test-conditional case $T=X_{n+1}$, $r_{w}$ is proportional to $w$.

Our theory is built around an ``oracle'' score $s^\star:\mathcal{X}\times \mathcal{Y}\to\mathbb{R}$, which serves as the population target of the learned score $s(x,y;z)$. 
\begin{assumption}\label{ass: pointwise_convergence}
    There exist a fixed population-level score function $s^\star(\cdot,\cdot)$ and a function $\delta_n(\varepsilon)$ such that $\mathrm{Pr}\left( |s(x,y;Z)-s^\star(x,y)|>\varepsilon \right)\leq\delta_n(\varepsilon)$ for any $x\in\mathcal{X}$ and $y\in\mathcal{Y}$. 
\end{assumption}
A key role of $s^\star$ is to remove the dependence on the full dataset $Z$. This oracle baseline highlights that the central problem is to quantify conditional coverage by comparing the learned score with its oracle counterpart. 
When the score is trained on an independent dataset 
as in split conformal, CQR, and RLCP, the oracle score can be obtained by letting the training sample size go to infinity.
For methods such as LCP and CC, a pre-trained base score $v(x,y)$ may be used, but the final score depends on the calibration sample through an estimated conditional distribution or conditional quantile. In this case, the oracle score is the population-level version of the learned score. For example, in LCP one may define $s^\star(x,y)=F_{v\mid X}(v(x,y)\mid x)$. 


Let $q^\star(Z;\alpha)= Q\bigl( 1-\alpha; \bigl\{\sum_{i=1}^{n+1}w(X_i)\bigr\}^{-1}\bigl\{\sum_{i=1}^{n}w(X_i)\delta_{s^\star(X_i,Y_i)}+w(X_{n+1})\delta_{\infty}\bigr\} \bigr)$ denote the \textit{oracle weighted conformal quantile} obtained by replacing $s$ with $s^\star$ in $q(Z;\alpha)$. Denote $F_{s^\star\mid \phi(t)}$ as the conditional distribution of $s^\star(X_{n+1},Y_{n+1})$ given $\phi(t)$. Then $Q(1-\alpha;F_{s^\star\mid \phi(t)})$ is the theoretical threshold associated with conditional coverage on the event $\phi(t)$ under the oracle score $s^\star$, whereas $q^\star(Z;\alpha)$ is the threshold produced by weighted conformal calibration over $\{s^\star(X_i,Y_i)\}_{i\in[n]}$.
We refer to $F_{s^\star\mid \phi(t)}$ and $Q(1-\alpha;F_{s^\star\mid \phi(t)})$ as the \textit{oracle conditional distribution} and \textit{oracle conditional quantile}, respectively, associated with $s^\star$.
In the test-conditional case, the oracle threshold reduces to the $(1-\alpha)$-quantile of the conditional distribution of $s^\star(X_{n+1},Y_{n+1})$ given $X_{n+1}=t$.
Later we will show that the discrepancy between $Q(1-\alpha;F_{s^\star\mid \phi(t)})$ and $q^\star(Z;\alpha)$ plays an important role in
conditional coverage. This viewpoint is analogous in spirit to Theorem~5.1 of \cite{angelopoulos2024theoretical}, but is adapted to the $w$-weighted representation under a conditional perspective. 



The following stability condition is imposed to control the influence of replacing a single observation in $Z$ on the score $s(x,y;Z)$
\cite{bian2023training, liang2025algorithmic}. 
\begin{assumption}[probabilistic uniform stability]\label{ass: stable score}
    For each $j\in[n+1]$ and any two replacement samples $z_{j,1},z_{j,2}\in\mathcal{X}\times\mathcal{Y}$, let $Z_{j,1}$ and $Z_{j,2}$ be the datasets obtained by replacing $(X_j,Y_j)$ in $Z$ with $z_{j,1},z_{j,2}$, respectively. 
    Assume there exists $\widetilde{\delta}_n\in[0,1)$ and $\widetilde{\varepsilon}_n>0$ such that for every $t\in\mathcal{T}$, 
    \begin{gather*}
        \mathrm{Pr}\Bigl( \sup_{z_{j,1},z_{j,2}:j\in[n+1]}\sup_{x,y}|s(x,y;Z_{j,1})-s(x,y;Z_{j,2})|\leq \widetilde{\varepsilon}_n\mid \phi(t) \Bigr)\geq 1-\widetilde{\delta}_n\,.
    \end{gather*}
\end{assumption}
When $s(x,y;Z)=v(x,y)$, 
$Z$ no longer affects the score function, and thus $\widetilde{\varepsilon}_n=\widetilde{\delta}_n=0$.
When $s(x,y;Z)$ depends on $Z$ and has the uniform stability property, $\widetilde{\varepsilon}_n=O(n^{-1})$ and $\widetilde{\delta}_n=0$ \cite{feldman2018generalization}. In contrast, 
for the kernel-estimation-based score constructed from $Z$ as in LCP, it can still have probabilistic uniform stability, where $\widetilde{\varepsilon}_n$ scales inversely with the effective local sample size, and $\widetilde{\delta}_n$ decays exponentially in the effective local sample size under mild conditions \cite{chen2019nearest}.

Denote by $F_{w\circ s^\star}(\cdot)$ the distribution of $s^\star(X_w,Y)$ given $w$, where $(X_w,Y)\sim P_{X,w}\times P_{Y\mid X}$. When $w(\cdot)\equiv 1$, this notation refers to the marginal distribution of $s^\star(X_{n+1},Y_{n+1})$. We impose the following mild regularity condition.

\begin{assumption}
\label{ass: oracle quantile regularity}
    There exist positive constants $M_{s^\star}$, $\underline{L}_{s^\star}$ and $L_t$ such that: (i) $|s^\star(X, Y)|\leq M_{s^\star}$; (ii) $F_{w\circ s^\star}$ has a density lower bounded by $\underline{L}_{s^\star}$ on its support; and (iii) for every $t\in\mathcal{T}$, $F_{s^\star\mid \phi(t)}$ is $L_t$-Lipschitz continuous.
\end{assumption}

The next theorem gives the main decomposition of conditional miscoverage.

\begin{theorem}\label{theo: main miscoverage}
    Suppose Assumptions~\ref{ass: data_distribution_and_score_invariance}--\ref{ass: oracle quantile regularity} hold. For each $t\in\mathcal{T}$, assume $s^\star(X_{n+1},Y_{n+1})$ is independent of $q^\star(Z;\alpha)$ conditional on $\phi(t)$. 
    Moreover, assume that $0\leq w(X_1)\leq M_w$ almost surely, and define $B_w = \mathbb{E}\{w(X_1)\mid w\}$ and $\sigma_w^2 = \mathbb{E}\{w^2(X_1)\mid w\}$, where $M_w$, $B_w$ and $\sigma_w^2$ are allowed to depend on $n$ and may vanish or diverge as $n\to\infty$.
    Then, for each $t\in\mathcal{T}$, there exists a constant $C_t > 0$ such that
    \begin{align}
        \notag&\left|\mathrm{Pr}\left( Y_{n+1}\in\widehat{C}(X_{n+1})\mid \phi(t) \right)-(1-\alpha)\right|\\
        \notag \leq& C_t \left\{ \varepsilon+n\delta_n(\varepsilon)+\widetilde{\varepsilon}_n+\widetilde{\delta}_n \right\}
        + C_t \mathbb{E}\left\{ \Gamma_n(w)\mid \phi(t) \right\} \\
        & +C_t \mathbb{E}\left\{\left|Q(1-\alpha;F_{s^\star\mid \phi(t)})-Q\left( 1-\alpha;F_{w\circ s^\star } \right)\right|\mid \phi(t) \right\},\label{eq: independent pointwise bound 1}
    \end{align}
    where $\Gamma_n(w)=(\sigma_w B_w^{-1}n^{-1/2}+M_w B_w^{-1}n^{-1})\sqrt{\log(n)}$. 
    Let $T_0$ be drawn independently from $P_{T}$. Then there exists a constant $C^\prime>0$ such that the last item in \eqref{eq: independent pointwise bound 1} can be further bounded by
    \begin{gather}
        C^\prime \mathbb{E}\left\{ r_w(T_0)\left|Q(1-\alpha;F_{s^\star\mid \phi(t)})-Q\left(1-\alpha;F_{s^\star\mid \phi(T_0)}\right)\right|\mid \phi(t) \right\}\,.\label{eq: independent pointwise bound 2}
    \end{gather}
\end{theorem}

The error bound in this theorem separates the overall conditional miscoverage into the following three interpretable terms:
\begin{itemize}
\item \textbf{Score-estimation error.} The term
$\varepsilon+n\delta_n(\varepsilon)+\widetilde{\varepsilon}_n+\widetilde{\delta}_n$
captures the discrepancy between the learned score $s(X_i,Y_i;Z)$ and its oracle target $s^\star(X_i,Y_i)$.

\item \textbf{Finite-sample calibration error.} The term $\mathbb{E}\{\Gamma_n(w)\mid \phi(t)\}$ reflects the error in estimating the quantile $Q(1-\alpha;F_{w\circ s^\star})$ induced by $w(\cdot)$ using the calibration scores $\{s^\star(X_i,Y_i)\}_{i\in[n]}$. The quantity
$\Gamma_n(w)=(\sigma_w B_w^{-1}n^{-1/2}+M_w B_w^{-1}n^{-1})\sqrt{\log(n)}$ characterizes the effective sample-size behavior induced by the weighting scheme $w$. In the special case $w(\cdot)\equiv1$, $\Gamma_n(w)$ reduces to the parametric rate $n^{-1/2}\sqrt{\log(n)}$.

\item \textbf{Intrinsic conditional-mismatch error.}
The last term 
$\mathbb{E}\{|Q(1-\alpha;F_{s^\star\mid \phi(t)})-Q(1-\alpha;F_{w\circ s^\star})|\mid \phi(t)\}$
quantifies the discrepancy between the oracle conditional quantile $Q(1-\alpha;F_{s^\star\mid \phi(t)})$ and the oracle threshold $Q(1-\alpha;F_{w\circ s^\star})$. This error persists even when the oracle score is known exactly and the calibration step is noise-free. When $w(\cdot)\equiv 1$, it reduces to the gap between the $(1-\alpha)$-quantile of the conditional distribution of $s^\star$ given $\phi(t)$ and the corresponding $(1-\alpha)$-quantile of its unconditional/marginal distribution. When $w(\cdot)$ is random, $Q\left( 1-\alpha;F_{w\circ s^\star } \right)$ is also random and this term is generally difficult to evaluate explicitly; therefore, we provide a computable upper bound in \eqref{eq: independent pointwise bound 2}.

\end{itemize}

Theorem~\ref{theo: main miscoverage} shows that any method aiming for asymptotic conditional validity should, through an appropriate choice of score function and weighting scheme, make the intrinsic conditional-mismatch error asymptotically negligible. In contrast, the remaining terms arise purely from estimation uncertainty: the score-estimation error term $\varepsilon+n\delta_n(\varepsilon)+\widetilde{\delta}_n+\widetilde{\varepsilon}_n$ vanishes as the relevant training and/or calibration sample sizes increase, while the finite-sample calibration error term $\Gamma_n(w)$ converges to zero as $n\to\infty$.

\begin{remark}\label{remark: LCPRLCP rate}
    In LCP, as the conditional c.d.f.~is estimated via a kernel estimator, the score-estimation error scales as $O\big((nh^d)^{-1/2}\sqrt{\log(n)}+h\log(1/h)\big)$. The finite-sample calibration error is dominated by the score-estimation error. Moreover, the intrinsic conditional-mismatch error is zero because both $F_{s^\star\mid \phi(t)}$ and $F_{w\circ s^\star}$ are the uniform distribution on $[0,1]$, and hence coincide exactly. 
    For RLCP, the score function is independent of $Z$ and thus the score-estimation error is zero. Furthermore, a detailed analysis shows that the finite-sample calibration error decays at rate $O((nh^d)^{-1/2}\sqrt{\log(n)})$, while the intrinsic conditional-mismatch error is of order $O(h)$; see Lemmas~\ref{lemma: kernel fun}--\ref{lemma: kernel fun bias} in the Supplementary Material. This yields the non-asymptotic bound of test-conditional miscoverage for RLCP, thereby extending the theoretical results in \cite{hore2025conformal}.
    Overall, up to logarithmic factors, both LCP and RLCP attain the rate $O\big((nh^d)^{-1/2}+h\big)$ for conditional miscoverage. 
\end{remark}

\subsection{Averaged Conditional Miscoverage under Weak Conditions}\label{sec: averaged conditional miscoverage}

Theorem~\ref{theo: main miscoverage} provides a pointwise result for each fixed index $t\in\mathcal{T}$, relying on the pointwise score-approximation condition in Assumption~\ref{ass: pointwise_convergence}. 
In some cases, such a pointwise condition may be too strong or difficult to verify. 
Instead, the relevant score-approximation error may only be controllable on average with respect to the distribution of $T$.
This motivates a weaker but still informative target,  \textit{averaged conditional miscoverage}, defined as
\begin{equation}
    \mathbb{E}\left|\mathrm{Pr}\left( Y_{n+1}\in\widehat{C}(X_{n+1})\mid \phi(T) \right)-(1-\alpha)\right|\,,\label{eq: general averaged conditional miscoverage}
\end{equation}
where the expectation is taken over the random index $T\sim P_T$. This criterion measures the average deviation from the target conditional coverage level across conditioning events. 

To isolate the dependence of $s(x,y;Z)$ on $(X_{n+1},Y_{n+1})$, we consider $s(x,y;Z_{n+1}^\prime)$, where $Z_{n+1}^\prime$ is obtained from $Z$ by replacing its $(n+1)$-th element $(X_{n+1},Y_{n+1})$ with an independent copy drawn from $P_X\times P_{Y\mid X}$. Conditional on $Z_{n+1}^\prime$, for each $t\in\mathcal{T}$, denote by $F_{s(\cdot,\cdot;Z_{n+1}^\prime)\mid \phi(t)}$ the conditional c.d.f.~of $s(X_{n+1},Y_{n+1};Z_{n+1}^\prime)$ given $\phi(t)$. Define the corresponding $(1-\alpha)$-quantiles by
\begin{equation}
    q_{s(\cdot,\cdot;Z_{n+1}^\prime)}(t)=Q(1-\alpha;F_{s(\cdot,\cdot;Z_{n+1}^\prime)\mid \phi(t)}),\qquad
q_{s^\star}(t)=Q(1-\alpha;F_{s^\star\mid \phi(t)})\,.\label{eq: quantiles for averaged}
\end{equation}
When $\phi(t)=\{X_{n+1}=t\}$, these reduce to the conditional quantiles given $X_{n+1}=t$. 

For any measurable function $f:\mathcal{T}\rightarrow\mathbb{R}$, define $\|f\|_{P_{T},p}=\{\int_{\mathcal{T}} |f(t)|^p dP_{T}(t)\}^{1/p}$. 
We need the following averaged quantile-approximation condition in place of the pointwise score-approximation condition in Assumption~\ref{ass: pointwise_convergence}.

\begin{assumption}\label{ass: L_p_convergence}
    There exist $p>1$ and a function $\delta_n(\varepsilon)$ such that $\mathrm{Pr}\bigl(\|q_{s(\cdot,\cdot;Z_{n+1}^\prime)}-q_{s^\star}\|_{P_{T},p}>\varepsilon\bigr) \leq \delta_n(\varepsilon)$.
\end{assumption}
Assumption~\ref{ass: L_p_convergence} quantifies the discrepancy between the learned and oracle score distributions through their conditional quantiles. Lemma~\ref{lem: pointwise implies lp} in the Supplementary Material shows that, under mild boundedness and regularity conditions, convergence in probability of $s(x,y;Z)$ to $s^\star(x,y)$ implies $\|q_{s(\cdot,\cdot;Z_{n+1}^\prime)}-q_{s^\star}\|_{P_T,p}\to 0$ in probability. 
Following the standard paradigm  for establishing averaged guarantees, 
we also assume that the score function $s(\cdot,\cdot;Z)$ belongs to a class $\mathcal{S}$ of measurable functions on $\mathcal{X}\times\mathcal{Y}$, and that $\mathcal{S}$ has finite complexity in the sense of pseudo-dimension \cite{vapnik2013nature}.

\begin{assumption}
\label{ass: space complexity}
The function class $\mathcal{S}$ is compact and has finite pseudo-dimension $\mathrm{Pdim}(\mathcal{S})$.
\end{assumption}

\begin{remark}\label{remark: cc instance}
    Consider the test-conditional setting $\phi(t)=\{X_{n+1}=t\}$ with $T=X_{n+1}$. 
    As an important example, consider the CC method (with marginal coverage guaranteed) proposed by \cite{gibbs2025conformal}, which can be written in our framework with a quantile-centered score, $s(x,y;Z)=v(x,y)-\widehat Q_\alpha(x;Z)$; see Section~\ref{sec: example details} in the Supplementary Material. 
 Here, $\widehat{Q}_\alpha(x;Z)\in\mathcal{F}$ estimates the conditional $(1-\alpha)$-quantile of $v(X,Y)$ given $X=x$, where $\mathcal{F}$ denotes a given function space. 
    Let $Q^\star_\alpha(x)$ denote the true conditional $(1-\alpha)$-quantile of $v(X,Y)$ given $X=x$, and define the oracle score
    $s^\star(x,y)=v(x,y)-Q^\star_\alpha(x)$, where the subscript follows the miscoverage level $\alpha$.
    Then the discrepancy measured in Assumption~\ref{ass: L_p_convergence} reduces to the estimation error of the learned conditional quantile:
    $$
    \|q_{s(\cdot,\cdot;Z_{n+1}^\prime)}-q_{s^\star}\|_{P_{T},p}
    =
    \|\widehat{Q}_\alpha(\cdot;Z_{n+1}^\prime)-Q^\star_\alpha(\cdot)\|_{P_{X},p}\,.
    $$   
    Viewing $\mathcal{S}=\{(x,y)\mapsto v(x,y)-f(x):f\in\mathcal{F}\}$ as a function class on $(X,Y)$,
    Lemma~\ref{lem:pdim_Sprime} in the Supplementary Material implies that
    $\mathrm{Pdim}(\mathcal{S})\leq \mathrm{Pdim}(\mathcal{F})$.
    Hence, Assumption~\ref{ass: space complexity} is satisfied whenever the estimator class for $\widehat{Q}_\alpha(\cdot;\cdot)$ has finite pseudo-dimension, as in finite-dimensional basis expansions or linear setting \cite{gibbs2025conformal}.
\end{remark}

The next theorem shows that Assumptions~\ref{ass: L_p_convergence} and~\ref{ass: space complexity} yield a non-asymptotic bound for averaged conditional miscoverage. Let $f_{s(\cdot,\cdot;Z_{n+1}^\prime)\mid \phi(t)}(\cdot)$ and $f_{s^\star\mid \phi(t)}(\cdot)$ denote the conditional densities corresponding to $F_{s(\cdot,\cdot;Z_{n+1}^\prime)\mid \phi(t)}$ and $F_{s^\star\mid \phi(t)}$, respectively.

\begin{theorem}\label{theo: averaged miscoverage}
    Suppose Assumptions~\ref{ass: data_distribution_and_score_invariance} and~\ref{ass: stable score}--\ref{ass: space complexity} hold. Assume further that the conditional densities $f_{s(\cdot,\cdot;Z_{n+1}^\prime)\mid \phi(t)}(\cdot)$ and $f_{s^\star\mid \phi(t)}(\cdot)$ are bounded below by $\underline L$ and above by $\overline L$ on their supports, for positive constants $\underline L$ and $\overline L$. Suppose also that $|s(X,Y;Z)|\leq \underline{M}$ and $\underline{M}\leq w(x)\leq \overline{M}$ uniformly over $x\in\mathcal{X}$ almost surely, for positive constants $\underline{M}$, and $\overline{M}$. Let $T_0,T_0^\prime \sim P_T$ be independent copies of $T$. There exists a constant $C>0$ such that
    \begin{align}
        \notag &\mathbb{E}\left|\mathrm{Pr}\left( Y_{n+1}\in\widehat{C}(X_{n+1})\mid \phi(T) \right)-(1-\alpha)\right|\\
        \notag \leq & C \left\{ \widetilde{\delta}_n+\widetilde{\varepsilon}_n+\delta_n(\varepsilon)+\varepsilon\mathbb{E}\left( \|r_w\|_{P_T,p/(p-1)} \right)\right\} + C\sqrt{n^{-1}\mathrm{Pdim}(\mathcal{S})\log(n)} \\
        & + C \mathbb{E}\left[ \left\{ r_w(T_0)+1\right\}\left|q_{s^\star}(T_0)-q_{s^\star}(T_0^\prime)\right| \right]\,.
        \label{eq: averaged miscoverage bound}
    \end{align}
\end{theorem}

The bound in~\eqref{eq: averaged miscoverage bound} can be viewed as an averaged counterpart of the pointwise bounds in~\eqref{eq: independent pointwise bound 1}--\eqref{eq: independent pointwise bound 2}, with the complexity of the score class entering the \textit{finite-sample calibration} error term $\sqrt{n^{-1}\mathrm{Pdim}(\mathcal{S})\log(n)}$ through $\mathrm{Pdim}(\mathcal S)$. We refer to $\mathbb{E}\left\{ r_w(T_0)+1\right\}\left|q_{s^\star}(T_0)-q_{s^\star}(T_0^\prime)\right|$ as the \emph{averaged intrinsic conditional-mismatch} error, and to $\widetilde{\delta}_n+\widetilde{\varepsilon}_n+\delta_n(\varepsilon)+\varepsilon\mathbb{E}\|r_w\|_{P_T,p/(p-1)}$ as the \emph{score-estimation} error.  The term
$\|r_w\|_{P_T,p/(p-1)}$ arises because $\varepsilon$ measures the convergence rate of $q_{s(\cdot,\cdot;Z_{n+1}^\prime)}$ under the $\|\cdot\|_{P_T,p}$ norm, whereas calibration via the weight $w(\cdot)$ requires measuring the convergence rate under the $\|\cdot\|_{P_{T,w},p}$ norm. This naturally induces an additional term capturing this discrepancy, quantified by the density ratio $r_w$. In particular, when $w(\cdot)\equiv 1$, $\mathbb{E}(\|r_w\|_{P_T,p/(p-1)})$ reduces to $1$.

To see the practical relevance of this unified result,
we next refine Theorem~\ref{theo: averaged miscoverage} for CC with an RKHS-based conditional quantile estimator. Consider the class
$\mathcal F=\{f_\kappa:f_\kappa(x)=\sum_{i=1}^{d_0}\kappa_i\eta_i(x),\ \|\kappa\|_2\leq B\}$, where $\eta(x)=(\eta_1(x),\ldots,\eta_{d_0}(x))^{\top}$ are basis functions, $\kappa=(\kappa_1,\ldots,\kappa_{d_0})^{\top}$, $d_0$ is a fixed integer, and $B$ is a positive constant. Define the pinball loss $\ell_{1-\alpha}(s_1,s_2)=\{1-\alpha-\mathbbm{1}(s_2\geq s_1)\}(s_1-s_2)$. The conditional quantile estimator is $\widehat Q_\alpha(x;Z)=f_{\widehat\kappa}(x)$, where
$\widehat{\kappa}=\arg\min_{\{\kappa:~\|\kappa\|_2\leq B\}}n^{-1}\sum_{i=1}^{n}\ell_{1-\alpha}(v(X_i,Y_i),f_\kappa(X_i))+\lambda\|\kappa\|_2^2$.
Define the risk by $R(f)=\mathbb E\{\ell_{1-\alpha}(v(X,Y),f(X))\}$.  Recall that $Q_\alpha^\star(\cdot)$ denotes the true conditional $(1-\alpha)$-quantile of $v(X,Y)$ given $X=x$ as in Remark~\ref{remark: cc instance}.
Let $\widehat{C}_{\rm CC}(X_{n+1})$ denote the resulting prediction set procduced by CC, as detailed in Section~\ref{sec: example details} of the Supplementary Material.

\begin{corollary}[Averaged conditional miscoverage of CC]\label{cor: cc L1 test}
    Suppose that, for constant $M>0$, $|v(X,Y)|\leq M$ and $\|\eta(X)\|_{2}\leq M$ almost surely for $(X,Y)\sim P_X\times P_{Y\mid X}$. Assume the conditional density of $v(X,Y)$ given $X=x$ is bounded below by $\underline L$ and above by $\overline L$ on its support, for positive constants $\underline L$ and $\overline L$.
    Then there exists a constant $C>0$ such that
    \begin{align*}
        &\mathbb{E}\left|\mathrm{Pr}\left( Y_{n+1}\in\widehat{C}_{\rm CC}(X_{n+1})\mid X_{n+1} \right)-(1-\alpha)\right| \\
        \leq & C \left[ \sqrt{\inf_{f\in\mathcal{F}}R(f)-R(Q_\alpha^\star)}+\{d_0\log(n)/n\}^{1/3} \right]\,.
    \end{align*}  
\end{corollary}

The term $\sqrt{\inf_{f\in\mathcal{F}}R(f)-R(Q_\alpha^\star)}$ represents the approximation error induced by restricting the conditional quantile estimator to the class $\mathcal F$, while $\{d_0\log(n)/n\}^{1/3}$ is the score-estimation error and vanishes as $n\to\infty$. Consequently, if the approximation error $\inf_{f\in\mathcal{F}}R(f)-R(Q_\alpha^\star)$ is negligible, the averaged conditional miscoverage converges to zero for CC. This provides a useful complement to the theory of CC \cite{gibbs2025conformal}.

Table~\ref{tab: conditional coverage methods summary} provides a brief summary of several major conformal methods targeting conditional coverage, together with the corresponding error decomposition for conditional miscoverage. Detailed descriptions and theoretical analyses of these methods are deferred to Section~\ref{sec: example details} of the Supplementary Material.

\begin{table*}[t]
    \caption{\small Summary of conditional miscoverage errors for several conformal methods. All rates ignore logarithmic factors. The order of $\delta_{\rm tr}$ depends on the estimation error of the pre-trained score and $\Delta_{\mathcal F}=\inf_{f\in\mathcal F} R(f)-R(Q_\alpha^\star)$. BatchGCP targets group-conditional coverage, whereas the other methods target test-conditional coverage.}
    \label{tab: conditional coverage methods summary}
    \begin{center}
    {\fontsize{7}{8}\selectfont{\setlength{\tabcolsep}{3pt}\begin{tabular}{ccccc}
    \toprule
    Method & Score-estimation & Finite-sample calibration & Conditional-mismatch & Total \\
    \midrule
    LCP \cite{guan2023localized} & $O\big((nh^d)^{-1/2}+h\big)$ & $O\bigl(n^{-1/2}\bigr)$ & $0$ & $O\big((nh^d)^{-1/2}+h\big)$ \\
    RLCP \cite{hore2025conformal} & $0$ & $O\big((nh^d)^{-1/2}\big)$ & $O(h)$ & $O\big((nh^d)^{-1/2}+h\big)$ \\
    CQR \cite{romano2019conformalized} & $O\bigl(\delta_{\rm tr}\bigr)$ & $O\bigl(n^{-1/2}\bigr)$ & $0$ & $O\bigl(\delta_{\rm tr}+n^{-1/2}\bigr)$ \\
    DCP \cite{Chernozhukov2021distribution} & $O\bigl(\delta_{\rm tr}\bigr)$ & $O\bigl(n^{-1/2}\bigr)$ & $0$ & $O\bigl(\delta_{\rm tr}+n^{-1/2}\bigr)$ \\
    GLCP \cite{min2025personalized} & $O\bigl(\delta_{\rm tr}\bigr)$ & $O\bigl(n^{-1/2}\bigr)$ & $0$ & $O\bigl(\delta_{\rm tr}+n^{-1/2}\bigr)$ \\
    BatchGCP \cite{jung2023batch} & $O\bigl(n^{-1/2}\bigr)$ & $O\bigl(n^{-1/2}\bigr)$ & $0$ & $O\bigl(n^{-1/2}\bigr)$ \\
    CC \cite{gibbs2025conformal} & $O\bigl((n/d_0)^{-1/3}\bigr)$ & $O\bigl(n^{-1/2}\bigr)$ & $O\bigl(\Delta_{\mathcal F}^{1/2}\bigr)$ & \makecell[c]{$O\bigl((n/d_0)^{-1/3}+\Delta_{\mathcal F}^{1/2}\bigr)$}\\
    \bottomrule
    \end{tabular}}}
    \end{center}
\end{table*}

\section{Applications and Extensions}\label{sec:application}

    The previous section developed the main conditional miscoverage theory, including both pointwise and averaged bounds, and showed how a broad class of conformal methods can be analyzed within a unified
    framework. We now demonstrate the usefulness of this unified theory by presenting several applications and extensions. 
    In Section~\ref{sec: model selection}, we first develop a conditional-coverage-oriented model selection procedure motivated by the conditional-miscoverage decomposition. 
    In Section~\ref{sec: Extensions to covariate shift}, we consider covariate shift, where our weighted formulation naturally combines localization with density-ratio correction. Finally, in Section~\ref{sec: extension structured data}, we extend the theory beyond i.i.d.~data through a weighted SymmPI formulation \cite{dobriban2025symmpi}, and apply it to community-conditional coverage on graph-structured data. 
    Some additional extensions, including a generalized RLCP framework, 
    an extension of localized conformal $p$-values \cite{wu2025conditional}, and a two-layer hierarchical example, are provided in Section~\ref{sec:supp_add_applications} of the Supplementary Material.
    
\subsection{Conditional-Coverage-Oriented Model Selection}\label{sec: model selection}

The error decomposition in Theorem~\ref{theo: main miscoverage} provides a principled basis for model selection when conditional coverage is the primary objective.
In practice, conditional conformal procedures require choosing among tuning parameters, score constructions, weighting schemes, and estimators of conditional distributions or conditional quantiles.
Although each fixed construction may achieve asymptotic conditional validity under suitable conditions, their finite-sample conditional performance can differ substantially.
We therefore seek to select, from a collection of candidate conformal constructions, the one that is most aligned with test-conditional validity.

Theorem~\ref{theo: main miscoverage} identifies the intrinsic conditional-mismatch error,
$\mathbb{E}\{|Q(1-\alpha;F_{s^\star\mid \phi(t)})-Q(1-\alpha;F_{w\circ s^\star})|\mid \phi(t)\}$,
as a key quantity that must be controlled for asymptotic conditional validity.
Analogously to the definition of $F_{w\circ s^\star}$, we denote by $F_{w\circ s(\cdot,\cdot;Z)}(\cdot)$ the distribution of $s(X_w,Y;Z)$ given $Z$ and $w$, where $(X_w,Y)\sim P_{X,w}\times P_{Y\mid X}$.  
This motivates a data-driven criterion that reflects this mismatch. Concretely, focusing on the test-conditional case, we measure the conditional-coverage performance of the set built with score $s(x,y;Z)$ by the gap between the conditional quantile of $s(X_{n+1},Y_{n+1};Z)$ given $X_{n+1}=t$ and the unconditional/marginal quantile $Q(1-\alpha;F_{w\circ s(\cdot,\cdot;Z)})$.
For $(X,Y)\sim P_X\times P_{Y\mid X}$, this motivates measuring $\mathbb{E}[\mathbbm{1}\{s(X,Y;Z)\leq Q(1-\alpha;F_{w\circ s(\cdot,\cdot;Z)})\}-(1-\alpha)\mid Z,w,X=t]$.

To empirically implement this idea, let $\mathcal{S}$ be a collection of candidate score functions. For each $s\in\mathcal{S}$, let $q_{s}(Z;\alpha)$ represent conformal quantile $q(Z;\alpha)$ using score $s(\cdot,\cdot;Z)$ and the corresponding weighting scheme $w(\cdot)$. Define $\hat{\zeta}_s(x,y;Z)=\mathbbm{1}\left(s(x,y;Z)\leq q_{s}(Z;\alpha)\right)-(1-\alpha)$. 
At the population/oracle level, test-conditional validity can be diagnosed through the conditional moment restriction $\mathbb E\{\hat\zeta_s(X,Y;Z)\mid Z,w,X=x\}=0$ for all $x$.  Thus, comparing candidate conformal constructions can be viewed as comparing the degree of violation of a conditional mean restriction. This perspective is closely related to classical nonparametric lack-of-fit tests for conditional mean restrictions, such as those of \cite{zheng1996consistent} and \cite{tripathi2003testing}. The procedure usually proceeds by first estimating the conditional expectation and then uses a lack-of-fit functional $\mathcal L(s;Z)$ to measure its discrepancy from zero. Different choices of $\mathcal L(s;Z)$ can be employed, such as squared-type or empirical likelihood-type discrepancies.


For a trial response $y$, we select the score function minimizing this criterion:
\[
\hat{s}^y=\underset{s\in\mathcal{S}}{\arg\min}\ \mathcal{L}\left(s;Z^y\right),
\]
and the conformal set with score selection is then
\begin{gather}
    \widehat{C}_{\rm Sel}(X_{n+1})=\left\{ y:\hat{s}^y(X_{n+1},y;Z^y)\leq q_{\hat{s}^y}(Z^y;\alpha) \right\}\,.\label{eq: conditional selection}
\end{gather}
If $\mathcal L(s;Z^y)$ is invariant under permutations of $Z^y$, then under the permutation-invariance of the candidate scores, the selected score $\hat s^y$ is also permutation invariant. Hence, the resulting conformal set $\widehat{C}_{\rm Sel}(X_{n+1})$ preserves marginal validity \cite{liang2024conformal}. At the same time, the selection criterion explicitly favors scores that are better aligned with test-conditional validity.
The main computational burden is that \eqref{eq: conditional selection} requires recomputing $\mathcal L(s;Z^y)$ and solving the minimization problem to obtain $\hat{s}^y$ for each trial value $y$. Section~\ref{sec: efficient sel} of the Supplementary Material provides an efficient approximation algorithm.


\subsection{Extensions of Conditional Methods under Covariate Shift}\label{sec: Extensions to covariate shift}

Most existing methods for conditional coverage are developed in the i.i.d.~setting. Under covariate shift, the calibration and test covariate distributions differ while $P_{Y\mid X}$ remains unchanged. 
Therefore, directly applying i.i.d.-based conditional conformal methods may fail to preserve marginal coverage. WCP \cite{tibshirani2019conformal} restores marginal validity through density-ratio weighting, but it does not guarantee small conditional miscoverage.

Our solution is to retain the general form of the prediction set in \eqref{eq: general Cvec}, while redesigning the weighting function $w(\cdot)$ to account for covariate shift. The design of $w(\cdot)$ is guided by two theoretical requirements: the marginal coverage guarantee under covariate shift and the conditional-miscoverage decomposition for $\widehat{C}(X_{n+1})$.

\begin{assumption}\label{ass: covariate shift assumption}
    (i) The calibration pairs $\{(X_i,Y_i)\}_{i\in[n]}$ are i.i.d.~from $P_{X,1}\times P_{Y\mid X}$ and the test pair $(X_{n+1},Y_{n+1})$ is from $P_{X,2}\times P_{Y\mid X}$, with the density ratio $r_X(x)=dP_{X,2}/dP_{X,1}(x)$ known; (ii) For $z=\{(x_1,y_1),\ldots,(x_n,y_n),(x_{n+1},y_{n+1})\}$, $s(x,y;z)$ is permutation invariant in $\{(x_i,y_i)\}_{i\in[n+1]}$.
\end{assumption}

Under Assumption~\ref{ass: covariate shift assumption}, the following covariate-shift analogue of Theorem~\ref{theo: marginal general} holds by incorporating $r_X(x)$ into the definition of $\mathcal{W}_X$.

\begin{theorem}\label{theo: marginal general cshift}
    Suppose Assumption~\ref{ass: covariate shift assumption} holds. For $X\sim P_{X,1}$, define $\mathcal{W}_X=\{w_0(\cdot)\geq 0:\mathbb{E}\{w_0(X)\mid w_0\}=1\text{ and }\mathbb{E}\{w_0(x)\}=r_X(x),\text{ for all } x\in\mathcal{X}\}$. Using the same definitions of $d_Z(w,w_0)$ and $\widehat{w}_{\max}(Z)$ as in Theorem~\ref{theo: marginal general}, we have
    \begin{align*} 
        1-\alpha-\inf_{w_0\in\mathcal{W}_X}\mathbb{E}\left\{ d_Z(w,w_0) \right\}&\leq \mathrm{Pr}\left( Y_{n+1}\in \widehat{C}(X_{n+1}) \right)\\
        &\leq 1-\alpha+\inf_{w_0\in\mathcal{W}_X}\mathbb{E}\left\{ d_Z(w,w_0) \right\}+\mathbb{E}\left\{ \widehat{w}_{\max}(Z) \right\}\,.
    \end{align*}
\end{theorem}
Theorem~\ref{theo: marginal general cshift} shows that marginal coverage is guaranteed under covariate shift when $\mathbb{E}\{\bar{w}(x)\}=r_X(x)$ for all $x\in\mathcal{X}$, where $\bar{w}(x)=w(x)/\mathbb{E}\{w(X)\mid w\}$ for $X\sim P_{X,1}$. This condition restores marginal validity, but it does not by itself ensure small conditional miscoverage. Since Theorem~\ref{theo: main miscoverage} gives bounds for coverage conditional on events determined by $X_{n+1}$, the form of the conditional-miscoverage decomposition is unchanged after replacing Assumption~\ref{ass: data_distribution_and_score_invariance} with Assumption~\ref{ass: covariate shift assumption}. A covariate-shift analogue of Theorem~\ref{theo: averaged miscoverage} for averaged conditional miscoverage is given in Theorem~\ref{theo: averaged miscoverage shift} in Section~\ref{sec:supp_covariate_shift} of the Supplementary Material, together with the corresponding specialization to CC. Therefore, in addition to satisfying the marginal condition above, the weight should also make the weighted quantile $Q(1-\alpha;F_{w\circ s^\star})$ close to the oracle conditional quantile $Q(1-\alpha;F_{s^\star\mid \phi(t)})$.

For test-conditional coverage, consider scores whose oracle counterparts remove conditional heterogeneity in the score distribution, as in CQR, LCP, DCP, and GLCP. For such scores, the oracle conditional quantile $Q(1-\alpha;F_{s^\star\mid X_{n+1}=t})$ is constant across $t$. In this case, setting $w(\cdot)=r_X(\cdot)$ suffices to remove the intrinsic conditional-mismatch error while restoring marginal validity. For scores that do not have this property, such as a pre-trained score $v(x,y)$ with $s^\star(x,y)=v(x,y)$, we combine density-ratio correction with localization. Specifically, let $\widetilde{X}$ be drawn from a distribution with density proportional to $K(\cdot,X_{n+1};h)$ conditional on $X_{n+1}$. Multiplying the RLCP local weight $K(\cdot,\widetilde{X};h)$ by $r_X(\cdot)$ gives
$w(x)=r_X(x)K(x,\widetilde{X};h)$,
which accounts for covariate shift through $r_X$ and localization through the kernel factor.

\subsection{Extension to Structured Data}\label{sec: extension structured data}

We extend our unified framework from i.i.d.~data to structurally symmetric data, following the SymmPI perspective \cite{dobriban2025symmpi}. The goal is to show that the same conditional-miscoverage decomposition continues to apply when calibration is performed over a symmetry group rather than over an i.i.d.~sample.

Let $Z\in\mathcal Z$ denote the full structured data object and let $\Omega(Z)$ denote its observed component. For example, if $Z=\{(X_1,Y_1),\ldots,(X_n,Y_n),(X_{n+1},Y_{n+1})\}$, then $\Omega(Z)=\{(X_1,Y_1),\ldots,(X_n,Y_n),X_{n+1}\}$. Suppose the structural symmetry is encoded by a finite group $\mathcal G$ acting on $\mathcal Z$ through $\rho:\mathcal G\times\mathcal Z\to\mathcal Z$. Let $V:\mathcal Z\to\widetilde{\mathcal Z}$ be a score map, and let $\widetilde\rho:\mathcal G\times\widetilde{\mathcal Z}\to\widetilde{\mathcal Z}$ be the induced action on the score space. In the i.i.d. setting, $\mathcal G$ is the permutation group on the sample indices and $\rho$ permutes the observations.
\begin{assumption}\label{ass: distributional invariance}
    The data $Z$ is $(\mathcal{G},\rho)$-distributionally invariant and the score map $V$ is $(\mathcal{G},\rho,\widetilde{\rho})$-distributionally equivariant in the sense that for every $g\in\mathcal{G}$,
    \begin{gather*}
        Z\overset{\rm d}{=}\rho(g,Z),\qquad V(\rho(g,Z))\overset{\rm d}{=}\widetilde{\rho}(g,V(Z))\,.
    \end{gather*}
\end{assumption}
As an example, when $Z$ consists of $n+1$ exchangeable pairs and $\mathcal G$ is the full permutation group, this reduces to the standard exchangeable setting. A more general version of this framework, which incorporates ratio functions to handle distribution shift, is developed in Section~\ref{sec: appendix weighted symmpi} of the Supplementary Material.

Let $\psi:\widetilde{\mathcal Z}\to\mathbb R$ extract the score associated with the unobserved component, so that $\psi(V(Z))$ is the test score. For instance, if $V(Z)=\{S(X_1,Y_1;Z),\ldots,S(X_{n+1},Y_{n+1};Z)\}$ with $Y_{n+1}$ unobserved, one may take $\psi(V(Z))=S(X_{n+1},Y_{n+1};Z)$. Define $\mathcal{G}_e=\{g\in\mathcal{G}:\psi(\widetilde{\rho}(g,V(z)))=\psi(V(z))\text{ for all } z\in\mathcal{Z}\}$ as the set of group elements that keep the test score in the test position. This set contains the identity and all other group elements that leave the extracted test score unchanged. Calibration is performed over $\{\psi(\widetilde{\rho}(g,V(Z))):g\in\mathcal{G}\setminus\mathcal{G}_e\}$, yielding
$$
q(V(Z);\alpha)=Q\bigl(1-\alpha;|\mathcal G|^{-1}\bigl\{{\textstyle\sum_{g\in\mathcal{G}\setminus\mathcal{G}_e}}\delta_{\psi(\widetilde{\rho}(g,V(Z)))}+|\mathcal{G}_e|\,\delta_{\infty}\bigr\}\bigr)\,,
$$
where $|\mathcal G|$ is the cardinality of $\mathcal{G}$.
Let $z_{\rm obs}=\Omega(Z)$; the prediction set takes the generic form
\[
\widehat{C}(z_{\rm obs})=\left\{z:\psi(V(z))\leq q(V(z);\alpha),\ \Omega(z)=z_{\rm obs}\right\}.
\]

\begin{theorem}\label{theo: marginal SymmPI}
    Suppose Assumption~\ref{ass: distributional invariance} holds. Then the marginal coverage of $\widehat{C}(z_{\rm obs})$ satisfies 
    \begin{gather*}
        1-\alpha\leq \mathrm{Pr}(Z\in \widehat{C}(\Omega(Z)))\leq 1-\alpha+\mathbb{E}\left\{\widehat{\gamma}_{\max}(Z)\right\}\,,
    \end{gather*}
    where $\widehat{\gamma}_{\max}(z)=\sup_c\left|\{g\in\mathcal{G}:\psi(\widetilde{\rho}(g,V(z)))=c\}\right|/|\mathcal{G}|$.
\end{theorem}

To analyze conditional coverage, let $V^\star:\mathcal Z\to\widetilde{\mathcal Z}$ be an oracle counterpart of $V$, analogous to the role of $s^\star$ in the i.i.d.~data framework. 
The oracle score map $V^\star$ serves as a population reference that removes the dependence of the learned map $V$ on the finite sample $Z$.
For a conditioning event $\phi(t)$, impose the following approximation condition.
\begin{assumption}\label{ass: independence score}
    For each $t\in\mathcal T$, we assume: (i) $\psi(V^\star(Z))$ is independent of $q(V^\star(Z);\alpha)$ conditional on $\phi(t)$; (ii) for every $\varepsilon>0$, defining
    \[
    D_\varepsilon=\left\{\left|\psi(V(Z))-\psi(V^\star(Z))\right|\vee\left|q(V(Z);\alpha)-q(V^\star(Z);\alpha)\right|\leq\varepsilon\right\},
    \]
    there exists a non-random function $\delta(\varepsilon)\in[0,1]$ such that $\mathrm{Pr}\left(D_\varepsilon\mid\phi(t)\right)\ge 1-\delta(\varepsilon)$.
\end{assumption}

For each $t\in\mathcal T$, let $F_{\psi(V^\star(\cdot))\mid\phi(t)}$ denote the conditional distribution of $\psi(V^\star(Z))$ given $\phi(t)$. Its $(1-\alpha)$-quantile $Q(1-\alpha;F_{\psi(V^\star(\cdot))\mid\phi(t)})$ is the theoretical threshold for conditional coverage under the oracle score map. Parallel to the i.i.d.~data case, let $F_{\psi(V^\star(\cdot))}$ denote the unconditional/marginal distribution of $\psi(V^\star(Z))$; its $(1-\alpha)$-quantile $Q(1-\alpha;F_{\psi(V^\star(\cdot))})$ is the population-level quantile that $q(V^\star(Z);\alpha)$ estimates. 
\begin{theorem}\label{theo:symmpi conditional miscoverage}
    Suppose Assumption~\ref{ass: independence score} holds. If for every $t\in\mathcal{T}$, $F_{\psi(V^\star(\cdot))\mid\phi(t)}$ is $L_t$-Lipschitz continuous for some positive constant $L_t$, then
    \begin{align}
        &\left|\mathrm{Pr}\left( Z\in\widehat C(\Omega(Z))\mid\phi(t) \right)-(1-\alpha)\right|\notag\\
        &\leq \delta(\varepsilon)+2L_t\varepsilon
        +L_t\,\mathbb{E}\bigl\{\left|q(V^\star(Z);\alpha)-Q(1-\alpha;F_{\psi(V^\star(\cdot))})\right|\mid\phi(t)\bigr\}\notag\\
        &\qquad\qquad\qquad +L_t\left|Q(1-\alpha;F_{\psi(V^\star(\cdot))})-Q(1-\alpha;F_{\psi(V^\star(\cdot))\mid\phi(t)})\right|. \label{eq: general symmpi error}
    \end{align}
\end{theorem}

Theorem~\ref{theo:symmpi conditional miscoverage} also decomposes the conditional miscoverage into three terms: the \textit{score-estimation error} $\delta(\varepsilon)+2L_t\varepsilon$, the \textit{finite-sample calibration error} through the gap between $q(V^\star(Z);\alpha)$ and its population counterpart $Q(1-\alpha;F_{\psi(V^\star(\cdot))})$, and the \textit{intrinsic conditional-mismatch error}, given by the gap between the population quantile $Q(1-\alpha;F_{\psi(V^\star(\cdot))})$ and the target conditional quantile $Q(1-\alpha;F_{\psi(V^\star(\cdot))\mid\phi(t)})$. This three-term structure mirrors the decomposition in Theorem~\ref{theo: main miscoverage} for i.i.d.~data. 

The graph setting below provides a concrete illustration, where the conditioning event is community membership. A second example based on a two-layer hierarchical model is given in Section~\ref{sec: two layer hierarchical} of the Supplementary Material. 
In both cases, the unified theory plays the same role: after controlling score estimation and finite-sample calibration, asymptotic conditional validity reduces to quantifying the structure-specific intrinsic conditional-mismatch error.

\subsubsection{Community-Conditional Conformal Prediction on Graphs}\label{sec: community conditional}

We study community-conditional coverage on graphs, motivated by real-world network structures with dense within-community and sparse between-community connections. In document classification on citation networks \cite{sen2008collective}, for example, one may want prediction sets that are calibrated not only marginally, but also within thematic communities \cite{clarkson2023distribution}. 
We adopt and extend the graph-data notation of \cite{lunde2025conformal}. Let
$Z=\{(W_1,X_1,Y_1),\ldots,(W_{n+1},X_{n+1},Y_{n+1}),A\}\in\mathcal{Z}$
be drawn from $P$, where $A=(A_{i,j})_{i,j\in[n+1]}\in\mathbb R^{(n+1)\times(n+1)}$ denotes the adjacency or connection matrix and $(W_i,X_i,Y_i)\in\mathcal W\times\mathcal X\times\mathcal Y$. Here, $W_i\in\mathcal W$ represents unobserved community-related information, such as community labels, community-level parameters, or latent positions in latent space models \cite{gao2022community}. Conditional on the latent variable $W_i$, the distributions of $X_i, Y_i$ and the edges involving node $i$ may vary across communities. The observation function is 
$\Omega(Z)=((X_1,Y_1),\ldots,(X_n,Y_n),X_{n+1},A)$. 
Let $\mathcal{G}$ be the permutation group on $n+1$ elements. For $g\in \mathcal{G}$, define the action $\rho$ on $Z$ by
\begin{align*}
     \rho(g, Z) =& \rho(g,\{(W_1,X_1,Y_1),\ldots,(W_{n+1},X_{n+1},Y_{n+1}),A\})\\
    =&\{(W_{g^{-1}(1)},X_{g^{-1}(1)},Y_{g^{-1}(1)}),\ldots,(W_{g^{-1}(n+1)},X_{g^{-1}(n+1)},Y_{g^{-1}(n+1)}),A^g\},
\end{align*}
where $A^g=(A_{i,j}^g)_{i,j\in[n+1]}$ is the permuted connection matrix with $A^g_{i,j}=A_{g^{-1}(i),g^{-1}(j)}$ for $i,j\in [n+1]$.
Thus, the joint exchangeability condition in \cite{lunde2025conformal} is equivalent to $(G,\rho)$-distributional invariance. Assume that there are $n_{\rm comm}$ non-overlapping communities, labeled by $[n_{\rm comm}]$. For a node-level latent variable $W$, let $\sigma(W)\in[n_{\rm comm}]$ denote its community label.
Assume $p_{\rm comm}=\inf_{\mathcal{I}\in[n_{\rm comm}]}\mathrm{Pr}(\sigma(W)=\mathcal{I})>0$.

To place this construction within the SymmPI framework, define the score map $V(Z)=(s(X_1,Y_1;Z),\ldots,s(X_{n+1},Y_{n+1};Z))$ and let $\widetilde\rho$ act on $V(Z)$ by permuting the score entries according to $g$, i.e., $\widetilde\rho(g,V(Z))=(s(X_{g^{-1}(1)},Y_{g^{-1}(1)};Z),\ldots,s(X_{g^{-1}(n+1)},Y_{g^{-1}(n+1)};Z))$. 
Here, $s(x,y;Z)$ is the conformal score used to construct the prediction set.
The test score extractor is $\psi(V(Z))=s(X_{n+1},Y_{n+1};Z)$. Under the joint exchangeability condition, $Z$ is $(\mathcal{G},\rho)$-distributionally invariant and $V$ is $(\mathcal{G},\rho,\widetilde\rho)$-distributionally equivariant.

Let $Z^y$ be obtained from $Z$ by replacing $Y_{n+1}$ with trial value $y$.
Assume the score is permutation invariant, in the sense that $s(x,y;z)=s(x,y;\rho(g,z))$ for any $g\in\mathcal{G}$. The SymmPI framework then yields a graph conformal prediction (GraphCP) set of the form 
\begin{gather*}
    \widehat{C}_{\rm GraphCP}(X_{n+1})=\Bigl\{ y:s(X_{n+1},y;Z^y)\leq Q\Bigl( 1-\alpha;(n+1)^{-1}\Bigl(\sum_{i=1}^{n}\delta_{s(X_i,Y_i;Z^y)}+\delta_{\infty}\Bigr) \Bigr) \Bigr\}.
\end{gather*}
We define the community-conditional coverage target as
\begin{gather*}
    \mathrm{Pr}\left( Y_{n+1}\in\widehat{C}_{\rm GraphCP}(X_{n+1})\mid\sigma(W_{n+1})=\mathcal{I} \right), \quad \mathcal{I}\in[n_{\rm comm}]\,.
\end{gather*}
Taking $\mathcal{T}=[n_{\rm comm}]$ and
$\phi(\mathcal{I})=\{\sigma(W_{n+1})=\mathcal{I}\}$, this target falls under Theorem~\ref{theo:symmpi conditional miscoverage}.
The remaining task is to construct score $s(\cdot,\cdot;Z)$ whose oracle community-conditional $(1-\alpha)$-quantile is the same as the corresponding unconditional/marginal $(1-\alpha)$-quantile.

Suppose a community detection algorithm applied to $A$ yields predicted labels 
$\widehat{\sigma}(A)=(\widehat{\sigma}(A)_1,\ldots,\widehat{\sigma}(A)_{n+1})\in\mathbb{N}^{n+1}$.
We assume that the community detection rule is permutation equivariant up to a relabeling of labels, which ensures that the score defined below is invariant to node indexing.
Let $v(x,y;Z):\mathcal X\times\mathcal Y\times\mathcal Z\to\mathbb R$ be a base score, either pre-trained or trained from $Z$. Given a full data object $Z$, define the within-community rank score of node $i$ by
\begin{gather*}
    s(X_i,Y_i;Z)=\dfrac{\sum_{j=1}^{n+1}\mathbbm{1}(\widehat{\sigma}(A)_j=\widehat{\sigma}(A)_i\text{ and }v(X_j,Y_j;Z)\leq v(X_i,Y_i;Z))}{\sum_{j=1}^{n+1}\mathbbm{1}(\widehat{\sigma}(A)_j=\widehat{\sigma}(A)_i)}\,.
\end{gather*}

The purpose of this within-community ranking is to remove community-level heterogeneity in the distribution of the base score. When the community detection algorithm consistently
recovers the true partition and $v(x,y;Z)$ converges to an oracle score $v^\star(x,y)$, the oracle conditional quantile no longer depends on $\mathcal{I}$.
Then the intrinsic conditional-mismatch error vanishes asymptotically in Theorem~\ref{theo:symmpi conditional miscoverage}.
The remaining error terms are driven by community
recovery, score estimation, and finite community sizes.

Since community labels are identifiable only up to permutation,  we measure the distance between the true labels $\sigma$ and the estimated labels $\widehat\sigma$ by the misclustering proportion
\begin{gather*}
    \ell_{\rm mis}(\sigma,\widehat{\sigma};\{W_i\}_{i\in[n+1]},A)=\min_{\pi\in\mathcal{G}_{n_{\rm comm}}}(n+1)^{-1}\sum_{i=1}^{n+1}\mathbbm{1}\left( \sigma(W_i)\neq\pi(\widehat{\sigma}(A)_i) \right)\,,
\end{gather*}
where $\mathcal{G}_{n_{\rm comm}}$ denotes the permutation group of $[n_{\rm comm}]$. We impose the following condition on community recovery and base-score approximation.

\begin{assumption}\label{ass: community estimate}
    For any $\varepsilon>0$, there exist $\delta_{\rm comm}(n,\varepsilon)$ and $\delta(n,\varepsilon)$ such that
    \begin{gather*}
        \mathrm{Pr}\left( \ell_{\rm mis}(\sigma,\widehat{\sigma};\{W_i\}_{i\in[n+1]},A)\leq\varepsilon \right)\geq 1-\delta_{\rm comm}(n,\varepsilon),\\
        \mathrm{Pr}\left( \max_{i\in[n+1]}|v(X_i,Y_i;Z)-v^\star(X_i,Y_i)|\leq\varepsilon \right)\geq 1-\delta(n,\varepsilon)\,.
    \end{gather*}
\end{assumption}
For a fixed design $\{(W_i,X_i)\}_{i\in[n+1]}$ under a latent space model for $A$, Proposition 7 in \cite{gao2022community} provides conditions under which Assumption \ref{ass: community estimate} holds. 
The next theorem shows that the community-conditional miscoverage converges to zero.
For each $\mathcal{I}\in[n_{\rm comm}]$, let $F_{v^\star\mid \sigma(W)=\mathcal{I}}$ denote the conditional distribution function of $v^\star(X,Y)$ given $\sigma(W)=\mathcal{I}$.

\begin{theorem}\label{theo: community conditional}
    Suppose Assumption~\ref{ass: community estimate} holds. Assume that for every $\mathcal I\in[n_{\rm comm}]$, $F_{v^\star\mid \sigma(W)=\mathcal{I}}$ is $L_{\mathcal I}$-Lipschitz continuous for some positive constant $L_{\mathcal I}$. Then for $\varepsilon\le p_{\rm comm}/4$, there exists a constant $C>0$ such that, for every $\mathcal I\in[n_{\rm comm}]$,
    \begin{align*}
        &\left|\mathrm{Pr}\left( Y_{n+1}\in\widehat{C}_{\rm GraphCP}(X_{n+1})\mid\sigma(W_{n+1})=\mathcal{I} \right)-(1-\alpha)\right|\\
        \leq&C\left\{ \varepsilon+n^{-1}+\delta_{\rm comm}(n,\varepsilon)+\delta(n,\varepsilon)+n_{\rm comm}\exp\left( -np_{\rm comm}\varepsilon^2 \right) \right\}\,.
    \end{align*}
\end{theorem}
Theorem~\ref{theo: community conditional} provides the graph-specific conditional miscoverage decomposition in Theorem~\ref{theo:symmpi conditional miscoverage}. The terms $\varepsilon+\delta_{\rm comm}(n,\varepsilon)+\delta(n,\varepsilon)$ capture community-recovery and score-estimation errors, while $n^{-1}+n_{\rm comm}\exp(-n p_{\rm comm}\varepsilon^2)$ is the finite-sample calibration error driven by the discreteness of conformal calibration and the minimum community size. The within-community rank score construction removes the intrinsic conditional-mismatch error asymptotically, as the oracle rank score has the same conditional distribution across communities.

\section{Numerical Experiments}
\label{sec: experimental simulation}

We conduct numerical experiments on both simulated and real-world data to validate the unified theory and demonstrate its applicability across different conditional conformal settings. The numerical studies include simulations of conditional methods under covariate shift, experiments on conditional-coverage-oriented model selection, and a real-data analysis on graph-structured data. Additional numerical experiments are provided in Section~\ref{sec: additional numerical} of the Supplementary Material.

\subsection{Synthetic Data Experiments}
\label{sec: Synthetic Data Experiments}

We consider the regression model $Y = \mu(X) + \epsilon(X)$, where $X = (X_1, \ldots, X_d)^{\top} \sim N_d(0, I_d)$ and $\mu(X)=2\sum_{i=1}^{d}X_i/d$. The noise term $\epsilon(X)$ follows one of the following data-generating processes (DGPs): \\
\noindent\textbf{DGP1}: $\epsilon(X) \mid X \sim N\big(0, \big(\sum_{i=1}^{d}\big||X_i|-\sqrt{2/\pi}\big|\big)^2/d\big)$, \\
\noindent\textbf{DGP2}: $\epsilon(X) \mid X \sim N\big(0, \big\{\sum_{i=1}^{d}\exp(|X_i|)\big\}^2/d\big)$, \\
\noindent\textbf{DGP3}: $\epsilon(X) \mid X \sim N\big(0, \sum_{i=1}^{d}|X_i|/d\big)$.

Across all experiments, conformal prediction sets are constructed under a common data-splitting protocol and the nominal coverage is $1-\alpha=90\%$. We first split the data into a training set $Z_{\mathrm{tr}}$ and a calibration set $Z$, where $Z$ is reserved only for calibration. When a method requires estimating conditional score distributions or conditional quantiles, $Z_{\mathrm{tr}}$ is further divided evenly into two parts: $Z_{\mathrm{tr},1}$ is used to train the base predictor by linear regression and construct the base score function $v(x,y)=|y-\hat{\mu}(x)|$, while $Z_{\mathrm{tr},2}$ is used for the corresponding conditional distribution or quantile estimation. In our implementation, we take $|Z_{\mathrm{tr}}|=2|Z|=2n$ with $n=500,1000$, and use an independent test set of size $n_{\rm te}=500$. Each experiment is repeated 50 times.

For each DGP, we draw the test covariates $X_{{\rm te},1},\ldots,X_{{\rm te},n_{\rm te}}$ once and keep them fixed across repetitions. Since the conditional distribution of $Y$ given $X=x$ is known in the simulation, the conditional coverage probability $\Pr(Y\in \widehat{C}(X)\mid X=x)$ can be computed for any realized prediction set $\widehat{C}(X)$. Let $\widehat{C}^{(1)}(X),\ldots,\widehat{C}^{(50)}(X)$ denote the prediction sets obtained over the 50 repetitions. We report the marginal coverage as
$(50n_{\rm te})^{-1}\sum_{j=1}^{50}\sum_{i=1}^{n_{\rm te}}\Pr(Y\in\widehat{C}^{(j)}(X)\mid X=X_{{\rm te},i})$,
and the conditional miscoverage error as
$n_{\rm te}^{-1}\sum_{i=1}^{n_{\rm te}}\left|50^{-1}\sum_{j=1}^{50}\Pr(Y\in\widehat{C}^{(j)}(X)\mid X=X_{{\rm te},i})-(1-\alpha)\right|$.

\subsubsection{Conditional Conformal Methods under Covariate Shift}\label{sec: covariate shift}

We first evaluate conditional conformal methods under covariate shift. The conditional distribution of $Y\mid X$ follows DGP1--DGP3. The training and calibration covariates are generated from $X\sim N_d(0,I_d)$, whereas the test covariates are generated from $X\sim N_d(0,\sigma^2 I_d)$, with $\sigma\in\{0.8,1.0,1.2\}$. We consider dimensions $d\in\{5,10,15,20\}$.

We compare 11 conformal procedures. These include five conditional methods without covariate-shift adjustment, their five density-ratio-weighted counterparts, and the standard WCP based on the base score. The five conditional methods are LCP, GLCP with Engression \cite{shen2025engression} as the conditional distribution estimator, and three CQR-type methods that center the base score by an estimated conditional score quantile: linear quantile regression (CQR-LR), quantile random forests (CQR-RF), and LightGBM \cite{ke2017lightgbm} (CQR-LGB).


\begin{table*}[htbp]
    \caption{Conditional miscoverage error and marginal coverage under covariate shift for $n=500$ and $d=10$. Methods with smaller conditional miscoverage error than WCP are highlighted in \textbf{bold} with $*$. Methods whose marginal coverage deviates from $1-\alpha$ by more than $0.02$ are highlighted in \textit{italic} with $\times$.}
    \label{tab: main cs 1}
    \begin{center}
    {\fontsize{7}{6}\selectfont{\setlength{\tabcolsep}{3.5pt}\begin{tabular}{cccccccccccccccccc}
    \toprule
    & & \multicolumn{5}{c}{No covariate-shift adjustment} & & \multicolumn{5}{c}{Density-ratio weighted}\\
    \cmidrule(lr){3-7} \cmidrule(lr){9-13} 
    & $\sigma$ & GLCP & LCP & CQR-LR & CQR-RF & CQR-LGB & & GLCP & LCP & CQR-LR & CQR-RF & CQR-LGB & & WCP \\
    \midrule
    & & \multicolumn{11}{c}{Conditional Miscoverage}\\
    DGP1 & $0.8$ & $0.128$ & $0.116$ & $0.128$ & $0.142$ & $0.131$ &  & $\textbf{0.099}^*$ & $\textbf{0.097}^*$ & $\textbf{0.099}^*$ & $\textbf{0.098}^*$ & $\textbf{0.095}^*$ &  & $0.101$\\
    & $1.0$ & $\textbf{0.099}^*$ & $\textbf{0.102}^*$ & $\textbf{0.103}^*$ & $\textbf{0.098}^*$ & $\textbf{0.097}^*$ &  & $\textbf{0.100}^*$ & $\textbf{0.102}^*$ & $\textbf{0.103}^*$ & $\textbf{0.098}^*$ & $\textbf{0.097}^*$ &  & $0.104$\\
    & $1.2$ & $0.150$ & $0.169$ & $0.159$ & $0.116$ & $0.138$ &  & $\textbf{0.102}^*$ & $\textbf{0.106}^*$ & $\textbf{0.105}^*$ & $\textbf{0.092}^*$ & $\textbf{0.100}^*$ &  & $0.107$\\
    DGP2 & $0.8$ & $\textbf{0.033}^*$ & $0.061$ & $0.067$ & $\textbf{0.038}^*$ & $\textbf{0.038}^*$ &  & $\textbf{0.020}^*$ & $\textbf{0.045}^*$ & $0.050$ & $\textbf{0.027}^*$ & $\textbf{0.019}^*$ &  & $0.050$\\
    & $1.0$ & $\textbf{0.033}^*$ & $\textbf{0.056}^*$ & $0.063$ & $\textbf{0.030}^*$ & $\textbf{0.039}^*$ &  & $\textbf{0.033}^*$ & $\textbf{0.056}^*$ & $0.063$ & $\textbf{0.030}^*$ & $\textbf{0.039}^*$ &  & $0.062$\\
    & $1.2$ & $\textbf{0.062}^*$ & $0.096$ & $0.109$ & $\textbf{0.048}^*$ & $\textbf{0.066}^*$ &  & $\textbf{0.046}^*$ & $\textbf{0.064}^*$ & $0.070$ & $\textbf{0.039}^*$ & $\textbf{0.057}^*$ &  & $0.070$\\
    DGP3 & $0.8$ & $\textbf{0.026}^*$ & $0.038$ & $0.041$ & $\textbf{0.019}^*$ & $\textbf{0.016}^*$ &  & $\textbf{0.026}^*$ & $\textbf{0.033}^*$ & $0.036$ & $\textbf{0.020}^*$ & $\textbf{0.015}^*$ &  & $0.035$\\
    & $1.0$ & $\textbf{0.025}^*$ & $\textbf{0.032}^*$ & $0.035$ & $\textbf{0.017}^*$ & $\textbf{0.014}^*$ &  & $\textbf{0.025}^*$ & $\textbf{0.032}^*$ & $0.035$ & $\textbf{0.017}^*$ & $\textbf{0.014}^*$ &  & $0.034$\\
    & $1.2$ & $0.032$ & $0.041$ & $0.047$ & $\textbf{0.015}^*$ & $\textbf{0.018}^*$ &  & $\textbf{0.024}^*$ & $\textbf{0.029}^*$ & $0.033$ & $\textbf{0.015}^*$ & $\textbf{0.013}^*$ &  & $0.032$\\
    \midrule
    & & \multicolumn{11}{c}{Marginal coverage}\\
    DGP1 & $0.8$ & $\textit{0.850}^\times$ & $\textit{0.872}^\times$ & $\textit{0.858}^\times$ & $\textit{0.831}^\times$ & $\textit{0.846}^\times$ &  & $0.893$ & $0.903$ & $0.903$ & $0.900$ & $0.902$ &  & $0.901$\\
    & $1.0$ & $0.899$ & $0.908$ & $0.907$ & $0.902$ & $0.906$ &  & $0.899$ & $0.908$ & $0.907$ & $0.902$ & $0.906$ &  & $0.907$\\
    & $1.2$ & $\textit{0.825}^\times$ & $\textit{0.813}^\times$ & $\textit{0.824}^\times$ & $\textit{0.865}^\times$ & $\textit{0.842}^\times$ &  & $0.896$ & $0.907$ & $0.904$ & $0.903$ & $0.903$ &  & $0.902$\\
    DGP2 & $0.8$ & $\textit{0.929}^\times$ & $\textit{0.957}^\times$ & $\textit{0.962}^\times$ & $\textit{0.935}^\times$ & $\textit{0.933}^\times$ &  & $0.898$ & $0.907$ & $0.906$ & $0.906$ & $0.905$ &  & $0.907$\\
    & $1.0$ & $0.897$ & $0.904$ & $0.905$ & $0.903$ & $0.906$ &  & $0.897$ & $0.904$ & $0.905$ & $0.903$ & $0.906$ &  & $0.905$\\
    & $1.2$ & $\textit{0.850}^\times$ & $\textit{0.829}^\times$ & $\textit{0.820}^\times$ & $\textit{0.865}^\times$ & $\textit{0.856}^\times$ &  & $0.897$ & $0.906$ & $0.905$ & $0.903$ & $0.907$ &  & $0.906$\\
    DGP3 & $0.8$ & $0.916$ & $\textit{0.932}^\times$ & $\textit{0.935}^\times$ & $0.906$ & $0.912$ &  & $0.899$ & $0.904$ & $0.905$ & $0.904$ & $0.905$ &  & $0.905$\\
    & $1.0$ & $0.895$ & $0.902$ & $0.902$ & $0.899$ & $0.902$ &  & $0.895$ & $0.902$ & $0.902$ & $0.899$ & $0.902$ &  & $0.903$\\
    & $1.2$ & $\textit{0.878}^\times$ & $\textit{0.874}^\times$ & $\textit{0.869}^\times$ & $0.898$ & $0.891$ &  & $0.893$ & $0.900$ & $0.900$ & $0.899$ & $0.901$ &  & $0.899$\\
    \bottomrule
    \end{tabular}}}
    \end{center}
\end{table*}

Table~\ref{tab: main cs 1} reports the results for $n=500$ and $d=10$; additional settings and implementation details are provided in Section~\ref{sec: method_details} of the Supplementary Material. When $\sigma\neq 1$, methods without covariate-shift adjustment can deviate substantially from the nominal marginal coverage level, whereas their density-ratio-weighted counterparts recover valid marginal coverage. Most density-ratio-weighted conditional methods also have smaller conditional miscoverage error than WCP, reflecting the benefit of combining density-ratio correction with conditional adaptation. The CQR-LR variant is less stable under the more complex DGPs, where the linear conditional quantile model is misspecified.

\subsubsection{Conditional-Coverage-Oriented Model Selection}\label{sec: selection experiment}

We evaluate the conditional-coverage-oriented model selection procedure proposed in Section~\ref{sec: model selection} under DGP1--DGP3. We consider 20 candidate conformal prediction sets in the CQR form by using different conditional quantile estimators and tuning parameter choices; implementation details are provided in Section~\ref{sec: detail selection experiment} of the Supplementary Material.

We use the localized conditional-moment statistic $\mathcal{L}(s;Z)$ motivated by \cite{tripathi2003testing}. Let
$a_{ij}(h)=K(X_i,X_j;h)/\sum_{\ell=1}^{n+1}K(X_i,X_\ell;h)$ for $i,j\in[n+1]$ and 
\begin{align*}
    & \mathcal{L}\left(s;Z\right)
    =\sum_{i=1}^{n+1}\sum_{j=1}^{n+1}a_{ij}(h)\log\Bigl(1+\dfrac{\lambda_i\hat{\zeta}_s(X_j,Y_j;Z)}{n+1}\Bigr) \\
    & \ {\rm s.t.}\ 
    \sum_{j=1}^{n+1}\dfrac{a_{ij}(h)\hat{\zeta}_s(X_j,Y_j;Z)}{n+1+\lambda_i\hat{\zeta}_s(X_j,Y_j;Z)}=0, \quad i=1,\ldots, n+1\,,
\end{align*}
where $\hat{\zeta}_s(x,y;Z)=\mathbbm{1}\{s(x,y;Z)\le q_s(Z;\alpha)\}-(1-\alpha)$ and $\lambda_i$ is chosen to satisfy the constraint for each $i\in[n+1]$.
The criterion is small when the conditional moments, corresponding to $\mathbb{E}\{\hat{\zeta}_s(X,Y;Z)\mid Z,X=X_i\}$, are close to zero.
The criterion depends on the bandwidth $h$. Following \cite{hore2025conformal}, we define the effective sample size
$n_{\mathrm{eff}}(h)=n \mathbb{E}\big[\mathbb{E}\{K(X_1,X_2;h)\mid X_1\}^2\big]/\mathbb{E}\{K^2(X_1,X_2;h)\}$.
Based on $Z_{\mathrm{tr},1}$ and $Z_{\mathrm{tr},2}$, we estimate this quantity by $\hat n_{\mathrm{eff}}(h)$. We choose bandwidths $\hat h_{30}$, $\hat h_{40}$, and $\hat h_{50}$ by matching the target effective sample sizes 30, 40, and 50, respectively, and compute $\mathcal{L}(s;Z)$ for each candidate score on the calibration set.

We aggregate the three losses based on $\hat h_{30}$, $\hat h_{40}$, and $\hat h_{50}$ using two selection rules: selecting the candidate with the smallest average loss (\textbf{AvgLoss}), and selecting the candidate with the smallest average rank across the three losses (\textbf{AvgRankLoss}). We compare them with two baselines: selecting the candidate with the smallest average prediction-set size on $Z$ (\textbf{EffSize}) and selecting a candidate uniformly at random (\textbf{Rand}). Experiments are conducted for $n=500,1000$ and $d\in\{1,\ldots,20\}$, with 50 repetitions for each setting. 

\begin{figure}[hbtp]
    \centering
    \includegraphics[scale=0.31]{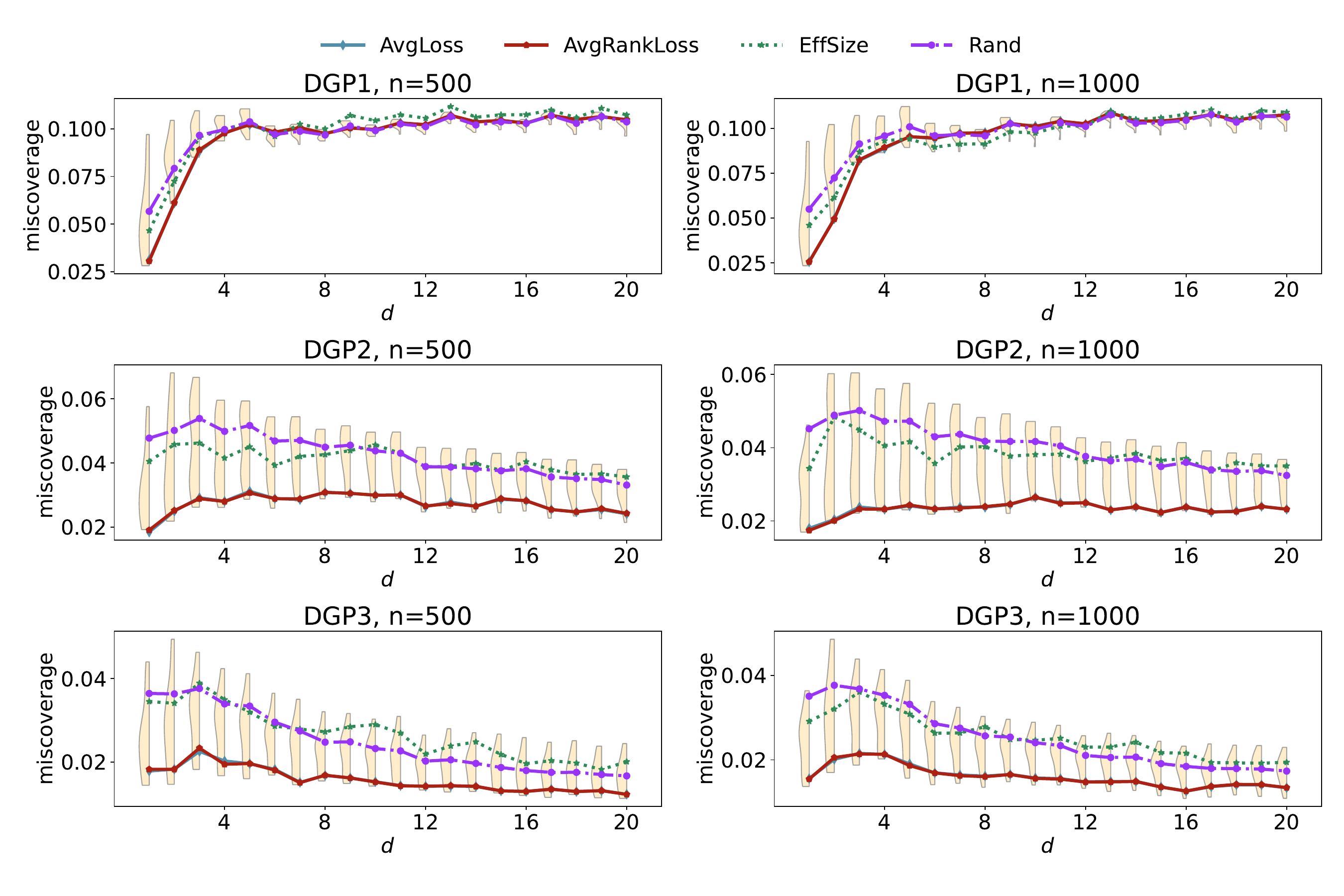}
    \caption{Comparison of selection strategies under DGP1--DGP3. Violin plots show the conditional miscoverage of candidate methods. Curves correspond to the selected candidates under AvgLoss, AvgRankLoss, EffSize, and Rand.}
    \label{fig: selection 123}
\end{figure}

Figure~\ref{fig: selection 123} reports violin plots of the conditional miscoverage across the 20 candidate methods, with the worst-performing 30\% of candidates omitted for visual clarity, and overlays the performance of the four selection strategies. \textbf{AvgLoss} and \textbf{AvgRankLoss} behave similarly and consistently select candidates close to the best-performing region of the candidate pool. In contrast, the \textbf{EffSize} and \textbf{Rand} baselines typically remain near the median of the candidate pool, showing that efficiency-driven or random selection does not reliably identify candidates with strong conditional coverage performance.

\subsection{Real Data Analysis on a Graph Dataset}\label{sec: graph real data}

We apply the community-conditional conformal prediction (GraphCP) method proposed in Section~\ref{sec: community conditional} to the Toloker graph dataset, available at \texttt{Kaggle.com}. The dataset contains a graph with 11,758 nodes and 519,000 undirected edges on a crowdsourcing platform, where nodes represent annotators and edges represent shared task annotations. Each node has numerical and categorical profile features, and the binary response indicates whether the annotator was banned. The goal is to predict ban status while providing prediction sets whose coverage adapts to the detected community structure.

The dataset provides 10 predefined random splits, each dividing the nodes into training, calibration, and test sets in a $2:1:1$ ratio. For each split, we train a LightGBM classifier using node features and neighborhood-derived features from the retained graph, and construct the base conformity score from the classifier's predicted probabilities. GraphCP and standard conformal prediction \cite[\textbf{StdCP}]{lunde2025conformal} use the same splits, base model, and score construction; they differ only in calibration. StdCP uses pooled conformal calibration, whereas GraphCP uses the community-aware calibration described in Section~\ref{sec: community conditional}. Details on graph construction, feature engineering, and model training are given in Section~\ref{sec: graph details} of the Supplementary Material.

To evaluate community-conditional coverage, we run the Louvain algorithm \cite{blondel2008fast} on the graph to obtain a reference community partition. Communities with fewer than 10 nodes are merged into a single outlier community, resulting in 8 final communities. These labels are used only for evaluation. Table~\ref{table: tolokerTS} reports the community proportions, community-wise miscoverage rates, and average prediction-set sizes over the 10 splits. The aggregate column is the weighted average of the community-wise conditional miscoverage results, using the proportions in the ``Prop.'' row as weights.

\begin{table*}[h]
    \caption{Community-wise miscoverage rates (in \textbf{bold}) and average prediction-set sizes (in parentheses) for predicting whether an annotator is banned. ``Prop.'' denotes the proportion of test nodes in each detected community.}
    \begin{center}
    {\fontsize{9}{9}\selectfont{\setlength{\tabcolsep}{3.5pt}\begin{tabular}{cccccccccccc}
    \toprule
    Community Index & 1 & 2 & 3 & 4 & 5 & 6 & 7 & 8 & Agg\vspace{4pt}\\ \midrule
    Prop. & $0.314$ & $0.207$ & $0.175$ & $0.156$ & $0.127$ & $0.010$ & $0.008$ & $0.003$ & $1.000$ \\
    GraphCP & $\textbf{0.036}$ & $\textbf{0.000}$ & $\textbf{0.006}$ & $\textbf{0.017}$ & $\textbf{0.010}$ & $\textbf{0.041}$ & $\textbf{0.014}$ & $\textbf{0.105}$ & $\textbf{0.0170}$ \\
    & $(0.990)$ & $(1.586)$ & $(1.428)$ & $(1.526)$ & $(1.243)$ & $(1.440)$ & $(1.425)$ & $(1.333)$ & $(1.3150)$ \\
    StdCP & $\textbf{0.065}$ & $\textbf{0.044}$ & $\textbf{0.015}$ & $\textbf{0.009}$ & $\textbf{0.011}$ & $\textbf{0.034}$ & $\textbf{0.023}$ & $\textbf{0.131}$ & $\textbf{0.0358}$ \\
    & $(1.048)$ & $(1.456)$ & $(1.347)$ & $(1.408)$ & $(1.247)$ & $(1.477)$ & $(1.425)$ & $(1.295)$ & $(1.2745)$ \\
    \bottomrule
    \end{tabular}}}
    \end{center}
    \label{table: tolokerTS}
\end{table*}

The results show that GraphCP reduces the aggregate miscoverage rate from $0.0358$ to $0.0170$, with only a modest increase in average prediction-set size. The improvement is most visible in several larger communities, such as communities 1--3, and also in community 8, where StdCP has the largest noncoverage rate. 
These findings demonstrate the practical effectiveness of the proposed GraphCP framework for achieving stronger conditional guarantees on a real-world graph dataset.

\section{Conclusion}\label{sec: conclusion}

This paper develops a unified framework for studying conditional coverage in conformal prediction. By representing diverse methods through weighted conformal quantiles, we derive pointwise and averaged conditional-miscoverage bounds that separate score-estimation, finite-sample calibration, and intrinsic conditional-mismatch errors. The theory clarifies existing conditional conformal methods and guides new procedures for model selection, conditional conformal methods under covariate shift, and structured data.
Several directions remain for future work. First, the covariate-shift theory assumes a known density ratio; allowing estimated ratios would require controlling how ratio-estimation error propagates into weighted calibration and conditional miscoverage. Second, conditional coverage becomes more challenging with high-dimensional covariates, where local neighborhoods are sparse and the relevant conditioning structure may be latent or low-dimensional. Extending the framework to learn such cases while preserving valid conditional-miscoverage control would broaden the applicability of conditional conformal methods.

\bibliographystyle{agsm}
\bibliography{ref}

@article{yang2025selection,
  title={Selection and aggregation of conformal prediction sets},
  author={Yang, Yachong and Kuchibhotla, Arun Kumar},
  journal={Journal of the American Statistical Association},
  volume={120},
  number={549},
  pages={435--447},
  year={2025},
  publisher={Taylor \& Francis}
}

@article{foygel2021limits,
  title={The limits of distribution-free conditional predictive inference},
  author={Foygel Barber, Rina and Candes, Emmanuel J and Ramdas, Aaditya and Tibshirani, Ryan J},
  journal={Information and Inference: A Journal of the IMA},
  volume={10},
  number={2},
  pages={455--482},
  year={2021},
  publisher={Oxford University Press}
}

@article{lei2018distribution,
  title={Distribution-free predictive inference for regression},
  author={Lei, Jing and G'Sell, Max and Rinaldo, Alessandro and Tibshirani, Ryan J. and Wasserman, Larry},
  journal={Journal of the American Statistical Association},
  volume={113},
  number={523},
  pages={1094--1111},
  year={2018},
  publisher={Taylor \& Francis}
}

@article{zhou2026conformal,
  title={Conformal Prediction Assessment: A Framework for Conditional Coverage Evaluation and Selection},
  author={Zhou, Zheng and Zhang, Xiangfei and Tao, Chongguang and Yang, Yuhong},
  journal={arXiv preprint arXiv:2603.27189},
  year={2026}
}

@article{lei2013,
author = { Jing   Lei  and  James   Robins  and  Larry   Wasserman },
title = {Distribution-Free Prediction Sets},
journal = {Journal of the American Statistical Association},
volume = {108},
number = {501},
pages = {278-287},
year  = {2013},
publisher = {Taylor & Francis}
}

@article{barber2026unifying,
  title={Unifying different theories of conformal prediction},
  author={Barber, Rina Foygel and Tibshirani, Ryan J},
  journal={Electronic Journal of Statistics},
  volume={20},
  number={1},
  pages={1428--1474},
  year={2026},
  publisher={The Institute of Mathematical Statistics and the Bernoulli Society}
}

@article{li2007quantile,
  title={Quantile regression in reproducing kernel Hilbert spaces},
  author={Li, Youjuan and Liu, Yufeng and Zhu, Ji},
  journal={Journal of the American Statistical Association},
  volume={102},
  number={477},
  pages={255--268},
  year={2007},
  publisher={Taylor \& Francis}
}

@article{meinshausen2006quantile,
  title={Quantile regression forests.},
  author={Meinshausen, Nicolai and Ridgeway, Greg},
  journal={Journal of machine learning research},
  volume={7},
  number={6},
  year={2006}
}

@inproceedings{NIPS2017_6449f44a,
 author = {Ke, Guolin and Meng, Qi and Finley, Thomas and Wang, Taifeng and Chen, Wei and Ma, Weidong and Ye, Qiwei and Liu, Tie-Yan},
 booktitle = {Advances in Neural Information Processing Systems},
 title = {LightGBM: A Highly Efficient Gradient Boosting Decision Tree},
 volume = {30},
 year = {2017}
}

@article{angelopoulos2024theoretical,
  title={Theoretical foundations of conformal prediction},
  author={Angelopoulos, Anastasios N and Barber, Rina Foygel and Bates, Stephen},
  journal={arXiv preprint arXiv:2411.11824},
  year={2024}
}

@article{zhou2002covering,
  title={The covering number in learning theory},
  author={Zhou, Ding-Xuan},
  journal={Journal of Complexity},
  volume={18},
  number={3},
  pages={739--767},
  year={2002},
  publisher={Elsevier}
}

@article{feldman2018generalization,
  title={Generalization bounds for uniformly stable algorithms},
  author={Feldman, Vitaly and Vondrak, Jan},
  journal={Advances in Neural Information Processing Systems},
  volume={31},
  year={2018}
}

@article{lunde2025conformal,
  title={Conformal prediction for network-assisted regression},
  author={Lunde, Robert and Levina, Elizaveta and Zhu, Ji},
  journal={Journal of the American Statistical Association},
  volume={120},
  number={551},
  pages={1633--1644},
  year={2025},
  publisher={Taylor \& Francis}
}

@article{gao2022community,
  title={Community detection in sparse latent space models},
  author={Gao, Fengnan and Ma, Zongming and Yuan, Hongsong},
  journal={Journal of Machine Learning Research},
  volume={23},
  number={322},
  pages={1--50},
  year={2022}
}

@article{guan2023localized,
  title={Localized conformal prediction: A generalized inference framework for conformal prediction},
  author={Guan, Leying},
  journal={Biometrika},
  volume={110},
  number={1},
  pages={33--50},
  year={2023},
  publisher={Oxford University Press}
}

@article{hore2025conformal,
  title={Conformal prediction with local weights: randomization enables robust guarantees},
  author={Hore, Rohan and Barber, Rina Foygel},
  journal={Journal of the Royal Statistical Society Series B: Statistical Methodology},
  volume={87},
  number={2},
  pages={549--578},
  year={2025},
  publisher={Oxford University Press UK}
}

@article{gibbs2025conformal,
  title={Conformal prediction with conditional guarantees},
  author={Gibbs, Isaac and Cherian, John J and Cand{\`e}s, Emmanuel J},
  journal={Journal of the Royal Statistical Society Series B: Statistical Methodology},
  pages={qkaf008},
  year={2025},
  publisher={Oxford University Press UK}
}

@inproceedings{romano2019conformalized,
  title={Conformalized quantile regression},
  author={Romano, Yaniv and Patterson, Evan and Cand\`es, Emmanuel},
  booktitle={Advances in Neural Information Processing Systems},
  volume={32},
  year={2019}
}

@article{dobriban2025symmpi,
  title={SymmPI: predictive inference for data with group symmetries},
  author={Dobriban, Edgar and Yu, Mengxin},
  journal={Journal of the Royal Statistical Society Series B: Statistical Methodology},
  pages={qkaf022},
  year={2025},
  publisher={Oxford University Press UK}
}

@article{Chernozhukov2021distribution,
  author = {Victor Chernozhukov  and Kaspar Wüthrich  and Yinchu Zhu },

  title = {Distributional conformal prediction},
  journal = {Proceedings of the National Academy of Sciences},
  volume = {118},
  number = {48},
  pages = {e2107794118},
  year = {2021},}

@article{tripathi2003testing,
  title={Testing conditional moment restrictions},
  author={Tripathi, Gautam and Kitamura, Yuichi},
  journal={The Annals of Statistics},
  volume={31},
  number={6},
  pages={2059--2095},
  year={2003},
  publisher={Institute of Mathematical Statistics}
}

@article{liang2024conformal,
  title={Conformal prediction after data-dependent model selection},
  author={Liang, Ruiting and Zhu, Wanrong and Barber, Rina Foygel},
  journal={Journal of the American Statistical Association},
  number={just-accepted},
  pages={1--26},
  year={2026},
  publisher={Taylor \& Francis}
}

@article{tibshirani2019conformal,
  title={Conformal prediction under covariate shift},
  author={Tibshirani, Ryan J and Foygel Barber, Rina and Cand{\`e}s, Emmanuel and Ramdas, Aaditya},
  journal={Advances in neural information processing systems},
  volume={32},
  year={2019}
}

@book{vovk2005algorithmic,
  title={Algorithmic learning in a random world},
  author={Vovk, Vladimir and Gammerman, Alexander and Shafer, Glenn},
  year={2005},
  publisher={New York: Springer Science \& Business Media}
}

@inproceedings{min2025personalized,
title={Personalized Federated Conformal Prediction with Localization},
author={Yinjie Min and Chuchen Zhang and Liuhua Peng and Changliang Zou},
booktitle={The Thirty-ninth Annual Conference on Neural Information Processing Systems},
year={2025},
}

@book{wainwright2019high,
  title={High-dimensional statistics: A non-asymptotic viewpoint},
  author={Wainwright, Martin J},
  volume={48},
  year={2019},
  publisher={Cambridge university press}
}

@article{10.3150/13-BEJ549,
author = {Johannes Lederer and Sara van de Geer},
title = {{New concentration inequalities for suprema of empirical processes}},
volume = {20},
journal = {Bernoulli},
number = {4},
publisher = {Bernoulli Society for Mathematical Statistics and Probability},
pages = {2020 -- 2038},
year = {2014}
}

@inproceedings{mohri2014learning,
  title={Learning theory and algorithms for revenue optimization in second price auctions with reserve},
  author={Mohri, Mehryar and Medina, Andres Munoz},
  booktitle={International conference on machine learning},
  pages={262--270},
  year={2014},
  organization={PMLR}
}

@inproceedings{jung2023batch,
  title={Batch Multivalid Conformal Prediction},
  author={Jung, Christopher and Noarov, Georgy and Ramalingam, Ramya and Roth, Aaron},
  booktitle={International Conference on Learning Representations},
  year={2023}
}

@article{shafer2008tutorial,
  title={A tutorial on conformal prediction.},
  author={Shafer, Glenn and Vovk, Vladimir},
  journal={Journal of Machine Learning Research},
  volume={9},
  number={3},
  year={2008}
}

@article{hu2024two,
  title={A two-sample conditional distribution test using conformal prediction and weighted rank sum},
  author={Hu, Xiaoyu and Lei, Jing},
  journal={Journal of the American Statistical Association},
  volume={119},
  number={546},
  pages={1136--1154},
  year={2024},
  publisher={Taylor \& Francis}
}

@inproceedings{wu2025conditional,
  title={Conditional Testing based on Localized Conformal p-values},
  author={Wu, Xiaoyang and Lu, Lin and Wang, Zhaojun and Zou, Changliang},
  booktitle={ICLR},
  year={2025}
}

@article{bates2023testing,
  title={Testing for outliers with conformal p-values},
  author={Bates, Stephen and Cand{\`e}s, Emmanuel and Lei, Lihua and Romano, Yaniv and Sesia, Matteo},
  journal={The Annals of Statistics},
  volume={51},
  number={1},
  pages={149--178},
  year={2023},
  publisher={Institute of Mathematical Statistics}
}

@article{shen2025engression,
  title={Engression: extrapolation through the lens of distributional regression},
  author={Shen, Xinwei and Meinshausen, Nicolai},
  journal={Journal of the Royal Statistical Society Series B: Statistical Methodology},
  volume={87},
  number={3},
  pages={653--677},
  year={2025},
  publisher={Oxford University Press UK}
}

@article{bousquet2002stability,
  title={Stability and generalization},
  author={Bousquet, Olivier and Elisseeff, Andr{\'e}},
  journal={Journal of machine learning research},
  volume={2},
  number={Mar},
  pages={499--526},
  year={2002}
}

@article{blondel2008fast,
  title={Fast unfolding of communities in large networks},
  author={Blondel, Vincent D and Guillaume, Jean-Loup and Lambiotte, Renaud and Lefebvre, Etienne},
  journal={Journal of statistical mechanics: theory and experiment},
  volume={2008},
  number={10},
  pages={P10008},
  year={2008}
}

@article{bian2023training,
  title={Training-conditional coverage for distribution-free predictive inference},
  author={Bian, Michael and Barber, Rina Foygel},
  journal={Electronic Journal of Statistics},
  volume={17},
  number={2},
  pages={2044--2066},
  year={2023},
  publisher={The Institute of Mathematical Statistics and the Bernoulli Society}
}

@article{liang2025algorithmic,
  title={Algorithmic stability implies training-conditional coverage for distribution-free prediction methods},
  author={Liang, Ruiting and Barber, Rina Foygel},
  journal={The Annals of Statistics},
  volume={53},
  number={4},
  pages={1457--1482},
  year={2025},
  publisher={Institute of Mathematical Statistics}
}

@inproceedings{clarkson2023distribution,
  title={Distribution free prediction sets for node classification},
  author={Clarkson, Jase},
  booktitle={International conference on machine learning},
  pages={6268--6278},
  year={2023},
  organization={PMLR}
}

@article{sen2008collective,
  title={Collective classification in network data},
  author={Sen, Prithviraj and Namata, Galileo and Bilgic, Mustafa and Getoor, Lise and Galligher, Brian and Eliassi-Rad, Tina},
  journal={AI magazine},
  volume={29},
  number={3},
  pages={93--93},
  year={2008}
}

@book{vapnik2013nature,
  title={The nature of statistical learning theory},
  author={Vapnik, Vladimir},
  year={2013},
  publisher={Springer science \& business media}
}

@article{braun2025conditional,
  title={Conditional Coverage Diagnostics for Conformal Prediction},
  author={Braun, Sacha and Holzm{\"u}ller, David and Jordan, Michael I and Bach, Francis},
  journal={arXiv preprint arXiv:2512.11779},
  year={2025}
}

@inproceedings{chen2019nearest,
  title={Nearest neighbor and kernel survival analysis: Nonasymptotic error bounds and strong consistency rates},
  author={Chen, George},
  booktitle={International Conference on Machine Learning},
  pages={1001--1010},
  year={2019},
  organization={PMLR}
}

@article{zheng1996consistent,
  title={A consistent test of functional form via nonparametric estimation techniques},
  author={Zheng, John Xu},
  journal={Journal of Econometrics},
  volume={75},
  number={2},
  pages={263--289},
  year={1996},
  publisher={Elsevier}
}

@article{ke2017lightgbm,
  title={LightGBM: A highly efficient gradient boosting decision tree},
  author={Ke, Guolin and Meng, Qi and Finley, Thomas and Wang, Taifeng and Chen, Wei and Ma, Weidong and Ye, Qiwei and Liu, Tie-Yan},
  journal={Advances in neural information processing systems},
  volume={30},
  year={2017}
}

@article{greblicki1984distribution,
  title={Distribution-free pointwise consistency of kernel regression estimate},
  author={Greblicki, W{\l}odzimierz and Krzy{\.z}ak, Adam and Pawlak, Miros{\l}aw},
  journal={The annals of Statistics},
  pages={1570--1575},
  year={1984},
  publisher={JSTOR}
}

\appendix
\renewcommand{\thesection}{S\arabic{section}}       
\renewcommand{\thefigure}{S\arabic{figure}}         
\renewcommand{\thetable}{S\arabic{table}}          
\renewcommand{\theequation}{S\arabic{section}.\arabic{equation}}  
\renewcommand{\thetheorem}{S\arabic{theorem}}
\renewcommand{\thelemma}{S\arabic{lemma}}
\renewcommand{\thealgorithm}{S\arabic{algorithm}}
\renewcommand{\thecorollary}{S\arabic{corollary}}

\setcounter{section}{0}
\setcounter{figure}{0}
\setcounter{table}{0}
\setcounter{equation}{0}
\setcounter{theorem}{0}
\setcounter{lemma}{0}
\setcounter{algorithm}{0}

\clearpage
\section*{Supplementary Materials}
This Supplementary Material includes technical details, extensions to weighted SymmPI framework, additional theoretical results and applications, proofs of main results, preliminary lemmas and their proofs, and is organized as follows.
Section~\ref{sec: example details} provides detailed descriptions and theoretical analysis of major conditional conformal methods in the i.i.d.~setting. Section~\ref{sec:supp_covariate_shift} presents additional results on conditional methods under the covariate shift setting. Section~\ref{sec: appendix weighted symmpi} introduces the weighted SymmPI framework. Section~\ref{sec:supp_add_applications} collects further applications and extensions, including GRLCP, localized conformal $p$-values, and the two-layer hierarchical example. Section~\ref{sec: app technical details} provides additional technical discussions and clarifications, while Section~\ref{sec:pre_lemmas} collects the preliminary lemmas and empirical-quantile tools used throughout the analysis. The proofs of main results, supplementary results and preliminary lemmas are deferred to Sections~\ref{sec:proof_main}--\ref{sec:proof_pre_lemmas}. Additional numerical results are reported in Section~\ref{sec: additional numerical}.

\startcontents
\printcontents{}{1}{\section*{Contents of the Supplement}}

\section{Detailed Analyses of Conditional Methods under the Unified Theory}\label{sec: example details}

In this section, we provide detailed descriptions and theoretical analyses of several representative conformal methods targeting conditional coverage in the i.i.d.~setting. These include LCP, RLCP, CQR, DCP, GLCP, BatchGCP, and CC. The goal is to show how each method can be expressed through the unified framework introduced in Section~\ref{sec: general framework}, and how the three components in the conditional-miscoverage decomposition arise in each case.

Throughout this section, let
\begin{gather*}
    Z=\{(X_1,Y_1),\ldots,(X_n,Y_n),(X_{n+1},Y_{n+1})\}
\end{gather*}
denote the calibration and test data, drawn i.i.d.~from $P_X\times P_{Y\mid X}$, where $X_i\in\mathcal X\subseteq\mathbb R^d$ and $Y_i\in\mathcal{Y}\subseteq\mathbb{R}$ for $i\in[n+1]$. For a trial response value $y$, write
\begin{gather*}
    Z^y=\{(X_1,Y_1),\ldots,(X_n,Y_n),(X_{n+1},y)\}\,.
\end{gather*}
In addition, denote $Y_i^y=Y_i$ for $i\in[n]$ and $Y_{n+1}^y=y$. Let $(X,Y)$ be a generic random pair drawn from $P_X\times P_{Y\mid X}$, independent of all other random objects.

For each method, it suffices to identify three objects: the score $s(x,y;Z)$, the weight function $w(\cdot)$, and the oracle score $s^\star(x,y)$. Once these objects are specified, the marginal-coverage property and the conditional-miscoverage bound follow by verifying the corresponding assumptions in the unified theory.

\subsection{Localized Conformal Prediction (LCP)}

\paragraph*{Details}
Let $v(\cdot,\cdot)$ be a pre-trained base score and let $K(\cdot,\cdot;h)$ be a kernel function with bandwidth $h$.
For a trial response value $y$, define the localized rank statistic
\begin{align*}
    \hat\beta_i^{\rm LCP}(y)=\frac{\sum_{j=1}^{n+1}K(X_i,X_j;h)\mathbbm{1}\{v(X_j,Y_j^y)\leq v(X_i,Y_i^y)\}}{\sum_{j\in[n+1]}K(X_i,X_j;h)}, \quad i\in[n+1]\,.
\end{align*}
This statistic estimates the conditional c.d.f.~of $v(X_i,Y_i^y)$ given $X_i$.
According to Section~A.4 in the Supplementary Material of \cite{hore2025conformal}, the LCP method proposed by \cite{guan2023localized} can be written as
\begin{align*}
    \widehat{C}_{\rm LCP}(X_{n+1})=\left\{ y:s(X_{n+1},y;Z^y)\leq Q\left(1-\alpha;\frac1{n+1}\Bigl\{\sum_{i=1}^{n}\delta_{s(X_i,Y_i;Z^y)}+\delta_{\infty}\Bigr\}\right)\right\}
\end{align*}
with $s(X_i,Y_i^y;Z^y)=\hat\beta_i^{\rm LCP}(y)$. Hence LCP is an instance of \eqref{eq: general Cvec} with $w(\cdot)\equiv 1$.

\paragraph*{Unified specification}
\begin{itemize}
    \item \textit{Score:} $\displaystyle s(x,y;Z)=\frac{\sum_{j=1}^{n+1}K(x,X_j;h)\mathbbm{1}\{v(X_j,Y_j)\leq v(x,y)\}}{\sum_{j=1}^{n+1}K(x,X_j;h)}$.
    \item \textit{Oracle score:} $s^\star(x,y)=F_{v\mid X}(v(x,y)\mid x)$, the conditional c.d.f.~of $v(X,Y)$ given $X=x$.
    \item \textit{Weight:} $w(\cdot)\equiv 1$.
\end{itemize}

\paragraph*{Marginal coverage}
Since $w(\cdot)\equiv 1$ and the score $s(x,y;Z)$ is permutation invariant in $\{(X_i,Y_i)\}_{i\in[n+1]}$, Theorem~\ref{theo: marginal general} guarantees marginal coverage for LCP.

\begin{corollary}[Test-conditional miscoverage of LCP]\label{corr: LCP test}
    Suppose Assumption~1 of \cite{guan2023localized} holds for the base score $v(\cdot,\cdot)$ and the kernel $K(\cdot,\cdot;h)$, with $\beta=d$ in that assumption. Then there exists a constant $C>0$ such that for every $t\in\mathcal{X}$, 
    \begin{align*}
        &\left|\mathrm{Pr}\left( Y_{n+1}\in\widehat{C}_{\rm LCP}(X_{n+1})\mid X_{n+1}=t \right)-(1-\alpha)\right|\\
        &\leq C \left\{ (nh^d)^{-1/2}\log^{1/2}(n)+h\log(1/h)\right\}\,.
    \end{align*}
\end{corollary}
\begin{proof}
    Since $s^\star(X,Y)$ is uniformly distributed on $[0,1]$ conditional on $X=t$ for any $t\in\mathcal{X}$, Assumption~\ref{ass: oracle quantile regularity} holds directly with constants independent of $t$.
    In addition,
    $Q(1-\alpha;F_{s^\star\mid X=t})=Q(1-\alpha;F_{s^\star})=1-\alpha$,
    so the intrinsic conditional-mismatch error is zero.

    By the proof of Theorem~3 in \cite{guan2023localized}, the main term $\inf_{\varepsilon>0}\{\varepsilon+n\delta_n(\varepsilon)\}$ in the score-estimation error is of order
    $O((nh^d)^{-1/2}\log^{1/2}(n)+h\log(1/h))$; see Equation~B.28 in the supplement of \cite{guan2023localized}.
    Following \cite{chen2019nearest}, or more directly the proof of Theorem~1 in \cite{greblicki1984distribution}, the stability contribution $\widetilde{\varepsilon}_n+\widetilde{\delta}_n$ for the kernel-based score can be bounded by
    $O((nh^d)^{-1}+\exp(-nh^d))$, which is negligible relative to the main score-estimation error term.
    Since $w(\cdot)\equiv 1$, the finite-sample calibration error in Theorem~\ref{theo: main miscoverage} is of order $O(n^{-1/2}\log^{1/2}(n))$, which is also dominated by the stated score-estimation error.
    Applying Theorem~\ref{theo: main miscoverage} yields the desired bound.
\end{proof}

\paragraph*{Analysis}
For LCP, the score-estimation error is the kernel c.d.f.~estimation error, of order $(nh^d)^{-1/2}\log^{1/2}(n)+h\log(1/h)$.
The finite-sample calibration error is $O(n^{-1/2}\log^{1/2}(n))$, which is dominated by the score-estimation error. The intrinsic conditional-mismatch error is identically zero because $s^\star(X,Y)\mid X=t\sim{\rm Uniform}(0,1)$ for all $t\in\mathcal X$. Therefore, the conditional-miscoverage error rate of LCP is
of order $(nh^d)^{-1/2}\log^{1/2}(n)+h\log(1/h)$,
which recovers the rate established in the original LCP analysis.

\subsection{Randomized Localized Conformal Prediction (RLCP)}

\paragraph*{Details}
Let $s(x,y)$ be a pre-trained score that is free of $Z$. RLCP \citep{hore2025conformal} introduces an auxiliary covariate $\widetilde X$ to randomize the local calibration weights. Conditional on $X_{n+1}$, draw $\widetilde X$ from the distribution with density proportional to $K(\cdot,X_{n+1};h)$. Given this draw, RLCP calibrates the pre-trained score using the kernel weights $K(X_i,\widetilde X;h)$, leading to the prediction set $\widehat{C}_{\rm RLCP}(X_{n+1})$ defined as
\begin{align*}
    \left\{y: s(X_{n+1},y) \leq Q\left(1-\alpha; \frac{\sum_{i=1}^{n}K(X_i,\widetilde X;h)\delta_{s(X_i,Y_i)} +K(X_{n+1},\widetilde X;h)\delta_{\infty}}{\sum_{i=1}^{n+1}K(X_i,\widetilde X;h)} \right)\right\}\,.
\end{align*}
Hence RLCP is an instance of \eqref{eq: general Cvec} with the random weight function $w(x)=K(x,\widetilde X;h)$.

\paragraph*{Unified specification}
\begin{itemize}
    \item \textit{Score:} $s(x,y;Z)=s(x,y)$ is a pre-trained score; 
    \item \textit{Oracle score:} $s^\star(x,y)=s(x,y)$.
    \item \textit{Weight:} $w(x)=K(x,\widetilde X;h)$.
\end{itemize}

\paragraph*{Marginal coverage}
Let $K_X(x)=\mathbb{E}\{K(x,X;h)\}$ and define the normalized weight function $\bar w(x)=K(x,\widetilde X;h)/K_X(\widetilde X)$. Then $d_Z(w,\bar w)=0$, and $\mathbb{E}\{\bar w(x)\}=1$ for every $x\in\mathcal{X}$. Hence $\bar w\in\mathcal W_X$, so Theorem~\ref{theo: marginal general} guarantees finite-sample marginal coverage for RLCP. See Section~\ref{sec: discussion of marginal} for more detailed proofs and discussions.

\begin{corollary}[Test-conditional miscoverage of RLCP]\label{corr: RLCP test}
    Assume that $X$ is supported on $[0,1]^d$ with density bounded within $[\underline{L}_1,\overline{L}_1]$, and that the kernel is of the form $K(x_1,x_2;h)=K_0(\|x_1-x_2\|_2/h)$, where $K_0(\cdot)$ is bounded, symmetric around $0$, decreases on $[0,\infty)$, and satisfies $uK_0(u)$ decreases on $(1,\infty)$, and $\int_0^\infty u^dK_0(u)\,du<\infty$. Assume further that (i) there exists $M>0$ such that $|s(X,Y)|\le M$ almost surely; (ii) for each $t\in\mathcal X$, the conditional distribution $F_{s\mid X=t}$ has a density lower bounded by $\underline{L}>0$ on its support; and (iii) the conditional quantile function $q(t)=Q(1-\alpha;F_{s\mid X=t})$ is Lipschitz on $\mathcal X$. 
    Then there exists a constant $C>0$ such that for every $t\in\mathcal{X}$,
    \begin{align*}
        &\left|\mathrm{Pr}\left( Y_{n+1}\in\widehat{C}_{\rm RLCP}(X_{n+1})\mid X_{n+1}=t \right)-(1-\alpha)\right|\leq C \left\{ (nh^d)^{-1/2}\log^{1/2}(n)+h\right\}\,.
    \end{align*}
\end{corollary}
\begin{proof}
    Since RLCP uses the pre-trained score $s(x,y)$, we have $s(x,y)=s^\star(x,y)$. Hence the score-estimation error of RLCP is zero.

    We first control the finite-sample calibration error in Theorem~\ref{theo: main miscoverage}. For the random weight function $w(x)=K(x,\widetilde X;h)$, Lemma~\ref{lemma: kernel fun} in Section~\ref{sec:pre_lemmas} gives
    $B_w=\mathbb{E}\{w(X)\mid \widetilde X\}=O(h^d)$ and
    $\sigma_w^2=\mathbb{E}\{w^2(X)\mid \widetilde X\}=O(h^d)$, while $M_w=O(1)$ because $K_0(\cdot)$ is bounded. Therefore, the finite-sample calibration error 
    $\Gamma_n(w)=O((nh^d)^{-1/2}\log^{1/2}(n))$.

    It remains to control the intrinsic conditional-mismatch error. By condition (iii), the oracle conditional quantile function $q(\cdot)$ is Lipschitz. Applying Lemma~\ref{lemma: kernel fun bias} with $f=q$ gives
    \begin{gather*}
        \mathbb{E}\left\{ B_w^{-1}w(X)\left|q(x)-q(X)\right|\mid X_{n+1}=t \right\}\le Ch\,
    \end{gather*}
    for some constant $C>0$ depending on the Lipschitz constant of $q(\cdot)$ and the kernel function.
    Conditions (i), (ii), and (iii) ensure that Assumption~\ref{ass: oracle quantile regularity} is satisfied for the oracle score, $F_{s^\star\mid X=t}$, and $F_{w\circ s^\star}$. Substituting these bounds into Theorem~\ref{theo: main miscoverage} yields the desired result.
\end{proof}

\paragraph*{Analysis}
The kernel conditions imposed in Corollary~\ref{corr: RLCP test} are standard in nonparametric kernel estimation: they imply a local effective sample size of order $nh^d$ and a localization bias of order $h$.
Since RLCP uses a pre-trained score, the score-estimation error is zero. The finite-sample calibration error is of order $(nh^d)^{-1/2}\log^{1/2}(n)$, reflecting the local effective sample size induced by the kernel weights. The intrinsic conditional-mismatch error is of order $h$. Therefore, the conditional-miscoverage rate of RLCP is
$(nh^d)^{-1/2}\log^{1/2}(n)+h$,
which reflects the usual tradeoff between local effective sample size and localization bias.

\subsection{Conformized Quantile Regression (CQR)}

\paragraph*{Details}
Let $v_{\alpha/2}(\cdot)$ and $v_{1-\alpha/2}(\cdot)$ be two pre-trained estimators of the conditional $\alpha/2$- and $(1-\alpha/2)$-quantiles of $Y$ given $X$. CQR \citep{romano2019conformalized} uses the pre-trained conformity score
\begin{gather*}
    s(x,y)=\max\{y-v_{1-\alpha/2}(x),\,v_{\alpha/2}(x)-y\}\,.
\end{gather*}
The above max-type score is the standard two-sided CQR score. One-sided variants can be written in the same framework as well; for example, if $\widehat Q_{\alpha}(x)$ denotes a pre-trained estimator of the conditional $(1-\alpha)$-quantile of $Y$ given $X=x$, then the score $s(x,y)=y-\widehat Q_{\alpha}(x)$ yields an upper one-sided CQR-type construction. In this subsection, we use the max-type two-sided score as the representative example.

For a trial response value $y$, applying this score to $Z^y$ yields the CQR prediction set
\[
\widehat{C}_{\rm CQR}(X_{n+1})=\left\{ y:s(X_{n+1},y)\leq Q\left(1-\alpha;\frac1{n+1}\left\{\sum_{i=1}^{n}\delta_{s(X_i,Y_i)}+\delta_{\infty}\right\}\right)\right\}\,.
\]
Hence CQR is an instance of \eqref{eq: general Cvec} with $w(\cdot)\equiv 1$.

\paragraph*{Unified specification}
\begin{itemize}
    \item \textit{Score:} $s(x,y;Z)=\max\{y-v_{1-\alpha/2}(x),\,v_{\alpha/2}(x)-y\}$.
    \item \textit{Oracle score:} $s^\star(x,y)=\max\{y-v^\star_{1-\alpha/2}(x),\,v^\star_{\alpha/2}(x)-y\}$, where $v^\star_{\tau}(\cdot)$ is the true conditional $\tau$-quantile of $Y$ given $X$.
    \item \textit{Weight:} $w(\cdot)\equiv 1$.
\end{itemize}

\paragraph*{Marginal coverage}
Since $w(\cdot)\equiv 1$ and the score is based on pre-trained models and is free of $Z$, Theorem~\ref{theo: marginal general} guarantees marginal coverage for CQR.

\begin{corollary}[Test-conditional miscoverage of CQR]\label{corr: CQR test}
    Assume that: (i) there exists a function $\delta_n(\varepsilon)$ such that, for every $\varepsilon>0$,
    \begin{gather*}
        \mathrm{Pr}\left(\sup_{x\in\mathcal X}\left[|v_{\alpha/2}(x)-v^\star_{\alpha/2}(x)|\vee|v_{1-\alpha/2}(x)-v^\star_{1-\alpha/2}(x)|\right]>\varepsilon\right)\le \delta_n(\varepsilon)\,;
    \end{gather*}
    (ii) there exists $M>0$ such that $|Y|\le M$ almost surely; (iii) the marginal distribution of $s^\star(X,Y)$ has density lower bounded by $\underline{L}>0$ on its support; and (iv) the conditional distribution $F_{Y\mid X=t}$ has density upper bounded by $\overline{L}<\infty$ for every $t\in\mathcal X$. 
    Then there exists a constant $C>0$ such that for every $t\in\mathcal{X}$,
    \begin{align*}
        &\left|\mathrm{Pr}\left( Y_{n+1}\in\widehat{C}_{\rm CQR}(X_{n+1})\mid X_{n+1}=t \right)-(1-\alpha)\right|\leq C \left\{ \delta_{\rm tr}+n^{-1/2}\log^{1/2}(n)\right\}\,,
    \end{align*}
    where $\delta_{\rm tr}=\inf_{\varepsilon>0}\{\varepsilon+\delta_n(\varepsilon)\}$ is the score-estimation error induced by the pre-trained quantile estimators.
\end{corollary}

\begin{proof}
    We first consider the score-estimation error. On the event
    \begin{gather*}
        \left\{ \sup_{x\in\mathcal X} \left[ |v_{\alpha/2}(x)-v^\star_{\alpha/2}(x)| \vee |v_{1-\alpha/2}(x)-v^\star_{1-\alpha/2}(x)| \right] \le \varepsilon \right\}\,,
    \end{gather*}
    we have for every $x\in\mathcal X$ and $y\in\mathcal Y$,
    \begin{align*}
        &|s(x,y)-s^\star(x,y)|\\
        =&\left| \max\{y-v_{1-\alpha/2}(x),\,v_{\alpha/2}(x)-y\} - \max\{y-v^\star_{1-\alpha/2}(x),\,v^\star_{\alpha/2}(x)-y\}
        \right|\\
        \le& |v_{1-\alpha/2}(x)-v^\star_{1-\alpha/2}(x)| \vee |v_{\alpha/2}(x)-v^\star_{\alpha/2}(x)|
        \le \varepsilon\,.
    \end{align*}
    Hence
    $\mathrm{Pr}\{\sup_{x,y}|s(x,y)-s^\star(x,y)|>\varepsilon\}\le \delta_n(\varepsilon)$.
    Since the fitted quantile estimators in CQR are pre-trained independently of the calibration sample, the stability terms in Theorem~\ref{theo: main miscoverage} satisfy $\widetilde{\varepsilon}_n=\widetilde{\delta}_n=0$.

    We next control the finite-sample calibration error. Since $w(\cdot)\equiv 1$, we have $M_w=B_w=\sigma_w=1$, and therefore
    $\Gamma_n(w)=O\bigl(n^{-1/2}\log^{1/2}(n)\bigr)$.

    It remains to evaluate the intrinsic conditional-mismatch error. Under the oracle score,
    \begin{gather*}
        s^\star(x,Y)=\max\{Y-v^\star_{1-\alpha/2}(x),\,v^\star_{\alpha/2}(x)-Y\}\,.
    \end{gather*}
    Therefore, for every fixed $t$,
    \begin{gather*}
        \mathrm{Pr}(s^\star(X,Y)\leq 0\mid X=t)
        =
        \mathrm{Pr}(v^\star_{\alpha/2}(x)\leq Y\leq v^\star_{1-\alpha/2}(x)\mid X=t)
        =
        1-\alpha\,,
    \end{gather*}
    which implies that $Q(1-\alpha;F_{s^\star\mid X=t})=0$. Since this oracle conditional quantile is the same for all $t\in\mathcal{X}$, the marginal oracle score distribution also has $(1-\alpha)$-quantile equal to $0$. Hence the intrinsic conditional-mismatch error is zero.

    Conditions (ii), (iii), and (iv) ensure that Assumption~\ref{ass: oracle quantile regularity} is satisfied for the oracle score, $F_{s^\star\mid X=t}$, and $F_{s^\star}$.
    According to Theorem~\ref{theo: main miscoverage} and taking minimum over $\varepsilon>0$ yields
    \begin{gather*}
        \left|\mathrm{Pr}\left( Y_{n+1}\in\widehat{C}_{\rm CQR}(X_{n+1})\mid X_{n+1}=t \right)-(1-\alpha)\right|
        \leq C\left\{\delta_{\rm tr}+n^{-1/2}\log^{1/2}(n)\right\}\,
    \end{gather*}
    for some positive constant $C$.
\end{proof}

\paragraph*{Analysis}
For CQR, the three components in Theorem~\ref{theo: main miscoverage} have simple forms. The score-estimation error is determined by the accuracy of the pre-trained quantile estimators and is summarized by $\delta_{\rm tr}$.
The finite-sample calibration error is $O(n^{-1/2}\log^{1/2}(n))$. Finally, under the oracle score based on $v^\star_{\alpha/2}$ and $v^\star_{1-\alpha/2}$, the conditional $(1-\alpha)$-quantile equals $0$ for every covariate value, so the intrinsic conditional-mismatch error is zero. Therefore, the conditional-miscoverage rate of CQR is $\delta_{\rm tr}+n^{-1/2}\log^{1/2}(n)$.
In many nonparametric settings, the score-estimation error is the dominant term.

\subsection{Distributional Conformal Prediction (DCP) and Generalized Localized Conformal Prediction (GLCP)}

\paragraph*{Details}
GLCP \citep{min2025personalized} constructs a conformity score by first estimating the conditional distribution of a base score $v(X,Y)$ given $X=x$. Let $\widehat F_{v\mid X}(\cdot\mid x)$ be an estimator of the conditional c.d.f.~of $v(X,Y)$ given $X=x$, trained independently of the calibration data, and let $\varphi:[0,1]\to[0,1]$ be a fixed transformation. GLCP uses the pre-trained score
\begin{gather*}
    s(x,y)=\varphi\bigl(\widehat F_{v\mid X}(v(x,y)\mid x)\bigr)\,.
\end{gather*} 
The resulting prediction set is
\[
\widehat{C}_{\rm GLCP}(X_{n+1})=\left\{ y:s(X_{n+1},y)\leq Q\left(1-\alpha;\frac1{n+1}\left\{\sum_{i=1}^{n}\delta_{s(X_i,Y_i)}+\delta_{\infty}\right\}\right)\right\}\,.
\]
Thus, GLCP is exactly \eqref{eq: general Cvec} with $w(\cdot)\equiv 1$ and a pre-trained score $s(x,y)$.
This formulation includes DCP \citep{Chernozhukov2021distribution} as a special case: taking $v(x,y)=y$ and $\varphi(u)=|1/2-u|$ recovers the DCP score.

\paragraph*{Unified specification}
\begin{itemize}
    \item \textit{Score:} $s(x,y;Z)=\varphi\bigl(\widehat F_{v\mid X}(v(x,y)\mid x)\bigr)$. 
    \item \textit{Oracle score:} $s^\star(x,y)=\varphi\bigl(F_{v\mid X}(v(x,y)\mid x)\bigr)$, where $F_{v\mid X}(\cdot\mid x)$ is the true conditional c.d.f.~of $v(X,Y)$ given $X=x$.
    \item \textit{Weight:} $w(\cdot)\equiv 1$.
\end{itemize}

\paragraph*{Marginal coverage}
Since $w(\cdot)\equiv 1$, Theorem~\ref{theo: marginal general} guarantees marginal coverage for GLCP and DCP.

\begin{corollary}[Test-conditional miscoverage of GLCP and DCP]\label{corr: GLCP test}
    Assume that: (i) there exists a function $\delta_n(\varepsilon)$ such that, for every $\varepsilon>0$,
    \begin{gather*}
        \mathrm{Pr}\left(\sup_{x\in\mathcal X,\,y\in\mathcal Y}\left|\widehat F_{v\mid X}(v(x,y)\mid x)-F_{v\mid X}(v(x,y)\mid x)\right|>\varepsilon\right)\le \delta_n(\varepsilon)\,;
    \end{gather*}
    (ii) $\varphi:[0,1]\to[0,1]$ is $L_\varphi$-Lipschitz continuous for a positive constant $L_\varphi$; and (iii) $\varphi(U)$ has a density bounded within $[\underline{L},\overline{L}]$ for $U\sim\mathrm{Uniform}(0,1)$. Then there exists a constant $C$ such that for every $t\in\mathcal{X}$,
    \begin{align*}
        &\left|\mathrm{Pr}\left( Y_{n+1}\in\widehat{C}_{\rm GLCP}(X_{n+1})\mid X_{n+1}=t \right)-(1-\alpha)\right| \leq C \left\{ \delta_{\rm tr}+n^{-1/2}\log^{1/2}(n) \right\}\,,
    \end{align*}
    where $\delta_{\rm tr}=\inf_{\varepsilon>0}\{\varepsilon+\delta_n(\varepsilon)\}$ is the score-estimation error. 
    In addition, the same conclusion holds for DCP with $v(x,y)=y$, $\varphi(u)=|1/2-u|$, and $L_\varphi=1$.
\end{corollary}
\begin{proof}
    On the event $\bigl\{\sup_{x\in\mathcal X,\,y\in\mathcal Y}\bigl|\widehat F_{v\mid X}(v(x,y)\mid x)-F_{v\mid X}(v(x,y)\mid x)\bigr|\le \varepsilon\bigr\}$, 
    the Lipschitz property of $\varphi$ gives, uniformly over $x\in\mathcal X$ and $y\in\mathcal Y$,
    \begin{align*}
        |s(x,y)-s^\star(x,y)|
        =&\left|\varphi\bigl(\widehat F_{v\mid X}(v(x,y)\mid x)\bigr)-\varphi\bigl(F_{v\mid X}(v(x,y)\mid x)\bigr)\right|\\
        \le& L_\varphi\left|\widehat F_{v\mid X}(v(x,y)\mid x)-F_{v\mid X}(v(x,y)\mid x)\right|
        \le L_\varphi\varepsilon\,.
    \end{align*}
    Hence $\mathrm{Pr}\left(\sup_{x,y}|s(x,y)-s^\star(x,y)|>L_\varphi\varepsilon\right)\le \delta_n(\varepsilon)$.
    Because $\widehat F_{v\mid X}$ is trained independently of the calibration sample, the stability terms in Theorem~\ref{theo: main miscoverage} satisfy $\widetilde{\varepsilon}_n=\widetilde{\delta}_n=0$. Since $w(\cdot)\equiv 1$, we have $M_w=B_w=\sigma_w=1$, and therefore $\Gamma_n(w)=O\bigl(n^{-1/2}\log^{1/2}(n)\bigr)$.

    It remains to identify the intrinsic conditional-mismatch error term. For any fixed $x$, let $U_x=F_{v\mid X}(v(x,Y)\mid x)$. By the probability integral transformation, $U_x\sim{\rm Uniform}(0,1)$. Under the oracle score, $s^\star(x,Y)=\varphi(U_x)$, so the conditional distribution of $s^\star(X,Y)$ given $X=x$ is the distribution of $\varphi(U)$ with $U\sim{\rm Uniform}(0,1)$, which does not depend on $x$. Therefore, $Q(1-\alpha;F_{s^\star\mid X=t})=Q(1-\alpha;F_{\varphi(U)})$ for every $t\in\mathcal{X}$,
    and the intrinsic conditional-mismatch error is zero. 
    Since $\varphi$ takes values in $[0,1]$, we also have $0\le s^\star(X,Y)\le 1$ almost surely. 
    Substituting these facts into Theorem~\ref{theo: main miscoverage} and minimizing over $\varepsilon$ yields
    \begin{gather*}
        \left|\mathrm{Pr}\left( Y_{n+1}\in\widehat{C}_{\rm GLCP}(X_{n+1})\mid X_{n+1}=t \right)-(1-\alpha)\right|
        \leq C\left\{\delta_{\rm tr}+n^{-1/2}\log^{1/2}(n)\right\}\,.
    \end{gather*}
    The conclusion of DCP follows by taking $v(x,y)=y$ and $\varphi(u)=|1/2-u|$.
\end{proof}

\paragraph*{Analysis}
For GLCP, the score-estimation error is due to the conditional distribution estimation, summarized by $\delta_{\rm tr}$. The finite-sample calibration error is $O(n^{-1/2}\log^{1/2}(n))$. The intrinsic conditional-mismatch error is also zero: by the probability integral transformation, $F_{v\mid X}(v(X,Y)\mid X=t)$ is uniformly distributed on $[0,1]$ for every $t\in\mathcal{X}$, so the oracle score has the same conditional distribution across all covariate values. Therefore, the conditional-miscoverage rate of GLCP is $\delta_{\rm tr}+n^{-1/2}\log^{1/2}(n)$.
The same conclusion applies to DCP as a special case.

\subsection{Batch Multivalid Conformal Prediction (BatchGCP)}\label{sec: batchgcp}

\paragraph*{Details}
BatchGCP \citep{jung2023batch} is designed to achieve group-conditional coverage over a finite collection of possibly intersecting groups. It can be viewed as a direct extension of quantile regression, which takes as input an arbitrary non-conformity score $v(x,y)$ and a finite collection of subgroup indicators $\mathcal H$, solves a single convex minimization problem, and returns a score adjustment with group-conditional guarantees for each subgroup.

Let each $h\in\mathcal H$ be a map from $\mathcal X$ to $\{0,1\}$, with subgroup event $\{h(X_{n+1})=1\}$. BatchGCP fits $g(x;a)=a_0+\sum_{h\in\mathcal H}a_h h(x)$ by minimizing the pinball loss at level $1-\alpha$: 
\begin{gather*}
    \widehat{a}(Z)={\rm argmin}_{a=\{a_h:h\in\mathcal H\cup\{0\}\}}
    \sum_{i=1}^{n+1}\ell_{1-\alpha}\{v(X_i,Y_i),g(X_i;a)\}\,,
\end{gather*}
where $\ell_{1-\alpha}(s_1,s_2)=\{1-\alpha-\mathbbm{1}(s_2\geq s_1)\}(s_1-s_2)$.
It then conformalizes the adjusted residual score
\begin{gather*}
    s(x,y;Z)=v(x,y)-g(x;\widehat{a}(Z))\,.
\end{gather*}

The resulting prediction set $\widehat{C}_{\rm BatchGCP}(X_{n+1})$ is defined as
\begin{gather*}
    \left\{y:s(X_{n+1},y;Z^y)\leq Q\left(1-\alpha;\frac1{n+1}\left\{\sum_{i=1}^{n}\delta_{s(X_i,Y_i;Z^y)}+\delta_{\infty}\right\}\right)
    \right\}\,.
\end{gather*}
Thus, BatchGCP is exactly \eqref{eq: general Cvec} with $w(\cdot)\equiv 1$ and a data-dependent score $s(\cdot,\cdot;Z^y)$. 
Its conditional target is group-conditional coverage on events of the form $\{h(X_{n+1})=1\}$.

\paragraph*{Unified specification}
\begin{itemize}
    \item \textit{Score:} $s(x,y;Z)=v(x,y)-g\{x;\widehat a(Z)\}$.
    \item \textit{Oracle score:} $s^\star(x,y)=v(x,y)-g(x;a^\star)$, where
    \[
        a^\star={\rm argmin}_{a=\{a_h:h\in\mathcal H\cup\{0\}\}}
        \mathbb{E}\left[\ell_{1-\alpha}\{v(X_1,Y_1),g(X_1;a)\}\right].
    \]
    The oracle adjustment $g(\cdot;a^\star)$ is chosen so that the oracle score $s^\star(x,y)$ has the target $(1-\alpha)$-quantile within each relevant subgroup.
    \item \textit{Weight:} $w(\cdot)\equiv 1$.
\end{itemize}

\paragraph*{Marginal coverage}
Since the weight function $w(\cdot)\equiv 1$, Theorem~\ref{theo: marginal general} guarantees marginal coverage of BatchGCP.

\begin{corollary}[Group-conditional miscoverage of BatchGCP]\label{corr: BatchGCP test}
    Assume that $|Y_{n+1}|\le M$ almost surely for some constant $M>0$. Then there exists a constant $C>0$ such that
    \begin{align*}
        \left|\mathrm{Pr}\left( Y_{n+1}\in\widehat{C}_{\rm BatchGCP}(X_{n+1})\mid h(X_{n+1})=1 \right)-(1-\alpha)\right|
        \leq C\,n^{-1/2}\log^{1/2}(n).
    \end{align*}
\end{corollary}

\begin{proof}
    Since
    \begin{gather*}
        s(x,y;Z)-s^\star(x,y)=g(x;a^\star)-g(x;\widehat a(Z))\,,
    \end{gather*}
    the score-estimation error term in Theorem~\ref{theo: main miscoverage} is controlled by the estimation error of the finite-dimensional coefficient vector $\widehat a(Z)$. Because $\mathcal H$ is finite, estimating $a^\star$ is a finite-dimensional quantile-regression problem, so this contribution is of order $O(n^{-1/2}\log^{1/2}(n))$.

    Since $w(\cdot)\equiv 1$, the finite-sample calibration error satisfies
    $\Gamma_n(w)=O(n^{-1/2}\log^{1/2}(n))$. Finally, by the definition of $a^\star$, the oracle score $s^\star(x,y)=v(x,y)-g(x;a^\star)$ has conditional $(1-\alpha)$-quantile equal to $0$ on each relevant subgroup. Hence the intrinsic conditional-mismatch error is zero. Substituting these bounds into Theorem~\ref{theo: main miscoverage} gives
    \begin{gather*}
        \left|\mathrm{Pr}\left( Y_{n+1}\in\widehat{C}_{\rm BatchGCP}(X_{n+1})\mid h(X_{n+1})=1 \right)-(1-\alpha)\right| \leq C\,n^{-1/2}\log^{1/2}(n).
    \end{gather*}
\end{proof}

\paragraph*{Analysis}
BatchGCP differs from pre-trained-score methods such as CQR or GLCP because the subgroup adjustment $\widehat a(Z)$ is learned from the same data that are later conformalized. Therefore, the score-estimation term does not vanish automatically. However, since the adjustment is indexed by finitely many subgroup indicators, estimating $\widehat a(Z)$ behaves like a finite-dimensional quantile-regression problem and contributes the parametric-order error $O(n^{-1/2}\log^{1/2}(n))$.

Once the oracle adjustment $g(\cdot;a^\star)$ removes subgroup-level quantile differences, the intrinsic conditional-mismatch error is zero. The remaining error comes from estimating the finite-dimensional grouped adjustment and from the usual conformal calibration step, both of order $O(n^{-1/2}\log^{1/2}(n))$. Thus, the conditional-miscoverage error rate of BatchGCP is $n^{-1/2}\log^{1/2}(n)$. This characterization is less sharp than the original finite-sample group-validity result in \cite{jung2023batch}, but it makes transparent how BatchGCP fits into the same three-term decomposition as the other conditional methods.

\subsection{Conditional Calibration (CC)}\label{sec: exmaple CC}

\paragraph*{Details}
The original CC procedure introduced in \cite{gibbs2025conformal} solves a quantile regression problem. Let $v:\mathcal{X}\times\mathcal{Y}\to\mathbb{R}$ be a base score, and let $\mathcal{F}$ be a class of functions from $\mathcal{X}$ to $\mathbb{R}$. The  quantile regression estimator is
\[
\widehat Q_\alpha(\cdot;Z^y)=\operatorname*{argmin}_{\widehat Q\in\mathcal F}\sum_{i=1}^{n+1}\ell_{1-\alpha}\bigl(v(X_i,Y_i^y),\widehat Q(X_i)\bigr),
\]
where $\ell_{1-\alpha}(s_1,s_2)=\{1-\alpha-\mathbbm{1}(s_2\geq s_1)\}(s_1-s_2)$ is the pinball loss at level $1-\alpha$. Define $s(x,y;Z^y)=v(x,y)-\widehat Q_\alpha(x;Z^y)$, then the CC prediction set in Equation~(2.5) of \cite{gibbs2025conformal} is
\begin{gather}\label{eq:supp_CC_01}
    \left\{ y:v(X_{n+1},y)\leq\widehat Q_\alpha(X_{n+1};Z^y)\right\}=\left\{ y:s(X_{n+1},y;Z^y)\leq 0\right\}\,.
\end{gather}

Since $\widehat Q_\alpha(\cdot;Z)$ minimizes the pinball loss, the subgradient optimality condition gives, for any $f\in\mathcal F$,
\begin{align*}
    0\in&\partial_\epsilon\left[\sum_{i=1}^{n+1}\ell_{1-\alpha}\bigl(v(X_i,Y_i),\widehat Q_\alpha(X_i;Z)+\epsilon f(X_i)\bigr)\right]\Big|_{\epsilon=0}\\
    =&\Bigg\{\sum_{i:v(X_i,Y_i)\neq\widehat Q_\alpha(X_i;Z)}f(X_i)\bigl[\alpha-\mathbbm{1}\{v(X_i,Y_i)>\widehat Q_\alpha(X_i;Z)\}\bigr]\\
    &\qquad+\sum_{i:v(X_i,Y_i)=\widehat Q_\alpha(X_i;Z)}\tau_i f(X_i):\tau_i\in[\alpha-1,\alpha]\Bigg\}\,.
\end{align*}
Choosing $\tau_i^\star\in[\alpha-1,\alpha]$ for $i$ with $v(X_i,Y_i)=\widehat{Q}_\alpha(X_i;Z)$ so that the subgradient contains zero yields
\begin{align*}
\sum_{i=1}^{n+1}f(X_i)\bigl[\alpha-\mathbbm{1}\{v(X_i,Y_i)>\widehat Q_\alpha(X_i;Z)\}\bigr]&=\sum_{i:v(X_i,Y_i)=\widehat Q_\alpha(X_i;Z)}(\alpha-\tau_i^\star)f(X_i)\,,\\
\sum_{i=1}^{n+1}f(X_i)\bigl[\alpha-\mathbbm{1}\{s(X_i,Y_i;Z)>0\}\bigr]&=\sum_{i:s(X_i,Y_i;Z)=0}(\alpha-\tau_i^\star)f(X_i)\,.
\end{align*}
Taking $f\equiv 1$ (the constant direction) gives
\[
\frac1{n+1}\sum_{i=1}^{n+1}\mathbbm{1}\{s(X_i,Y_i;Z)\leq 0\}
=1-\alpha+\frac1{n+1}\sum_{i:s(X_i,Y_i;Z)=0}(\alpha-\tau_i^\star).
\]
Thus, up to the usual tie adjustment, the zero threshold coincides with the empirical $(1-\alpha)$-quantile of the scores $\{s(X_i,Y_i^y,Z^y\}_{i\in[n+1]}$. Therefore, the CC construction in \eqref{eq:supp_CC_01} can be written as
\begin{gather}
    \left\{ y:s(X_{n+1},y;Z^y)\leq Q\left(1-\alpha;\frac1{n+1}\left\{\sum_{i=1}^{n}\delta_{s(X_i,Y_i;Z^y)}+\delta_{\infty}\right\}\right)\right\}\,,\label{eq: CC general}
\end{gather}
which is exactly \eqref{eq: general Cvec} with $w(\cdot)\equiv 1$. 

\cite{gibbs2025conformal} considers a penalized version of CC. In Corollary~\ref{cor: cc L1 test}, we work with a finite-dimensional penalized quantile regression estimator so that the estimation error of the learned conditional quantile can be controlled directly. Specifically, let
\begin{gather*}
\mathcal F=\left\{f_\kappa:f_\kappa(x)=\sum_{i=1}^{d_0}\kappa_i\eta_i(x),\ \|\kappa\|_2\leq B\right\},\qquad \widehat Q_\alpha(x;Z)=f_{\widehat\kappa}(x),
\end{gather*}
where
\begin{gather*}
\widehat{\kappa}={\rm argmin}_{\{\kappa:\|\kappa\|_2\leq B\}}\left\{
\frac{1}{n}\sum_{i=1}^{n}\ell_{1-\alpha}\bigl(v(X_i,Y_i),f_\kappa(X_i)\bigr)+\lambda\|\kappa\|_2^2\right\}\,.
\end{gather*}
Then Equation~(3.1) of \cite{gibbs2025conformal} gives the prediction set
\begin{gather}
    \left\{ y:v(X_{n+1},y)\leq \widehat Q_\alpha(X_{n+1};Z^y) \right\}=\left\{ y:s(X_{n+1},y;Z^y)\leq 0 \right\}\,,\label{eq: CC penalized origin}
\end{gather}
where the centered score is defined as
\begin{gather*}
    s(x,y;Z)=v(x,y)-\widehat Q_\alpha(x;Z)\,.
\end{gather*}

In the unified analysis, we use the same centered score but replace the zero threshold in \eqref{eq: CC penalized origin} by the empirical conformal quantile in \eqref{eq: CC general}. This calibrated version retains the finite-sample marginal coverage guarantee from Theorem~\ref{theo: marginal general}, whereas the plug-in zero-threshold set need not. This is the CC version used throughout our unified analysis.

\paragraph*{Unified specification}
\begin{itemize}
    \item \textit{Score:} $s(x,y;Z)=v(x,y)-\widehat Q_\alpha(x;Z)$, where $\widehat Q_\alpha(\cdot;\cdot)$ is a fitted conditional $(1-\alpha)$-quantile estimator of the base score $v(X,Y)$ given $X$.
    \item \textit{Oracle score:} $s^\star(x,y)=v(x,y)-Q^\star_\alpha(x)$, where $Q^\star_\alpha(x)$ is the oracle conditional $(1-\alpha)$-quantile of $v(X,Y)$ given $X=x$.
    \item \textit{Weight:} $w(\cdot)\equiv 1$.
\end{itemize}

\paragraph*{Marginal coverage}
For the calibrated CC construction in \eqref{eq: CC general}, the weight is $w(\cdot)\equiv 1$. Hence Theorem~\ref{theo: marginal general} gives exact finite-sample marginal coverage.

The main i.i.d.~result for CC is the averaged conditional-miscoverage bound in Corollary~\ref{cor: cc L1 test}. This is the natural form of the unified theory for CC, since the relevant error is the global estimation error of $\widehat Q_\alpha$ rather than a pointwise bound at each fixed covariate value. The extension to covariate shift is given in Section~\ref{sec:supp_covariate_shift}.

\section{Additional Results on Conditional Methods under Covariate Shift}\label{sec:supp_covariate_shift}

This section supplements Section~\ref{sec: Extensions to covariate shift} of the main text by providing additional results for conditional conformal methods under covariate shift. We work within the same unified prediction-set framework \eqref{eq: general Cvec}, but replace the i.i.d.~data setting with the covariate-shift setting in Assumption~\ref{ass: covariate shift assumption}. Throughout this section, the data are
\begin{gather*}
    Z=\{(X_1,Y_1),\ldots,(X_n,Y_n),(X_{n+1},Y_{n+1})\},
\end{gather*}
with $\{(X_i,Y_i)\}_{i\in[n]}$ drawn i.i.d.\ from $P_{X,1}\times P_{Y\mid X}$ and $(X_{n+1},Y_{n+1})$ independently from $P_{X,2}\times P_{Y\mid X}$. The density ratio $r_X(\cdot)=dP_{X,2}/dP_{X,1}(\cdot)$ is assumed to be known.

\subsection{Averaged Conditional Miscoverage under Covariate Shift}\label{sec: averaged miscoverage shift}

We begin with the covariate-shift analogue of Theorem~\ref{theo: averaged miscoverage}. The result follows the same framework as in the main text, but the averaging measure and the weighted calibration distribution now need to account for the difference between $P_{X,1}$ and $P_{X,2}$.

For the conditioning variable $T=T(X)$, let $P_{T,1}$, $P_{T,2}$, and $P_{T,w}$ denote the distributions induced by $X\sim P_{X,1}$, $X\sim P_{X,2}$, and $X\sim P_{X,w}$, respectively, where $dP_{X,w}(x)=w(x)dP_{X,1}(x)$. Define
$r(\cdot)=dP_{T,2}/dP_{T,1}(\cdot)$ and $r_w(\cdot)=dP_{T,w}/dP_{T,1}(\cdot)$.

\begin{theorem}[Averaged conditional miscoverage under covariate shift]\label{theo: averaged miscoverage shift}
    Replace Assumption~\ref{ass: data_distribution_and_score_invariance} in Theorem~\ref{theo: averaged miscoverage} with Assumption~\ref{ass: covariate shift assumption}, and suppose that all other conditions in Theorem~\ref{theo: averaged miscoverage} hold. Let $T_0,T_0^\prime \sim P_{T,1}$ be independent. Then there exists a constant $C>0$ such that
    \begin{align*}
        &\mathbb{E}\left|\mathrm{Pr}\left( Y_{n+1}\in\widehat{C}(X_{n+1})\mid \phi(T) \right)-(1-\alpha)\right|\\
        \leq &C \left[ \widetilde{\delta}_n+\widetilde{\varepsilon}_n+\delta_n(\varepsilon)+\varepsilon\left\{ \|r\|_{P_{T,1},p/(p-1)}+\mathbb{E}\left( \|r_w\|_{P_{T,1},p/(p-1)} \right) \right\} \right]\\
        &+C \left( \sqrt{n^{-1}\mathrm{Pdim}(\mathcal{S})\log(n)}+\mathbb{E}\left[ \left\{ r_w(T_0)+r(T_0) \right\}r(T_0^{\prime})\left|q_{s^\star}(T_0)-q_{s^\star}(T_0^{\prime})\right| \right] \right)\,.
    \end{align*}
\end{theorem}

Compared with Theorem~\ref{theo: averaged miscoverage}, Theorem~\ref{theo: averaged miscoverage shift} contains an additional term $\varepsilon\|r\|_{P_{T,1},p/(p-1)}$, which arises because the averaged conditional target is evaluated under the test distribution $P_{T,2}$ rather than the calibration distribution $P_{T,1}$. When $w=r_X$, the ratio $r_w$ coincides with $r$, and hence $\|r_w\|_{P_{T,1},p/(p-1)}=\|r\|_{P_{T,1},p/(p-1)}$. The intrinsic conditional-mismatch error is also averaged with the density-ratio corrections $r_w$ and $r$. 

\subsection{Conditional Calibration (CC) under Covariate Shift}\label{sec: cc shift}

\paragraph*{Details}
Under covariate shift, the centered CC score is still defined by $s(x,y;Z)=v(x,y)-\widehat Q_\alpha(x;Z)$,
where $\widehat Q_\alpha(\cdot;Z)$ estimates the conditional $(1-\alpha)$-quantile of the base score $v(X,Y)$ given $X$. The difference from the i.i.d.~case is that calibration should now target the test covariate distribution. This is achieved by weighting the empirical score distribution by the density ratio $r_X=dP_{X,2}/dP_{X,1}$.

The weighted subgradient optimality condition, applied in the direction $f(x)=r_X(x)$, yields
\[
\frac{\sum_{i=1}^{n+1}r_X(X_i)\mathbbm{1}\{s(X_i,Y_i;Z)\leq 0\}}{\sum_{i=1}^{n+1}r_X(X_i)}
=1-\alpha+\frac{\sum_{i:s(X_i,Y_i;Z)=0}r_X(X_i)(\alpha-\tau_i^\star)}{\sum_{i=1}^{n+1}r_X(X_i)}.
\]
Let $\hat s_{(1)}\leq\cdots\leq\hat s_{(n+1)}$ be the order statistics of $\{s(X_i,Y_i;Z)\}_{i\in[n+1]}$, and let $\hat s_{(\hat n)}$ be the $(1-\alpha)$-quantile of distribution
$\{\sum_{i=1}^{n+1}r_X(X_i)\}^{-1}\sum_{i=1}^{n+1}r_X(X_i)\delta_{s(X_i,Y_i;Z)}$.
Following the same argument as in the i.i.d.\ case (Section~\ref{sec: exmaple CC}), the subgradient inequality implies $\hat s_{(\hat n)}\leq 0<\hat s_{(\hat n+1)}$, so with a slight tie adjustment, the prediction set $\bigl\{ y:v(X_{n+1},y)\leq \widehat Q_\alpha(X_{n+1};Z^y) \bigr\}=\left\{ y:s(X_{n+1},y;Z^y)\leq 0 \right\}$ can be rewritten as
\begin{align}\label{eq:supp_CC_02}
    \widehat{C}_{\rm CC}(X_{n+1})=\left\{ y:s(X_{n+1},y;Z)\leq Q\left(1-\alpha;\frac{\sum_{i=1}^{n+1}r_X(X_i)\delta_{s(X_i,Y_i;Z)}}{\sum_{i=1}^{n+1}r_X(X_i)}\right)\right\}\,,
\end{align}
which is exactly \eqref{eq: general Cvec} with $w(x)=r_X(x)$. When there are no ties at $0$, the threshold $0$ is exactly the $r_X$-weighted $(1-\alpha)$-quantile of the centered scores $\{s(X_i,Y_i;Z)\}_{i\in[n+1]}$.

\cite{gibbs2025conformal} also considers a covariate-shift version defined directly through zero thresholding over $s(X_{n+1},y;Z^y)$. For the averaged conditional-miscoverage analysis for CC under covariate shift, we employ the same centered score $s(x,y;Z)=v(x,y)-\widehat Q_\alpha(x;Z)$ but work with the explicitly calibrated version in \eqref{eq:supp_CC_02}. This calibrated construction is an instance of \eqref{eq: general Cvec} with $w(x)=r_X(x)$, and therefore satisfies marginal coverage guarantee under covariate shift by Theorem~\ref{theo: marginal general cshift}. This makes clear the distinction between the original zero-threshold construction and the weighted conformal calibration used in our unified analysis under covariate shift.

\paragraph*{Unified specification}
\begin{itemize}
    \item \textit{Score:} $s(x,y;Z)=v(x,y)-\widehat Q_\alpha(x;Z)$, where $\widehat Q_\alpha(\cdot;\cdot)$ is a fitted conditional $(1-\alpha)$-quantile estimator of the base score $v(X,Y)$ given $X$.
    \item \textit{Oracle score:} $s^\star(x,y)=v(x,y)-Q^\star_\alpha(x)$, where $Q^\star_\alpha(x)$ is the oracle conditional $(1-\alpha)$-quantile of $v(X,Y)$ given $X=x$.
    \item \textit{Weight:} $w(x)=r_X(x)$, the known density ratio.
\end{itemize}

\paragraph*{Marginal coverage}
Since $w(\cdot)=r_X(\cdot)$ satisfies the marginal-validity condition in Theorem~\ref{theo: marginal general cshift}, the finite-sample marginal coverage guarantee holds.

\begin{corollary}\label{cor: cc L1 test shift} 
    Suppose the conditions of Corollary~\ref{cor: cc L1 test} hold after replacing $P_X$ by $P_{X,1}$. Assume that the density ratio $r_X(\cdot)$ satisfies $\underline{M}\leq r_X(x)\leq \overline{M}$ for all $x\in\mathcal X$. Then there exists a constant $C>0$ such that
    \begin{align*}
        &\mathbb{E}\left|\mathrm{Pr}\left( Y_{n+1}\in\widehat{C}_{\rm CC}(X_{n+1})\mid X_{n+1} \right)-(1-\alpha)\right|\\
        \leq & C \left\{ \sqrt{\inf_{f\in\mathcal{F}}R(f)-R(Q_\alpha^\star)}+\{d_0\log(n)/n\}^{1/3} \right\}\,.
    \end{align*}  
\end{corollary}

\paragraph*{Analysis}
Corollary~\ref{cor: cc L1 test shift} gives the averaged conditional-miscoverage bound for CC under covariate shift. The density ratio $r_X$ corrects the calibration distribution toward the test covariate, and its boundedness keeps the rate of the averaged bound unchanged. As in the i.i.d.~case, the rate consists of the approximation error of the conditional quantile class,
$\sqrt{\inf_{f\in\mathcal F}R(f)-R(Q_\alpha^\star)}$, and the estimation error of the learned conditional quantile, $\{d_0\log(n)/n\}^{1/3}$.

\subsection{Weighted Conformal Prediction (WCP)}\label{supsec: WCP}

WCP \citep{tibshirani2019conformal} is the standard density-ratio-weighted conformal method under covariate shift. In this subsection, we show that WCP is a special case of the unified prediction set \eqref{eq: general Cvec} with $w(\cdot)=r_X(\cdot)$, and then use the conditional-miscoverage decomposition to derive its test-conditional miscoverage bound.

\paragraph*{Details}
Let $s(x,y)$ be a pre-trained score function that is independent of $Z$.
WCP uses the known density ratio $r_X(x)=dP_{X,2}/dP_{X,1}(x)$ to calibrate the test score against the weighted empirical distribution of the calibration scores:
\[
\widehat{C}_{\rm WCP}(X_{n+1})=\left\{ y:s(X_{n+1},y)\leq Q\left( 1-\alpha;\dfrac{\sum_{i=1}^{n}r_X(X_i)\delta_{s(X_i,Y_i)}+r_X(X_{n+1})\delta_{\infty}}{\sum_{i=1}^{n+1}r_X(X_i)} \right) \right\}\,,
\]
Thus, WCP is exactly \eqref{eq: general Cvec} with $s(x,y;Z)=s(x,y)$ and the weight function $w(x)=r_X(x)$.

\paragraph*{Unified specification}
\begin{itemize}
    \item \textit{Score:} $s(x,y;Z)=s(x,y)$, where $s(x,y)$ is a pre-trained score independent of $Z$.
    \item \textit{Oracle score:} $s^\star(x,y)$, the population target of $s(x,y)$.
    \item \textit{Weight:} $w(x)=r_X(x)$.
\end{itemize}

\paragraph*{Marginal coverage}
Since $w(\cdot)=r_X(\cdot)$ is deterministic and satisfies the marginal-validity condition in Theorem~\ref{theo: marginal general cshift}, WCP preserves marginal coverage under covariate shift up to the usual tie term $\mathbb{E}\{\widehat{w}_{\max}(Z)\}$. When $r_X(\cdot)$ is bounded and the weighted scores are almost surely distinct, this term is of order $n^{-1}$.

\begin{corollary}[Test-conditional miscoverage of WCP]\label{corr: WCP test}
    Assume that: (i) there exists a function $\delta_n(\varepsilon)$ such that, for every $\varepsilon>0$,
    \begin{gather*}
        \mathrm{Pr}\left(\sup_{x\in\mathcal X,\,y\in\mathcal Y}|s(x,y)-s^\star(x,y)|>\varepsilon\right)\le \delta_n(\varepsilon)\,;
    \end{gather*}
    (ii) there exists a constant $M>0$ such that $|s^\star(X_{n+1},Y_{n+1})|\le M$ almost surely; (iii) the distribution $F_{r_X\circ s^\star}$ has a density lower bounded by $\underline{L}>0$ on its support; (iv) for each $t\in\mathcal X$, the conditional distribution $F_{s^\star\mid X=t}$ has a density upper bounded by $\overline{L}<\infty$ on its support; and (v) the density ratio satisfies $0<\underline{M}\le r_X(x)\le \overline{M}<\infty$ for all $x\in\mathcal X$.
    Then there exists a constant $C$ such that, for every $t\in\mathcal X$,
    \begin{align*}
        &\left|\mathrm{Pr}\left( Y_{n+1}\in\widehat{C}_{\rm WCP}(X_{n+1})\mid X_{n+1}=t \right)-(1-\alpha)\right|\\
        \leq&C \left\{ \delta_{\rm tr}+n^{-1/2}\log^{1/2}(n)+\left|Q(1-\alpha;F_{s^\star\mid X=t})-Q\left(1-\alpha;F_{r_X\circ s^\star}\right)\right| \right\}\,,
    \end{align*}
    where $\delta_{\rm tr}=\inf_{\varepsilon>0}\{\varepsilon+\delta_n(\varepsilon)\}$ is the score-estimation error.
\end{corollary}

\begin{proof}
    Condition (i) gives the score-estimation error, which contributes $\delta_{\rm tr}=\inf_{\varepsilon>0}\{\varepsilon+\delta_n(\varepsilon)\}$ after optimizing over $\varepsilon$. Since WCP uses a pre-trained score, the score does not depend on the calibration sample, and therefore the stability terms satisfy $\widetilde{\varepsilon}_n=\widetilde{\delta}_n=0$.

    For the weight $w(x)=r_X(x)$, condition (v) implies $M_w\le \overline{M}$, $B_w=\mathbb{E}\{r_X(X_1)\}=1$, and $\sigma_w^2=\mathbb{E}\{r_X^2(X_1)\}\le \overline{M}^2$, where $X_1\sim P_{X,1}$. Hence the finite-sample calibration error satisfies
    $\Gamma_n(w)=O(n^{-1/2}\log^{1/2} n)$.

    Conditions (ii)--(iv) ensure Assumption~\ref{ass: oracle quantile regularity} for the oracle score, $F_{s^\star\mid X=t}$, and $F_{r_X\circ s^\star}$. Unlike methods whose oracle scores remove conditional heterogeneity in the score distribution, WCP with a generic pre-trained score does not necessarily eliminate the intrinsic conditional-mismatch error, which is
    $\left|Q(1-\alpha;F_{s^\star\mid X=t})-Q\left(1-\alpha;F_{r_X\circ s^\star}\right)\right|$.
    Substituting these three components into Theorem~\ref{theo: main miscoverage} under Assumption~\ref{ass: covariate shift assumption} yields the claimed bound.
\end{proof}

\paragraph*{Analysis}
Corollary~\ref{corr: WCP test} shows that the test-conditional miscoverage of WCP has three components. The score-estimation error is $\delta_{\rm tr}$, and the finite-sample weighted calibration error is of order $n^{-1/2}\log^{1/2}(n)$. The intrinsic conditional-mismatch error is
\[
    \left|Q(1-\alpha;F_{s^\star\mid X=t})-Q(1-\alpha;F_{r_X\circ s^\star})\right|,
\]
which measures the gap between the target oracle conditional quantile at $t$ and the oracle quantile induced by density-ratio weighting. Thus, the weight $r_X(\cdot)$ restores marginal validity under covariate shift, but small test-conditional miscoverage further requires this weighted oracle quantile to align with the target conditional quantile.

\section{Weighted SymmPI Framework}\label{sec: appendix weighted symmpi}

We develop a weighted version of the SymmPI framework \citep{dobriban2025symmpi}, extending the unweighted formulation in Section~\ref{sec: extension structured data} of the main text. The extension incorporates nonnegative ratio functions into the distributional invariance and equivariance, thereby allowing the framework to cover distributional shift under group actions, including the covariate-shift setting for exchangeable vector data as a special case.


We consider a general data object $Z\in\mathcal Z$, whose distribution belongs to a class $\mathcal P$, and a score map $V:\mathcal Z\to\widetilde{\mathcal Z}$. The SymmPI construction calibrates the test score against the collection of scores obtained from a symmetry group action on the full data object. Let $\mathcal G$ be a finite group equipped with a probability measure $P_{\mathcal G}$. We write $\rho:\mathcal G\times\mathcal Z\to\mathcal Z$ for the group action on the data space, and $\widetilde\rho:\mathcal G\times\widetilde{\mathcal Z}\to\widetilde{\mathcal Z}$ for the corresponding action on the score space. Thus, $\rho$ acts on the data object, while $\widetilde\rho$ describes the induced action on $V(Z)$.

\begin{example}[Exchangeable vector data]\label{example: exchangeable}
    Let $Z=\{(X_1,Y_1),\ldots,(X_{n+1},Y_{n+1})\}$ be the vector data, where $(X_i,Y_i)\in\mathcal X\times\mathcal Y$ for $i\in[n+1]$, and suppose that the $n+1$ pairs are exchangeable. Let $V(Z)=(S_1,\ldots,S_{n+1})$ be the score map with $S_i=s(X_i,Y_i;Z),\ i\in[n+1]$ for a score function $s:\mathcal X\times\mathcal Y\times\mathcal Z\to\mathbb R$. In this case, $\mathcal G$ is the permutation group on $[n+1]$, $\rho$ acts by permuting the labeled pairs, and $\widetilde\rho$ acts by applying the same permutation to the score vector $V(Z)$.
\end{example}

In the unweighted case, the SymmPI framework assumes distributional invariance,
$Z\overset{\mathrm d}{=}\rho(g,Z)$, together with score distribution equivariance,
$V(\rho(g,Z))\overset{\mathrm d}{=}\widetilde\rho(g,V(Z))$. Under distribution shift, these may no longer hold. We therefore introduce nonnegative ratio functions
$r_g:\mathcal Z\to\mathbb R_+$ and $\widetilde r_g:\widetilde{\mathcal Z}\to\mathbb R_+$. 
We say that $Z$ is $(\mathcal G,\rho)$-distributionally invariant with ratio $r_g$ if
\begin{align*}
    Z\overset{\mathrm d}{=}r_g\circ\rho(g,Z)\,.
\end{align*}
Similarly, the score map $V$ is $(\mathcal G,\rho,\widetilde\rho)$-distributionally equivariant with ratios $(r_g,\widetilde r_g)$ if
\begin{align*}
    V\bigl(r_g\circ\rho(g,Z)\bigr)
    \overset{\mathrm d}{=}
    \widetilde r_g\circ\widetilde\rho(g,V(Z))\,.
\end{align*}
Here, $r_g\circ\rho(g,Z)$ denotes the law of $\rho(g,Z)$ reweighted by the factor $r_g(\rho(g,Z))$, and $\widetilde r_g\circ\widetilde\rho(g,V(Z))$ is interpreted analogously on the score space. More explicitly, if $X\sim P_X$ has density $f_X$ and $r:\mathcal X\to\mathbb R_+$ is integrable with positive integral, then $r\circ X$ denotes a random variable with density
\begin{align*}
    \frac{r(x)f_X(x)}{\int_{\mathcal X}r(x')f_X(x')\,dx'}\,.
\end{align*}
When all ratio functions are identically one, the definitions reduce to the usual distributional invariance and equivariance conditions. More generally, the ratio functions allow the calibration step to account for the distributional change induced by the group action, rather than requiring $\rho(g,Z)$ to have exactly the same law as $Z$.

\begin{assumption}[Weighted distributional invariance and equivariance]\label{ass: Exchangeability}
    The data $Z$ are $(\mathcal{G},\rho)$-distributionally invariant weighted by ratio $r_g$, and the score map $V$ is $(\mathcal{G},\rho,\widetilde{\rho})$-distributionally equivariant weighted by ratios $(r_g,\widetilde{r}_g)$.
\end{assumption}

\begin{example}[Covariate-shift vector data]\label{example: nonexchangeable}
    Suppose $Z$ is as in Example~\ref{example: exchangeable}. Let $(X_i,Y_i)$, $i\in[n]$, be i.i.d.~from $P_{X,1}\times P_{Y\mid X}$, and let the test pair $(X_{n+1},Y_{n+1})$ be independent with law $P_{X,2}\times P_{Y\mid X}$. Assume that the density ratio $r_X=dP_{X,2}/dP_{X,1}$ is known. Let $V(Z)=(S_1,\ldots,S_{n+1},Z)$, where $S_i=s(X_i,Y_i;Z)$. The data are not exchangeable in the usual unweighted sense. Under the permutation action in Example~\ref{example: exchangeable}, the required ratio is
    \begin{align*}
        r_g(Z)=\frac{r_X(X_{n+1})}{r_X(X_{g(n+1)})}\,,
    \end{align*}
    and the corresponding score-space ratio is $\widetilde r_g(V(Z))\equiv r_g(Z)$. This recovers the weighted conformal prediction \citep{tibshirani2019conformal} under covariate shift.

    We verify this ratio from the weighted SymmPI perspective. The calibration pairs are drawn from $P_{X,1}\times P_{Y\mid X}$, whereas the test pair is drawn from $P_{X,2}\times P_{Y\mid X}$. Denote the corresponding densities by $f_{X,1}$, $f_{X,2}$, and $f_{Y\mid X}$. For a fixed realization $z_0=((x_1,y_1),\ldots,(x_{n+1},y_{n+1}))$, the joint density of $Z$ at $z_0$ is
    \begin{align*}
        &f_{X,2}(x_{n+1})f_{Y\mid X}(y_{n+1}\mid x_{n+1})\prod_{i=1}^{n}f_{X,1}(x_i)f_{Y\mid X}(y_i\mid x_i)\\
        =& r_X(x_{n+1})\prod_{i=1}^{n+1}f_{X,1}(x_i)f_{Y\mid X}(y_i\mid x_i)\,.
    \end{align*}
    Under the permuted data object $\rho(g,Z)$, the observation in the test position is indexed by $g(n+1)$, so the corresponding factor is $r_X(x_{g(n+1)})$. Hence
    \begin{align*}
        dZ\big|_{z_0} = \frac{r_X(x_{n+1})}{r_X(x_{g(n+1)})}\cdot d\rho(g,Z)\big|_{\rho(g,Z)=z_0}\,.
    \end{align*}
    Consequently, $Z$ satisfies Assumption~\ref{ass: Exchangeability} with ratios
    $r_g(Z)=r_X(X_{n+1})/r_X(X_{g(n+1)})$ and $\widetilde r_g\equiv r_g$.
\end{example}

We now define the conformal construction. In many structured-data settings, the full data object is only partially observed. Let $\Omega:\mathcal Z\to\mathcal O$ be the observation map, so that $\Omega(Z)$ denotes the observed component, and let $\psi:\widetilde{\mathcal Z}\to\mathbb R$ extract the score corresponding to the unobserved or test component. For a candidate completion $z$ with $\Omega(z)$ equal to the observed data, plausibility is assessed through the value of $\psi(V(z))$. The conformal threshold is then computed from the collection of transformed scores $\{\psi(\widetilde\rho(g,V(z))):g\in\mathcal G\}$, with the appropriate ratio weights.

Under Assumption~\ref{ass: Exchangeability}, the weighted distributional equivariance of $V$ induces a weighted calibration distribution for $\psi(V(z))$. Define
\begin{align*}
    \mathcal{G}_e=\{g\in\mathcal{G}:\psi(\widetilde{\rho}(g,V(z)))=\psi(V(z))\text{ for all }z\in\mathcal{Z}\}\,.
\end{align*}
This set contains the group elements that leave the extracted test score unchanged. As in conformal prediction, these elements are assigned an atom at $\infty$, so that the test score itself is not used as an ordinary calibration score. Specifically, $q(V(z);\alpha)$ is defined as the $(1-\alpha)$-quantile of the distribution
\begin{gather*}
    \frac{\sum_{g\in\mathcal{G}\setminus\mathcal{G}_e}\mathrm{Pr}(G=g)\widetilde{r}_g\bigl(\widetilde{\rho}(g,V(z))\bigr)\delta_{\psi(\widetilde{\rho}(g,V(z)))}+\sum_{g\in\mathcal{G}_e}\mathrm{Pr}(G=g)\widetilde{r}_g\bigl(\widetilde{\rho}(g,V(z))\bigr)\delta_\infty}{\sum_{g\in\mathcal{G}}\mathrm{Pr}(G=g)\widetilde{r}_g\bigl(\widetilde{\rho}(g,V(z))\bigr)}
\end{gather*}
where $G\sim P_{\mathcal G}$. For an observed value $z_{\rm obs}\in\mathcal O$, the resulting prediction set is
\begin{equation}\label{eq:conformal set general 0}
    \widehat{C}(z_{\rm obs})=\left\{ y:\psi\bigl(V(z^y)\bigr)\leq q(V(z^y);\alpha),\ \Omega(z^y)=z_{\rm obs}\right\}\,,
\end{equation}
where $z^y$ is any completion satisfying $\Omega(z^y)=z_{\rm obs}$. 
In the exchangeable vector-data case, this construction reduces to the usual conformal prediction set. There, the score map is $V(Z)=(s(X_1,Y_1;Z),\ldots,s(X_{n+1},Y_{n+1};Z))$, and the permutation action only changes which entry is treated as the test score. Hence the actual test score $s(X_{n+1},Y_{n+1};Z)$ is compared with the empirical quantile $Q(1-\alpha;(n+1)^{-1}(\sum_{i=1}^n\delta_{s(X_i,Y_i;Z)}+\delta_{\infty}))$. In the covariate-shift vector-data case, the same construction recovers the weighted conformal prediction set. Thus, the weighted conformal framework in the main text can be viewed as a vector-data instance of the weighted SymmPI formulation.

\begin{remark}[Role of the weighted SymmPI formulation]
    The weighted SymmPI formulation separates three ingredients that recur across applications: the symmetry encoded by $(\mathcal G,\rho,\widetilde\rho)$, the partial observation structure encoded by $\Omega$ and $\psi$, and the distributional correction encoded by $(r_g,\widetilde r_g)$. This separation is useful for structured data, where the relevant symmetry may be induced by graph relabeling, hierarchical block exchangeability, or other problem-specific transformations. Once these components are specified, the prediction set in \eqref{eq:conformal set general 0} and the conditional-coverage decomposition below apply in a common form. Thus, exchangeable vector data, covariate-shift vector data, graph data, and the two-layer hierarchical model can be treated within the same weighted symmetry-based formulation.
\end{remark}


The next theorem establishes marginal validity for the weighted SymmPI construction.

\begin{theorem}\label{theo: marginal SymmPI weighted}
    Suppose Assumption~\ref{ass: Exchangeability} holds. Then
    \begin{align*}
        1-\alpha\leq \mathrm{Pr}(Z\in \widehat{C}(\Omega(Z)))\leq 1-\alpha+\mathbb{E}\left\{\widehat{\gamma}_{\max}(Z)\right\}\,,
    \end{align*}
    where $\widehat{\gamma}_{\max}(z)=\sup_c\frac{\sum_{g\in\mathcal{G},\,\psi(\widetilde{\rho}(g,V(z)))=c}\mathrm{Pr}(G=g)\widetilde{r}_{g}(\widetilde{\rho}(g,V(z)))}{\sum_{g\in\mathcal{G}}\mathrm{Pr}(G=g)\widetilde{r}_{g}(\widetilde{\rho}(g,V(z)))}$.
\end{theorem}

The proof of Theorem~\ref{theo: marginal SymmPI weighted} follows the same argument as its unweighted counterpart in the main text, with the ratio $\widetilde{r}_g$ playing the role of the weight in the conformal calibration quantile.

We now analyze conditional coverage. Let $V^\star:\mathcal Z\to\widetilde{\mathcal Z}$ be an oracle counterpart of $V$, analogous to the oracle score $s^\star$ in the i.i.d.~data setting in the main text. Fix a conditioning event $\phi(t)$ associated with the test component. We impose the following condition.

\begin{assumption}[Score approximation for structured data]\label{ass: independence score weighted}
    For each $t\in\mathcal{T}$, assume that: (i) $\psi(V^\star(Z))$ is independent of $q(V^\star(Z);\alpha)$ conditional on $\phi(t)$; (ii) for any $\varepsilon>0$, defining
    \begin{align*}
        D_\varepsilon=\left\{\left|\psi(V(Z))-\psi(V^\star(Z))\right|\vee\left|q(V(Z);\alpha)-q(V^\star(Z);\alpha)\right|\leq\varepsilon\right\}\,,
    \end{align*}
    there exists a non-random function $\delta(\varepsilon)\in[0,1]$ such that $\mathrm{Pr}\left( D_\varepsilon\mid\phi(t) \right)\geq 1-\delta(\varepsilon)$.
\end{assumption}

\begin{remark}\label{sec: disc of ass independence score}
    Assumption~\ref{ass: independence score weighted} requires the oracle test score $\psi(V^\star(Z))$ to be independent of the oracle calibration threshold $q(V^\star(Z);\alpha)$ conditional on $\phi(t)$. This is the structured-data analogue of the independence condition in Theorem~\ref{theo: main miscoverage} for the i.i.d.~data, where $s^\star(X_{n+1},Y_{n+1})$ is assumed to be independent of $q^\star(Z;\alpha)$ given $\phi(t)$. In both cases, the condition separates the oracle test score from the calibration threshold. This requirement is natural in the present construction because the conformal quantile replaces the test score by an atom at $\infty$.

    This requirement can be weakened by conditioning on an auxiliary oracle object. For example, in the structured-data setting it is enough to assume that $\psi(V^\star(Z))$ is independent of $q(V^\star(Z);\alpha)$ given $(\phi(t),V^\star)$. 
    Similarly, in the i.i.d.~data setting it is enough to assume that $s^\star(X_{n+1},Y_{n+1})$ is independent of $q^\star(Z;\alpha)$ given $\phi(t)$ and $s^\star$.
    The same refinement applies to the approximation event $D_{\varepsilon}$; one may replace $\mathrm{Pr}(D_\varepsilon\mid\phi(t))\geq 1-\delta(\varepsilon)$ by
    \begin{align*}
        \mathrm{Pr}(D_\varepsilon\mid\phi(t),V^\star)\geq 1-\delta(\varepsilon)\,.
    \end{align*}

    These refinements are useful when the oracle score or score map is pre-trained independently of the calibration sample. If $V^\star=V$ in the structured-data setting, the score-estimation error term $\delta(\varepsilon)+2L_t\varepsilon$ in Theorems~\ref{theo:symmpi conditional miscoverage} and~\ref{theo:symmpi conditional miscoverage weighted} vanishes. Analogously, if the score in the i.i.d.~data setting is already oracle, the score-estimation error term $\varepsilon+n\delta_n(\varepsilon)+\widetilde{\varepsilon}_n+\widetilde{\delta}_n$ in Theorem~\ref{theo: main miscoverage} is zero.
\end{remark}

For each $t\in\mathcal T$, let $F_{\psi(V^\star(\cdot))\mid\phi(t)}$ denote the conditional distribution of $\psi(V^\star(Z))$ given $\phi(t)$. Its $(1-\alpha)$-quantile is the oracle conditional threshold. To incorporate the weights induced by $r_g$, define the oracle weighted c.d.f.
\begin{align*}
    &F_{r\circ\psi(V^\star(\cdot))}(u)\\
    =&\sum_{g\in\mathcal{G}}\mathrm{Pr}(G=g)\frac{\mathbb{E}\left\{\widetilde{r}_g\bigl(\widetilde{\rho}(g,V^\star(Z))\bigr)\mathbbm{1}(\psi(\widetilde{\rho}(g,V^\star(Z)))\leq u)\right\}}{\mathbb{E}\left\{\widetilde{r}_g\bigl(\widetilde{\rho}(g,V^\star(Z))\bigr)\right\}}\,.
\end{align*}
The discrepancy between $Q(1-\alpha;F_{r\circ\psi(V^\star(\cdot))})$ and $Q(1-\alpha;F_{\psi(V^\star(\cdot))\mid\phi(t)})$ is the structure-specific conditional-mismatch error.

\begin{theorem}\label{theo:symmpi conditional miscoverage weighted}
    Suppose Assumptions~\ref{ass: Exchangeability} and~\ref{ass: independence score weighted} hold. If, for every $t\in\mathcal T$, $F_{\psi(V^\star(\cdot))\mid\phi(t)}$ is $L_t$-Lipschitz continuous with constant $L_t<\infty$, then
    \begin{align}
        &\left|\mathrm{Pr}\left( \psi(V(Z))\leq q(V(Z);\alpha)\mid\phi(t) \right)-(1-\alpha)\right|\notag\\
        \leq & \delta(\varepsilon)+2L_t\varepsilon
        +L_t\mathbb{E}\Bigl\{ \left|q(V^\star(Z);\alpha)-Q(1-\alpha;F_{r\circ\psi(V^\star(\cdot))})\right|\wedge L_t^{-1}\notag\\
        &\qquad\qquad\qquad+\left|Q(1-\alpha;F_{r\circ\psi(V^\star(\cdot))})-Q(1-\alpha;F_{\psi(V^\star(\cdot))\mid\phi(t)})\right|\mid\phi(t) \Bigr\}\,. \label{eq: general symmpi error weighted}
    \end{align}
\end{theorem}

Theorem~\ref{theo:symmpi conditional miscoverage weighted} extends the unweighted decomposition in Theorem~\ref{theo:symmpi conditional miscoverage} of the main text to the weighted setting. The bound separates conditional miscoverage into three terms. The first term, $\delta(\varepsilon)+2L_t\varepsilon$, is the score-estimation error. The second term is the finite-sample calibration error, measured by the gap between $q(V^\star(Z);\alpha)$ and $Q(1-\alpha;F_{r\circ\psi(V^\star(\cdot))})$. The third term is the structure-specific conditional-mismatch error, measured by the gap between $Q(1-\alpha;F_{r\circ\psi(V^\star(\cdot))})$ and the target oracle conditional quantile $Q(1-\alpha;F_{\psi(V^\star(\cdot))\mid\phi(t)})$. This last term is the analogue of the intrinsic conditional-mismatch error in Theorem~\ref{theo: main miscoverage}.
When all ratio functions are identically one, the weighted oracle quantile $Q(1-\alpha;F_{r\circ\psi(V^\star(\cdot))})$ reduces to the unweighted oracle quantile $Q(1-\alpha;F_{\psi(V^\star(\cdot))})$. In this case, Theorem~\ref{theo:symmpi conditional miscoverage weighted} reduces to Theorem~\ref{theo:symmpi conditional miscoverage} in the main text.

After this decomposition is established, the remaining analysis is structure-specific. One needs to control the finite-sample calibration term and the structure-specific conditional-mismatch term using the relevant symmetry structure. The graph application in Section~\ref{sec: community conditional} of the main text and the two-layer hierarchical model in Section~\ref{sec: two layer hierarchical} below illustrate this.

\section{Additional Applications}\label{sec:supp_add_applications}

This section presents three additional applications of the unified framework. We first introduce a generalized version of RLCP \citep{hore2025conformal} based on auxiliary weighting, including its covariate-shift extension. We then show how the same weighted calibration idea yields localized conformal $p$-values. Finally, we study a two-layer hierarchical model as another structured-data example under the SymmPI framework.

\subsection{Generalized Randomized Localized Conformal Prediction (GRLCP)}\label{sec: grlcp details}

In this section, we extend RLCP \cite{hore2025conformal} to a broader auxiliary-weighting construction that can also accommodate covariate shift. We call this approach Generalized Randomized Localized Conformal Prediction (GRLCP). We work under the covariate-shift setting in Assumption~\ref{ass: covariate shift assumption}: the calibration data $(X_1,Y_1),\ldots,(X_n,Y_n)$ are drawn i.i.d.~from $P_{X,1}\times P_{Y\mid X}$, while the test pair $(X_{n+1},Y_{n+1})$ is drawn independently from $P_{X,2}\times P_{Y\mid X}$. Let $f_{X,j}$ denote the density of $P_{X,j}$ for $j=1,2$, and assume that the density ratio $r_X(x)=f_{X,2}(x)/f_{X,1}(x)$ is known. Let $s:\mathcal X\times\mathcal Y\to\mathbb R$ be a pre-trained score function.

The construction is based on a random weight function induced by an auxiliary covariate. Let $k:\mathcal X\times\mathcal X\to\mathbb R_+$ be a fixed function, independent of the data, satisfying
\begin{gather*}
    k_{\mathcal X}(x_2)=\int_{\mathcal X} k(x_1,x_2)\,dx_1<\infty,\qquad x_2\in\mathcal X.
\end{gather*}
Given $X_{n+1}$, draw an auxiliary covariate $\widetilde X$ from the distribution with density proportional to $k(\cdot, X_{n+1})$. Specifically, the density of $\widetilde{X}$ at $\widetilde{x}$ is 
\begin{gather*}
    \dfrac{k(\widetilde{x},X_{n+1})}{k_\mathcal{X}(X_{n+1})}\,.
\end{gather*}
The joint density of $(X_{n+1},\widetilde X)$ at $(x,\widetilde x)$ is therefore
\begin{gather*}
    f_{X,2}(x)\frac{k(\widetilde x,x)}{k_{\mathcal X}(x)}.
\end{gather*}
Conditioning on $\widetilde X=\widetilde x$, the density of $X_{n+1}$ is proportional to $f_{X,2}(x)k(\widetilde x,x)/k_{\mathcal X}(x)$. Relative to the calibration covariate density $f_{X,1}$, this induces the weight
\begin{gather*}
    w(x)=\frac{f_{X,2}(x)}{f_{X,1}(x)}\frac{k(\widetilde X,x)}{k_{\mathcal X}(x)} = r_X(x)\frac{k(\widetilde X,x)}{k_{\mathcal X}(x)}.
\end{gather*}
Thus GRLCP is the weighted conformal set \eqref{eq: general Cvec} with pre-trained score $s(x,y)$ and random weight function $w(\cdot)$:
\begin{gather*}
    \widehat{C}_{\rm GRLCP}(X_{n+1})=\left\{ y:s(X_{n+1},y)\leq Q\left( 1-\alpha;\dfrac{\sum_{i=1}^{n}w(X_i)\delta_{s(X_i,Y_i)}+w(X_{n+1})\delta_{\infty}}{\sum_{i=1}^{n+1}w(X_i)} \right) \right\}.
\end{gather*}
When $r_X\equiv 1$ and $k$ is chosen as the kernel used in RLCP, this construction reduces to RLCP.

We next verify marginal coverage. For $X\sim P_{X,1}$ independent of $w$, define
\begin{gather*}
    b(\widetilde x)
    =
    \int_{\mathcal X} f_{X,2}(u)\frac{k(\widetilde x,u)}{k_{\mathcal X}(u)}\,du,
    \qquad
    \mathbb{E}\{w(X)\mid w\}=b(\widetilde X),
\end{gather*}
and let $\bar w(x)=w(x)/b(\widetilde X)$. 
Since $\bar w(x)$ is obtained from $w(x)$ by dividing by the same normalization factor $b(\widetilde X)$ for all $x$, the normalized weights used in the conformal quantile are unchanged, and hence $d_Z(w,\bar w)=0$.
Moreover, for each fixed $x\in\mathcal X$,
\begin{align*}
    \mathbb{E}\{\bar{w}(x)\}&=\mathbb{E}\left\{ \dfrac{r_X(x)k(\widetilde{X},x)}{k_\mathcal{X}(x)b(\widetilde{X})} \right\}=\dfrac{r_X(x)}{k_\mathcal{X}(x)}\mathbb{E}\left[ \mathbb{E}\left\{ \dfrac{k(\widetilde{X},x)}{b(\widetilde{X})}\mid X_{n+1} \right\} \right]\\
    &=\dfrac{r_X(x)}{k_\mathcal{X}(x)}\mathbb{E} \left\{ \int_{\widetilde{x}\in\mathcal{X}}\dfrac{k(\widetilde{x},x)k(\widetilde{x},X_{n+1})}{b(\widetilde{x})k_\mathcal{X}(X_{n+1})}d\widetilde{x} \right\}\\
    &=\dfrac{r_X(x)}{k_\mathcal{X}(x)}\int_{x_{n+1}\in\mathcal{X}}f_{X,2}(x_{n+1}) \int_{\widetilde{x}\in\mathcal{X}}\dfrac{k(\widetilde{x},x)k(\widetilde{x},x_{n+1})}{b(\widetilde{x})k_\mathcal{X}(x_{n+1})}d\widetilde{x}dx_{n+1}\\
    &=\dfrac{r_X(x)}{k_\mathcal{X}(x)}\int_{\widetilde{x}\in\mathcal{X}}\dfrac{k(\widetilde{x},x)}{b(\widetilde{x})}\int_{x_{n+1}\in\mathcal{X}}f_{X,2}(x_{n+1})\frac{k(\widetilde{x},x_{n+1})}{k_\mathcal{X}(x_{n+1})} dx_{n+1}d\widetilde{x}\\
    &=\dfrac{r_X(x)}{k_\mathcal{X}(x)}\int_{\widetilde{x}\in\mathcal{X}}k(\widetilde{x},x)d\widetilde{x}=r_X(x)\,.
\end{align*}
Hence $\bar w\in\mathcal W_X$, and Theorem~\ref{theo: marginal general cshift} yields the finite-sample marginal coverage guarantee for GRLCP under covariate shift.

We now study test-conditional coverage. Since the score is pre-trained, the oracle score is $s^\star(x,y)=s(x,y)$, so the score-estimation and stability terms in Theorem~\ref{theo: main miscoverage} vanish. The remaining terms are the finite-sample calibration error induced by the random weights and the intrinsic conditional-mismatch error. The following corollary provides a specialization when $r_X\equiv 1$, which is also the setting of RLCP.

\begin{corollary}[Test-conditional miscoverage of GRLCP]\label{corr: grlcp}
    Assume $r_X\equiv 1$. In addition, assume that: (i) there exist positive constants $M_{s^\star}$ and $\underline L_{s^\star}$ such that $|s^\star(X,Y)|\leq M_{s^\star}$ almost surely; (ii) the distribution $F_{w\circ s^\star}$ has density lower bounded by $\underline L_{s^\star}$ on $[-M_{s^\star},M_{s^\star}]$; (iii) $F_{s^\star\mid X=t}$ is $L$-Lipschitz; and (iv) $w(X)\leq M_w$ almost surely. Define
    $B_w = \mathbb{E}\{w(X_1)\mid w\}$, $\sigma_w^2 = \mathbb{E}\{w^2(X_1)\mid w\}$, where $\sigma_{w}B_{w}^{-1}n^{-1/2}\log^{1/2}(n)=O(a_n)$ and $M_{w}B_{w}^{-1}n^{-1}\log^{1/2}(n)=O(b_n)$ satisfying $a_n=o(1),b_n=o(1)$.
    Then, there exist constants $C > 0$ such that for every $t\in\mathcal{X}$,
    \begin{align}
        \notag&\left|\mathrm{Pr}\left( Y_{n+1}\in\widehat{C}_{\rm GRLCP}(X_{n+1})\mid X_{n+1}=t \right)-(1-\alpha)\right|\\
        \leq &C \left[ a_n+b_n+\mathbb{E}\left\{ B_w^{-1}w(X)\left|Q(1-\alpha;F_{s^\star\mid X=t})-Q(1-\alpha;F_{w\circ s^\star})\right|\mid X_{n+1}=t \right\} \right]\,.\notag
    \end{align}
\end{corollary}
\begin{proof}
    Since GRLCP uses a pre-trained score, $s(x,y)=s^\star(x,y)$. Hence the score-estimation and stability terms in Theorem~\ref{theo: main miscoverage} are zero. The assumed bounds on $\sigma_wB_w^{-1}$ and $M_wB_w^{-1}$ imply that the finite-sample calibration term satisfies $\Gamma_n(w)=O(a_n+b_n)$. 
    The remaining term in Theorem~\ref{theo: main miscoverage} is the intrinsic conditional-mismatch contribution
    \begin{gather*}
        \mathbb{E}\left\{ B_w^{-1}w(X)\left|Q(1-\alpha;F_{s^\star\mid X=t})-Q(1-\alpha;F_{w\circ s^\star})\right|\mid X_{n+1}=t \right\}.
    \end{gather*}
    Substituting these components into Theorem~\ref{theo: main miscoverage} gives the result.
\end{proof}

Corollary~\ref{corr: grlcp} shows that the test-conditional miscoverage error of GRLCP has two remaining components, because the pre-trained score makes the score-estimation term vanish. The finite-sample calibration error is controlled by $a_n+b_n$, which reflects the dispersion of the random weights. If, for a given $\widetilde X$, the function $k(\widetilde X,\cdot)$ concentrates too much mass on a small region, then $\sigma_w$ or $M_w$ may be large relative to $B_w$, preventing this term from vanishing. The intrinsic conditional-mismatch term is
\begin{gather*}
    \mathbb{E}\left\{B_w^{-1}w(X)\left|Q(1-\alpha;F_{s^\star\mid X=t})-Q(1-\alpha;F_{w\circ s^\star})\right| \mid X_{n+1}=t \right\}.
\end{gather*}
This term is small when the auxiliary weighting scheme places more mass on calibration covariates whose oracle conditional quantiles are close to the target oracle conditional quantile at $t$.

\subsection{Extension of Localized Conformal $p$-values}\label{sec: app localized p}

This section extends the weighted conformal construction from prediction sets to conformal $p$-values. Conformal $p$-values, introduced by \cite{vovk2005algorithmic} and discussed in \cite{shafer2008tutorial}, are based on the relative rank of conformity scores. 
Here we view them as another consequence of the same weighted conformal framework used for prediction sets.
This perspective is useful in many localized settings such as outlier detection \cite{bates2023testing}, two-sample testing \cite{hu2024two}, and conditional testing \cite{wu2025conditional}.

In conditional inference, the key question is whether the $p$-value adapts to the local test distribution, rather than merely satisfying marginal validity. The localized conformal $p$-values of \cite{wu2025conditional} are one representative example. Our goal is to show that the weighted calibration used for prediction sets gives an analogous construction for $p$-values. To cover both vector data and structured data, we formulate the $p$-value directly under the weighted SymmPI framework.

Following the notation in Section~\ref{sec: appendix weighted symmpi}, for an observed sample $z_{\rm obs}$ and any completion $z$ satisfying $\Omega(z)=z_{\rm obs}$, define the conformal $p$-value as
\begin{gather}
    \widehat{p}(z)=\mathrm{Pr}\left( \psi(V(z))<\psi(\widetilde{r}_{G}\circ\widetilde{\rho}(G,V(z))) \right)\,,\label{eq: general p value}
\end{gather}
where the probability is taken over $G\sim P_{\mathcal G}$. The validity target is therefore the tail inequality $\mathrm{Pr}(\widehat p(Z)\leq\alpha)\leq\alpha$, rather than a coverage statement for a prediction set.

The next theorem establishes marginal validity of \eqref{eq: general p value}. It also shows that, in the i.i.d.~data setting, the conditional deviation of the $p$-value from level $\alpha$ is exactly the same as the conditional miscoverage of the corresponding weighted conformal prediction set.

\begin{theorem}\label{theo: valid p value}
    Suppose Assumption~\ref{ass: Exchangeability} holds. Then $\widehat p$ is valid in the sense that
    \[
        \mathrm{Pr}(\widehat p(Z)\leq\alpha)\leq\alpha.
    \]
    In addition, consider the i.i.d.~data setting of Section~\ref{sec: general framework}, and let $\widehat C(X_{n+1})$ be the conformal set in \eqref{eq: general Cvec} constructed from the same score $s(x,y;Z)$ and weight function $w(\cdot)$. Then, for any conditioning event $\phi(t)$,
    \begin{gather*}
        \left|\mathrm{Pr}(\widehat p(Z)\leq\alpha\mid \phi(t))-\alpha\right| = \left|\mathrm{Pr}\left(
        Y_{n+1}\in\widehat C(X_{n+1})\mid \phi(t)\right)-(1-\alpha)\right|.
    \end{gather*}
\end{theorem}

Therefore, any pointwise or averaged conditional-miscoverage bound for $\widehat C(X_{n+1})$ immediately yields the corresponding conditional bound for $\widehat p(Z)$.

\begin{example}[Vector data]
    Consider the vector-data construction under covariate shift in Section~\ref{sec:supp_covariate_shift}. If the weight function is $w(x)=r_X(x)$, then any score used by CQR, LCP, GLCP, or CC in Section~\ref{sec: example details} gives the weighted conformal $p$-value
    \begin{gather*}
        \widehat p(Z) = \frac{\sum_{i=1}^{n+1} w(X_i)\mathbbm{1}\{s(X_{n+1},Y_{n+1};Z)<s(X_i,Y_i;Z)\}
        }{\sum_{i=1}^{n+1}w(X_i)}.
    \end{gather*}
    Next consider the localized random-weighting construction in Section~\ref{sec: Extensions to covariate shift}. For a pre-trained score $s(x,y)$ and auxiliary covariate $\widetilde X$, take $w(x)=r_X(x)K(x,\widetilde X;h)$. When $r_X\equiv 1$, this reduces to the RLCP weighting scheme. The corresponding $p$-value is
    \begin{gather*}
        \widehat p(Z)
        =
        \frac{
        \sum_{i=1}^{n+1}
        r_X(X_i)K(X_i,\widetilde X;h)
        \mathbbm{1}\{s(X_{n+1},Y_{n+1};Z)<s(X_i,Y_i;Z)\}
        }{
        \sum_{i=1}^{n+1}r_X(X_i)K(X_i,\widetilde X;h)
        }.
    \end{gather*}
    This agrees with the weighted localized conformal $p$-value construction in \cite{wu2025conditional}, up to the additional de-randomization step for marginal validity.
\end{example}

\subsection{Two-layer Hierarchical Model Case}\label{sec: two layer hierarchical}

This section considers a two-layer hierarchical model, which provides another structured-data example under the SymmPI framework. This setting, also considered in \cite{dobriban2025symmpi}, differs from the i.i.d.~case because observations are grouped into branches generated from latent branch-level distributions. The conditioning event is indexed by the branch-level distribution together with the test covariate. Theorem~\ref{theo:symmpi conditional miscoverage} then yields the same three-term decomposition as before: score-estimation error, finite-sample calibration error, and a structure-specific conditional-mismatch error term.

Specifically, for the first layer, draw $P_1,\ldots,P_K$ independently from a distribution $\Pi$ on a distribution space $\mathcal P$. These distributions specify the branch-level subpopulations. Conditional on $P_k$, the second-layer observations in branch $k$ are i.i.d.~random variables
$Z_i^{(k)}=(X_i^{(k)},Y_i^{(k)}),\ i\in[N]$, with common distribution $P_k$. Let
$Z^{(k)}=(Z_1^{(k)},\ldots,Z_N^{(k)})$
denote the $k$-th block, and let $Z=(Z^{(1)},\ldots,Z^{(K)})$ denote the full data object, with total size $n=KN$. For a single index $a\in[n]$, write $a=(k-1)N+i$ for the unique $k\in[K]$ and $i\in[N]$; then the $a$-th observation is $Z_i^{(k)}$. Consider the $K$-th branch as the test branch, with $X^{(K)}_N$ being the test covariate and $Y^{(K)}_N$ unobserved.

Let $\mathcal G_n$ be the group of permutations on $[n]$, and let $\mathcal G=\Lambda_{K,N}\subset\mathcal G_n$ be the subgroup of permutations that map each block $Z^{(k)}$ bijectively onto another block $Z^{(k')}$. Since this example is unweighted, set $r_g\equiv \widetilde r_g\equiv 1$. Let $\mathcal Z$ denote the space of all possible data objects, and define the action $\rho:\mathcal G\times\mathcal Z\to\mathcal Z$ by
\begin{gather*}
    \rho(g,Z)=(Z_{g^{-1}(1)},\ldots,Z_{g^{-1}(n)})\,,
\end{gather*}
for $g\in\mathcal G$ and $Z\in\mathcal Z$, where $Z_a$ represents $Z_i^{(k)}$ for $a=(k-1)N+i$. As shown in \cite{dobriban2025symmpi}, the data object $Z$ is distributionally invariant under this group action $\rho$ on group $\mathcal{G}$.

Define the data space for each block $Z^{(k)}$ by $\mathcal{B}$, and consider a score function $s:\mathcal X\times\mathcal Y\times\mathcal B\times\mathcal Z\to\mathbb R$.
Assume that $s(x,y;z^{(k)},Z)$ is invariant under permutations within the block argument $z^{(k)}$ and under the corresponding permutations of the full data object $Z$. For $a=(k-1)N+i$, with $k\in[K]$ and $i\in[N]$, define
\begin{gather*}
    S_a=s(X_i^{(k)},Y_i^{(k)};Z^{(k)},Z),
\end{gather*}
and set $V(Z)=(S_1,\ldots,S_n)$ as the score map. Define the group action $\widetilde\rho:\mathcal G\times\mathbb R^n\to\mathbb R^n$ by
\begin{gather*}
    \widetilde\rho(g,(S_1,\ldots,S_n))=(S_{g^{-1}(1)},\ldots,S_{g^{-1}(n)})\,,
\end{gather*}
for $g\in\mathcal{G}$ and $\widetilde{Z}=(S_1,\ldots,S_n)\in\mathbb{R}^n$.
By the definition of $\mathcal G$ and the invariance of the score function, if $g^{-1}(a)=(k'-1)N+i'$ for some $k'\in[K]$ and $i'\in[N]$, then
\begin{gather*}
    s(X_i^{(k)},Y_i^{(k)};Z^{(k)},Z)
    =
    s(X_{i'}^{(k')},Y_{i'}^{(k')};Z^{(k')},\rho(g,Z))\,.
\end{gather*}
Thus the score map $V$ is distributionally equivariant under $\widetilde\rho$, with $\widetilde r_g\equiv 1$.

We take the test score to be the last coordinate of the score vector, $\psi(\widetilde Z)=S_n$ for $\widetilde Z=(S_1,\ldots,S_n)$. The observation map is
\begin{gather*}
    \Omega(Z)=\bigl(Z^{(1)},\ldots,Z^{(K-1)},(Z_1^{(K)},\ldots,Z_{N-1}^{(K)},X_N^{(K)})\bigr)\,.
\end{gather*}
For an observed value $z_{\rm obs}$, the conformal prediction set is
\begin{gather*}
    \widehat C(z_{\rm obs}) = \left\{y: \psi(V(Z^y))\leq q(V(Z^y);\alpha),\ \Omega(Z^y)=z_{\rm obs}\right\}\,,
\end{gather*}
where $q(\widetilde{Z};\alpha)=Q(1-\alpha; n^{-1}(\sum_{a=1}^{n-1}\delta_{S_a}+\delta_\infty))$ for $\widetilde{Z}=(S_1,\ldots,S_n)$ and $Z^y$ is obtained by replacing $Y_N^{(K)}$ with trial data $y$. 
For simplicity, we write $\widehat C(\Omega(Z))$ as $\widehat C_{\rm HI}(X_N^{(K)})$ in the following analysis. 
In this setting, conditioning only on $X_N^{(K)}=x$ gives the coverage probability
\begin{gather*}
    \mathrm{Pr}\left(Y_N^{(K)}\in\widehat{C}_{\rm HI}(X_N^{(K)})\mid X_N^{(K)}=x\right)\,.
\end{gather*}
This probability still averages the probability over the first-layer draw $P_K\sim\Pi$. Thus a better target in the hierarchical model conditions also on the branch-level distribution $P_K$ of the test block. We therefore take the test space to be $\mathcal T=\mathcal P\times\mathcal X$ and define the conditioning event
\begin{align*}
    \phi(p_K,x)=\{P_K=p_K,X_N^{(K)}=x\}\,.
\end{align*}
The corresponding hierarchical conditional coverage probability is
\begin{gather*}
    \mathrm{Pr}\left(Y_N^{(K)}\in\widehat{C}_{\rm HI}(X_N^{(K)})
    \mid X_N^{(K)}=x,\ P_K=p_K \right)\,.
\end{gather*}

To apply Theorem~\ref{theo:symmpi conditional miscoverage}, we specify the oracle score map. Conditional on the first-layer draws $P_1=p_1,\ldots,P_K=p_K$, the oracle score map $V^\star(Z)$ is block-specific: observations in block $k$ use the oracle score $s_k^\star$ associated with $P_k$. Thus,
\begin{align*}
    V^\star(Z)=( s^\star _1(X_1^{(1)},Y_1^{(1)}),\ldots, s^\star _1(X_N^{(1)},Y_N^{(1)}),\ldots, s^\star _K(X_1^{(K)},Y_1^{(K)}),\ldots, s^\star _K(X_N^{(K)},Y_N^{(K)}))\,.
\end{align*}
With this choice, the score-estimation error compares $s(x,y;Z^{(k)},Z)$ with $s_k^\star(x,y)$ within each block. The finite-sample calibration error compares the empirical calibration quantile with its population counterpart, and the structure-specific conditional-mismatch error compares that population quantile with the oracle hierarchical conditional quantile given $(P_K,X_N^{(K)})$.

We next state the stability condition used for the hierarchical score. It is the analogue of the uniform stability condition in the i.i.d.~setting, adapted to perturbations of both the full data object and the block.

\begin{assumption}[probabilistic uniform stability]\label{ass: stable score hierarchical}
    Let $Z_{1},Z_{2}\in\mathcal Z$ be two data objects obtained by replacing the $(k-1)N+i$-th element of $Z$ with $Z_{i,1}^{(k)}$ and $Z_{i,2}^{(k)}$, respectively. Let $Z^{(j,1)},Z^{(j,2)}\in\mathcal B$ be $j$-th blocks of $Z_{1}$ and $Z_{2}$, respectively. There exists $\widetilde{\delta}_n$ such that, conditional on any event $D_0=\{X_N^{(K)}=x,P_1=p_1,\ldots,P_K=p_K\}$,
    \begin{align*}
        &\mathrm{Pr}\left( \sup_{Z_{i,l}^{(k)}:k\in[K],i\in[N],l\in[2]}\sup_{x,y}|s(x,y;Z^{(k,1)},Z_{1})-s(x,y;Z^{(k,2)},Z_{2})|\leq \widetilde{\varepsilon}_n\mid D_0 \right)\\
        \geq &1-\widetilde{\delta}_n\,.
    \end{align*}
\end{assumption}
Define $\phi_k(p_k,x)=\{P_k=p_k,X_N^{(k)}=x\}$ as the counterpart of $\phi$. 
Let $F_{s_k^\star\mid\phi_k(p,x)}$ denote the conditional c.d.f.~of $s_k^\star(X_N^{(k)},Y_N^{(k)})$ for $(X_N^{(k)},Y_N^{(k)})\sim P_k=p$ given $X_N^{(k)}=x$, and let $F_{s_k^\star}$ denote its marginal c.d.f. 
Note that $F_{s_k^\star}$ is itself a random distribution function, since it depends on the random branch distribution $P_k\sim\Pi$. Averaging over $P_k\sim\Pi$, define
 $F_{s^\star}(u)=\mathbb{E}\{F_{s_k^\star}(u)\}$. Because the branch distributions $P_k$ are identically distributed, the averaged c.d.f. $F_{s^\star}$ does not depend on the particular index $k$. 
Let 
\begin{gather*}
    q_{s^\star}=Q(1-\alpha;F_{s^\star}),\qquad q_{s_k^\star}(p,x)=Q(1-\alpha;F_{s_k^\star\mid\phi_k(p,x)})\,.
\end{gather*}
Define the variance of $F_{s_k^\star}(q_{s^\star})$ by $\sigma_{s^\star}^2$, where the randomness is induced by the first-layer draw $P_k\sim\Pi$. 
\begin{remark}\label{remark: hier score quantile}
    If $Q(1-\alpha;F_{s_k^\star})$ is constant across $P_k$, then $F_{s_k^\star}$ and $F_{s^\star}$ share the same $(1-\alpha)$-quantile. Under the usual continuity condition at this quantile, this implies $F_{s_k^\star}(q_{s^\star})=1-\alpha$ for all $P_k$, and hence $\sigma_{s^\star}=0$. 
\end{remark}

Once the hierarchical construction is expressed in the SymmPI notation, Theorem~\ref{theo:symmpi conditional miscoverage} can be applied. Under stability and convergence conditions adapted to the hierarchical setting, we obtain the following specialization.
    
\begin{theorem}\label{theo: pointwise hierarchical}
    Suppose Assumption~\ref{ass: stable score hierarchical} holds. Assume there exists a function $\delta_n(\varepsilon)$ such that,
    \begin{gather*}
        \mathrm{Pr}\left( |s(x,y;Z^{(k)},Z)- s^\star _k(x,y)|>\varepsilon\mid P_1,\ldots,P_K \right)\leq \delta_n(\varepsilon)\qquad\text{for all }k\in[K]\,.
    \end{gather*}
    Assume further that $|s_k^\star(X_1^{(k)},Y_1^{(k)})|\leq M^*$ almost surely, $F_{s_k^\star\mid\phi_k(p,x)}$ is $L$-Lipschitz uniformly in $k,p,x$, and $F_{s_k^\star}$ has density lower bounded by $\underline L$ on its support. Then there exists a constant $C>0$ such that
    \begin{align*}
        &\left|\mathrm{Pr}\left( Y_N^{(K)}\in\widehat{C}_{\rm HI}(X_N^{(K)})\mid X_N^{(K)}=x, P_K=p_K \right)-(1-\alpha)\right|\\
        \leq&C\left\{ KN\delta_n(\varepsilon)+\widetilde{\delta}_n+\varepsilon+\widetilde{\varepsilon}_n+\dfrac{N}{\sqrt{K}}\wedge\sqrt{\dfrac{\log(K)}{N}}+\left|q_{s^\star}-q_{s_K^\star}(p_K,x)\right|+\dfrac{\sigma_{s^\star}}{\sqrt{K}} \right\}\,.
    \end{align*}
\end{theorem}

Theorem~\ref{theo: pointwise hierarchical} identifies the terms that a suitable hierarchical score should control. In particular, the structure-specific conditional-mismatch term contains
$\left|q_{s^\star}-q_{s_K^\star}(p_K,x)\right|+\sigma_{s^\star}/\sqrt{K}$.
Thus, Remark~\ref{remark: hier score quantile} indicates that it is desirable to construct the score so that the quantile $q_{s_k^\star}(p,x)$ is constant across branch-level distributions, and in this case the structure-specific conditional-mismatch term reduces to zero. We give two examples of such constructions. Let $\widehat Q_a^{(k)}(x)$ be an estimator of the conditional $(1-a)$-quantile of $Y$ given $X=x$ under branch $k$, and let $\widehat F_{Y\mid X}^{(k)}(y\mid x)$ be an estimator of the conditional c.d.f.~of $Y$ given $X=x$ under branch $k$. The score may be chosen as
\begin{align*}
    s(x,y;Z^{(k)},Z)&=\left|\widehat{F}_{Y\mid X}^{(k)}(y\mid x)-1/2\right| \quad \text{~or}\\
    s(x,y;Z^{(k)},Z)&=\left\{ y-\widehat{Q}_{1-\alpha/2}^{(k)}(x) \right\}\vee\left\{ \widehat{Q}_{\alpha/2}^{(k)}(x)-y \right\}\,.
\end{align*}
For the corresponding oracle scores, both $\left|q_{s^\star}-q_{s_K^\star}(p_K,x)\right|$ and $\sigma_{s^\star}$ are equal to zero.
Since the distribution of $KN$ data points is invariant under permutation $\Lambda_{K,N}$, and the distribution of $N$ data points from the $K$-th branch is also invariant under permutation on $N$ samples, a natural question arises: when should we use $KN$ data points rather than $N$ data points for conformal calibration. 

Even without calibrating on data from other branches, those data can still be used for score estimation. Hence the score-estimation error remains of the same order, up to at most a factor depending on $K$. Consequently, only the $NK^{-1/2}\wedge N^{-1/2}\log^{1/2}(K)$ term needs to be considered. When $N$ is large while $K$ is relatively small, $N^{-1/2}\log^{1/2}(K)$ is of the same order as $N^{-1/2}\log^{1/2}(N)$, achieving nearly parametric optimality; in this regime, the dominant term is governed by the score-estimation error, and calibrating with additional branch data provides almost no benefit. On the other hand, when $K$ is large and $N$ is small, $NK^{-1/2}$ can be significantly smaller than $N^{-1/2}\log^{1/2}(N)$. In such a scenario, calibrating using all branch data becomes beneficial.

\section{Technical Details and Discussions}\label{sec: app technical details}

This section collects technical details that support the main theory. We first discuss the marginal-validity condition for random weighting schemes, then show that pointwise score convergence implies the averaged quantile-approximation condition used in the averaged conditional miscoverage theory. We also include a practical implementation note for conditional-coverage-oriented model selection.

\subsection{Discussion of Marginal Validity}\label{sec: discussion of marginal}

In Theorem~\ref{theo: marginal general}, the upper bound is governed by the expectation of the maximum point mass
\begin{gather*}
    \widehat{w}_{\max}(Z)=\sup_{c\in\{s(X_i,Y_i;Z):1\leq i\leq n+1\}}\dfrac{\sum_{i=1}^{n+1}w(X_i)\mathbbm{1}\{s(X_i,Y_i;Z)=c\}}{\sum_{j=1}^{n+1}w(X_j)}\,,
\end{gather*}
which is the largest atom of the weighted empirical score distribution used for conformal calibration. In the no-tie case with $w\equiv 1$, this term is $1/(n+1)$, recovering established findings in split conformal prediction literature \cite{vovk2005algorithmic}.

When ties are present, several observations may have the same score value, and the maximum weight $\widehat{w}_{\rm max}(Z)$ can exceed $1/(n+1)$. This leads to a more conservative conformal prediction set whose coverage probability surpasses the nominal level $1-\alpha$. 
For instance, for the CC method \cite{gibbs2025conformal} construct with a finite dimensional RKHS $\mathcal{F}$, since the score function corresponds to the solution of a quantile regression, multiple ties may occur at $0$. Consequently, the associated $\widehat{w}_{\rm max}(Z)$ is constrained by $d_0/(n+1)$, where $d_0=\mathrm{Pdim}(\mathcal{F})$ is also the upper bound on the number of ties at $0$.

When $w$ is nontrivial, the maximum weight $\widehat{w}_{\rm max}(Z)$ is determined by the ratio of the largest to the smallest weights. Specifically, $\widehat{w}_{\rm max}(Z)\leq\max_{1\leq i\leq n+1}w(X_i)/\sum_{j=1}^{n+1}w(X_j)$. When $w(\cdot)$ is bounded within $[\underline{M},\overline{M}]$ where $0<\underline{M}\leq \overline{M}<\infty$, $\widehat{w}_{\rm max}(Z)$ is controlled by $\overline{M}/(n\underline{M})$. However, if $w(\cdot)$ is unbounded, $\widehat{w}_{\rm max}$ can become arbitrarily large, potentially leading to overly conservative conformal prediction sets. 

We next explain the first claim below Theorem~\ref{theo: marginal general}. For any possibly random $w$, define $\bar{w}(x)=w(x)/\int w(x) dP_X(x)$ as its normalized version. Then 
\begin{gather*}
    \dfrac{w(X_j)}{\sum_{i=1}^{n+1}w(X_i)} = \dfrac{\bar{w}(X_j)}{\sum_{i=1}^{n+1}\bar{w}(X_i)}\,.
\end{gather*}
Consequently,
\begin{gather*}
    \inf_{w_0\in\mathcal{W}_X}\mathbb{E}\left\{ d_Z(w,w_0) \right\}=\inf_{w_0\in\mathcal{W}_X}\mathbb{E}\left\{ d_Z(\bar{w},w_0) \right\}\,.
\end{gather*}
Since $\mathbb{E}\left\{ \bar{w}(X)\mid \bar{w} \right\}=1$, the remaining condition for $\bar{w}\in\mathcal{W}_X$ is exactly $\mathbb{E}\left\{ \bar{w}(x) \right\}=1$ for almost every $x\in\mathcal{X}$. 
Thus, verifying $\inf_{w_0\in\mathcal W_X}\mathbb E\{d_Z(w,w_0)\}=0$ reduces to checking whether the normalized random weight satisfies $\mathbb{E}\left\{ \bar{w}(x) \right\}=1$ for almost every $x\in\mathcal{X}$.

\paragraph*{BaseLCP} For BaseLCP, $\bar{w}(x)=K(x,X_{n+1};h)/K_X(X_{n+1})$. 
Let $x^\star\in\arg\max_x K_X(x)$. Then $K_X(X_{n+1})\leq K_X(x^\star)$ almost surely, and hence
\begin{align*}
    \mathbb{E}\{\bar{w}(x^*)\} = \mathbb{E}\left\{\frac{K(x^*,X_{n+1};h)}{K_X(X_{n+1})}\right\} \geq \mathbb{E}\left\{\frac{K(x^*,X_{n+1};h)}{K_X(x^*)}\right\} = 1\,,
\end{align*}
where the equality holds when $K_X(X_{n+1})=K_X(x^*)$ almost surely, and this is not a trivial condition to satisfy. 

\paragraph*{RLCP} For RLCP, $\bar{w}(x)=K(x,\widetilde{X};h)/K_X(\widetilde{X})$, where, conditional on $X_{n+1}$, the auxiliary covariate $\widetilde X$ is drawn from a distribution with density $K(\widetilde x,X_{n+1};h)/K_0$. Then, for each fixed $x$,
\begin{align*}
    \mathbb{E}\{\bar{w}(x)\} &= \mathbb{E} \left[ \mathbb{E}\left\{ \bar{w}(x)\mid X_{n+1} \right\} \right]=\mathbb{E} \left[ \mathbb{E}\left\{ \dfrac{K(x,\widetilde{X};h)}{K_X(\widetilde{X})}\mid X_{n+1} \right\} \right]\\
    &=\mathbb{E} \left\{ \int_{\widetilde{x}}\dfrac{K(x,\widetilde{x};h)K(\widetilde{x},X_{n+1};h)}{K_X(\widetilde{X})K_0}d\widetilde{x} \right\}\\
    &=\int_{\widetilde{x}}\dfrac{K(x,\widetilde{x};h)}{K_X(\widetilde{x})K_0}\mathbb{E}\left\{K(\widetilde{x},X_{n+1};h)\right\}d\widetilde{x} \\
    &=\int_{\widetilde{x}}\dfrac{K(x,\widetilde{x};h)}{K_0}d\widetilde{x}=1\,.
\end{align*}
Thus the randomized weight in RLCP satisfies the marginal coverage condition in Theorem~\ref{theo: marginal general}.

\subsection{Pointwise Convergence Implies Averaged Convergence}\label{sec: detail averaged miscoverage}

Theorem~\ref{theo: averaged miscoverage} is stated under an averaged quantile-approximation assumption. The next lemma shows that this assumption follows from pointwise convergence of the learned score under mild conditions.

\begin{lemma}\label{lem: pointwise implies lp}
    Suppose $P_{X,1}=P_{X,2}=P_X$ and $P_{T,1}=P_{T,2}=P_T$. 
    Let $Z_{n+1}=Z$ explicitly specify the dependence of $Z$ on $n$. 
    Assume that: (i) for almost every $(x,y)\in\mathcal X\times\mathcal Y$, $s(x,y;Z_{n+1})\to s^\star(x,y)$ in probability; (ii) there exists $\overline M<\infty$ such that $|s(x,y;Z_{n+1})|\vee |s^\star(x,y)|\leq \overline M$ for all $x,y$ almost surely; and (iii) for every $t\in\mathcal T$, the c.d.f.~$F_{s^\star\mid\phi(t)}$ is continuous on its support and has density lower bounded by $\underline L>0$ on that support.
    Then, for every $p\geq 1$,
    \begin{gather*}
        \|q_{s(\cdot,\cdot;Z_{n+1})}-q_{s^\star}\|_{P_T,p}\rightarrow 0 \qquad\text{in probability}\,.
    \end{gather*}
\end{lemma}

\begin{proof}
Write $q_n(t)=q_{s(\cdot,\cdot;Z_{n+1})}(t)$ and $ q^\star(t)=q_{s^\star}(t)$.
The proof contains two steps. We first show that $q_n(t)\to q^\star(t)$ in probability for each fixed $t$ in a full-measure set, and then integrate this conclusion over $P_T$, using the uniform boundedness from assumption (ii).

Let $(X^{(t)},Y^{(t)})\sim \mu_t$ be independent of all randomness in $Z_{n+1}$, where $\mu_t$ is the distribution of $(X_{n+1},Y_{n+1})$ conditional on $T_{n+1}=t$. 
We first show that $q_n(t)\to q^\star(t)$ in probability for this fixed $t$. Fix an arbitrary subsequence $\{n_m\}_{m\ge 1}$. For any $\varepsilon>0$, assumption (i) and dominated convergence under the probability measure $\mu_t$ give
\begin{align*}
    \int \mathrm{Pr}\left(\left|s(x,y;Z_{n_m+1})-s^\star(x,y)\right|>\varepsilon\right)\mu_t(dx,dy)\to 0\,.
\end{align*}
Equivalently,
\begin{align*}
    \mathrm{Pr}\left(\left|s(X^{(t)},Y^{(t)};Z_{n_m+1})-s^\star(X^{(t)},Y^{(t)})\right|>\varepsilon\right)\to 0\,.
\end{align*}
By the subsequence principle for convergence in probability, the subsequence $\{n_m\}$ admits a further subsequence $\{n_{m_k}\}$ such that
\begin{align*}
    s(X^{(t)},Y^{(t)};Z_{n_{m_k}+1})\to s^\star(X^{(t)},Y^{(t)})
\end{align*}
almost surely on the product space carrying $(X^{(t)},Y^{(t)})$ and the original data randomness. Fubini's theorem therefore yields an event $\Omega_t$ of probability one such that, for every $\omega\in\Omega_t$,
\begin{align*}
    s(x,y;Z_{n_{m_k}+1}(\omega))\to s^\star(x,y)
\end{align*}
for $\mu_t$-almost every $(x,y)$.

Fix $\omega\in\Omega_t$. For any $u\in[-\overline{M},\overline{M}]$, bounded convergence under $\mu_t$ implies
\begin{align*}
    F_t\bigl(u;Z_{n_{m_k}+1}(\omega)\bigr)\to F_t^\star(u)\,.
\end{align*}
Because $F_t^\star$ is continuous on the compact interval $[-\overline{M},\overline{M}]$, P\'olya's theorem yields
\begin{align*}
    \sup_{u\in[-\overline{M},\overline{M}]}\left|F_t\bigl(u;Z_{n_{m_k}+1}(\omega)\bigr)-F_t^\star(u)\right|\to 0\,.
\end{align*}
Since the density of $F_t^\star$ is lower bounded by $\underline L$ on $[-\overline{M},\overline{M}]$, the usual quantile perturbation bound gives
\begin{align*}
    \left|q_{n_{m_k}}(t)-q^\star(t)\right|
    \le \underline L^{-1}\sup_{u\in[-\overline{M},\overline{M}]}\left|F_t\bigl(u;Z_{n_{m_k}+1}(\omega)\bigr)-F_t^\star(u)\right|\,.
\end{align*}
Hence $q_{n_{m_k}}(t)\to q^\star(t)$ almost surely. Since the original subsequence $\{n_m\}$ was arbitrary, another application of the subsequence principle shows that $q_n(t)\to q^\star(t)$ in probability for $P_T$-almost every $t$.

We now turn to prove the averaged quantile-approximation result. By assumption (ii), both quantiles lie in $[-\overline{M},\overline{M}]$, so $\left|q_n(t)-q^\star(t)\right|^p\le (2\overline{M})^p$. For each fixed $t$ in the full-measure set above, boundedness and convergence in probability imply
\begin{align*}
    \mathbb{E}\left|q_n(t)-q^\star(t)\right|^p\to 0\,.
\end{align*}
Applying dominated convergence in $t$ and Fubini's theorem, we obtain
\begin{align*}
    \mathbb{E}\left\|q_{s(\cdot,\cdot;Z_{n+1})}-q_{s^\star}\right\|_{P_T,p}^p = \int \mathbb{E}\left|q_n(t)-q^\star(t)\right|^p\,dP_T(t)\to 0\,.
\end{align*}
Markov's inequality therefore gives
\begin{align*}
    \mathrm{Pr}\left(\left\|q_{s(\cdot,\cdot;Z_{n+1})}-q_{s^\star}\right\|_{P_T,p}>\varepsilon\right)\le \varepsilon^{-p} \mathbb{E}\left\|q_{s(\cdot,\cdot;Z_{n+1})}-q_{s^\star}\right\|_{P_T,p}^p \to 0\,,
\end{align*}
which proves the claimed convergence in probability. 
\end{proof}

\subsection{Efficient Implementation of Conditional-Coverage-Oriented Model Selection}\label{sec: efficient sel}

Motivated by Assumption~\ref{ass: stable score}, we provide a computational simplification of the proposed conditional-coverage-oriented model selection approach in Section~\ref{sec: model selection}. If replacing the last element of $Z$ changes $s(x,y;Z)$ by at most $O_p(\widetilde{\varepsilon}_n)$, then, for a fixed pair $(x_0,y_0)$, the quantity $s(x_0,y_0;Z^y)-q_s(Z^y;\alpha)$ varies little when the trial value $y$ changes. Consequently, the selection criterion $\mathcal{L}(s;Z^y)$ is nearly unchanged when evaluated on the calibration sample alone. Under uniform stability of the score \citep{feldman2018generalization}, the resulting discrepancy is typically of order $\widetilde{\varepsilon}_n+\widetilde{\delta}_n$, which is often $n^{-1}$.

Motivated by this observation, let $Z_{:n}$ denote the first $n$ elements of $Z$, and use $\mathcal{L}(s;Z_{:n})$ as a surrogate for $\{\mathcal{L}(s;Z^y):y\in\mathcal Y\}$. Define $\hat{s}=\arg\min_{s\in\mathcal{S}}\mathcal{L}(s;Z_{:n})$. The resulting efficient approximation to the conformal set is
\begin{gather}
    \widehat{C}_{\rm Eff-Sel}(X_{n+1})=\left\{ y:\hat{s}(X_{n+1},y;Z_{:n})\leq q_{\hat{s}}(Z_{:n};\alpha) \right\}\,.\label{eq: efficient conditional selection}
\end{gather}
Under the same stability condition, the marginal coverage gap incurred by $\widehat{C}_{\rm Eff\text{-}Sel}(X_{n+1})$ instead of the conformal set with exact selection is expected to be of order $\widetilde{\varepsilon}_n+\widetilde{\delta}_n$. We use this as a practical implementation guideline.

\section{Preliminary Lemmas}\label{sec:pre_lemmas}

This section collects several preliminary lemmas used in the proofs of the main results. Their proofs are deferred to Section~\ref{sec:proof_pre_lemmas}.


We first introduce a deterministic notion used in the empirical-quantile arguments. Let $F:\mathbb R\to\mathbb R$ be a continuous and nondecreasing function. Suppose there exist finite constants $a_0<b_0$ such that
$F(a_0)=\inf_{t\in\mathbb R}F(t)$ and $F(b_0)=\sup_{t\in\mathbb R}F(t)$.
For $s\in(a_0,b_0)$, we say that $F$ is \emph{sub-Lipschitz at $s$ with constant $\underline L>0$} if
\begin{gather*}
    |F(s^\prime)-F(s)|\geq \underline{L}|s^\prime-s|,\quad \forall s^\prime\in[a_0,b_0].
\end{gather*}

This condition gives local identifiability of the quantile: a uniform error bound for a c.d.f.~can then be translated into a bound for the corresponding quantile. The next lemma gives a sufficient condition that guarantees this property.

\begin{lemma}\label{lemma: sub Lipschitz weighted}
    Let $F:\mathbb R\to\mathbb R$ be a continuous and nondecreasing function. Suppose there exist finite constants $a_0<b_0$ such that
    $F(a_0)=\inf_{t\in\mathbb R}F(t)$ and $F(b_0)=\sup_{t\in\mathbb R}F(t)$.
    Let $s\in(a_0,b_0)$. If $F$ is continuously differentiable at $s$ with derivative $f(s)>0$, then there exists $\underline L>0$ such that $F$ is sub-Lipschitz at $s$ with constant $\underline L$.
\end{lemma}


We next state two weighted empirical-process bounds. The first is a weighted DKW-type inequality, used in the proof of Lemma~\ref{theo: weighted quantile error} to control the stochastic error of the weighted empirical distribution. Both lemmas are stated for a fixed weight function $a(\cdot)$. Consequently, in later theorem proofs involving random weight functions, these bounds should be applied conditionally on the realized random weight function.

\begin{lemma}\label{lemma: weighted dkw}
    Let $\{(X_i^{(0)},S_i^{(0)})\}_{i\in[n]}$ be i.i.d.~random pairs from a joint distribution $P^{(0)}$, where $X_i^{(0)}\in\mathcal{X}$ and $S_i^{(0)}\in\mathbb{R}$. Let $a:\mathcal X\to\mathbb R_+$ be a deterministic nonnegative weight function,  possibly depending on $n$ and satisfying $\sigma_a^2=\mathbb{E}[\{a(X_1^{(0)})\}^2]<\infty$. Define
    \begin{gather*}
        A_n^{(0)}(u) = \frac{1}{n}\sum_{i=1}^{n}a(X_i^{(0)})\mathbbm{1}\{S_i^{(0)}\leq u\}, \quad
        A^{(0)}(u) = \mathbb E\left\{a(X^{(0)})\mathbbm{1}\{S^{(0)}\leq u\}\right\}\,.
    \end{gather*}
    Then, for any $\varepsilon>0$, the following DKW-type bound holds:
    \begin{align*}
        &\mathrm{Pr}\left(\sup_{u\in\mathbb R}\left|A_n^{(0)}(u)-A^{(0)}(u)\right|>16\sigma_a\sqrt{\frac{\log n}{n}}+\varepsilon\right) \\
        \leq &
        \begin{cases}
            \exp\left\{-n\varepsilon^2/(8\sigma_a^2)\right\},
            & \varepsilon<2\sigma_a^2,\\
            \exp\left(-n\varepsilon/4\right),
            & \varepsilon\geq 2\sigma_a^2.
        \end{cases}
    \end{align*}
\end{lemma}

When the score ranges over a function class, the previous empirical-process bound must hold uniformly over that class. The next lemma provides the uniform weighted DKW-type bound used in Theorem~\ref{theo: sup quantile error}.

\begin{lemma}\label{lemma: sup weighted dkw}
    Let $X_1^{(0)},\ldots,X_n^{(0)}$ be i.i.d.~from $P_X^{(0)}$, supported on $\mathcal{X}$. Let $a:\mathcal X\to\mathbb R_+$ be a deterministic nonnegative weight function satisfying the conditions of Lemma~\ref{lemma: weighted dkw}, and assume that
    $0<a(X^{(0)})\leq M_a$ almost surely for a positive constant $M_a$.
    Let $\mathcal M$ be a class of real-valued functions on $\mathcal X$ with finite pseudo-dimension $\mathrm{Pdim}(\mathcal M)$. Define
    \begin{align*}
        A_{n,m}^{(0)}(u) &= \frac{1}{n}\sum_{i=1}^n
        a(X_i^{(0)})\mathbbm{1}\{m(X_i^{(0)})\leq u\},\\
        A_{n,m}^{(0),+}(u) &= \frac{1}{n}\sum_{i=1}^{n-1}
        a(X_i^{(0)})\mathbbm{1}\{m(X_i^{(0)})\leq u\},\\
        A_m^{(0)}(u) &= \mathbb E\left\{
        a(X^{(0)})\mathbbm{1}\{m(X^{(0)})\leq u\}\right\},
    \end{align*}
    where $X^{(0)}\sim P_X^{(0)}$ and $m\in\mathcal M$. Then, for any $\varepsilon>0$,
    \begin{align*}
        &\mathrm{Pr}\left(
        \sup_{m\in\mathcal M,\,u\in\mathbb R}
        \left|A_{n,m}^{(0)}(u)-A_m^{(0)}(u)\right|
        \leq 8M_a \sqrt{\frac{\mathrm{Pdim}(\mathcal M)\log n}{n}} +\varepsilon \right)\\
        \geq & 1-\exp\left\{-\frac{n\varepsilon^2}{2M_a^2}\right\},
    \end{align*}
    and
    \begin{align*}
        &\mathrm{Pr}\left(
        \sup_{m\in\mathcal M,\,u\in\mathbb R}
        \left|A_{n,m}^{(0),+}(u)-A_m^{(0)}(u)\right|
        \leq 9M_a \sqrt{\frac{\mathrm{Pdim}(\mathcal M)\log n}{n}} +\varepsilon \right)\\
        \geq & 1-\exp\left\{-\frac{n\varepsilon^2}{2M_a^2}\right\}\,.
    \end{align*}
\end{lemma}

The next two lemmas quantify the kernel terms used in the RLCP and GRLCP analyses. 

\begin{lemma}\label{lemma: kernel fun}
    Suppose $X$ is supported on $\mathcal X=[0,1]^d$ without loss of generality and has density bounded between $\underline L_1$ and $\overline L_1$. Assume 
    $K(\cdot,\cdot;h)$ is in the form of $K(x_1,x_2;h)=K_0\left( \|x_1-x_2\|_2/h \right)$, where $K_0(\cdot)$ is a bounded univariate kernel function that is symmetric around $0$, and satisfies: (\textit{a}) $K_0(u)$ decreases when $u$ increases for $u\geq 0$; (\textit{b}) $uK_0(u)$ decreases when $u$ increases for $u>1$; (\textit{c}) $\int_{0}^{\infty} u^dK_0(u)du<\infty$.
    Then, for any $\ell\geq 1$, there exist constants $0<\underline L_2\leq \overline L_2<\infty$ such that
    \begin{equation}
        \notag \underline{L}_2 \leq \mathbb{E}\left[ h^{-d}\left\{ K(X,x_0;h) \right\}^{\ell} \right] \leq \overline{L}_2\,.
    \end{equation}
    for any $x_0\in\mathcal{X}=[0,1]^{d}$.
\end{lemma}

\begin{lemma}\label{lemma: kernel fun bias}
    Suppose the conditions of Lemma~\ref{lemma: kernel fun} hold. Let $f:\mathcal X\to\mathbb R$ be $L$-Lipschitz continuous for a positive constant $L$. For $x_0\in\mathcal X$, let $\widetilde X$ be drawn from the distribution with density $K(x_0,\cdot;h)/\int_{\mathcal{X}}K(x_0,x;h)dx$.
    If $X$ satisfies the conditions in Lemma~\ref{lemma: kernel fun}, then there exists a constant $\overline L_3<\infty$ such that
    \begin{gather*}
        \mathbb{E} \left\{ \left|f(x_0)-f(X)\right|K(\widetilde{X},X;h) \right\}\leq \overline{L}_3h^{d+1}\,.
    \end{gather*}
\end{lemma}

We next state a lemma for conditional quantile estimation. It converts excess pinball risk into an $L_p$ error bound for the learned conditional quantile, which is the form required in the averaged conditional-miscoverage analysis for CC.

\begin{lemma}\label{lemma: quantile regression error}
    Let $(X,Y)\sim P_X\times P_{Y\mid X}$ and let $\alpha\in(0,1)$. Let
    $Q_\alpha(x)=\inf\{y:\mathrm{Pr}(Y\leq y\mid X=x)\geq 1-\alpha\}$ be the conditional quantile function.
    Assume that, conditional on $X=x$, the conditional distribution $Y$ is supported on $[-M_1(x),M_2(x)]$, and that its density $f_x(\cdot)$ is lower bounded by $\underline L(x)>0$ on this support. Suppose that, for some $p\in[1,2]$, $0<\|1/\underline{L}\|_{P_X,p/(2-p)}<\infty$. 
    Let $\widehat Q(\cdot)$ be any function satisfying $\widehat Q(x)\in[-M_1(x),M_2(x)]$ for all $x$. Then
    \begin{gather*}
        2\|1/\underline{L}\|_{P_X,p/(2-p)}\mathbb{E}\left\{ \ell_{1-\alpha}(Y,\widehat{Q}(X))-\ell_{1-\alpha}(Y,Q_\alpha(X)) \right\}\geq \|Q_\alpha-\widehat{Q}\|_{P_X,p}^2\,.
    \end{gather*}
    where $\ell_{1-\alpha}(s_1,s_2)=\{1-\alpha-\mathbbm{1}(s_2\geq s_1)\}(s_1-s_2)$ is the pinball loss.
\end{lemma}


Finally, the averaged conditional-miscoverage theorem is stated for a score class $\mathcal S$, whereas CC is naturally parameterized by a regression class $\mathcal F$. The next lemma transfers the pseudo-dimension bound from $\mathcal F$ to the induced score class, and is used to verify Assumption~\ref{ass: space complexity} for CC-type constructions.

\begin{lemma}\label{lem:pdim_Sprime}
    Let $\mathcal F$ be a class of real-valued functions on $\mathcal X$ with finite pseudo-dimension, and suppose that $f\in\mathcal F$ implies $-f\in\mathcal F$. For a fixed function $v:\mathcal X\times\mathcal Y\to\mathbb R$, define
    $\mathcal{S}=\{(x,y)\mapsto v(x,y)-f(x): f\in\mathcal{F}\}$ as a class of functions on $\mathcal{X}\times\mathcal{Y}$.
    Then $\mathrm{Pdim}(\mathcal{S}) \le \mathrm{Pdim}(\mathcal{F})$. 
\end{lemma}

\subsection{Extended Results on Empirical Quantile}\label{sec: extend result empirical quantile}

In this subsection, we provide some extended results on weighted empirical quantiles, which will be used in our theoretical analysis.

Let $\{(X_i^{(0)},S_i^{(0)})\}_{i=1}^n$ be i.i.d.~copies of a random pair $(X^{(0)},S^{(0)})$, where $X^{(0)}\in\mathcal X$ and $S^{(0)}$ is supported on $[-M_S,M_S]$. Let $a:\mathcal X\to\mathbb R_+$ be a deterministic nonnegative weight function, possibly depending on $n$. Define the weighted empirical distribution
\begin{gather*}
    P_{n,S}^{a} = \left\{\sum_{i=1}^n a(X_i^{(0)})\right\}^{-1}
    \left\{\sum_{i=1}^n a(X_i^{(0)})\delta_{S_i^{(0)}}\right\}\,,
\end{gather*}
with weighted empirical c.d.f.
\begin{gather*}
    F_n^a(u) = P_{n,S}^{a}(S^{(0)} \leq u) = \frac{\sum_{i=1}^n a(X_i^{(0)})\mathbbm{1}\{S_i^{(0)}\leq u\}}{\sum_{i=1}^n a(X_i^{(0)})}\,.
\end{gather*}
The corresponding population weighted c.d.f. is
\begin{gather*}
    F^a(u) = \frac{\mathbb E\{a(X^{(0)})\mathbbm{1}(S^{(0)}\leq u)\}}{\mathbb E\{a(X^{(0)})\}}\,.
\end{gather*}
We also define the version obtained by replacing the last score $S_n^{(0)}$ with $\infty$:
\begin{gather*}
    F_n^{a,+}(u) = \frac{\sum_{i=1}^{n-1}a(X_i^{(0)})\mathbbm{1}\{S_i^{(0)}\leq u\}}{\sum_{i=1}^n a(X_i^{(0)})}\,.
\end{gather*}

For $\alpha\in(0,1)$, let $\xi_\alpha$ be the population weighted quantile satisfying $F^a(\xi_\alpha)=\alpha$, and define the sample weighted quantiles
\begin{gather*}
    \widehat{\xi}_\alpha = \inf\{u:F_n^a(u)\geq\alpha\},
    \qquad
    \widehat{\xi}_\alpha^+ = \inf\{u:F_n^{a,+}(u)\geq\alpha\}\,.
\end{gather*}

\begin{assumption}\label{ass: weight function}
    The weight function $a$ satisfies: (i) $0<a(X^{(0)})\leq M_a$ almost surely for some positive constant $M_a$; and (ii) the map $u\mapsto\mathbb E\{a(X^{(0)})\mathbbm{1}(S^{(0)}\leq u)\}$ is continuously differentiable at $\xi_\alpha$ with positive derivative.
\end{assumption}

\begin{remark}
    Assumption~\ref{ass: weight function}(i) controls the variation of the weight function $a(\cdot)$. Assumption~\ref{ass: weight function}(ii) gives local identifiability of the population weighted quantile.
\end{remark}

Let $B_a=\mathbb E\{a(X^{(0)})\}$ and $\sigma_a^2=\mathbb E\{a^2(X^{(0)})\}.$
The following lemma bounds the expected difference between the empirical weighted quantile and its population counterpart.

\begin{lemma}\label{theo: weighted quantile error}
    Suppose $|S^{(0)}|\leq M_S$ almost surely for some positive constant $M_S$. If Assumption~\ref{ass: weight function} holds and $n\sigma_a^2/(2M_a^2)\geq1$, then there exists a constant $C>0$ such that both
    $\mathbb E|\widehat{\xi}_\alpha-\xi_\alpha|$ and
    $\mathbb E\{|\widehat{\xi}_\alpha^+-\xi_\alpha|\wedge(2M_S)\}$ are bounded by
    \begin{gather}
        C\left\{\exp\left(-\frac{nB_a^2}{8\sigma_a^2+2M_aB_a}\right) + \frac{\sigma_a\log^{1/2}(n)}{B_an^{1/2}} + \frac{(\sigma_a^2+M_aB_a)\log^{1/2}(n)}{B_a^2n} + \frac{M_a\sigma_a}{B_a^2n^{3/2}}\right\}\,.
        \label{eq: weighted quantile error bound}
    \end{gather}
\end{lemma}

\begin{remark}\label{remark: weighted quantile error remark}
    When $B_a$, $\sigma_a$, and $M_a$ are all of constant order, \eqref{eq: weighted quantile error bound} gives the rate $O(n^{-1/2}\log^{1/2}n)$. In the RLCP setting, if the kernel function $K(\cdot,\cdot;h)$ and the bandwidth $h$ satisfy the conditions in Lemma~\ref{lemma: kernel fun}, then $M_a=O(1)$, $B_a=O(h^d)$, and $\sigma_a=O(h^{d/2})$. Hence, if $nh^d\to\infty$, the bound in \eqref{eq: weighted quantile error bound} becomes $O((nh^d)^{-1/2}\log^{1/2}n)$ for RLCP.
\end{remark}

We next consider a family of real-valued functions $\mathcal M$ on $\mathcal X$. For $m\in\mathcal M$, define
\begin{align*}
    F_{n,m}^{a}(u) & = \frac{\sum_{i=1}^n a(X_i^{(0)})\mathbbm{1}\{m(X_i^{(0)})\leq u\}}{\sum_{i=1}^n a(X_i^{(0)})}, \\
    F_{n,m}^{a,+}(u) & = \frac{\sum_{i=1}^{n-1}a(X_i^{(0)})\mathbbm{1}\{m(X_i^{(0)})\leq u\}}{\sum_{i=1}^n a(X_i^{(0)})} \quad \text{~and} \\
    F_m^a(u) & = \frac{\mathbb E\{a(X^{(0)})\mathbbm{1}\{m(X^{(0)})\leq u\}\}}{\mathbb E\{a(X^{(0)})\}}\,.
\end{align*}
Note that $F_{n,m}^{a,+}(u)$ is defined by replacing $m(X_n^{(0)})$ of the weighted empirical distribution $F_{n,m}^{a}(u)$ with $\infty$.
Define the corresponding empirical and population quantiles as
\begin{gather*}
    \widehat{\xi}_{\alpha,m} = \inf\{u:F_{n,m}^a(u)\geq\alpha\}, ~~
    \widehat{\xi}_{\alpha,m}^+ = \inf\{u:F_{n,m}^{a,+}(u)\geq\alpha\}, ~~
    \xi_{\alpha,m} = \inf\{u:F_m^a(u)\geq\alpha\}\,.
\end{gather*}

The next assumption gives quantile identifiability and controls the complexity of $\mathcal M$.

\begin{assumption}\label{ass: function class pdim}
    (i) For every $m\in\mathcal M$, $F_m^a(u)$ is differentiable in an $\epsilon_0$-neighborhood of $\xi_{\alpha,m}$, with derivative lower bounded by $\underline L>0$; and (ii) the class $\mathcal M$ has finite pseudo-dimension, $\mathrm{Pdim}(\mathcal M)<\infty$.
\end{assumption}


The following lemma gives a uniform bound for the difference between the empirical quantiles and the population quantile over $\mathcal M$.

\begin{lemma}\label{theo: sup quantile error}
    Suppose that for every $m\in\mathcal M$, $|m(X^{(0)})|\leq M_m$ almost surely for some positive constant $M_m$. Assume that Assumption~\ref{ass: weight function}(i) and Assumption~\ref{ass: function class pdim} hold. Then there exists a constant $C>0$ such that
    \begin{gather*}
        \mathbb E\left\{\sup_{m\in\mathcal M}
        |\widehat{\xi}_{\alpha,m}-\xi_{\alpha,m}|\right\}
        \vee \mathbb E\left\{2M_m\wedge \sup_{m\in\mathcal M} |\widehat{\xi}_{\alpha,m}^+-\xi_{\alpha,m}| \right\} \leq \frac{CM_a}{B_a} \sqrt{\frac{\mathrm{Pdim}(\mathcal M)\log n}{n}}\,.
    \end{gather*}
\end{lemma}

Lemma~\ref{theo: sup quantile error} is a uniform version of Lemma~\ref{theo: weighted quantile error}. The variance term $\sigma_a^2$ is not retained after taking the supremum over $\mathcal M$; it is replaced by the uniform bound $M_a$. For a bounded weight function with $B_a$ bounded away from zero, the resulting rate is
$O(\{\mathrm{Pdim}(\mathcal M)\log n/n\}^{1/2})$, which aligns with Lemma~\ref{theo: weighted quantile error}.

\section{Proofs for Main Results}\label{sec:proof_main}

\subsection{Proof of Theorem~\ref{theo: marginal general} }
\begin{proof}
    Theorem~\ref{theo: marginal general} is a special case of Theorem~\ref{theo: marginal general cshift} with $P_{X,1}=P_{X,2}=P_X$. The conclusion therefore follows directly from Theorem~\ref{theo: marginal general cshift}; see Section~\ref{sec:supp_proof_theo_marginal general cshift} for its proof.
\end{proof}

\subsection{Proof of Theorem \ref{theo: main miscoverage}}
\begin{proof}
    Fix $t\in\mathcal{T}$. Let $Z_i^\prime$ be obtained from $Z$ by replacing the $i$-th and $(n+1)$-th elements with independent copies, while keeping all other elements unchanged. Then $Z\overset{\rm d}{=}Z_i^\prime$. Define
    \begin{gather*}
        D_{n,\varepsilon}(t)=\left\{ \sup_{1\leq i\leq n+1}|s(X_i,Y_i;Z_i^\prime)-s^\star(X_i,Y_i)|\leq\varepsilon \right\}\,.
    \end{gather*}
    For each $i\in[n+1]$, conditional on $\phi(t)$ and on the realized pair $(X_i,Y_i)$, the score is evaluated at a fixed input pair and at a data argument $Z_i^\prime$ that is independent of $(X_i,Y_i)$. Since the same bound $\delta_n(\varepsilon)$ holds uniformly over $(x,y)$, we have
    \begin{gather*}
        \mathrm{Pr}\left( |s(X_i,Y_i;Z_i^\prime)-s^\star(X_i,Y_i)|>\varepsilon\mid \phi(t) \right)\leq\delta_n(\varepsilon)\,.
    \end{gather*}
    Hence, by the union bound,
    \begin{gather*}
        \mathrm{Pr}\left( D_{n,\varepsilon}(t)\mid \phi(t) \right)\geq 1-(n+1)\delta_n(\varepsilon)\,.
    \end{gather*}
    Also define
    \begin{gather*}
        \widetilde{D}_n(t)=\left\{ \sup_{Z_{i}^\prime:i\in[n+1]}\sup_{x,y}|s(x,y;Z_i^\prime)-s(x,y;Z)|\leq \widetilde{\varepsilon}_n \right\}\,.
    \end{gather*}
    By Assumption~\ref{ass: stable score},
    \begin{gather*}
        \mathrm{Pr}\left( \widetilde{D}_n(t)\mid \phi(t) \right)\geq 1-\widetilde{\delta}_n\,.
    \end{gather*}

    On the event $\widetilde{D}_n(t)\cap D_{n,\varepsilon}(t)$,
    \begin{align*}
        \sup_{1\leq i\leq n+1}|s(X_i,Y_i;Z)-s^\star(X_i,Y_i)|\leq \varepsilon+\widetilde{\varepsilon}_n\,.
    \end{align*}
    Let $\Delta_\varepsilon=\varepsilon+\widetilde{\varepsilon}_n$. Then $|s(X_i,Y_i;Z)-s^\star(X_i,Y_i)|\leq \Delta_\varepsilon$ for every $i\in[n+1]$. Therefore, for every $u\in\mathbb{R}$,
    \begin{align*}
        &\sum_{i=1}^{n+1}\frac{w(X_i)}{\sum_{j=1}^{n+1}w(X_j)}\mathbbm{1}\{s^\star(X_i,Y_i)\le u-\Delta_\varepsilon\}\\
        \leq&
        \sum_{i=1}^{n+1}\frac{w(X_i)}{\sum_{j=1}^{n+1}w(X_j)}\mathbbm{1}\{s(X_i,Y_i;Z)\le u\}\\
        \leq&
        \sum_{i=1}^{n+1}\frac{w(X_i)}{\sum_{j=1}^{n+1}w(X_j)}\mathbbm{1}\{s^\star(X_i,Y_i)\le u+\Delta_\varepsilon\}\,.
    \end{align*}
    By monotonicity of the quantile map,
    \begin{gather*}
        q^\star(Z;\alpha)-\Delta_\varepsilon\leq q(Z;\alpha)\leq q^\star(Z;\alpha)+\Delta_\varepsilon\,,
    \end{gather*}
    and therefore $|q(Z;\alpha)-q^\star(Z;\alpha)|\leq \Delta_\varepsilon$. Define
    \begin{gather*}
        D_\varepsilon(t)=\left\{ |s(X_{n+1},Y_{n+1};Z)-s^\star(X_{n+1},Y_{n+1})|\vee|q(Z;\alpha)-q^\star(Z;\alpha)|\leq \Delta_\varepsilon \right\}\,.
    \end{gather*}
    Since $\widetilde{D}_n(t)\cap D_{n,\varepsilon}(t)\subset D_\varepsilon(t)$, we have
    \begin{gather*}
        \mathrm{Pr}\left( D_\varepsilon(t)\mid \phi(t) \right)\geq 1-(n+1)\delta_n(\varepsilon)-\widetilde{\delta}_n\,.
    \end{gather*}

    We now bound the conditional coverage directly. For the lower bound,
    \begin{align*}
        &\mathrm{Pr}\left( Y_{n+1}\in\widehat{C}(X_{n+1})\mid \phi(t) \right)\\
        =&\mathrm{Pr}\left( s(X_{n+1},Y_{n+1};Z)\leq q(Z;\alpha)\mid \phi(t) \right)\\
        \geq&\mathrm{Pr}\left( s^\star(X_{n+1},Y_{n+1})+\Delta_\varepsilon\leq q^\star(Z;\alpha)-\Delta_\varepsilon,D_\varepsilon(t)\mid \phi(t) \right)\\
        \geq&\mathrm{Pr}\left( s^\star(X_{n+1},Y_{n+1})+\Delta_\varepsilon\leq q^\star(Z;\alpha)-\Delta_\varepsilon\mid \phi(t) \right)-\mathrm{Pr}\left( D_\varepsilon(t)^c\mid \phi(t) \right)\\
        =&\mathbb{E}\left\{ F_t^\star(q^\star(Z;\alpha)-2\Delta_\varepsilon)\mid \phi(t) \right\}-(n+1)\delta_n(\varepsilon)-\widetilde{\delta}_n\,,
    \end{align*}
    where the last equality follows from the assumed conditional independence between $s^\star(X_{n+1},Y_{n+1})$ and $q^\star(Z;\alpha)$ given $\phi(t)$, together with the fact that the conditional distribution of $s^\star(X_{n+1},Y_{n+1})$ given $\phi(t)$ is $F_t^\star$.
    Since $F_t^\star$ is $L_t$-Lipschitz, it follows that
    \begin{align*}
        &\mathbb{E}\left\{ F_t^\star(q^\star(Z;\alpha)-2\Delta_\varepsilon)\mid \phi(t) \right\}-(n+1)\delta_n(\varepsilon)-\widetilde{\delta}_n\\
        \geq&1-\alpha-(n+1)\delta_n(\varepsilon)-\widetilde{\delta}_n-2L_t\Delta_\varepsilon
        -L_t\mathbb{E}\left\{ \left|q^\star(Z;\alpha)-Q(1-\alpha;F_t^\star)\right|\wedge L_t^{-1}\mid \phi(t) \right\}\,.
    \end{align*}
    Similarly, for the upper bound,
    \begin{align*}
        &\mathrm{Pr}\left( Y_{n+1}\in\widehat{C}(X_{n+1})\mid \phi(t) \right)\\
        \leq&(n+1)\delta_n(\varepsilon)+\widetilde{\delta}_n
        +\mathbb{E}\left\{ F_t^\star(q^\star(Z;\alpha)+2\Delta_\varepsilon)-F_t^\star(Q(1-\alpha;F_t^\star))+1-\alpha\mid \phi(t) \right\}\\
        \leq&1-\alpha+(n+1)\delta_n(\varepsilon)+\widetilde{\delta}_n+2L_t\Delta_\varepsilon
        +L_t\mathbb{E}\left\{ \left|q^\star(Z;\alpha)-Q(1-\alpha;F_t^\star)\right|\wedge L_t^{-1}\mid \phi(t) \right\}\,.
    \end{align*}

    Combining the lower and upper bounds gives
    \begin{align*}
        &\left|\mathrm{Pr}\left( Y_{n+1}\in\widehat{C}(X_{n+1})\mid \phi(t) \right)-(1-\alpha)\right|\\
        \leq&(n+1)\delta_n(\varepsilon)+\widetilde{\delta}_n+2L_t\Delta_\varepsilon
        +L_t\mathbb{E}\left\{ \left|q^\star(Z;\alpha)-Q(1-\alpha;F_t^\star)\right|\wedge L_t^{-1} \mid \phi(t) \right\}\,.
    \end{align*}
    By the triangle inequality,
    \begin{align*}
        &\mathbb{E}\left\{ \left|q^\star(Z;\alpha)-Q(1-\alpha;F_t^\star)\right|\wedge L_t^{-1}\mid \phi(t) \right\}\\
        \leq&\mathbb{E}\left\{ \left|q^\star(Z;\alpha)-Q\left(1-\alpha;F_{w\circ s^\star}\right)\right|\wedge L_t^{-1}\mid \phi(t) \right\}\\
        &+\mathbb{E}\left\{\left|Q(1-\alpha;F_t^\star)-Q\left(1-\alpha;F_{w\circ s^\star}\right)\right|\mid \phi(t) \right\}\,.
    \end{align*}
     Since $q^\star(Z;\alpha)$ is the weighted empirical quantile targeting $Q(1-\alpha;F_{w\circ s^\star})$, Lemma~\ref{theo: weighted quantile error} implies that there exists a constant $C_1>0$ such that
    \begin{gather*}
        \mathbb{E}\left\{ \left|q^\star(Z;\alpha)-Q(1-\alpha;F_{w\circ s^\star})\right|\wedge L_t^{-1}\mid w \right\}\leq C_1\Gamma_n(w)\,.
    \end{gather*}
    Taking iterated conditional expectations then yields
    \begin{gather*}
        \mathbb{E}\left\{ \left|q^\star(Z;\alpha)-Q(1-\alpha;F_{w\circ s^\star})\right|\wedge L_t^{-1}\mid \phi(t) \right\}\leq C_1\mathbb{E}\left\{\Gamma_n(w)\mid \phi(t)\right\}\,.
    \end{gather*}
    Absorbing constants into $C_t$ proves \eqref{eq: independent pointwise bound 1}.
    
    It remains to prove the refinement of the intrinsic conditional-mismatch error term
    \begin{gather*}
        \mathbb{E}\left\{\left|Q(1-\alpha;F_t^\star)-Q\left(1-\alpha;F_{w\circ s^\star}\right)\right|\mid \phi(t) \right\}\,.
    \end{gather*}
    Let $T_0$ be drawn from $P_T$ independently of all other random objects. Denote
    \begin{gather*}
        q_t^\star=Q(1-\alpha;F_t^\star),\qquad q_{T_0}^\star=Q(1-\alpha;F_{T_0}^\star)\,.
    \end{gather*}
    Since $F_t^\star$ is $\overline{L}_{s^\star}$-Lipschitz for each $t\in\mathcal{T}$, it is continuous and therefore satisfies $F_{T_0}^\star(q_{T_0}^\star)=1-\alpha$ almost surely. Also,
    \begin{gather*}
        F_{w\circ s^\star}(u)=\mathbb{E}\left\{ F_{T_0}^\star(u)r_w(T_0)\mid w \right\}\,.
    \end{gather*}
    Then
    \begin{align*}
        &\left|F_{w\circ s^\star}(q_t^\star)-(1-\alpha)\right|\\
        =& \left|\mathbb{E}\left\{ F_{T_0}^\star(q_t^\star)r_w(T_0)\mid w \right\}
        -\mathbb{E}\left\{ F_{T_0}^\star(q_{T_0}^\star)r_w(T_0)\mid w \right\}\right|\\
        \leq&\overline{L}_{s^\star}\mathbb{E}\left\{ \left|Q(1-\alpha;F_t^\star)-Q(1-\alpha;F_{T_0}^\star)\right|r_w(T_0) \mid w \right\}\,,
    \end{align*}
    where the expectation is taken over $T_0\sim P_T$. Since the density of $F_{w\circ s^\star}$ on $(-M_{s^\star},M_{s^\star})$ is lower bounded by $\underline{L}_{s^\star}$, we obtain
    \begin{align*}
        &\left|Q(1-\alpha;F_t^\star)-Q(1-\alpha;F_{w\circ s^\star})\right|\\
        \leq&\dfrac{\overline{L}_{s^\star}}{\underline{L}_{s^\star}}
        \mathbb{E}\left\{ \left|Q(1-\alpha;F_t^\star)-Q(1-\alpha;F_{T_0}^\star)\right|r_w(T_0) \mid w \right\}\,.
    \end{align*}
    Taking expectation over $w$ conditional on $\phi(t)$ proves \eqref{eq: independent pointwise bound 2}. This completes the proof of the theorem.
\end{proof}

\subsection{Proof of Theorem~\ref{theo: averaged miscoverage}}

\begin{proof}
    Theorem~\ref{theo: averaged miscoverage} follows directly from Theorem~\ref{theo: averaged miscoverage shift} with $r \equiv 1$ and $P_{T,1}=P_T$. The proof of Theorem~\ref{theo: averaged miscoverage shift} can be found in Section~\ref{sec:supp_proof_averaged miscoverage shift}.
\end{proof}

\subsection{Proof of Theorem~\ref{theo: marginal general cshift}}\label{sec:supp_proof_theo_marginal general cshift}
\begin{proof}
    By the definition of $\widehat C(\cdot)$, the marginal coverage probability is
    \begin{align}
        \notag&\mathrm{Pr}\left( Y_{n+1}\in \widehat{C}(X_{n+1}) \right)\\
        \notag=&\mathrm{Pr}\left( s(X_{n+1},Y_{n+1};Z)\leq Q\left( 1-\alpha; \dfrac{\sum_{i=1}^{n}w(X_i)\delta_{s(X_i,Y_i;Z)}+w(X_{n+1})\delta_{\infty}}{\sum_{i=1}^{n+1}w(X_i)} \right) \right)\\
        =&\mathbb{E} \left\{ \mathbbm{1}\left( s(X_{n+1},Y_{n+1};Z)\leq Q\left( 1-\alpha; \dfrac{\sum_{i=1}^{n+1}w(X_i)\delta_{s(X_i,Y_i;Z)}}{\sum_{i=1}^{n+1}w(X_i)} \right) \right) \right\}\,.\label{eq: marginal general coverage eq1}
    \end{align}
    Let $w_0(\cdot)$ be a random function independent of $X\sim P_{X,1}$ satisfying $\mathbb{E}\{w_0(X)\mid w_0\}=1$ and $\mathbb{E}\{w_0(x)\}=r_X(x)$. Given a realization of $w_0(\cdot)$, draw $X_{n+1}'$ from the distribution with density proportional to $w_0(x)dP_{X,1}(x)$. Since $\mathbb{E}\{w_0(X)\mid w_0\}=1$, $w_0(x)dP_{X,1}(x)$ defines a valid density. Moreover, the marginal distribution of $X_{n+1}'$ satisfies
    \begin{gather*}
        \mathbb{E}\{w_0(x)dP_{X,1}(x)\}=\mathbb{E}\{w_0(x)\}dP_{X,1}(x)=r_X(x)dP_{X,1}(x)=dP_{X,2}(x)\,,
    \end{gather*}
    and hence $X_{n+1}'\sim P_{X,2}$. Define the collection of all such random functions by
    \begin{gather*}
        \mathcal{W}_X=\{w_0(\cdot)\geq 0:\mathbb{E}\{w_0(X)\mid w_0\}=1\text{ and }\mathbb{E}\{w_0(x)\}=r_X(x),\forall x\in\mathcal{X}\}\,.
    \end{gather*}

    Fix $w_0(\cdot)\in\mathcal{W}_X$ and denote the unordered collection of observations in $Z$ by $\mathrm{set}(Z)=\{(X_i,Y_i):1\leq i\leq n+1\}$. Conditional on $(w_0,\mathrm{set}(Z))$, define an auxiliary index $J$ by
    \begin{gather*}
        \mathrm{Pr}\left( J=j\mid w_0,\mathrm{set}(Z) \right)=\dfrac{w_0(X_j)}{\sum_{i=1}^{n+1}w_0(X_i)},\qquad j\in[n+1]\,.
    \end{gather*}
    By Assumption~\ref{ass: covariate shift assumption}, this auxiliary draw represents the conditional law of the test score in the sense that $s(X_{n+1},Y_{n+1};Z)\mid (w_0,\mathrm{set}(Z))$ has the same distribution as $s(X_J,Y_J;Z)\mid (w_0,\mathrm{set}(Z))$. Equivalently, conditional on $(w_0,\mathrm{set}(Z))$, the distribution of $s(X_{n+1},Y_{n+1};Z)$ is discrete with mass $\bigl\{ \sum_{j=1}^{n+1}w_0(X_j) \bigr\}^{-1}w_0(X_i)$ at the value $s(X_i,Y_i;Z)$.
    Therefore, by the law of total expectation, \eqref{eq: marginal general coverage eq1} equals
    \begin{align*}
        &\mathbb{E}\left[ \mathbb{E} \left\{ \mathbbm{1}\left( s(X_{n+1},Y_{n+1};Z)\leq Q\left( 1-\alpha; \dfrac{\sum_{i=1}^{n+1}w(X_i)\delta_{s(X_i,Y_i;Z)}}{\sum_{i=1}^{n+1}w(X_i)} \right) \right) \mid w_0,\mathrm{set}(Z)\right\} \right]\\
        &=\mathbb{E} \left\{ \sum_{j=1}^{n+1}\dfrac{w_0(X_j)}{\sum_{i=1}^{n+1}w_0(X_i)}\mathbbm{1}\left( s(X_{j},Y_{j};Z)\leq Q\left( 1-\alpha; \dfrac{\sum_{i=1}^{n+1}w(X_i)\delta_{s(X_i,Y_i;Z)}}{\sum_{i=1}^{n+1}w(X_i)} \right) \right) \right\}\,.
    \end{align*}
    By the definition of the quantile function,
    \begin{align*}
        1-\alpha&\leq \sum_{j=1}^{n+1}\dfrac{w(X_j)}{\sum_{i=1}^{n+1}w(X_i)}\mathbbm{1}\left( s(X_{j},Y_{j};Z)\leq Q\left( 1-\alpha; \dfrac{\sum_{i=1}^{n+1}w(X_i)\delta_{s(X_i,Y_i;Z)}}{\sum_{i=1}^{n+1}w(X_i)} \right) \right)\\
        &\leq 1-\alpha+\widehat{w}_{\rm max}(Z)\,,
    \end{align*}
    where $\widehat{w}_{\max}(Z)=\sup_{c}\bigl\{ \sum_{j=1}^{n+1}w(X_j) \bigr\}^{-1}\bigl\{ \sum_{i=1}^{n+1}w(X_i)\mathbbm{1}(s(X_i,Y_i;Z)=c) \bigr\}$. Thus
    \begin{align*}
        &\left|\sum_{j=1}^{n+1}\left\{ \dfrac{w(X_j)}{\sum_{i=1}^{n+1}w(X_i)}-\dfrac{w_0(X_j)}{\sum_{i=1}^{n+1}w_0(X_i)} \right\}\mathbbm{1}\left( s(X_{j},Y_{j};Z)\leq Q\left( 1-\alpha; \dfrac{\sum_{i=1}^{n+1}w(X_i)\delta_{s(X_i,Y_i;Z)}}{\sum_{i=1}^{n+1}w(X_i)} \right) \right)\right|\\
        &\leq \sum_{j=1}^{n+1}\left|\dfrac{w(X_j)}{\sum_{i=1}^{n+1}w(X_i)}-\dfrac{w_0(X_j)}{\sum_{i=1}^{n+1}w_0(X_i)}\right|\overset{\rm def.}{=}d_Z(w,w_0)\,.
    \end{align*}
    Therefore,
    \begin{align*}
        1-\alpha-\mathbb{E}\left\{ d_Z(w,w_0) \right\}&\leq \mathrm{Pr}\left( Y_{n+1}\in \widehat{C}(X_{n+1}) \right)\\
        &\leq 1-\alpha+\mathbb{E}\left\{ \widehat{w}_{\max}(Z) \right\}+\mathbb{E}\left\{ d_Z(w,w_0) \right\}\,.
    \end{align*}
    Taking the infimum over $w_0\in\mathcal{W}_X$ gives
    \begin{align*} 
        1-\alpha-\inf_{w_0\in\mathcal{W}_X}\mathbb{E}\left\{ d_Z(w,w_0) \right\}&\leq \mathrm{Pr}\left( Y_{n+1}\in \widehat{C}(X_{n+1}) \right)\\
        &\leq 1-\alpha+\mathbb{E}\left\{ \widehat{w}_{\max}(Z) \right\}+\inf_{w_0\in\mathcal{W}_X}\mathbb{E}\left\{ d_Z(w,w_0) \right\}\,,
    \end{align*}
    which completes the proof of this theorem.
\end{proof}

\subsection{Proof of Theorem~\ref{theo: marginal SymmPI}}
\begin{proof}
    Theorem~\ref{theo: marginal SymmPI} follows immediately from Theorem~\ref{theo: marginal SymmPI weighted}, whose proof can be found in Section~\ref{sec:supp_proof_marginal_symmpi_weighted}.
\end{proof}

\subsection{Proof of Theorem~\ref{theo:symmpi conditional miscoverage}}\label{sec: proof of symmpi conditional miscoverage}
\begin{proof}
    Theorem~\ref{theo:symmpi conditional miscoverage} follows directly from Theorem~\ref{theo:symmpi conditional miscoverage weighted}; see Section~\ref{sec:supp_proof_conditional_symmpi_weighted} for its proof.
\end{proof}

\subsection{Proof of Theorem~\ref{theo: community conditional}}
\begin{proof}
For any $\mathcal{I}\in [n_{\rm comm}]$, let $F_{v^\star\mid\sigma(W)=\mathcal{I}}(\cdot)$ denote the conditional distribution function of $v^\star(X,Y)$ given $\sigma(W)=\mathcal{I}$. We first bound the difference between the score $s(X_{n+1},Y_{n+1};Z)$ and $F_{v^\star\mid\sigma(W)=\mathcal{I}}(v^\star(X_{n+1},Y_{n+1}))$.

\noindent\textbf{Step 1:} Define the empirical distribution of the oracle base score within community $\mathcal{I}$ by
\begin{gather*}
    \widetilde{F}_{v^\star\mid\sigma(W)=\mathcal{I}}^{(1)}(u)=\dfrac{\sum_{j=1}^{n+1}\mathbbm{1}(\sigma(W_j)=\mathcal{I},v^\star(X_j,Y_j)\leq u)}{\sum_{j=1}^{n+1}\mathbbm{1}(\sigma(W_j)=\mathcal{I})}\,.
\end{gather*}
By the DKW inequality, conditional on $\sum_{j=1}^{n+1}\mathbbm{1}\{\sigma(W_j)=\mathcal{I}\}$, for any $\varepsilon>0$,
\begin{align*}
    &\mathrm{Pr}\left( \sup_{u}|F_{v^\star\mid\sigma(W)=\mathcal{I}}(u)-\widetilde{F}_{v^\star\mid\sigma(W)=\mathcal{I}}^{(1)}(u)|\leq\varepsilon \right)\\
    \geq & 1-2\exp\left( -2\left\{ \sum_{j=1}^{n+1}\mathbbm{1}(\sigma(W_j)=\mathcal{I}) \right\}\varepsilon^2 \right)\,.
\end{align*}
Let $p_\mathcal{I}=\mathrm{Pr}(\sigma(W)=\mathcal{I})\geq p_{\rm comm}>0$. Hoeffding's inequality gives
\begin{gather*}
    \mathrm{Pr}\left( \sum_{j=1}^{n+1}\mathbbm{1}(\sigma(W_j)=\mathcal{I})>(n+1)p_\mathcal{I}/2 \right)\geq 1-\exp\left( -np_\mathcal{I}^2/2 \right)\\
    \mathrm{Pr}\left( \forall\mathcal{I}\in[n_{\rm comm}],\sum_{j=1}^{n+1}\mathbbm{1}(\sigma(W_j)=\mathcal{I})>(n+1)p_\mathcal{I}/2 \right)\geq 1-n_{\rm comm}\exp\left( -np_{\rm comm}^2/2 \right)\,,
\end{gather*}
where the last inequality is derived by union bound. Define the event
\begin{gather*}
    D_1=\left\{ \forall\mathcal{I}\in[n_{\rm comm}],\sum_{j=1}^{n+1}\mathbbm{1}(\sigma(W_j)=\mathcal{I})>(n+1)p_\mathcal{I}/2 \right\}\,.
\end{gather*}
On $D_1$, applying the preceding DKW bound uniformly over $\mathcal I$ yields
\begin{align*}
    &\mathrm{Pr}\left( \forall\mathcal{I}\in[n_{\rm comm}],\sup_{u}\left|F_{v^\star\mid\sigma(W)=\mathcal{I}}(u)-\widetilde{F}_{v^\star\mid\sigma(W)=\mathcal{I}}^{(1)}(u)\right|\leq\varepsilon\mid D_1 \right)\\
    \geq&1-2n_{\rm comm}\exp\left( -np_{\rm comm}\varepsilon^2 \right)\,.
\end{align*}
Consequently,
\begin{align*}
    &\mathrm{Pr}\left( \forall\mathcal{I}\in[n_{\rm comm}],\sup_{u}\left|F_{v^\star\mid\sigma(W)=\mathcal{I}}(u)-\widetilde{F}_{v^\star\mid\sigma(W)=\mathcal{I}}^{(1)}(u)\right|\leq\varepsilon \right)\\
    \geq&1-3n_{\rm comm}\exp\left( -np_{\rm comm}\varepsilon^2 \right),
\end{align*}
where the last inequality holds when $\varepsilon^2\leq p_{\rm comm}$.

\noindent\textbf{Step 2:} First define the optimal rearangement of $\widehat{\sigma}$ as:
\begin{gather*}
    \widehat{\pi}=\underset{\pi\in\mathcal{G}_{n_{\rm comm}}}{\arg\min}\ \ell_{\rm mis}(\sigma,\widehat{\sigma};\{W_i\}_{i=1}^{n+1},A)\,.
\end{gather*}
Define the empirical distribution of the oracle base score within estimated community $\mathcal{I}$ by
\begin{gather*}
    \widetilde{F}_{v^\star\mid\sigma(W)=\mathcal{I}}^{(2)}(u)=\dfrac{\sum_{j=1}^{n+1}\mathbbm{1}(\widehat{\pi}(\widehat{\sigma}(A)_j)=\mathcal{I},v^\star(X_j,Y_j)\leq u)}{\sum_{j=1}^{n+1}\mathbbm{1}(\widehat{\pi}(\widehat{\sigma}(A)_j)=\mathcal{I})}\,.
\end{gather*}
We next bound the difference between $\widetilde{F}_{v^\star\mid\sigma(W)=\mathcal{I}}^{(2)}(u)$ and $\widetilde{F}_{v^\star\mid\sigma(W)=\mathcal{I}}^{(1)}(u)$. By the definition of $\widehat{\pi}(\cdot)$, we can derive
\begin{align*}
    &\left|\widetilde{F}_{v^\star\mid\sigma(W)=\mathcal{I}}^{(1)}(u)-\widetilde{F}_{v^\star\mid\sigma(W)=\mathcal{I}}^{(2)}(u)\right|\\
    \leq&\dfrac{\sum_{j=1}^{n+1}\mathbbm{1}(\widehat{\pi}(\widehat{\sigma}(A)_j)=\mathcal{I},v^\star(X_j,Y_j)\leq u)\left|\sum_{j=1}^{n+1}\mathbbm{1}(\widehat{\pi}(\widehat{\sigma}(A)_j)=\mathcal{I})-\mathbbm{1}(\sigma(W_j)=\mathcal{I})\right|}{\sum_{j=1}^{n+1}\mathbbm{1}(\sigma(W_j)=\mathcal{I})\sum_{j=1}^{n+1}\mathbbm{1}(\widehat{\pi}(\widehat{\sigma}(A)_j)=\mathcal{I})}\\
    \leq&\dfrac{\sum_{j=1}^{n+1}\mathbbm{1}(\widehat{\pi}(\widehat{\sigma}(A)_j)=\mathcal{I})\left|\sum_{j=1}^{n+1}\left\{ \mathbbm{1}(\widehat{\pi}(\widehat{\sigma}(A)_j)=\mathcal{I})-\mathbbm{1}(\sigma(W_j)=\mathcal{I}) \right\}\mathbbm{1}(v^\star(X_j,Y_j)\leq u)\right|}{\sum_{j=1}^{n+1}\mathbbm{1}(\sigma(W_j)=\mathcal{I})\sum_{j=1}^{n+1}\mathbbm{1}(\widehat{\pi}(\widehat{\sigma}(A)_j)=\mathcal{I})}\\
    \leq&\dfrac{2\sum_{j=1}^{n+1}\mathbbm{1}(\widehat{\pi}(\widehat{\sigma}(A)_j)=\mathcal{I})\sum_{j=1}^{n+1}\mathbbm{1}(\widehat{\pi}(\widehat{\sigma}(A)_j)\neq\sigma(W_j))}{\sum_{j=1}^{n+1}\mathbbm{1}(\sigma(W_j)=\mathcal{I})\sum_{j=1}^{n+1}\mathbbm{1}(\widehat{\pi}(\widehat{\sigma}(A)_j)=\mathcal{I})}\\
    =&\dfrac{2(n+1)\ell_{\rm mis}(\sigma,\widehat{\sigma};\{W_i\}_{i=1}^{n+1},A)}{\sum_{j=1}^{n+1}\mathbbm{1}(\sigma(W_j)=\mathcal{I})}\,.
\end{align*}
Denote the event of identifying true community with $\varepsilon$ precision as
\begin{gather*}
    D_2=\{\ell_{\rm mis}(\sigma,\widehat{\sigma};\{W_i\}_{i=1}^{n+1},A)\leq\varepsilon\},
\end{gather*}
which has probability at least $1-\delta_{\rm comm}(n,\varepsilon)$. On $D_1\cap D_2$,
\begin{gather*}
    \left|\widetilde{F}_{v^\star\mid\sigma(W)=\mathcal{I}}^{(1)}(u)-\widetilde{F}_{v^\star\mid\sigma(W)=\mathcal{I}}^{(2)}(u)\right|\leq\dfrac{4\varepsilon}{p_{\mathcal{I}}}\,.
\end{gather*}
Hence
\begin{align*}
    &\mathrm{Pr}\left( \forall\mathcal{I}\in[n_{\rm comm}],\sup_{u}\left|F_{v^\star\mid\sigma(W)=\mathcal{I}}(u)-\widetilde{F}_{v^\star\mid\sigma(W)=\mathcal{I}}^{(2)}(u)\right|\leq\varepsilon(1+4/p_\mathcal{I}) \right)\\
    \geq&1-3n_{\rm comm}\exp\left( -np_{\rm comm}\varepsilon^2 \right)-\delta_{\rm comm}(n,\varepsilon)\,.
\end{align*}

\noindent\textbf{Step 3:} 
Define the empirical distribution of the estimated base score within estimated community $\mathcal{I}$ by
\begin{gather*}
    \widehat{F}_{v(\cdot,\cdot;Z)\mid\sigma(W)=\mathcal{I}}(u)=\dfrac{\sum_{j=1}^{n+1}\mathbbm{1}\{\widehat{\pi}(\widehat{\sigma}(A)_j)=\mathcal{I},v(X_j,Y_j;Z)\leq u\}}{\sum_{j=1}^{n+1}\mathbbm{1}\{\widehat{\pi}(\widehat{\sigma}(A)_j)=\mathcal{I}\}}\,.
\end{gather*}
Define the score-estimation event
\begin{gather*}
    D_3=\left\{ \sup_{1\leq i\leq n+1}|v(X_i,Y_i;Z)-v^\star(X_i,Y_i)|\leq\varepsilon \right\}\,,
\end{gather*}
which has probability at least $1-\delta(n,\varepsilon)$. On $D_3$, simple algebra gives
\begin{gather*}
    \widehat{F}_{v(\cdot,\cdot;Z)\mid\sigma(W)=\mathcal{I}}(u-\varepsilon)\leq \widetilde{F}_{v^\star\mid\sigma(W)=\mathcal{I}}^{(2)}(u)\leq \widehat{F}_{v(\cdot,\cdot;Z)\mid\sigma(W)=\mathcal{I}}(u+\varepsilon)\,.
\end{gather*}
Thus, on $D_1\cap D_2\cap D_3$,
\begin{align*}
    F_{v^\star\mid\sigma(W)=\mathcal{I}}(u-\varepsilon)-\varepsilon(1+4/p_\mathcal{I})
    &\leq \widehat{F}_{v(\cdot,\cdot;Z)\mid\sigma(W)=\mathcal{I}}(u)\\
    &\leq F_{v^\star\mid\sigma(W)=\mathcal{I}}(u+\varepsilon)+\varepsilon(1+4/p_\mathcal{I})\,.
\end{align*}
Since $F_{v^\star\mid\sigma(W)=\mathcal{I}}$ is $L_\mathcal{I}$-Lipschitz continuous, it follows that, on $D_1\cap D_2\cap D_3$,
\begin{gather*}
    \left|\widehat{F}_{v(\cdot,\cdot;Z)\mid\sigma(W)=\mathcal{I}}(u)-F_{v^\star\mid\sigma(W)=\mathcal{I}}(u)\right|\leq (1+4/p_\mathcal{I}+L_\mathcal{I})\varepsilon\,.
\end{gather*}
Define $L_0=\sup_{\mathcal{I}\in[n_{\rm comm]}}\{1+4/p_\mathcal{I}+L_\mathcal{I}\}$. Then
\begin{align*}
    &\mathrm{Pr}\left( \forall\mathcal{I}\in[n_{\rm comm}],\sup_{u}\left|\widehat{F}_{v(\cdot,\cdot;Z)\mid\sigma(W)=\mathcal{I}}(u)-F_{v^\star\mid\sigma(W)=\mathcal{I}}(u)\right|\leq L_0\varepsilon \right)\\
    \geq&1-3n_{\rm comm}\exp\left( -np_{\rm comm}\varepsilon^2 \right)-\delta_{\rm comm}(n,\varepsilon)-\delta(n,\varepsilon)\,.
\end{align*}

\noindent\textbf{Step 4:} Consider $(W_{n+1},X_{n+1},Y_{n+1})$ satisfying $\sigma(W_{n+1})=\mathcal{I}$. If $\widehat{\pi}(\widehat{\sigma}(A)_{n+1})=\mathcal{I}$, then
$s(X_{n+1},Y_{n+1};Z)=\widehat{F}_{v(\cdot,\cdot;Z)\mid\sigma(W)=\mathcal{I}}(v(X_{n+1},Y_{n+1};Z))$.
On $D_1\cap D_2\cap D_3$, it follows that
\begin{align*}
    &\left|s(X_{n+1},Y_{n+1};Z)-F_{v^\star\mid\sigma(W)=\mathcal{I}}(v^\star(X_{n+1},Y_{n+1}))\right|\\
    \leq&\left|\widehat{F}_{v(\cdot,\cdot;Z)\mid\sigma(W)=\mathcal{I}}(v(X_{n+1},Y_{n+1};Z))-F_{v^\star\mid\sigma(W)=\mathcal{I}}(v(X_{n+1},Y_{n+1};Z))\right|\\
    &+\left|F_{v^\star\mid\sigma(W)=\mathcal{I}}(v(X_{n+1},Y_{n+1};Z))-F_{v^\star\mid\sigma(W)=\mathcal{I}}(v^\star(X_{n+1},Y_{n+1}))\right|\\
    \leq&L_0\varepsilon+L_\mathcal{I}\left|v(X_{n+1},Y_{n+1};Z)-v^\star(X_{n+1},Y_{n+1})\right|\leq 2L_0\varepsilon\,.
\end{align*}

For any observed dataset $Z\in D_1\cap D_2\cap D_3$, assume that the observations are unordered and denote the resulting unordered set by $Z_{\{\}}=\{(W_i,X_i,Y_i):i\in[n+1]\}$. Since the graph model is permutation invariant, conditional on $Z_{\{\}}$, each labeling of the $n+1$ nodes is equally admissible. Thus the test node is uniformly distributed over the nodes in this unordered sample. Conditional on $Z_{\{\}}$, suppose there are at least $\widehat{n}_\mathcal{I}^{(1)}$ possible values of $(W_{n+1},X_{n+1},Y_{n+1})$ satisfying $\sigma(W_{n+1})=\mathcal{I}$. By the definition of $D_1$, $\widehat{n}_\mathcal{I}^{(1)}>(n+1)p_\mathcal{I}/2$. By the definition of $D_2$, at least $\widehat{n}_\mathcal{I}^{(2)}=\widehat{n}_\mathcal{I}^{(1)}-(n+1)\varepsilon$ of these $\widehat{n}_\mathcal{I}^{(1)}$ possible values are assigned to community $\mathcal{I}$, namely $\widehat{\pi}(\widehat{\sigma}(A)_j)=\mathcal{I}$. Hence, for a fixed unordered set $Z_{\{\}}$ with $Z\in D_1\cap D_2\cap D_3$,
\begin{gather*}
    \mathrm{Pr}\left( \widehat{\pi}(\widehat{\sigma}(A)_{n+1})=\mathcal{I}\mid Z_{\{\}},\sigma(W_{n+1})=\mathcal{I} \right)\geq \dfrac{\widehat{n}_\mathcal{I}^{(1)}-(n+1)\varepsilon}{\widehat{n}_\mathcal{I}^{(1)}}> 1-2\varepsilon/p_\mathcal{I}\,.
\end{gather*}
Define $D_\varepsilon^{(1)}=\left\{ \left|s(X_{n+1},Y_{n+1};Z)-F_{v^\star\mid\sigma(W)=\mathcal{I}}(v^\star(X_{n+1},Y_{n+1}))\right|\leq 2L_0\varepsilon \right\}$. The preceding inequality implies
\begin{align*}
    &\mathrm{Pr}\left( D_\varepsilon^{(1)}\mid \sigma(W_{n+1})=\mathcal{I} \right)\\
    \geq&(1-2\varepsilon/p_{\rm comm})\left\{ 1-3n_{\rm comm}\exp\left( -np_{\rm comm}\varepsilon^2 \right)-\delta_{\rm comm}(n,\varepsilon)-\delta(n,\varepsilon) \right\}\,.
\end{align*}

Define the finite-score empirical quantile
\begin{gather*}
    \overline{q}(V(Z);\alpha)=Q\left( 1-\alpha;(n+1)^{-1}\sum_{i=1}^{n+1}\delta_{s(X_i,Y_i;Z)} \right)
\end{gather*}
and recall that the graph conformal construction uses
\begin{gather*}
    q(V(Z);\alpha)=Q\left( 1-\alpha;(n+2)^{-1}\left\{ \sum_{i=1}^{n+1}\delta_{s(X_i,Y_i;Z)}+\delta_\infty \right\} \right)\,.
\end{gather*}

We first control $\overline{q}(V(Z);\alpha)$ and then show that the additional $\delta_\infty$ changes the quantile by only $O(n^{-1})$. On $D_1 \cap D_2$, for any $\mathcal{I}\in[n_{\rm comm}]$, at least $\widehat{n}_\mathcal{I}^{(2)}\geq(n+1)(p_{\rm comm}/2-\varepsilon)$ samples $(W_i, X_i, Y_i)$ satisfy $\sigma(W_i) = \mathcal{I}$ and are assigned to community $\mathcal I$. Denote their index set by $\widehat{\mathfrak{I}}_\mathcal{I}$. By the definition of $s(X_i,Y_i;Z)$ and $|\widehat{\mathfrak{I}}_\mathcal{I}|=\widehat{n}_\mathcal{I}^{(2)}$, we have
\begin{gather*}
    \left\{ s(X_i,Y_i;Z):i\in\widehat{\mathfrak{I}}_\mathcal{I} \right\}=\left\{ 1/\widehat{n}_\mathcal{I}^{(2)},\ldots,(\widehat{n}_\mathcal{I}^{(2)}-1)/\widehat{n}_\mathcal{I}^{(2)}, 1 \right\}\,.
\end{gather*}
Hence the union of all finite score values contains the grid $\{j/m_{\min}:1\le j\le m_{\min}\}$ with
\begin{gather*}
    m_{\min}=\min_{\mathcal{I}\in[n_{\rm comm}]}\widehat{n}_\mathcal{I}^{(2)}\,.
\end{gather*}
Therefore, every adjacent gap in the ordered finite score values is at most $m_{\min}^{-1}$, and by the definition of the empirical quantile function,
\begin{gather*}
    |\overline{q}(V(Z);\alpha)-(1-\alpha)|\leq m_{\min}^{-1}\,.
\end{gather*}

Let $S_{(1)}\leq\cdots\leq S_{(n+1)}$ be the ordered finite scores $\{s(X_i,Y_i;Z)\}_{i\in[n+1]}$. Since the added $\delta_\infty$ only appends one extra point at the end of the ordered sample, $\overline{q}(V(Z);\alpha)=S_{(k)}$ and $q(V(Z);\alpha)=S_{(k^\prime)}$ for
\begin{gather*}
    k=\left\lceil (1-\alpha)(n+1)\right\rceil,\qquad k^\prime=\left\lceil (1-\alpha)(n+2)\right\rceil\,,
\end{gather*}
and necessarily $k^\prime-k\in\{0,1\}$. Since every adjacent finite-score gap is at most $m_{\min}^{-1}$, it follows that
\begin{gather*}
    |q(V(Z);\alpha)-\overline{q}(V(Z);\alpha)|\leq m_{\min}^{-1}.
\end{gather*}
Consequently,
\begin{gather*}
    |q(V(Z);\alpha)-(1-\alpha)|\leq 2m_{\min}^{-1}\leq \dfrac{8}{np_{\rm comm}}\,,
\end{gather*}
where the last inequality holds when $\varepsilon\leq p_{\rm comm}/4$ on $D_1\cap D_2$. Define
\begin{gather*}
    D^{(2)}=\left\{ |q(V(Z);\alpha)-(1-\alpha)|\leq \dfrac{8}{np_{\rm comm}} \right\}\,.
\end{gather*}
Then, for $\varepsilon\leq p_{\rm comm}/4$,
\begin{align*}
    &\mathrm{Pr}\left( D_\varepsilon^{(1)}\cap D^{(2)}\mid \sigma(W_{n+1})=\mathcal{I} \right)\\
    \geq&(1-2\varepsilon/p_{\rm comm})\left\{ 1-3n_{\rm comm}\exp\left( -np_{\rm comm}\varepsilon^2 \right)-\delta_{\rm comm}(n,\varepsilon)-\delta(n,\varepsilon) \right\}\,.
\end{align*}

Using Theorem~\ref{theo:symmpi conditional miscoverage}, for $\varepsilon\leq p_{\rm comm}/4$, there exists a constant $C>0$ such that
\begin{align*}
    &\left|\mathrm{Pr}\left( Y_{n+1}\in\widehat{C}_{\rm GraphCP}(X_{n+1})\mid\sigma(W_{n+1})=\mathcal{I} \right)-(1-\alpha)\right|\\
    \leq&1-(1-2\varepsilon/p_{\rm comm})\left\{ 1-3n_{\rm comm}\exp\left( -np_{\rm comm}\varepsilon^2 \right)-\delta_{\rm comm}(n,\varepsilon)-\delta(n,\varepsilon) \right\}\\
    &+C\left\{\varepsilon+\dfrac{1}{np_{\rm comm}}\right\}\\
    \leq& C\left\{ \varepsilon+n^{-1}+\delta_{\rm comm}(n,\varepsilon)+\delta(n,\varepsilon)+n_{\rm comm}\exp\left( -np_{\rm comm}\varepsilon^2 \right) \right\}.
\end{align*}
This completes the proof of the theorem.
\end{proof}

\subsection{Proof of Corollary~\ref{cor: cc L1 test}}
\begin{proof}
    Corollary~\ref{cor: cc L1 test} follows immediately from Corollary~\ref{cor: cc L1 test shift}; see Section~\ref{sec:supp_proof_cc_L1_test_shift} for its proof.
\end{proof}

\section{Proofs for Supplementary Results}\label{sec:proof_supp}

\subsection{Proof of Theorem~\ref{theo: averaged miscoverage shift}}\label{sec:supp_proof_averaged miscoverage shift}
\begin{proof}
Let $P_{T,1}$, $P_{T,2}$, and $P_{T,w}$ denote the laws of $T(X)$ induced by $X\sim P_{X,1}$, $X\sim P_{X,2}$, and $X\sim P_{X,w}$, respectively, with density ratios $r=dP_{T,2}/dP_{T,1}$ and $r_w=dP_{T,w}/dP_{T,1}$.

The proof contains two steps. We first derive an event-conditional miscoverage bound for a fixed label $t$, following the same template as the i.i.d.~averaged result while keeping track of the weighted calibration distribution under covariate shift. We then average this bound over $T\sim P_{T,2}$ and express the resulting terms relative to $P_{T,1}$ through the density ratios $r$ and $r_w$.

Let $Z_{n+1}^\prime$ be obtained by replacing the $(n+1)$-th element of $Z$ with an independent copy $(X_{n+1}^\prime,Y_{n+1}^\prime)$ drawn from $P_{X,2}\times P_{Y\mid X}$; then $Z\overset{\rm d}{=}Z_{n+1}^\prime$. Given $Z_{n+1}^\prime$ and $t\in\mathcal{T}$, denote the conditional c.d.f.s of $s(X_{n+1},Y_{n+1};Z_{n+1}^\prime)$ and $s^\star(X_{n+1},Y_{n+1})$ given $\phi(t)$ by $F_t(\cdot;Z_{n+1}^\prime)$ and $F_t^\star(\cdot)$, and denote the corresponding densities by $f_t(\cdot;Z_{n+1}^\prime)$ and $f_t^\star(\cdot)$. For a random label $T$, write $F_T(\cdot;Z_{n+1}^\prime)$ and $f_T(\cdot;Z_{n+1}^\prime)$ for these functions evaluated at the realized label $T$. The upper density bound implies that $F_t(\cdot;Z_{n+1}^\prime)$ and $F_t^\star(\cdot)$ are $\overline{L}$-Lipschitz.

Denote the marginal c.d.f.~of $s(X,Y;Z_{n+1}^\prime)$ under independent $(X,Y)\sim P_{X,2}\times P_{Y\mid X}$, conditional on fixed $Z_{n+1}^\prime$, by $F_{r\circ S}(\cdot;Z_{n+1}^\prime)$, and denote the corresponding c.d.f.~of the oracle score $s^\star(X,Y)$ by $F_{r\circ s^\star}(\cdot)$. Also denote the marginal c.d.f.~of $s(X,Y;Z_{n+1}^\prime)$ under independent $(X,Y)\sim P_{X,w}\times P_{Y\mid X}$, conditional on fixed $Z_{n+1}^\prime$, by $F_{w\circ S}(\cdot;Z_{n+1}^\prime)$, and denote the corresponding c.d.f.~of the oracle score $s^\star(X,Y)$ by $F_{w\circ s^\star}(\cdot)$.
Let $T^\prime$ be an i.i.d.~copy of $T\sim P_{T,2}$, and let $T_0$ be drawn independently from $P_{T,1}$.

Define the event
\begin{gather*}
    D_{n,\varepsilon}=\left\{ \|q_{s(\cdot,\cdot;Z_{n+1}^\prime)}-q_{s^\star}\|_{P_{T,1},p}\leq\varepsilon \right\}\,,
\end{gather*}
which satisfies $\mathrm{Pr}(D_{n,\varepsilon})\geq 1-\delta_n(\varepsilon)$. By the assumed independence of $\phi(t)$ from the first $n$ elements of $Z$, the event $D_{n,\varepsilon}$ is independent of $\phi(t)$. Also define
\begin{gather*}
    \widetilde{D}_n=\left\{ \sup_{z_{n+1}^\prime}\sup_{x,y}|s(x,y;Z)-s(x,y;Z_{n+1}^\prime)|\leq \widetilde{\varepsilon}_n \right\}\,,
\end{gather*}
which satisfies $\mathrm{Pr}(\widetilde{D}_n\mid \phi(t))\geq 1-\widetilde{\delta}_n$ by Assumption~\ref{ass: stable score}. On $\widetilde{D}_n$, for any $i\in[n+1]$,
\begin{gather*}
    |s(X_i,Y_i;Z)-s(X_i,Y_i;Z_{n+1}^\prime)|\leq \widetilde{\varepsilon}_n\,.
\end{gather*}
Hence
\begin{align*}
    \Bigg|&Q\left( 1-\alpha;\left\{ \sum_{i=1}^{n+1}w(X_i) \right\}^{-1}\left\{ \sum_{i=1}^{n}w(X_i)\delta_{s(X_i,Y_i;Z)}+w(X_{n+1})\delta_{\infty} \right\} \right)-\\
    &Q\left( 1-\alpha;\left\{ \sum_{i=1}^{n+1}w(X_i) \right\}^{-1}\left\{ \sum_{i=1}^{n}w(X_i)\delta_{s(X_i,Y_i;Z_{n+1}^\prime)}+w(X_{n+1})\delta_{\infty} \right\} \right)\Bigg|
    \leq \widetilde{\varepsilon}_n\,,
\end{align*}
where, if both quantiles are infinite, the difference is interpreted as $0$.

Define
\begin{gather*}
    q_\alpha(Z,Z_{n+1}^\prime)=Q\left( 1-\alpha;\left\{ \sum_{i=1}^{n+1}w(X_i) \right\}^{-1}\left\{ \sum_{i=1}^{n}w(X_i)\delta_{s(X_i,Y_i;Z_{n+1}^\prime)}+w(X_{n+1})\delta_{\infty} \right\} \right)\,,
\end{gather*}
and
\begin{gather*}
    q_{\alpha}(Z_{n+1}^\prime)=Q\left( 1-\alpha;\dfrac{\sum_{i=1}^{n}w(X_i)\delta_{s(X_i,Y_i;Z_{n+1}^\prime)}+w(X_{n+1}^\prime)\delta_{s(X_{n+1}^\prime,Y_{n+1}^\prime;Z_{n+1}^\prime)}}{\sum_{i=1}^{n}w(X_i)+w(X_{n+1}^\prime)} \right).
\end{gather*}
Since $w(X_i)\in[\underline{M},\overline{M}]$ almost surely for every $i\in[n+1]$, there exists $\gamma_n=O(n^{-1})$ such that
\begin{gather*}
    q_{\alpha+\gamma_n}(Z_{n+1}^\prime)\leq q_\alpha(Z,Z_{n+1}^\prime)\leq q_{\alpha-\gamma_n}(Z_{n+1}^\prime)\,.
\end{gather*}

\noindent\textbf{Step 1: decomposition of conditional miscoverage.}
First,
\begin{align*}
    &\mathrm{Pr}\left( Y_{n+1}\in\widehat{C}(X_{n+1})\mid \phi(t) \right)\\
    \leq&\widetilde{\delta}_n+\delta_n(\varepsilon)+\mathrm{Pr}\Bigg(\widetilde{D}_n,D_{n,\varepsilon},s(X_{n+1},Y_{n+1};Z)\leq\\
    &\qquad\qquad Q\left( 1-\alpha;\left\{ \sum_{i=1}^{n+1}w(X_i) \right\}^{-1}\left\{ \sum_{i=1}^{n}w(X_i)\delta_{s(X_i,Y_i;Z)}+w(X_{n+1})\delta_{\infty} \right\} \right)\mid \phi(t)\Bigg)\\
    \leq&\widetilde{\delta}_n+\delta_n(\varepsilon)+\mathrm{Pr}\left( \widetilde{D}_n,D_{n,\varepsilon},s(X_{n+1},Y_{n+1};Z_{n+1}^\prime)\leq 2\widetilde{\varepsilon}_n+q_\alpha(Z,Z_{n+1}^\prime)\mid \phi(t) \right)\\
    \leq&\widetilde{\delta}_n+\delta_n(\varepsilon)+\mathrm{Pr}\left( D_{n,\varepsilon},s(X_{n+1},Y_{n+1};Z_{n+1}^\prime)\leq 2\widetilde{\varepsilon}_n+q_{\alpha-\gamma_n}(Z_{n+1}^\prime)\mid \phi(t) \right)\\
    =&\widetilde{\delta}_n+\delta_n(\varepsilon)+\mathbb{E}\left\{ \mathbbm{1}\left(D_{n,\varepsilon},s(X_{n+1},Y_{n+1};Z_{n+1}^\prime)\leq 2\widetilde{\varepsilon}_n+q_{\alpha-\gamma_n}(Z_{n+1}^\prime)\right)\mid \phi(t) \right\}\,.
\end{align*}
Decomposing the conditional expectation gives
\begin{align}
    &\mathbb{E}\left\{ \mathbbm{1}\left(D_{n,\varepsilon},s(X_{n+1},Y_{n+1};Z_{n+1}^\prime)\leq 2\widetilde{\varepsilon}_n+q_{\alpha-\gamma_n}(Z_{n+1}^\prime)\right)\mid \phi(t) \right\}\notag\\
    =&\mathbb{E}\left[ \mathbb{E}\left\{ \mathbbm{1}\left(D_{n,\varepsilon},s(X_{n+1},Y_{n+1};Z_{n+1}^\prime)\leq 2\widetilde{\varepsilon}_n+q_{\alpha-\gamma_n}(Z_{n+1}^\prime)\right)\mid Z_{n+1}^\prime,\phi(t) \right\} \right]\notag\\
    =&\mathbb{E}\left[ \mathbbm{1}(D_{n,\varepsilon})\mathbb{E}\left\{ \mathbbm{1}\left(s(X_{n+1},Y_{n+1};Z_{n+1}^\prime)\leq 2\widetilde{\varepsilon}_n+q_{\alpha-\gamma_n}(Z_{n+1}^\prime)\right)\mid Z_{n+1}^\prime,\phi(t) \right\} \right]\label{eq: general lp explain 1}\\
    =&\mathbb{E}\left\{ \mathbbm{1}(D_{n,\varepsilon})F_t\left( 2\widetilde{\varepsilon}_n+q_{\alpha-\gamma_n}(Z_{n+1}^\prime);Z_{n+1}^\prime \right)\mid\phi(t) \right\}\notag\\
    \leq& 2\overline{L}\widetilde{\varepsilon}_n+\mathbb{E}\left\{ \mathbbm{1}(D_{n,\varepsilon})F_t\left( q_{\alpha-\gamma_n}(Z_{n+1}^\prime);Z_{n+1}^\prime \right) \mid\phi(t)\right\},\label{eq: general lp explain 2}
\end{align}
where the inner expectation in \eqref{eq: general lp explain 1} is taken with respect to the conditional law of $s(X_{n+1},Y_{n+1};Z_{n+1}^\prime)$ given $\phi(t)$ and $Z_{n+1}^\prime$.
Also,
\begin{align}
    &F_t\left( q_{\alpha-\gamma_n}(Z_{n+1}^\prime);Z_{n+1}^\prime \right)\notag\\
    =&F_t\Bigl( q_{\alpha-\gamma_n}(Z_{n+1}^\prime)-Q\left( 1-\alpha+\gamma_n;F_{w\circ S}(\cdot;Z_{n+1}^\prime) \right)\notag\\
    & ~~~~~~~~~~~~~~~~~~~~~~ +Q\left( 1-\alpha+\gamma_n;F_{w\circ S}(\cdot;Z_{n+1}^\prime) \right);Z_{n+1}^\prime \Bigr)\notag\\
    \leq& \left\{ \overline{L}\left|q_{\alpha-\gamma_n}(Z_{n+1}^\prime)-Q\left( 1-\alpha+\gamma_n;F_{w\circ S}(\cdot;Z_{n+1}^\prime) \right)\right| \right\}\wedge 1\notag\\
    & ~~~~~~~~~~~~~~~~~~~~~~ +F_t\left(Q\left( 1-\alpha+\gamma_n;F_{w\circ S}(\cdot;Z_{n+1}^\prime) \right);Z_{n+1}^\prime \right)\,.\label{eq: general lp explain 3}
\end{align}

Define $B_w=\mathbb{E}\left\{ w(X_1)\mid w \right\}\geq \underline{M}$. Combining \eqref{eq: general lp explain 2} and \eqref{eq: general lp explain 3}, and applying Lemma~\ref{theo: sup quantile error}, gives
\begin{align*}
    &\mathrm{Pr}\left( Y_{n+1}\in\widehat{C}(X_{n+1})\mid \phi(t) \right)\\
    \leq&\widetilde{\delta}_n+\delta_n(\varepsilon)+2\overline{L}\widetilde{\varepsilon}_n
    +\overline{L}(2\overline{M}\overline{L}\vee 1)C_2(\overline{M}/\underline{M})\sqrt{n^{-1}\mathrm{Pdim}(\mathcal{S})\log(n)}\\
    &+\mathbb{E}\left\{ \mathbbm{1}(D_{n,\varepsilon})F_t\left(Q\left( 1-\alpha+\gamma_n;F_{w\circ S}(\cdot;Z_{n+1}^\prime) \right);Z_{n+1}^\prime \right) \right\}\,.
\end{align*}
By the definition of the quantile function and the $\overline{L}$-Lipschitz property of $F_t(\cdot;Z_{n+1}^\prime)$,
\begin{align*}
    &\mathbb{E}\left\{ \mathbbm{1}(D_{n,\varepsilon})F_t\left(Q\left( 1-\alpha+\gamma_n;F_{w\circ S}(\cdot;Z_{n+1}^\prime) \right);Z_{n+1}^\prime \right) \right\}\\
    \leq&1-\alpha+\overline{L}\mathbb{E}\left\{ \mathbbm{1}(D_{n,\varepsilon})\left|Q\left( 1-\alpha+\gamma_n;F_{w\circ S}(\cdot;Z_{n+1}^\prime) \right)-q_{s(\cdot,\cdot;Z_{n+1}^\prime)}(t)\right| \right\}\,.
\end{align*}
Thus
\begin{align*}
    &\mathrm{Pr}\left( Y_{n+1}\in\widehat{C}(X_{n+1})\mid \phi(t) \right)\\
    \leq&1-\alpha+\widetilde{\delta}_n+\delta_n(\varepsilon)+2\overline{L}\widetilde{\varepsilon}_n
    +\overline{L}(2\overline{M}\overline{L}\vee 1)C_2(\overline{M}/\underline{M})\sqrt{n^{-1}\mathrm{Pdim}(\mathcal{S})\log(n)}\\
    &+\overline{L}\mathbb{E}\left\{ \mathbbm{1}(D_{n,\varepsilon})\left|Q\left( 1-\alpha+\gamma_n;F_{w\circ S}(\cdot;Z_{n+1}^\prime) \right)-q_{s(\cdot,\cdot;Z_{n+1}^\prime)}(t)\right| \right\}\,.
\end{align*}
Similarly,
\begin{align*}
    &\mathrm{Pr}\left( Y_{n+1}\in\widehat{C}(X_{n+1})\mid \phi(t) \right)\\
    \geq&1-\alpha-2\overline{L}\widetilde{\varepsilon}_n-\widetilde{\delta}_n-\delta_n(\varepsilon)-\overline{L}(2\overline{M}\overline{L}\vee 1)C_2(\overline{M}/\underline{M})\sqrt{n^{-1}\mathrm{Pdim}(\mathcal{S})\log(n)}\\
    &-\overline{L}\mathbb{E}\left\{ \mathbbm{1}(D_{n,\varepsilon})\left|Q\left( 1-\alpha-\gamma_n;F_{w\circ S}(\cdot;Z_{n+1}^\prime) \right)-q_{s(\cdot,\cdot;Z_{n+1}^\prime)}(t)\right| \right\}\,.
\end{align*}
Taking expectation over $T\sim P_{T,2}$ gives
\begin{align*}
    &\mathbb{E}\left|\mathrm{Pr}\left( Y_{n+1}\in\widehat{C}(X_{n+1})\mid \phi(T) \right)-(1-\alpha)\right|\\
    \leq&2\overline{L}\widetilde{\varepsilon}_n+\widetilde{\delta}_n+\delta_n(\varepsilon)+\overline{L}(2\overline{M}\overline{L}\vee 1)C_2(\overline{M}/\underline{M})\sqrt{n^{-1}\mathrm{Pdim}(\mathcal{S})\log(n)}\\
    &+\overline{L}\mathbb{E}\left[ \mathbb{E}\left\{ \mathbbm{1}(D_{n,\varepsilon})\left|Q\left( 1-\alpha+\gamma_n;F_{w\circ S}(\cdot;Z_{n+1}^\prime) \right)-q_{s(\cdot,\cdot;Z_{n+1}^\prime)}(T)\right|\mid Z_{n+1}^\prime \right\} \right]\vee\\
    &\quad \overline{L}\mathbb{E}\left[ \mathbb{E}\left\{ \mathbbm{1}(D_{n,\varepsilon})\left|Q\left( 1-\alpha-\gamma_n;F_{w\circ S}(\cdot;Z_{n+1}^\prime) \right)-q_{s(\cdot,\cdot;Z_{n+1}^\prime)}(T)\right|\mid Z_{n+1}^\prime \right\} \right]\,.
\end{align*}

\noindent\textbf{Step 2: bounding the remaining quantile-alignment term.}
We treat the upper-sign case, where the quantile is taken at $1-\alpha+\gamma_n$; the proof for the lower-sign case is identical. Given $Z_{n+1}^\prime$, the constant $Q(1-\alpha+\gamma_n;F_{w\circ S}(\cdot;Z_{n+1}^\prime))$ satisfies
\begin{align*}
    &\mathbb{E}\left\{ F_{T_w}\left(Q\left( 1-\alpha+\gamma_n;F_{w\circ S}(\cdot;Z_{n+1}^\prime) \right);Z_{n+1}^\prime \right) \right\}\\
    =&F_{w\circ S}\left(Q\left( 1-\alpha+\gamma_n;F_{w\circ S}(\cdot;Z_{n+1}^\prime) \right);Z_{n+1}^\prime \right)=1-\alpha+\gamma_n\,,
\end{align*}
where the expectation is taken over $T_w\sim P_{T,w}$ independent of all other random objects. Therefore
\begin{align*}
    \mathbb{E}\left\{ F_{T_w}\left(Q\left( 1-\alpha+\gamma_n;F_{w\circ S}(\cdot;Z_{n+1}^\prime) \right);Z_{n+1}^\prime \right)-F_{T_w}\left(q_{s(\cdot,\cdot;Z_{n+1}^\prime)}(T_w);Z_{n+1}^\prime \right) \right\}=\gamma_n\,.
\end{align*}

By the mean-value theorem, for each realized $T_w$ there exists $\xi(T_w,Z_{n+1}^\prime)$ between $Q(1-\alpha+\gamma_n;F_{w\circ S}(\cdot;Z_{n+1}^\prime))$ and $q_{s(\cdot,\cdot;Z_{n+1}^\prime)}(T_w)$ such that
\begin{align*}
    &F_{T_w}\left(Q\left( 1-\alpha+\gamma_n;F_{w\circ S}(\cdot;Z_{n+1}^\prime) \right);Z_{n+1}^\prime \right)-F_{T_w}\left(q_{s(\cdot,\cdot;Z_{n+1}^\prime)}(T_w);Z_{n+1}^\prime \right)\\
    =&f_{T_w}\left(\xi(T_w,Z_{n+1}^\prime);Z_{n+1}^\prime\right)\left\{ Q\left( 1-\alpha+\gamma_n;F_{w\circ S}(\cdot;Z_{n+1}^\prime) \right)-q_{s(\cdot,\cdot;Z_{n+1}^\prime)}(T_w) \right\}\,.
\end{align*}
Under the standing density assumption, $f_{T_w}(\xi(T_w,Z_{n+1}^\prime);Z_{n+1}^\prime)\in[\underline{L},\overline{L}]$. Taking expectation over $T_w$ yields
\begin{align*}
    &Q\left( 1-\alpha+\gamma_n;F_{w\circ S}(\cdot;Z_{n+1}^\prime) \right)\\
    =&\dfrac{\mathbb{E} \left\{ f_{T_w}(\xi(T_w,Z_{n+1}^\prime);Z_{n+1}^\prime)q_{s(\cdot,\cdot;Z_{n+1}^\prime)}(T_w)\mid Z_{n+1}^\prime \right\}+\gamma_n}{\mathbb{E} \left\{ f_{T_w}(\xi(T_w,Z_{n+1}^\prime);Z_{n+1}^\prime)\mid Z_{n+1}^\prime \right\}}\,.
\end{align*}
Hence
\begin{align*}
    &\mathbb{E}\left\{ \mathbbm{1}(D_{n,\varepsilon})\left|Q\left( 1-\alpha+\gamma_n;F_{w\circ S}(\cdot;Z_{n+1}^\prime) \right)-q_{s(\cdot,\cdot;Z_{n+1}^\prime)}(T)\right|\mid Z_{n+1}^\prime \right\}\\
    \leq&\mathbbm{1}(D_{n,\varepsilon})\mathbb{E}\Bigg[ \Bigg|Q\left( 1-\alpha+\gamma_n;F_{w\circ S}(\cdot;Z_{n+1}^\prime) \right) \\
    & \qquad\qquad\qquad -\dfrac{\mathbb{E} \left\{ f_{T_w}(\xi(T_w,Z_{n+1}^\prime);Z_{n+1}^\prime)q_{s^\star}(T_w)\mid Z_{n+1}^\prime \right\}}{\mathbb{E} \left\{ f_{T_w}(\xi(T_w,Z_{n+1}^\prime);Z_{n+1}^\prime)\mid Z_{n+1}^\prime \right\}}\Bigg|\mid Z_{n+1}^\prime \Bigg]\\
    &+\mathbbm{1}(D_{n,\varepsilon})\mathbb{E}\left[ \left|\dfrac{\mathbb{E} \left\{ f_{T_w}(\xi(T_w,Z_{n+1}^\prime);Z_{n+1}^\prime)q_{s^\star}(T_w)\mid Z_{n+1}^\prime \right\}}{\mathbb{E} \left\{ f_{T_w}(\xi(T_w,Z_{n+1}^\prime);Z_{n+1}^\prime)\mid Z_{n+1}^\prime \right\}}-q_{s^\star}(T^\prime)\right|\mid Z_{n+1}^\prime \right]\\
    &+\mathbbm{1}(D_{n,\varepsilon})\mathbb{E}\left\{ \left|q_{s^\star}(T)-q_{s^\star}(T^\prime)\right|\mid Z_{n+1}^\prime \right\} \\
    & +\mathbbm{1}(D_{n,\varepsilon})\mathbb{E}\left\{ \left|q_{s^\star}(T)-q_{s(\cdot,\cdot;Z_{n+1}^\prime)}(T)\right|\mid Z_{n+1}^\prime \right\}\,.
\end{align*}

By changing measure from $P_{T,2}$ to $P_{T,1}$ and from $P_{T,w}$ to $P_{T,1}$, and then applying the triangle inequality,
\begin{align*}
    &\mathbb{E}\left\{ \mathbbm{1}(D_{n,\varepsilon})\left|Q\left( 1-\alpha+\gamma_n;F_{w\circ S}(\cdot;Z_{n+1}^\prime) \right)-q_{s(\cdot,\cdot;Z_{n+1}^\prime)}(T)\right|\mid Z_{n+1}^\prime \right\}\\
    \leq&\mathbbm{1}(D_{n,\varepsilon})\dfrac{\mathbb{E} \left\{ r_w(T_0)f_{T_0}(\xi(T_0,Z_{n+1}^\prime);Z_{n+1}^\prime)\left| q_{s(\cdot,\cdot;Z_{n+1}^\prime)}(T_0)-q_{s^\star}(T_0) \right|\mid Z_{n+1}^\prime \right\}}{\mathbb{E} \left\{ f_{T_w}(\xi(T_w,Z_{n+1}^\prime);Z_{n+1}^\prime)\mid Z_{n+1}^\prime \right\}}\\
    &+\mathbbm{1}(D_{n,\varepsilon})\dfrac{\mathbb{E} \left\{ f_{T_w}(\xi(T_w,Z_{n+1}^\prime);Z_{n+1}^\prime)\left|q_{s^\star}(T_w)-q_{s^\star}(T^\prime)\right|\mid Z_{n+1}^\prime \right\}}{\mathbb{E} \left\{ f_{T_w}(\xi(T_w,Z_{n+1}^\prime);Z_{n+1}^\prime)\mid Z_{n+1}^\prime \right\}}\\
    &+\mathbbm{1}(D_{n,\varepsilon})\mathbb{E}\left\{ \left|q_{s^\star}(T)-q_{s^\star}(T^\prime)\right|\mid Z_{n+1}^\prime \right\}\\
    &+\mathbbm{1}(D_{n,\varepsilon})\mathbb{E}\left\{ r(T_0)\left|q_{s^\star}(T_0)-q_{s(\cdot,\cdot;Z_{n+1}^\prime)}(T_0)\right|\mid Z_{n+1}^\prime \right\}
    +\gamma_n\underline{L}^{-1}\\
    \leq& \mathbbm{1}(D_{n,\varepsilon})(\overline{L}/\underline{L})\mathbb{E}\left\{ r_w(T_0)\left|q_{s^\star}(T_0)-q_{s(\cdot,\cdot;Z_{n+1}^\prime)}(T_0)\right|\mid Z_{n+1}^\prime \right\}\\
    &+\mathbbm{1}(D_{n,\varepsilon})\mathbb{E}\left\{ r(T_0)\left|q_{s^\star}(T_0)-q_{s(\cdot,\cdot;Z_{n+1}^\prime)}(T_0)\right|\mid Z_{n+1}^\prime \right\}\\
    &+\mathbbm{1}(D_{n,\varepsilon})(\overline{L}/\underline{L})\mathbb{E}\left[ \left\{ r_w(T_0)+r(T_0) \right\}r(T_0^\prime)\left|q_{s^\star}(T_0^{\prime})-q_{s^\star}(T_0)\right| \right]
    +\gamma_n\underline{L}^{-1}\,,
\end{align*}
where the expectations in the last inequality are with respect to independent $T_0,T_0^\prime\sim P_{T,1}$. The same bound holds with $1-\alpha+\gamma_n$ replaced by $1-\alpha-\gamma_n$. By H\"older's inequality,
\begin{align*}
    &\mathbbm{1}(D_{n,\varepsilon})\mathbb{E}\left\{ r(T_0)\left|q_{s^\star}(T_0)-q_{s(\cdot,\cdot;Z_{n+1}^\prime)}(T_0)\right|\mid Z_{n+1}^\prime \right\}\\
    \leq&\mathbbm{1}(D_{n,\varepsilon})\left[ \mathbb{E}\left\{ r(T_0)^{p/(p-1)}\mid Z_{n+1}^\prime \right\} \right]^{(p-1)/p}\left[ \mathbb{E}\left\{ \left|q_{s^\star}(T_0)-q_{s(\cdot,\cdot;Z_{n+1}^\prime)}(T_0)\right|^p\mid Z_{n+1}^\prime \right\} \right]^{1/p}\\
    \leq&\|r\|_{P_{T,1},p/(p-1)}\varepsilon\,.
\end{align*}
Similarly,
\begin{align*}
    &\mathbbm{1}(D_{n,\varepsilon})(\overline{L}/\underline{L})\mathbb{E}\left\{ r_w(T_0)\left|q_{s^\star}(T_0)-q_{s(\cdot,\cdot;Z_{n+1}^\prime)}(T_0)\right|\mid Z_{n+1}^\prime \right\}\\
    \leq&(\overline{L}/\underline{L})\mathbb{E}\left( \|r_w\|_{P_{T,1},p/(p-1)} \right)\varepsilon\,.
\end{align*}
Substituting these bounds into the preceding inequality, we obtain
\begin{align*}
    &\mathbb{E}\left|\mathrm{Pr}\left( Y_{n+1}\in\widehat{C}(X_{n+1})\mid \phi(T) \right)-(1-\alpha)\right|\\
    \leq&2\overline{L}\widetilde{\varepsilon}_n+\widetilde{\delta}_n+\delta_n(\varepsilon)+\overline{L}(2\overline{M}\overline{L}\vee 1)C_2(\overline{M}/\underline{M})\sqrt{n^{-1}\mathrm{Pdim}(\mathcal{S})\log(n)}+\gamma_n\overline{L}\underline{L}^{-2}\\
    &+(\overline{L}/\underline{L})\mathbb{E}\left[ \|r+r_w\|_{P_{T,1},p/(p-1)}\varepsilon+\left\{ r_w(T_0)+r(T_0) \right\}r(T_0^\prime)\left|q_{s^\star}(T_0^\prime)-q_{s^\star}(T_0)\right| \right]\,.
\end{align*}
Since $\gamma_n=O(n^{-1})$, its contribution is absorbed by the other terms. 
Finally, we have
\begin{align*}
    &\mathbb{E}\left|\mathrm{Pr}\left( Y_{n+1}\in\widehat{C}(X_{n+1})\mid \phi(T) \right)-(1-\alpha)\right|\\
    \leq& C\left[ \widetilde{\delta}_n+\widetilde{\varepsilon}_n+\delta_n(\varepsilon)+\left\{ \|r\|_{P_{T,1},p/(p-1)}+\mathbb{E}\left( \|r_w\|_{P_{T,1},p/(p-1)} \right) \right\}\varepsilon \right]\\
    & +C \left( \sqrt{n^{-1}\mathrm{Pdim}(\mathcal{S})\log(n)}+\mathbb{E}\left[ \left\{ r_w(T_0)+r(T_0) \right\}r(T_0^{\prime})\left|q_{s^\star}(T_0)-q_{s^\star}(T_0^{\prime})\right| \right] \right)\,,
\end{align*}
which completes the proof of this theorem.
\end{proof}

\subsection{Proof of Theorem~\ref{theo: marginal SymmPI weighted}}\label{sec:supp_proof_marginal_symmpi_weighted}
\begin{proof}
    Under the weighted distributional invariance of $Z$ with ratio $r_g$ and the weighted distributional equivariance of $V$ with ratios $(r_g,\widetilde r_g)$, the marginal coverage probability of $\widehat{C}(z_{\rm obs})$ can be written as
    \begin{align*}
        \mathrm{Pr}(Z\in\widehat{C}(\Omega(Z)))&=\mathrm{Pr}(\psi(V(Z))\leq q(V(Z);\alpha))\\
        &=\mathrm{Pr}\left( \psi(\widetilde{r}_{G}\circ\widetilde{\rho}(G,V(Z)))\leq q(V(Z);\alpha) \right)\\
        &=\mathbb{E}\left\{ \mathrm{Pr}\left( \psi(\widetilde{r}_{G}\circ\widetilde{\rho}(G,V(Z)))\leq q(V(Z);\alpha)\mid Z \right) \right\}\,.
    \end{align*}
    For a given $Z$, $q(V(Z);\alpha)$ is defined as the $(1-\alpha)$-quantile of the distribution of $\psi(\widetilde{r}_{G}\circ\widetilde{\rho}(G,V(Z)))$ for $G\sim P_{\mathcal G}$, with the elements in $\mathcal{G}_e$ replaced by $\infty$. Therefore,
    \begin{gather*}
        \mathrm{Pr}\left( \psi(\widetilde{r}_{G}\circ\widetilde{\rho}(G,V(Z)))\leq q(V(Z);\alpha)\mid Z \right)\geq 1-\alpha\,.
    \end{gather*}
    For a given $Z$, the corresponding weighted distribution, $\psi(\widetilde{r}_{G}\circ\widetilde{\rho}(G,V(Z)))$ for $G\sim P_\mathcal{G}$, with the above replacement, is
    \begin{gather*}
        \dfrac{\sum_{g\in\mathcal{G}\setminus\mathcal{G}_e}\mathrm{Pr}(G=g)\widetilde{r}_{g}(\widetilde{\rho}(g,V(Z)))\delta_{\psi(\widetilde{\rho}(g,V(Z)))}+\sum_{g\in\mathcal{G}_e}\mathrm{Pr}(G=g)\widetilde{r}_{g}(\widetilde{\rho}(g,V(Z)))\delta_{\infty}}{\sum_{g\in\mathcal{G}}\mathrm{Pr}(G=g)\widetilde{r}_{g}(\widetilde{\rho}(g,V(Z)))}\,.
    \end{gather*}
    By the definition of the quantile function, for a discrete distribution, the probability that a random variable is less than or equal to its $(1-\alpha)$-quantile is at most $1-\alpha$ plus the point mass at that quantile. Hence,
    \begin{gather*}
        1-\alpha\leq\mathrm{Pr}\left( \psi(\widetilde{r}_{G}\circ\widetilde{\rho}(G,V(Z)))\leq q(V(Z);\alpha)\mid Z \right)\leq 1-\alpha+\widehat{\gamma}_{\max}(Z)\,,
    \end{gather*}
    where $\widehat{\gamma}_{\max}(Z)$ is a random variable that depends on $Z$ and is defined as
    \begin{align*}
        \widehat{\gamma}_{\max}(Z)&=\sup_{c\in\{\psi(\widetilde{\rho}(g,V(Z))):g\in\mathcal{G}\setminus\mathcal{G}_e\}}\dfrac{\sum_{g\in\mathcal{G},\psi(\widetilde{\rho}(g,V(Z)))=c}\mathrm{Pr}(G=g)\widetilde{r}_{g}(\widetilde{\rho}(g,V(Z)))}{\sum_{g\in\mathcal{G}}\mathrm{Pr}(G=g)\widetilde{r}_{g}(\widetilde{\rho}(g,V(Z)))}\\
        &=\sup_c\dfrac{\sum_{g\in\mathcal{G},\psi(\widetilde{\rho}(g,V(Z)))=c}\mathrm{Pr}(G=g)\widetilde{r}_{g}(\widetilde{\rho}(g,V(Z)))}{\sum_{g\in\mathcal{G}}\mathrm{Pr}(G=g)\widetilde{r}_{g}(\widetilde{\rho}(g,V(Z)))}\,.
    \end{align*}
    Finally, taking expectations gives
    \begin{gather*}
        1-\alpha\leq \mathrm{Pr}(Z\in\widehat{C}(\Omega(Z)))\leq 1-\alpha+\mathbb{E}\left\{ \widehat{\gamma}_{\max}(Z) \right\}\,.
    \end{gather*}
    This completes the proof of the theorem.
\end{proof}

\subsection{Proof of Theorem~\ref{theo:symmpi conditional miscoverage weighted}}\label{sec:supp_proof_conditional_symmpi_weighted}
\begin{proof}
    Fix $t\in\mathcal{T}$ and write 
    $q_t^\star=Q(1-\alpha;F_{\psi(V^\star(\cdot))\mid\phi(t)})$.
    The proof follows the same structure as the i.i.d.~data case in Theorem~\ref{theo: main miscoverage}, adapted to the weighted SymmPI construction.

    For the lower bound,
    \begin{align*}
        &\mathrm{Pr}\left( \psi(V(Z))\leq q(V(Z);\alpha)\mid \phi(t) \right)\\
        =&\mathrm{Pr}\left( \psi(V(Z))\leq q(V(Z);\alpha),D_\varepsilon^c\mid \phi(t) \right)+\mathrm{Pr}\left( \psi(V(Z))\leq q(V(Z);\alpha),D_\varepsilon\mid \phi(t) \right)\\
        \geq&\mathrm{Pr}\left( \psi(V^\star(Z))+\varepsilon\leq q(V^\star(Z);\alpha)-\varepsilon,D_\varepsilon\mid \phi(t) \right)\\
        \geq&\mathrm{Pr}\left( \psi(V^\star(Z))+\varepsilon\leq q(V^\star(Z);\alpha)-\varepsilon\mid \phi(t) \right)-\mathrm{Pr}\left( D_\varepsilon^c\mid \phi(t) \right)\\
        =&\mathbb{E}\left\{ F_{\psi(V^\star(\cdot))\mid\phi(t)}(q(V^\star(Z);\alpha)-2\varepsilon)\mid \phi(t) \right\}-\delta(\varepsilon)\,,
    \end{align*}
    where the last equality follows from the conditional independence in Assumption~\ref{ass: independence score weighted}. Since $F_{\psi(V^\star(\cdot))\mid\phi(t)}$ is Lipschitz continuous, it is continuous, and hence
    $F_{\psi(V^\star(\cdot))\mid\phi(t)}(q_t^\star)=1-\alpha$. Therefore,
    \begin{align*}
        &\mathbb{E}\left\{ F_{\psi(V^\star(\cdot))\mid\phi(t)}(q(V^\star(Z);\alpha)-2\varepsilon)\mid \phi(t) \right\}-\delta(\varepsilon)\\
        \geq&1-\alpha-\delta(\varepsilon)-L_t\mathbb{E}\left[\left\{\left|q(V^\star(Z);\alpha)-q_t^\star\right|+2\varepsilon\right\}\wedge L_t^{-1}\mid \phi(t)\right]\\
        \geq&1-\alpha-\delta(\varepsilon)-2L_t\varepsilon-L_t\mathbb{E}\left\{\left|q(V^\star(Z);\alpha)-q_t^\star\right|\wedge L_t^{-1}\mid \phi(t)\right\}\,.
    \end{align*}

    Similarly, for the upper bound,
    \begin{align*}
        &\mathrm{Pr}\left( \psi(V(Z))\leq q(V(Z);\alpha)\mid \phi(t) \right)\\
        =&\mathrm{Pr}\left( \psi(V(Z))\leq q(V(Z);\alpha),D_\varepsilon^c\mid \phi(t) \right)+\mathrm{Pr}\left( \psi(V(Z))\leq q(V(Z);\alpha),D_\varepsilon\mid \phi(t) \right)\\
        \leq&\delta(\varepsilon)+\mathrm{Pr}\left( \psi(V^\star(Z))-\varepsilon\leq q(V^\star(Z);\alpha)+\varepsilon\mid \phi(t) \right)\\
        =&\delta(\varepsilon)+\mathbb{E}\left\{ F_{\psi(V^\star(\cdot))\mid\phi(t)}(q(V^\star(Z);\alpha)+2\varepsilon)-F_{\psi(V^\star(\cdot))\mid\phi(t)}(q_t^\star)+1-\alpha\mid \phi(t) \right\}\\
        \leq&1-\alpha+\delta(\varepsilon)+2L_t\varepsilon+L_t\mathbb{E}\left\{\left|q(V^\star(Z);\alpha)-q_t^\star\right|\wedge L_t^{-1}\mid \phi(t)\right\}.
    \end{align*}

    Combining the lower and upper bounds gives
    \begin{align*}
        &\left|\mathrm{Pr}\left( \psi(V(Z))\leq q(V(Z);\alpha)\mid \phi(t) \right)-(1-\alpha)\right|\\
        \leq&\delta(\varepsilon)+2L_t\varepsilon+L_t\mathbb{E}\left\{\left|q(V^\star(Z);\alpha)-q_t^\star\right|\wedge L_t^{-1}\mid \phi(t)\right\}.
    \end{align*}
    By the triangle inequality and the inequality $(a+b)\wedge c\leq (a\wedge c)+b$ for $a,b,c\geq0$, we obtain
    \begin{align*}
        &\mathbb{E}\left\{\left|q(V^\star(Z);\alpha)-q_t^\star\right|\wedge L_t^{-1}\mid \phi(t)\right\}\\
        \leq&\mathbb{E}\left\{\left|q(V^\star(Z);\alpha)-Q(1-\alpha;F_{r\circ\psi(V^\star(\cdot))})\right|\wedge L_t^{-1}\mid \phi(t)\right\}\\
        &+\left|Q(1-\alpha;F_{r\circ\psi(V^\star(\cdot))})-Q(1-\alpha;F_{\psi(V^\star(\cdot))\mid\phi(t)})\right|\,.
    \end{align*}
    Substituting this bound into the previous inequality proves the theorem.
\end{proof}

\subsection{Proof of Theorem~\ref{theo: valid p value}}
\begin{proof}
By the definition of $\widehat{p}(z)$ in formula \eqref{eq: general p value}, we can rewrite the event $\{\widehat{p}(z)\leq\alpha\}$ as:
\begin{align*}
    \{\widehat{p}(z)\leq\alpha\}&=\left\{ \mathrm{Pr}\left( \psi(V(z))\leq\psi(\widetilde{r}_{G}\circ\widetilde{\rho}(G,V(z))) \right)<\alpha \right\}\\
    &=\left\{ \psi(V(z))> Q\left( 1-\alpha;\psi(\widetilde{r}_{G}\circ\widetilde{\rho}(G,V(z))) \right) \right\}\,.
\end{align*}
Since $\psi(V(Z)) \overset{\rm d.}{=} \psi(\widetilde{r}_{G}\circ\widetilde{\rho}(G,V(Z)))$ and by the definition of quantile function, we obtain that
\begin{align*}
    \mathrm{Pr}\left( \widehat{p}(Z)\leq\alpha \right)&=\mathrm{Pr}\left( \psi(V(Z))>Q\left( 1-\alpha;\psi(\widetilde{r}_{G}\circ\widetilde{\rho}(G,V(Z))) \right) \right)\\
    &=1-\mathrm{Pr}\left( \psi(V(Z))\leq Q\left( 1-\alpha;\psi(\widetilde{r}_{G}\circ\widetilde{\rho}(G,V(Z))) \right) \right)\\
    &\leq \alpha\,,
\end{align*}
which indicates $\widehat{p}$ is valid.

When the exchangeable data construction in Section~\ref{sec: general framework} is used with weight function $w(\cdot)$, for $G\sim P_{\mathcal G}$, $\psi(\widetilde{r}_{G}\circ\widetilde{\rho}(G,V(Z)))$ has the following distribution:
\begin{gather*}
    \dfrac{\sum_{i=1}^{n+1}w(X_i)\delta_{s(X_i,Y_i;Z)}}{\sum_{i=1}^{n+1}w(X_i)}\,.
\end{gather*}
Therefore the $p$-value can be rewritten as:
\begin{gather*}
    \widehat{p}(Z)=\dfrac{\sum_{i=1}^{n+1}w(X_i)\mathbbm{1}\left( s(X_{n+1},Y_{n+1};Z)< s(X_i,Y_i;Z) \right)}{\sum_{i=1}^{n+1}w(X_i)}\,.
\end{gather*}
Thus
\begin{align*}
    \left\{ \widehat{p}(Z)\leq\alpha \right\}&=\left\{ \dfrac{\sum_{i=1}^{n+1}w(X_i)\mathbbm{1}\left( s(X_i,Y_i;Z)\leq s(X_{n+1},Y_{n+1};Z) \right)}{\sum_{i=1}^{n+1}w(X_i)}\geq 1-\alpha \right\}\\
    &=\left\{ Y_{n+1}\in\widehat{C}(X_{n+1}) \right\}\,.
\end{align*}
Therefore, for any conditioning event $\phi(t)$ with positive probability,
\begin{gather*}
\left|\mathrm{Pr}(\widehat{p}(Z)\leq\alpha\mid \phi(t))-\alpha\right|=\left|\mathrm{Pr}\left( Y_{n+1}\in\widehat{C}(X_{n+1})\mid \phi(t) \right)-(1-\alpha)\right|\,.
\end{gather*}
\end{proof}

\subsection{Proof of Theorem~\ref{theo: pointwise hierarchical}}
\begin{proof}
    The proof contains two parts. We first verify that the hierarchical score satisfies the same approximation and stability ingredients required by Theorem~\ref{theo:symmpi conditional miscoverage}. We then bound the remaining structure-specific conditional-mismatch error, where the argument splits according to whether within-branch sample size $N$ or the number of branches $K$ is the more effective source of concentration.

    Assume $Z_{k,i,1}$ has the $(k-1)N+i$-th and $NK$-th elements $Z_{i,1}^{(k)}$ and $Z_{N,1}^{(K)}$ being independent copies of $Z_{i}^{(k)}$ and $Z_{N}^{(K)}$, with other elements being identical to $Z$. Also define $Z_{k,i,1}^{(j)}$ as the $j$-th block of $Z_{k,i,1}$. Denote the event $D_{n,\varepsilon}$ as:
    \begin{align*}
        D_{n,\varepsilon}=&\left\{ \sup_{k\in[K],i\in[N]}\left|s(X_i^{(k)},Y_i^{(k)};Z_{k,i,1}^{(k)},Z_{k,i,1})- s^\star _k(X_i^{(k)},Y_i^{(k)})\right| \leq \varepsilon \right\}\cap D_0\\
        \subset&\underset{k\in[K],i\in[N]}{\cap}\left\{ \left|s(X_i^{(k)},Y_i^{(k)};Z_{k,i,1}^{(k)},Z_{k,i,1})- s^\star _k(X_i^{(k)},Y_i^{(k)})\right| \leq \varepsilon \right\}\cap D_0\,,
    \end{align*}
    where $D_0=\{X_N^{(K)}=x,P_1=p_1,\ldots,P_K=p_K\}$. By the definition of $Z_{k,i,1}$, $(X_i^{(k)},Y_i^{(k)})$ is independent of $Z_{k,i,1}^{(k)}$ and $Z_{k,i,1}$. By the union bound, we can derive that $D_{n,\varepsilon}$ has conditional probability:
    \begin{gather*}
        \mathrm{Pr}\left( D_{n,\varepsilon}\mid X_N^{(K)}=x,P_1=p_1,\ldots,P_K=p_K \right)=\mathrm{Pr}\left( D_{n,\varepsilon}\mid D_0 \right)\geq 1-KN\delta_n(\varepsilon)\,.
    \end{gather*}
    Denote following event:
    \begin{align*}
        \widetilde{D}_n=&\left\{ \sup_{Z_{k,i,1}^{(k)},Z_{N,1}^{(K)}:i\in[N],k\in[K]}\sup_{x,y}\left|s(x,y;Z_{k,i,1}^{(k)},Z_{k,i,1})-s(x,y;Z^{(k)},Z)\right|\leq 2\widetilde{\varepsilon}_n \right\}\cap D_0\,.
    \end{align*}
    By Assumption~\ref{ass: stable score hierarchical}, it has conditional probability:
    \begin{gather*}
        \mathrm{Pr}\left( \widetilde{D}_n\mid X_N^{(K)}=x,P_1=p_1,\ldots,P_K=p_K \right)=\mathrm{Pr}\left( \widetilde{D}_n\mid D_0 \right)\geq 1-2\widetilde{\delta}_n\,.
    \end{gather*}
    Thus on the event $\widetilde{D}_n\cap D_{n,\varepsilon}$, we have:
    \begin{gather*}
        \left|s(X_i^{(k)},Y_i^{(k)};Z^{(k)},Z)- s^\star _k(X_i^{(k)},Y_i^{(k)})\right|\leq \varepsilon + 2\widetilde{\varepsilon}_n\,.
    \end{gather*}
    Further define 
    \begin{gather*}
        D_\varepsilon=\left\{ \left|\psi(V(Z))-\psi(V^\star(Z))\right|\vee\left|q(V(Z);\alpha)-q(V^\star(Z);\alpha)\right|\leq \varepsilon + 2\widetilde{\varepsilon}_n \right\}\,,
    \end{gather*}
    and we can derive that $\widetilde{D}_n\cap D_{n,\varepsilon}\subset D_\varepsilon\cap D_0$ with probability
    \begin{gather}
        \mathrm{Pr}\left( D_\varepsilon\mid X_N^{(K)}=x,P_1=p_1,\ldots,P_K=p_K \right)=\mathrm{Pr}\left( D_\varepsilon\mid D_0 \right)\geq 1-KN\delta_n(\varepsilon)-2\widetilde{\delta}_n\,.\label{eq: hier ineq 1}
    \end{gather}
    Moreover, conditional on $D_0$, the oracle test score $s_K^\star(X_N^{(K)},Y_N^{(K)})$ is independent of $q(V^\star(Z);\alpha)$, because the latter is computed from the remaining oracle scores with the test position replaced by $\infty$. 
    Therefore, averaging the previous inequality~\eqref{eq: hier ineq 1} over $P_1,\ldots,P_{K-1}$ shows that Assumption~\ref{ass: independence score} holds with $\delta(\varepsilon)=KN\delta_n(\varepsilon)+2\widetilde{\delta}_n$. Since $F_{s_K^\star\mid\phi(p_K,x)}$ is $L$-Lipschitz, Theorem~\ref{theo:symmpi conditional miscoverage} reduces the problem to controlling the structure specific conditional-mismatch error:
    \begin{align*}
        &\left|\mathrm{Pr}\left( Y_N^{(K)}\in\widehat{C}_{\rm HI}(X_N^{(K)})\mid X_N^{(K)}=x, P_K=p_K \right)-(1-\alpha)\right|\\
        \leq&KN\delta_n(\varepsilon)+2\widetilde{\delta}_n+2L\varepsilon+4L\widetilde{\varepsilon}_n+L\mathbb{E}\left\{ \left|q(V^\star(Z);\alpha)-q_{s_K^\star}(p_K,x)\right| \right\}\,.
    \end{align*}
    To be specific, the $q(V^\star(Z);\alpha)$ reduces to:
    \begin{gather*}
        q(V^\star(Z);\alpha)=Q\left( 1-\alpha;(KN)^{-1}\sum_{(k,i)\neq (K,N)}\delta_{ s^\star _k(X_i^{(k)},Y_i^{(k)})}+\delta_\infty \right)\,,
    \end{gather*}
    which satisfies following inequality:
    \begin{align*}
        &Q\left( 1-\alpha;(KN)^{-1}\sum_{k\in[K],i\in[N]}\delta_{ s^\star _k(X_i^{(k)},Y_i^{(k)})} \right)\leq q(V^\star(Z);\alpha)\\
        & ~~~~~~~~~~~~~~~~~~~~~~~~~~~~~~~~~~~ \leq Q\left( 1-\alpha_n;(KN)^{-1}\sum_{k\in[K],i\in[N]}\delta_{ s^\star _k(X_i^{(k)},Y_i^{(k)})} \right)\,,
    \end{align*}
    where $\alpha_n=\alpha-1/(KN+1)$. For notation simplicity, we denote $ s^\star _k(X_i^{(k)},Y_i^{(k)})$ as $ S^\star _{k,i}$ and $ S^\star_{k,i}$ has c.d.f. $F_{s_k^\star}$ given the specific score function $s^\star _k(\cdot,\cdot)$. 
    
    \noindent\textbf{Case 1: $N\geq \log(K)$.} In this regime, the empirical score distribution within each branch is sufficiently accurate, so we compare the oracle conformal quantile with the population mixture quantile through concentration of the branch-level empirical c.d.f.s. 
    Denote empirical distribution of $\{ S^\star _{k,i}\}_{i\in[N]}$ and $\{ S^\star _{k,i}\}_{k\in[K],i\in[N]}$ as 
    \begin{gather*}
        \widehat{F}_{s_k^\star}(u)=N^{-1}\sum_{i=1}^N\mathbbm{1}( S^\star _{k,i}\leq u)\\
        \widehat{F}_{s^\star}(u)=K^{-1}\sum_{k\in[K]}\widehat{F}_{s_k^\star}(u)=(KN)^{-1}\sum_{k\in[K],i\in[N]}\mathbbm{1}( S^\star _{k,i}\leq u)\\
        \widetilde{F}_{s^\star}(u)=K^{-1}\sum_{k\in[K]}F_{s_k^\star}(u)\,,~F_{s^\star}(u)=\mathbb{E} \left\{ F_{s_k^\star}(u) \right\}\,,
    \end{gather*}
    where the last expectation is taken over $ s_k^\star $ with random $P_K\sim\Pi$. Then
    \begin{gather}
        Q\left( 1-\alpha;\widehat{F}_{s^\star} \right)\leq q(V^\star(Z);\alpha)\leq Q\left( 1-\alpha_n;\widehat{F}_{s^\star} \right)\,.\label{eq: hier ineq 2}
    \end{gather}
    By the Dvoretzky-Kiefer-Wolfowitz (DKW) inequality, for any $\epsilon>0$, we can derive that:
    \begin{gather*}
        \mathrm{Pr}\left( \sup_{u\in\mathbb{R}}|\widehat{F}_{s_k^\star}(u)-F_{s_k^\star}(u)|>\epsilon \right)\leq 2\exp(-2N\epsilon^2)\,.
    \end{gather*}
    As $\widehat{F}_{s_k^\star}$ are independent with each other, we have:
    \begin{align*}
        &\mathrm{Pr}\left( \sup_{u\in\mathbb{R},k\in[K]}|\widehat{F}_{s_k^\star}(u)-F_{s_k^\star}(u)|>\epsilon \right)\\
        =&1-\prod_{k\in[K]}\mathrm{Pr}\left( \sup_{u\in\mathbb{R}}|\widehat{F}_{s_k^\star}(u)-F_{s_k^\star}(u)|\leq\epsilon \right)\leq 1-\left[ \{1-2\exp(-2N\epsilon^2)\}\vee 0 \right]^K\,.
    \end{align*}
    Since $\sup_{u\in\mathbb{R},k\in[K]}|\widehat{F}_{s_k^\star}(u)-F_{s_k^\star}(u)|\leq\epsilon$ indicates that $\sup_{u\in\mathbb{R}}|\widehat{F}_{s^\star}(u)-F_{s^\star}(u)|\leq\epsilon$, we can further derive that:
    \begin{gather*}
        \mathrm{Pr}\left( \sup_{u\in\mathbb{R}}|\widehat{F}_{s^\star}(u)-\widetilde{F}_{s^\star}(u)|>\epsilon \right)\leq 1-\left[ \{1-2\exp(-2N\epsilon^2)\}\vee 0 \right]^K\,.
    \end{gather*}
    Denote $A_\epsilon=\{\sup_{u\in\mathbb{R}}|\widehat{F}_{s^\star}(u)-\widetilde{F}_{s^\star}(u)|\leq\epsilon\}$ with probability  
    \begin{gather*}
        \mathrm{Pr}(A_\epsilon)\geq \left[ \{1-2\exp(-2N\epsilon^2)\}\vee 0 \right]^K\geq 1-2K\exp(-2N\epsilon^2)\,.
    \end{gather*}
    For any $u\in\mathbb{R}$, on event $A_\epsilon$ it holds that
    \begin{gather*}
        \widetilde{F}_{s^\star}(u)-\epsilon\leq\widehat{F}_{s^\star}(u)\leq \widetilde{F}_{s^\star}(u)+\epsilon\,.
    \end{gather*}
    Define $\beta_k=F_{s_k^\star}(q_{s^\star})$ and $\overline{\beta}=K^{-1}\sum_{k\in[K]}\beta_k$, where $\beta_k$ is random with randomness induced by $P_k\sim\Pi$. Since $F_{s_k^\star}$ has a density lower bounded by $\underline{L}$ on its support, we obtain that for any $\gamma>0$, it holds that:
    \begin{gather*}
        F_{s_k^\star}(q_{s^\star}+\gamma)\geq\beta_k+\underline{L}\gamma\,,~F_{s_k^\star}(q_{s^\star}-\gamma)\leq\beta_k-\underline{L}\gamma\\
        \widetilde{F}_{s^\star}(q_{s^\star}+\gamma)\geq\overline{\beta}+\underline{L}\gamma\,,~\widetilde{F}_{s^\star}(q_{s^\star}-\gamma)\leq\overline{\beta}-\underline{L}\gamma\,,
    \end{gather*}
    which indicates on the event $A_\epsilon$ the following holds: 
    \begin{gather*}
        \widehat{F}_{s^\star}(q+\gamma)\geq\overline{\beta}+\underline{L}\gamma-\epsilon\,,~\widehat{F}_{s^\star}(q-\gamma)\leq\overline{\beta}-\underline{L}\gamma+\epsilon\,.
    \end{gather*}
    Simple algebra yields that:
    \begin{gather*}
        q_{s^\star}+\left\{ |\overline{\beta}-(1-\alpha_n)|+\epsilon \right\}/\underline{L}\geq Q\left( \overline{\beta}+|\overline{\beta}-(1-\alpha_n)|;\widehat{F}_{s^\star} \right)\geq Q\left( 1-\alpha_n;\widehat{F}_{s^\star} \right)\\
        q_{s^\star}-\left\{ |\overline{\beta}-(1-\alpha)|+\epsilon \right\}/\underline{L}\leq Q\left( \overline{\beta}-|\overline{\beta}-(1-\alpha)|;\widehat{F}_{s^\star} \right)\leq Q\left( 1-\alpha;\widehat{F}_{s^\star} \right)\,.
    \end{gather*}
    Therefore, on the event $A_\epsilon$, by inequality~\eqref{eq: hier ineq 2}, we have:
    \begin{gather*}
        -|\overline{\beta}-(1-\alpha)|-\epsilon\leq \underline{L}\left\{ q(V^\star(Z);\alpha)-q_{s^\star} \right\}\leq |\overline{\beta}-(1-\alpha_n)|+\epsilon\,,
    \end{gather*}
    which indicates 
    \begin{align*}
        &\mathrm{Pr}\left( \underline{L}\left|q(V^\star(Z);\alpha)-q\right|-|\overline{\beta}-(1-\alpha_n)|\vee|\overline{\beta}-(1-\alpha)|>\epsilon \right)\\
        \leq& 1-\left[ \{1-2\exp(-2N\epsilon^2)\}\vee 0 \right]^K\,.
    \end{align*}
    Therefore, by triangle inequality and equality of expectations we can calculate the expectation of $\left|q(V^\star(Z);\alpha)-q_{s_K^\star}(p_K,x)\right|$ as follows:
    \begin{align*}
        &\mathbb{E}\left\{ \left|q(V^\star(Z);\alpha)-q_{s_K^\star}(p_K,x)\right| \right\}\\
        \leq& \left|q_{s^\star}-q_{s_K^\star}(p_K,x)\right|+\mathbb{E}\left\{ \left|q(V^\star(Z);\alpha)-q_{s^\star}\right| \right\}\,.
    \end{align*}
    Simple algebra yields that:
    \begin{align*}
        &\mathbb{E}\left\{ \underline{L}\left|q(V^\star(Z);\alpha)-q_{s^\star}\right|-|\overline{\beta}-(1-\alpha_n)|\vee|\overline{\beta}-(1-\alpha)| \right\}\\
        \leq&\int_0^\infty\mathrm{Pr}\left( \underline{L}\left|q(V^\star(Z);\alpha)-q\right|-|\overline{\beta}-(1-\alpha_n)|\vee|\overline{\beta}-(1-\alpha)|>\epsilon \right)d\epsilon\\
        \leq&\int_0^\infty 1-\left[ \{1-2\exp(-2N\epsilon^2)\}\vee 0 \right]^Kd\epsilon\\
        =&(2N)^{-1/2}\int_0^\infty 1-\left[ \{1-2\exp(-\epsilon^2)\}\vee 0 \right]^Kd\epsilon\,,
    \end{align*}
    where the last equality is obtained by change of variable $\epsilon\to\epsilon/\sqrt{2N}$. Further we can derive that:
    \begin{align*}
        &\int_0^\infty 1-\left[ \{1-2\exp(-\epsilon^2)\}\vee 0 \right]^Kd\epsilon\\
        =&\int_0^{2\sqrt{\log(K)}}1-\left[ \{1-2\exp(-\epsilon^2)\}\vee 0 \right]^Kd\epsilon\\
        &+\int_{2\sqrt{\log(K)}}^\infty 1-\left[ \{1-2\exp(-\epsilon^2)\}\vee 0 \right]^Kd\epsilon\\
        \leq&2\sqrt{\log(K)}+\int_{2\sqrt{\log(K)}}^\infty 1-\{1-2\exp(-\epsilon^2)\}^Kd\epsilon\\
        \leq&2\sqrt{\log(K)}+\int_0^\infty 1-\{1-2\exp(-\epsilon^2-4\log(K)-4\epsilon\sqrt{log(K)})\}^Kd\epsilon\\
        \leq&2\sqrt{\log(K)}+\int_0^\infty 1-\{1-2K^{-4}\exp(-\epsilon^2)\}^Kd\epsilon\\
        \leq&2\sqrt{\log(K)}+\int_0^\infty 2K^{-3}\exp(-\epsilon^2)d\epsilon\\
        =&2\sqrt{\log(K)}+\sqrt{\pi}K^{-3}\,,
    \end{align*}
    where the second inequality is obtained by change of variable $\epsilon\to\epsilon+2\sqrt{\log(K)}$ and the last inequality is obtained by Bernoulli's inequality. 
    Finally we can conclude that there exisits constant $C>0$ such that the following holds:
    \begin{align*}
        &L\mathbb{E}\left\{ \left|q(V^\star(Z);\alpha)-q_{s_K^\star}(p_K,x)\right| \right\}\\
        \leq& L\left|q_{s^\star}-q_{s_K^\star}(p_K,x)\right|+L\underline{L}^{-1}\sqrt{\dfrac{2\pi \log(K)}{N}}+L\underline{L}^{-1}\mathbb{E}\left\{ |\overline{\beta}-(1-\alpha_n)|\vee|\overline{\beta}-(1-\alpha)| \right\}\\
        \leq& C \left[ \left|q_{s^\star}-q_{s_K^\star}(p_K,x)\right|+\sqrt{\dfrac{\log(K)}{N}}+\mathbb{E}\left\{ |\overline{\beta}-(1-\alpha)| \right\} \right]\,.
    \end{align*}
    By the definition of $\overline{\beta}$ and $q_{s^\star}$ we obtain that $\mathbb{E}\overline{\beta}=1-\alpha$. Since $\beta_k$ is independent of each other, simple algebra using concentration inequality indicates there exists constant $C^\prime>0$ such that:
    \begin{gather*}
        \mathbb{E}\left\{ |\overline{\beta}-(1-\alpha)| \right\} \leq C^\prime K^{-1/2}\sigma_{s^\star}\,.
    \end{gather*}
    For notation simplicity, we use constant $C$ to absorb other constants and we can conclude that:
    \begin{align*}
        &\left|\mathrm{Pr}\left( Y_N^{(K)}\in\widehat{C}_{\rm HI}(X_N^{(K)})\mid X_N^{(K)}=x, P_K=p_K \right)-(1-\alpha)\right|\\
        \leq&C\left\{ KN\delta_n(\varepsilon)+\widetilde{\delta}_n+\varepsilon+\widetilde{\varepsilon}_n+\left|q_{s^\star}-q_{s_K^\star}(p_K,x)\right|+\sqrt{\dfrac{\log(K)}{N}}+\dfrac{\sigma_{s^\star}}{\sqrt{K}} \right\}\,.
    \end{align*}

    \noindent\textbf{Case 2: $N<\log(K)$.} In this regime, concentration within each branch is weaker, so it is more efficient to average across branches first and then control the resulting empirical c.d.f. 
    Denote the empirical distribution of $\{ S^\star _{k,i}\}_{k\in[K]}$ for $i\in[N]$ as
    \begin{gather*}
        \widehat{H}_i(u)=K^{-1}\sum_{k=1}^K\mathbbm{1}( S^\star _{k,i}\leq u)\,.
    \end{gather*}
    Therefore $\widehat{F}_{s^\star}(u)=N^{-1}\sum_{i=1}^{N}\widehat{H}_i(u)$. 
    Since $S^\star_{k,i}$ are i.i.d. sampled from c.d.f. $F_{s^\star}(\cdot)$ for a given $i\in[N]$, thus by the DKW inequality, for any $\epsilon>0$, we can derive that:
    \begin{gather*}
        \mathrm{Pr}\left( \sup_{u\in\mathbb{R}}|\widehat{H}_i(u)-F_{s^\star}(u)|>\epsilon \right)\leq 2\exp(-2K\epsilon^2)\,.
    \end{gather*}
    Then by union bound we obtain that
    \begin{gather*}
        \mathrm{Pr}\left( \sup_{u\in\mathbb{R},i\in[N]}|\widehat{H}_i(u)-F_{s^\star}(u)|>\epsilon \right)\leq 2N\exp(-2K\epsilon^2)\,,
    \end{gather*}
    which indicates 
    \begin{gather*}
        \mathrm{Pr}\left( \sup_{u\in\mathbb{R}}|\widehat{F}_{s^\star}(u)-F_{s^\star}(u)|>\epsilon \right)\leq 2N\exp(-2K\epsilon^2)\,.
    \end{gather*}
    We can calculate following integral:
    \begin{align*}
        &\int_0^\infty 2N\exp(-2K\epsilon^2)d\epsilon=N\sqrt{\frac{\pi}{2K}}\,.
    \end{align*}
    Thus similar to the previous case, we can derive that there exists constant $C>0$ such that:
    \begin{align*}
        &\left|\mathrm{Pr}\left( Y_N^{(K)}\in\widehat{C}_{\rm HI}(X_N^{(K)})\mid X_N^{(K)}=x, P_K=p_K \right)-(1-\alpha)\right|\\
        \leq&C\left\{ KN\delta_n(\varepsilon)+\widetilde{\delta}_n+\varepsilon+\widetilde{\varepsilon}_n+\left|q_{s^\star}-q_{s_K^\star}(p_K,x)\right|+\dfrac{N+\sigma_{s^\star}}{\sqrt{K}} \right\}\,.
    \end{align*}

    Therefore, we can conclude that there exists constant $C>0$ such that the following holds:
    \begin{align*}
        &\left|\mathrm{Pr}\left( Y_N^{(K)}\in\widehat{C}_{\rm HI}(X_N^{(K)})\mid X_N^{(K)}=x, P_K=p_K \right)-(1-\alpha)\right|\leq \\
        \leq&C\left\{ KN\delta_n(\varepsilon)+\widetilde{\delta}_n+\varepsilon+\widetilde{\varepsilon}_n+\dfrac{N}{\sqrt{K}}\wedge\sqrt{\dfrac{\log(K)}{N}}+\left|q_{s^\star}-q_{s_K^\star}(p_K,x)\right|+\dfrac{\sigma_{s^\star}}{\sqrt{K}} \right\}\,.
    \end{align*}
\end{proof}

\subsection{Proof of Corollary~\ref{cor: cc L1 test shift}}\label{sec:supp_proof_cc_L1_test_shift}
\begin{proof}

We apply Theorem~\ref{theo: averaged miscoverage shift} with $T=X_{n+1}$ and $\phi(t)=\{X_{n+1}=t\}$. For the CC construction under covariate shift, the weight function is $w(x)=r_X(x)\in[\underline{M},\overline{M}]$. Thus $r(t)\leq \overline{M}/\underline{M}$ and $r_w(t)\leq \overline{M}/\underline{M}$ hold for every $t$. Let
\begin{gather*}
    s^\star(x,y)=v(x,y)-Q_\alpha^\star(x)\,.
\end{gather*}
Then $q_{s^\star}(t)=0$ for every $t$, and therefore
\begin{gather*}
    \mathbb{E}\left|q_{s^\star}(T)-q_{s^\star}(T^\prime)\right|=0
\end{gather*}
always holds. We now verify the assumptions and compute the corresponding terms in Theorem~\ref{theo: averaged miscoverage shift}.

First, we verify Assumption~\ref{ass: space complexity}. Since $\{\kappa\in\mathbb{R}^{d_0}:\|\kappa\|_2\le B\}$ is closed and bounded in the finite-dimensional Euclidean space $\mathbb{R}^{d_0}$, it is compact. The map $\kappa\mapsto f_\kappa(\cdot)=\sum_{i=1}^{d_0}\kappa_i\eta_i(\cdot)$ is linear and continuous. Hence $\mathcal{S}$, as the image of a compact set under a continuous mapping, is compact, for example under the sup-norm. Moreover, $\mathcal{F}$ is a linear function class parameterized by $\kappa\in\mathbb{R}^{d_0}$, so its pseudo-dimension is bounded by the parameter dimension, namely $\mathrm{Pdim}(\mathcal{F})\le d_0<\infty$. By Lemma~\ref{lem:pdim_Sprime}, we also have $\mathrm{Pdim}(\mathcal{S})\le d_0<\infty$.

Since $\|\eta(X)\|_2\le M$ holds almost surely and $\|\kappa\|_2\le B$, we have
\begin{gather*}
    |f_\kappa(X)|\le \|\kappa\|_2\,\|\eta(X)\|_2\le BM
\end{gather*}
almost surely, which implies that $|\widehat{Q}_\alpha(\cdot;Z)|$ is bounded by $BM$. Therefore,
\begin{gather*}
    |s(X,Y;Z)|=|v(X,Y)-\widehat{Q}_\alpha(X;Z)| \le |v(X,Y)|+|\widehat{Q}_\alpha(X;Z)| \le (B+1)M\,.
\end{gather*}
Since $\phi(t)=\{X_{n+1}=t\}$ depends only on the test covariate, it is independent of the first $n$ elements of $Z$. By the definition of $q_{s(\cdot,\cdot;Z_{n+1}^\prime)}$, we have
\begin{gather*}
    \|q_{s(\cdot,\cdot;Z_{n+1}^\prime)}-q_{s^\star}\|_{P_{T,1},p}
    = \| \widehat{Q}_\alpha(\cdot;Z_{n+1}^\prime)- Q_\alpha^\star(\cdot) \|_{P_{X,1},p}\,.
\end{gather*}
Because $Z_{n+1}^\prime\overset{\rm d}{=}Z$, the same bound holds with $Z_{n+1}^\prime$ replaced by $Z$. We assume that there exists $\delta_n(\varepsilon)$ such that
\begin{gather*}
    \mathrm{Pr}\bigl( \| \widehat{Q}_\alpha(\cdot;Z)- Q_\alpha^\star(\cdot) \|_{P_{X,1},p}>\varepsilon \bigr)\leq \delta_n(\varepsilon)\,,
\end{gather*}
and defer the analysis of $\delta_n(\varepsilon)$ to the end of the proof.

Since the conditional c.d.f.~of $s(X,Y;Z)$ given $X=x$ is a location-shifted version of that of $v(X,Y)$ given $X=x$, it has the same density bounds as $v(X,Y)\mid X=x$. Therefore, the event-conditional densities required by Theorem~\ref{theo: averaged miscoverage shift} are bounded within $[\underline{L},\overline{L}]$, and the remaining assumptions of that theorem are satisfied.

Next, we control the $\widetilde{\delta}_n$ term in Assumption~\ref{ass: stable score}. The regularized quantile regression estimator is
\begin{gather*}
    \widehat{\kappa}(Z) =\arg\min_{\|\kappa\|_2\leq B}\left\{\frac{1}{n}\sum_{i=1}^n \ell_{1-\alpha}\!\left(v(X_i,Y_i),f_\kappa(X_i)\right)+\lambda\|\kappa\|_2^2\right\}\,,
\end{gather*}
with $\lambda>0$. Since the objective is strongly convex in $\kappa$, the minimizer is unique and the solution map $Z\mapsto\widehat{\kappa}(Z)$ is uniformly stable under the replacement of a single sample. In particular, for any $j$ and any $z_{j,1},z_{j,2}$, standard stability results for regularized empirical risk minimization \cite{bousquet2002stability} imply
\begin{gather*}
    \|\widehat{\kappa}(Z_{j,1})-\widehat{\kappa}(Z_{j,2})\|_2 \le \frac{C M}{\lambda n}\,,
\end{gather*}
for a universal constant $C>0$. Using $\|\eta(x)\|_2\le M$, it follows that
\begin{gather*}
    \sup_x\big|\widehat{Q}_\alpha(x;Z_{j,1})-\widehat{Q}_\alpha(x;Z_{j,2})\big| \le M\|\widehat{\kappa}(Z_{j,1})-\widehat{\kappa}(Z_{j,2})\|_2 \le \frac{C M^2}{\lambda n}\,.
\end{gather*}
Consequently,
\begin{gather*}
    \sup_{x,y}\big|s(x,y;Z_{j,1})-s(x,y;Z_{j,2})\big| = \sup_x\big|\widehat{Q}_\alpha(x;Z_{j,1})-\widehat{Q}_\alpha(x;Z_{j,2})\big| \le \frac{C M^2}{\lambda n}
\end{gather*}
holds deterministically. Hence Assumption~\ref{ass: stable score} holds with $\widetilde{\delta}_n=0$ and $\widetilde{\varepsilon}_n=C\lambda^{-1}n^{-1}$. Applying Theorem~\ref{theo: averaged miscoverage shift} gives
\begin{align*}
    &\mathbb{E}\left|\mathrm{Pr}\left( Y_{n+1}\in\widehat{C}_{\rm CC}(X_{n+1})\mid X_{n+1} \right)-(1-\alpha)\right|\\
    \leq& C \left\{ \delta_n(\varepsilon)+\varepsilon+d_0^{1/2}n^{-1/2}\log^{1/2}(n)+\lambda^{-1}n^{-1} \right\}\,.
\end{align*}

It remains to calculate the rate of $\delta_n(\varepsilon)+\varepsilon$. Define $Q_\alpha^\circ=\arg\min_{f\in\mathcal{F}}R(f)$. For $f_\kappa\in\mathcal F$, define
\begin{gather*}
    h(x,u;f_\kappa)=\mathbbm{1}\{u\in[Q_\alpha^\circ(x),f_\kappa(x)]\cup[f_\kappa(x),Q_\alpha^\circ(x)]\}|u-f_\kappa(x)|\,.
\end{gather*}
Since the conditional density of $v(X,Y)$ given $X=x$ is lower bounded by the positive constant $\underline{L}$, it follows from \cite{li2007quantile} that there exists $c>0$ such that
\begin{gather*}
    [\mathbb{E}\{h(X,v(X,Y);f_\kappa)\}]^{1/2}\geq c\mathbb{E}|f_\kappa(X)-Q_\alpha^\circ(X)|
\end{gather*}
holds for all $f_\kappa\in\mathcal{F}$. By the definition of $Q_\alpha^\circ=\arg\min_{f\in\mathcal{F}}R(f)$, assume that $f_{\kappa^\circ}\in \mathcal{F}$ satisfies $Q_\alpha^\circ=f_{\kappa^\circ}$. Define $J_0=\|\kappa^\circ\|_2^2\vee 1\leq B^2\vee 1$.

Based on Theorem~2 of \cite{li2007quantile}, together with the properties of finite-dimensional linear reproducing kernel Hilbert spaces (RKHS) in \cite{zhou2002covering}, the following inequality holds for $\lambda\leq c_1 \{d_0\log(n)/n\}^{2/3}/(2J_0)$:
\begin{gather*}
    \mathrm{Pr}\left( R(f_{\widehat{\kappa}})-R(f_{\kappa^\circ})\geq c_1 \{d_0\log(n)/n\}^{2/3} \right)
    \leq 3.5\exp\left( -c_2n(J_0\lambda)^{3/2} \right)\,,
\end{gather*}
where $c_1,c_2>0$ are constants. Taking $\lambda=c_1 \{d_0\log(n)/n\}^{2/3}/(2J_0)$, for a sufficiently large constant $c_1$, we have
\begin{gather*}
    \lambda^{-1}n^{-1}=O\bigl(\{d_0\log(n)/n\}^{-2/3}n^{-1}\bigr)
    =O\bigl(\{nd_0^2\log^2(n)\}^{-1/3}\bigr)
\end{gather*}
and
\begin{align*}
    & \mathrm{Pr}\left( R(f_{\widehat{\kappa}})-R(f_{\kappa^\circ})\geq c_1 \{d_0\log(n)/n\}^{2/3} \right)\\
    \leq& 3.5\exp\left( -c_3(c_1/2J_0)^{3/2}\log^2(n) \right)=o(n^{-1})\,.
\end{align*}
Therefore,
\begin{gather*}
    \mathrm{Pr}\left( R(f_{\widehat{\kappa}})-R(Q_\alpha^\star)\geq R(Q_\alpha^\circ)-R(Q_\alpha^\star) + c_1 \{d_0\log(n)/n\}^{2/3} \right)=o(n^{-1})\,.
\end{gather*}

By Lemma~\ref{lemma: quantile regression error}, for the admissible $p\in[1,2]$, there exists a constant $C>0$ such that
\begin{align*}
    & \mathrm{Pr}\left( \| \widehat{Q}_\alpha(\cdot;Z)- Q_\alpha^\star(\cdot) \|_{P_{X,1},p}
    > C\left[ \{R(Q_\alpha^\circ)-R(Q_\alpha^\star)\}^{1/2}+ \{d_0\log(n)/n\}^{1/3} \right] \right)\\
    \leq& o(n^{-1})\,.
\end{align*}
Hence, taking $\varepsilon=C\left[ \{R(Q_\alpha^\circ)-R(Q_\alpha^\star)\}^{1/2}+ \{d_0\log(n)/n\}^{1/3} \right]$
gives $\delta_n(\varepsilon)=o(n^{-1})$. Combining this with the preceding miscoverage bound, and absorbing all constants into $C$, yields
\begin{align*}
    &\mathbb{E}\left|\mathrm{Pr}\left( Y_{n+1}\in\widehat{C}_{\rm CC}(X_{n+1})\mid X_{n+1} \right)-(1-\alpha)\right|\\
    \leq& C \left\{ \sqrt{R(Q_\alpha^\circ)-R(Q_\alpha^\star)}+\sqrt[3]{d_0\log(n)/n} \right\}\,.
\end{align*}
This completes the proof of the theorem.
\end{proof}

\section{Proofs of Preliminary Lemmas}\label{sec:proof_pre_lemmas}

\subsection{Proof of Lemma~\ref{lemma: sub Lipschitz weighted}}

\begin{proof}
    For $s^\prime\neq s$, define the abstract slope function as
    \begin{gather*}
        k_s(s^\prime)=\dfrac{|F(s^\prime)-F(s)|}{|s^\prime-s|}\,,
    \end{gather*}
    and set $k_s(s)=f(s)$.
    Since $F(\cdot)$ is continuously differentiable at $s$ with derivative $f(s)>0$,
    \begin{gather*}
        \lim_{s^\prime\rightarrow s}\dfrac{|F(s^\prime)-F(s)|}{|s^\prime-s|}=f(s)\,,
    \end{gather*}
    so $k_s(\cdot)$ is continuous at $s$. 
    For $s'\neq s$, continuity of $F(\cdot)$ implies continuity of $k_s(\cdot)$ away from $s$. Hence $k_s(\cdot)$ is continuous on the compact interval $[a,b]$.

    By the extreme value theorem, $k_s(\cdot)$ attains its minimum on $[a,b]$. Moreover, this minimum is strictly positive. Indeed, $k_s(s)=f(s)>0$. If $s'\neq s$ and $F(s')=F(s)$, then monotonicity of $F$ would imply that $F$ is constant on the interval between $s'$ and $s$, contradicting the fact that the derivative at $s$ is strictly positive. Thus $k_s(s')>0$ for all $s'\in[a,b]$.

    Define $\underline L=\inf_{s'\in[a,b]} k_s(s')>0$.
    Then, for every $s'\in[a,b]$, $|F(s')-F(s)|\geq \underline L |s'-s|$.
    Therefore, $F$ is sub-Lipschitz at $s$ with constant $\underline L$.
\end{proof}

\subsection{Proof of Lemma \ref{lemma: weighted dkw}}

\begin{proof}
    Define the function class $\mathcal{F}_a^{(0)}=\{(x,\varsigma)\mapsto a(x)\mathbbm{1}(\varsigma\leq u):u\in\mathbb{R}\}$. Since for any given $(x,\varsigma)$, $a(x)\mathbbm{1}(\varsigma\leq u)$ is monotonically increasing with respect to the index $u$, the VC dimension of the function class $\mathcal{F}_a^{(0)}$ is $1$. To be specific, for any two distinct points $(x_1,\varsigma_1,v_1)$ and $(x_2,\varsigma_2,v_2)$, denote $u_i=\inf\{u\in\mathbb{R}:v_i\leq a(x_i)\mathbbm{1}(\varsigma_i\leq u)\}$. If $u_1=u_2$, obviously $\mathcal{F}_a^{(0)}$ cannot scatter these two points. Without loss of generality, assume $u_1<u_2$. Then we cannot find $u$ such that $v_1>a(x_1)\mathbbm{1}(\varsigma_1\leq u)$ and $v_2\leq a(x_2)\mathbbm{1}(\varsigma_2\leq u)$ hold simultaneously, which indicates that $\mathcal{F}_a^{(0)}$ cannot scatter these two points. The proof therefore reduces to combining a standard VC/Rademacher bound for $\mathcal F_a^{(0)}$ with the empirical-process inequality in \cite{10.3150/13-BEJ549}.

    Define the empirical Rademacher complexity of $\mathcal{F}_a^{(0)}$ on $(X_1^{(0)},S_1^{(0)}),\ldots,(X_n^{(0)},S_n^{(0)})$ as
    \begin{gather*}
        \widehat{\mathcal{R}}_n(\mathcal{F}_a^{(0)})\overset{\rm def.}{=}\mathbb{E}\left\{ \sup_{f\in\mathcal{F}_a^{(0)}}\left|n^{-1}\sum_{i=1}^{n}\tau_if(X_i^{(0)},S_i^{(0)})\right|\mid (X_1^{(0)},S_1^{(0)}),\ldots,(X_n^{(0)},S_n^{(0)}) \right\}\,,
    \end{gather*}
    where $\tau_1,\ldots,\tau_n$ are i.i.d.~Rademacher random variables with $\mathrm{Pr}(\tau_i=1)=\mathrm{Pr}(\tau_i=-1)=1/2$ and the expectation is taken over all $\tau_i$. Define the Rademacher complexity of $\mathcal{F}_a^{(0)}$ as the expectation of the empirical Rademacher complexity:
    \begin{gather*}
        \mathcal{R}_n(\mathcal{F}_a^{(0)})\overset{\rm def.}{=}\mathbb{E}\left\{ \widehat{\mathcal{R}}_n(\mathcal{F}_a^{(0)}) \right\}=\mathbb{E}\left\{ \sup_{f\in\mathcal{F}_a^{(0)}}\left|n^{-1}\sum_{i=1}^{n}\tau_if(X_i^{(0)},S_i^{(0)})\right| \right\}\,.
    \end{gather*}

    As $\mathcal{F}_a^{(0)}$ has finite VC dimension $1$, it has polynomial discrimination of order one by Sauer--Shelah's lemma. Lemma~4.14 of \cite{wainwright2019high} gives
    \begin{gather}
        \widehat{\mathcal{R}}_n(\mathcal{F}_a^{(0)})\leq 4\sup_{f\in\mathcal{F}_a^{(0)}}\sqrt{n^{-1}\sum_{i=1}^{n}f^2(X_i^{(0)},S_i^{(0)})}\sqrt{\dfrac{\log(n)}{n}}\,.\label{eq: weighted dkw 1}
    \end{gather}
    By the definition of $\mathcal F_a^{(0)}$, $\sup_{f\in\mathcal{F}_a^{(0)}}\sum_{i=1}^{n}f^2(X_i^{(0)},S_i^{(0)})\leq\sum_{i=1}^{n}a^2(X_i^{(0)})$. Taking expectations in \eqref{eq: weighted dkw 1} and using Jensen's inequality yields
    \begin{align*}
        \mathcal{R}_n(\mathcal{F}_a^{(0)})&\leq 4\mathbb{E}\sqrt{n^{-1}\sum_{i=1}^{n}a^2(X_i^{(0)})}\sqrt{n^{-1}\log(n)}\\
        &\leq 4\sigma_a\sqrt{n^{-1}\log(n)}\,,
    \end{align*}
    where $\sigma_a^2=\mathbb{E}\{a^2(X^{(0)})\}$. Moreover,
    \begin{gather*}
        \sup_{f\in\mathcal{F}_a^{(0)}}\sqrt{\mathbb{E}f^2(X^{(0)},S^{(0)})}
        =\sup_{u\in\mathbb R}\sqrt{\mathbb{E}\{a^2(X^{(0)})\mathbbm{1}(S^{(0)}\leq u)\}}
        \leq \sigma_a\,.
    \end{gather*}
    Applying Equation~(1) of \cite{10.3150/13-BEJ549} to the class $\mathcal F_a^{(0)}$, for any $\gamma>0$ and any $\eta>0$,
    \begin{align*}
        &\mathrm{Pr}\left( \sup_{u\in\mathbb R} \left|A_n^{(0)}(u)-A^{(0)}(u)\right| > 2(1+\gamma)\mathcal R_n(\mathcal F_a^{(0)})+\sigma_a\sqrt{2\eta}+(\gamma^{-1}+1/3)\eta \right)\\
        &\qquad\leq \exp(-n\eta).
    \end{align*}

    Taking $\gamma=1$ and $\eta=\sigma_a^2x/2$ for $x\in(0,1)$ gives
    \begin{align*}
        \exp(-n\sigma_a^2x/2)\geq&\mathrm{Pr}\left( \sup_{u\in\mathbb{R}}\left|A_n^{(0)}(u)-A^{(0)}(u)\right|>4\mathcal{R}_n(\mathcal{F}_a^{(0)})+\sigma_a\sqrt{2\eta}+4\eta/3 \right)\\
        \geq&\mathrm{Pr}\left( \sup_{u\in\mathbb{R}}\left|A_n^{(0)}(u)-A^{(0)}(u)\right|>4\mathcal{R}_n(\mathcal{F}_a^{(0)})+\sigma_a^2(\sqrt{x}+x) \right)\\
        \geq&\mathrm{Pr}\left( \sup_{u\in\mathbb{R}}\left|A_n^{(0)}(u)-A^{(0)}(u)\right|>4\mathcal{R}_n(\mathcal{F}_a^{(0)})+2\sigma_a^2\sqrt{x} \right)\,.
    \end{align*}
    Let $\varepsilon=2\sigma_a^2\sqrt{x}<2\sigma_a^2$. The previous inequality becomes
    \begin{gather*}
        \mathrm{Pr}\left( \sup_{u\in\mathbb{R}}\left|A_n^{(0)}(u)-A^{(0)}(u)\right|>4\mathcal{R}_n(\mathcal{F}_a^{(0)})+\varepsilon \right)\leq\exp(-n\varepsilon^2/(8\sigma_a^2)),\\
        \mathrm{Pr}\left( \sup_{u\in\mathbb{R}}\left|A_n^{(0)}(u)-A^{(0)}(u)\right|>16\sigma_a\sqrt{n^{-1}\log(n)}+\varepsilon \right)\leq\exp(-n\varepsilon^2/(8\sigma_a^2)).
    \end{gather*}
    Likewise, when $x\geq1$, take $\eta=\sigma_a^2x/2$. Then $\sigma_a\sqrt{2\eta}+4\eta/3=\sigma_a^2\sqrt{x}+2\sigma_a^2x/3\leq 2\sigma_a^2x$. Let $\varepsilon=2\sigma_a^2x\geq2\sigma_a^2$. The same argument gives
    \begin{gather*}
        \mathrm{Pr}\left( \sup_{u\in\mathbb{R}}\left|A_n^{(0)}(u)-A^{(0)}(u)\right|>16\sigma_a\sqrt{n^{-1}\log(n)}+\varepsilon \right)\leq\exp(-n\varepsilon/4)\,.
    \end{gather*}
    Thus, we finish the proof of this lemma.
\end{proof}

\subsection{Proof of Lemma~\ref{lemma: sup weighted dkw}}

\begin{proof}
    Define
    \begin{gather*}
        \mathcal F^{(0)} = \left\{ x\mapsto \mathbbm{1}\{m(x)\leq u\}:m\in\mathcal M,\ u\in\mathbb R
        \right\}\,,\\
        \mathcal F_a^{(0)} = \left\{ x\mapsto a(x)\mathbbm{1}\{m(x)\leq u\}:m\in\mathcal M,\ u\in\mathbb R \right\}\,.
    \end{gather*}
    
    The proof contains two parts: first, a uniform empirical-process bound over the weighted class $\mathcal{F}_a^{(0)}$; second, a contraction step reducing the complexity of $\mathcal{F}_a^{(0)}$ to that of the indicator class $\mathcal F^{(0)}$. Since $0<a(X^{(0)})\leq M_a$ almost surely, the class $\mathcal{F}_a^{(0)}$ is uniformly bounded by $M_a$. The pseudo-dimension assumption on $\mathcal M$ implies that the VC dimension of $\mathcal F^{(0)}$ is at most $\mathrm{Pdim}(\mathcal M)$. Based on the definition of $\mathcal{F}_a^{(0)}$, we can derive that
    \begin{gather*}
        \sup_{m\in\mathcal M,\,u\in\mathbb R}
        \left|A_{n,m}^{(0)}(u)-A_m^{(0)}(u)\right| = \sup_{f_a\in\mathcal F_a^{(0)}}
        \left| \frac1n\sum_{i=1}^n f_a(X_i^{(0)}) - \mathbb E\{f_a(X^{(0)})\} \right|.
    \end{gather*}
    Define empirical Rademacher complexities
    \begin{gather*}
        \widehat{\mathcal{R}}_n(\mathcal{F}^{(0)})\overset{\rm def.}{=}\mathbb{E}\left\{ \sup_{f\in\mathcal{F}^{(0)}}\left|n^{-1}\sum_{i=1}^{n}\tau_if(X_i^{(0)})\right|\mid X_1^{(0)},\ldots,X_n^{(0)} \right\}\\
        \widehat{\mathcal{R}}_n(\mathcal{F}_a^{(0)})\overset{\rm def.}{=}\mathbb{E}\left\{ \sup_{f_a\in\mathcal{F}_a^{(0)}}\left|n^{-1}\sum_{i=1}^{n}\tau_if_a(X_i^{(0)})\right|\mid X_1^{(0)},\ldots,X_n^{(0)} \right\}\,,
    \end{gather*}
    where $\tau_1,\ldots,\tau_n$ are i.i.d. Rademacher random variable.
    Define the expectations of empirical Rademacher complexity as $\mathcal{R}_n(\mathcal{F}^{(0)})=\mathbb{E}\{\widehat{\mathcal{R}}_n(\mathcal{F}^{(0)})\}$ and $\mathcal{R}_n(\mathcal{F}_a^{(0)})=\mathbb{E}\{\widehat{\mathcal{R}}_n(\mathcal{F}_a^{(0)})\}$. For any $\varepsilon>0$, Theorem 4.10 of \cite{wainwright2019high} yields
    \begin{align*}
        & \mathrm{Pr}\left( \sup_{f_a\in\mathcal{F}_a^{(0)}}\left|n^{-1}\sum_{i=1}^{n}f_a(X_i^{(0)})-\mathbb{E}\{f_a(X^{(0)})\}\right|\leq 2\mathcal{R}_n(\mathcal{F}_a^{(0)})+\varepsilon \right) \\
        \geq & 1-\exp\left( -\dfrac{n\varepsilon^2}{2M_a^2} \right)\,.
    \end{align*}
    It remains to bound $\mathcal R_n(\mathcal F_a^{(0)})$.
    First we notice that
    \begin{align*}
        &\mathbb{E}\left\{ \sup_{f_a\in\mathcal{F}_a^{(0)}}\left|n^{-1}\sum_{i=1}^{n}\tau_if_a(X_i^{(0)})\right|\mid X_1^{(0)},\ldots,X_n^{(0)} \right\}\\
        =&\mathbb{E}\left\{ \sup_{f\in\mathcal{F}^{(0)}}\left|n^{-1}\sum_{i=1}^{n}\tau_i a(X_i^{(0)})f(X_i^{(0)})\right|\mid X_1^{(0)},\ldots,X_n^{(0)} \right\}
    \end{align*}
    where for any given $X_i^{(0)}$ and any $f\in\mathcal{F}^{(0)}$, the function $f(X_i^{(0)})\mapsto a(X_i^{(0)})f(X_i^{(0)})$ is $M_a$-Lipschitz. Thus by Lemma 8 of \cite{mohri2014learning}, the contraction leads to
    \begin{gather*}
        \widehat{\mathcal{R}}_n(\mathcal{F}_a^{(0)})\leq M_a\widehat{\mathcal{R}}_n(\mathcal{F}^{(0)}), \quad \mathcal{R}_n(\mathcal{F}_a^{(0)})\leq M_a\mathcal{R}_n(\mathcal{F}^{(0)})\,.
    \end{gather*}
    Furthermore, since the function space $\mathcal{F}^{(0)}$ has a finite VC dimension of $\mathrm{Pdim}(\mathcal{M})$, it follows from Lemma 4.14 in \cite{wainwright2019high} that we can bound the empirical Rademacher complexity as follows:
    \begin{gather*}
        \widehat{\mathcal{R}}_n(\mathcal{F}^{(0)})\leq 4\sup_{f\in\mathcal{F}^{(0)}}\sqrt{n^{-1}\sum_{i=1}^{n}f^2(X_i^{(0)})}\sqrt{\dfrac{\mathrm{Pdim}(\mathcal{M})\log(n)}{n}}\leq 4\sqrt{\dfrac{\mathrm{Pdim}(\mathcal{M})\log(n)}{n}}\,,
    \end{gather*}
    which indicates
    \begin{gather*}
        \mathcal{R}_n(\mathcal{F}_a^{(0)})\leq 4M_a\sqrt{\dfrac{\mathrm{Pdim}(\mathcal{M})\log(n)}{n}}\,.
    \end{gather*}
    Together, we conclude the probabilistic bound for $\sup_{m\in\mathcal{M},u\in\mathbb{R}}\left|A_{n,m}^{(0)}(u)-A_m^{(0)}(u)\right|$:
    \begin{align*}
        & \mathrm{Pr}\left( \sup_{m\in\mathcal{M},u\in\mathbb{R}}\left|A_{n,m}^{(0)}(u)-A_m^{(0)}(u)\right|\leq 8M_a\sqrt{\dfrac{\mathrm{Pdim}(\mathcal{M})\log(n)}{n}}+\varepsilon \right) \\
        \geq & 1-\exp\left( -\dfrac{n\varepsilon^2}{2M_a^2} \right)\,.
    \end{align*}
    For $A_{n,m}^{(0),+}(u) = n^{-1} \sum_{i=1}^{n-1} a(X_i^{(0)}) \mathbbm{1}(m(X_i^{(0)}) \leq u)$, the following holds almost surely for $n>1$:
    \begin{gather*}
        \sup_{m\in\mathcal{M},u\in\mathbb{R}}\left|A_{n,m}^{(0)}(u)-A_{n,m}^{(0),+}(u)\right|\leq n^{-1}M_a\leq M_a\sqrt{\dfrac{\mathrm{Pdim}(\mathcal{M})\log(n)}{n}}\,.
    \end{gather*}
    Combining this with the previous inequality and applying the triangle inequality gives
    \begin{align*}
        & \mathrm{Pr}\left( \sup_{m\in\mathcal{M},u\in\mathbb{R}}\left|A_{n,m}^{(0),+}(u)-A_m^{(0)}(u)\right|\leq 9M_a\sqrt{\dfrac{\mathrm{Pdim}(\mathcal{M})\log(n)}{n}}+\varepsilon \right) \\
        \geq & 1-\exp\left( -\dfrac{n\varepsilon^2}{2M_a^2} \right)\,.
    \end{align*}
    We finish the proof of this lemma.
\end{proof}

\subsection{Proof of Lemma~\ref{lemma: kernel fun}}
\begin{proof}
    Under the assumptions of the kernel function, by changing to spherical coordinates for the following integral, we obtain
    \begin{equation}
        \notag \int_{\mathbb{R}^{d}}K_0(\|x\|_2)dx \leq \int_{0}^{\infty}u^{d-1}K_0(u)du < \infty\,.
    \end{equation}
    It follows that for any $\ell\geq 1$, 
    \begin{equation*}
        \int_{\mathbb{R}^d}\left\{ K_0(\|x\|_2) \right\}^{\ell}dx\leq \left\{ K_0(0) \right\}^{\ell-1}\int_{\mathbb{R}^d}K_0(\|x\|_2)dx<\infty\,.
    \end{equation*}

    Suppose $f_X(x)$ is the density function of $X$ and $\underline{L}_1\leq f_X(x)\leq \overline{L}_1$ for any $x\in\mathcal{X}$. Consequently, for any $x_0\in\mathcal{X}$,
    \begin{eqnarray}
        \notag \mathbb{E}\left[ h^{-d}\left\{ K(X,x_0;h) \right\}^{\ell} \right] & = & \int_\mathcal{X}h^{-d}\left\{ K_0\left( \|x-x_0\|_2/h \right) \right\}^{\ell}f_X(x)dx\\
        \notag  & = & \int_{\{x:~x_0+hx\in\mathcal{X}\}}\left\{ K_0(\|x\|_2) \right\}^{\ell}f_X(x_0+hx)dx\\
        \notag & \leq & \overline{L}_1\int_{\mathbb{R}^d}\left\{ K_0(\|x\|_2) \right\}^{\ell}dx<\infty\,.
    \end{eqnarray}
    In addition, since $K_0(u)$ is decreasing in $u$,
    \begin{eqnarray}
        \notag \mathbb{E}\left[ h^{-d}\left\{ K(X,x_0;h) \right\}^{\ell} \right] & \geq & \{K_0(1)\}^{\ell}h^{-d}\int_{\{x\in\mathcal{X}:~\|x-x_0\|_2\leq h\}}f_X(x)dx\,,
    \end{eqnarray}
    where $h^{-d}\int_{\{x\in\mathcal{X}:~\|x-x_0\|_2\leq h\}}f_X(x)dx=h^{-d}\mathrm{Pr}\left(\|X-x_0\|_2\leq h\right)$ is bounded below by a positive constant due to $f_X(x)\geq \underline{L}_1$ for any $x\in\mathcal{X}$. 
    Therefore, there exist positive constants $\underline{L}_2$ and $\overline{L}_2$ such that
    \begin{equation}
        \notag \underline{L}_2 \leq \mathbb{E}\left[ h^{-d}\left\{ K(X,x_0;h) \right\}^{\ell} \right] \leq \overline{L}_2\,.
    \end{equation}
\end{proof}

\subsection{Proof of Lemma~\ref{lemma: kernel fun bias}}
\begin{proof}
    The Lipschitz property and triangle inequality yield:
    \begin{align*}
        \left|f(x_0)-f(X)\right|\leq L\|x_0-X\|_2\leq L\|x_0-\widetilde{X}\|_2+L\|\widetilde{X}-X\|_2\,.
    \end{align*}
    We can decompose the target as follows:
    \begin{align*}
        &\mathbb{E} \left\{ \left|f(x_0)-f(X)\right|K(\widetilde{X},X;h) \right\}\\
        \leq &\mathbb{E}\left[ \mathbb{E} \left\{ \left|f(x_0)-f(X)\right|K(\widetilde{X},X;h)\mid\widetilde{X} \right\} \right]\\
        \leq &LE \left[ \|x_0-\widetilde{X}\|_2\mathbb{E} \left\{ K(\widetilde{X},X;h)\mid\widetilde{X} \right\} \right] + LE\left[ \mathbb{E} \left\{ \|\widetilde{X}-X\|_2K(\widetilde{X},X;h)\mid\widetilde{X} \right\} \right]\,.
    \end{align*}
    First by Lemma~\ref{lemma: kernel fun}, we have $\mathbb{E} \left\{ K(\widetilde{X},X;h)\mid\widetilde{X} \right\}\leq\overline{L}_2h^d$. Then we calculate $E\|x_0-\widetilde{X}\|_2$ by:
    \begin{align*}
        E\|x_0-\widetilde{X}\|_2=&\dfrac{\int_{\mathcal{X}}\|x_0-x\|_2K(x_0,x;h)dx}{\int_{\mathcal{X}}K(x_0,x;h)dx}\\
        \leq&\dfrac{\int_{\mathcal{X}}\|x_0-x\|_2K(x_0,x;h)dx}{\underline{L}_2h^d}\\
        =&\dfrac{\int_{\mathcal{X}}\|x_0-x\|_2K_0(\|x_0-x\|_2/h)dx}{\underline{L}_2h^d}\\
        =&\dfrac{h^d\int_{x_0+hx\in\mathcal{X}}h\|x\|_2K_0(\|x\|_2)dx}{\underline{L}_2h^d}\\
        \leq&h\underline{L}_2^{-1}\int_{\mathbb{R}^d} \|x\|_2K_0(\|x\|_2)dx\\
        \leq&h\underline{L}_2^{-1}\int_{0}^{\infty}u^dK_0(u)du\,,
    \end{align*}
    where the first inequality can be obtained by assuming a uniform distribution over $\mathcal{X}$ in Lemma~\ref{lemma: kernel fun}. Thus we obtain
    \begin{gather*}
        LE \left[ \|x_0-\widetilde{X}\|_2\mathbb{E} \left\{ K(\widetilde{X},X;h)\mid\widetilde{X} \right\} \right]\leq L\underline{L}_2^{-1}\overline{L}_2h^{d+1}\int_{0}^{\infty}u^dK_0(u)du\,.
    \end{gather*}
    For the other part, we consider $x\in\mathcal{X}$ and $f_X(x)$ as the density of $X$. Thus we can calculate:
    \begin{align*}
        \mathbb{E} \left\{ \|x-X\|_2K(x,X;h) \right\}=&\int_{\mathcal{X}}\|x-x^\prime\|_2K_0(\|x-x^\prime\|_2/h)f_X(x^\prime)dx^\prime\\
        \leq&\overline{L}_1\int_{\mathcal{X}}\|x-x^\prime\|_2K_0(\|x-x^\prime\|_2/h)dx^\prime\\
        =&\overline{L}_1h^{d+1}\int_{x+hx^\prime\in\mathcal{X}}\|x^\prime\|_2K_0(\|x^\prime\|_2)dx^\prime\\
        \leq& \overline{L}_1h^{d+1}\underline{L}_2^{-1}\int_{\mathbb{R}^d} \|x\|_2K_0(\|x\|_2)dx\\
        \leq& \overline{L}_1h^{d+1}\int_{0}^{\infty}u^dK_0(u)du\,.
    \end{align*}
    Together we obtain
    \begin{align*}
        \mathbb{E} \left\{ \left|f(x_0)-f(X)\right|K(\widetilde{X},X;h) \right\}\leq L(\overline{L}_1+\underline{L}_2^{-1}\overline{L}_2)h^{d+1}\int_{0}^{\infty}u^dK_0(u)du\,.
    \end{align*}
    Define $\overline{L}_3=L(\overline{L}_1+\underline{L}_2^{-1}\overline{L}_2)\int_{0}^{\infty}u^dK_0(u)du$ and we conclude that
    \begin{gather*}
        \mathbb{E} \left\{ \left|f(x_0)-f(X)\right|K(\widetilde{X},X;h) \right\}\leq \overline{L}_3h^{d+1}\,.
    \end{gather*}
\end{proof}

\subsection{Proof of Lemma~\ref{lemma: quantile regression error}}
\begin{proof}
Write
\begin{gather*}
    \phi_x(t)=\mathbb{E}\left\{ \ell_{1-\alpha}(Y,t)\mid X=x \right\},\qquad
    F_x(t)=\mathrm{Pr}(Y\le t\mid X=x)\,.
\end{gather*}
Since $Y\mid X=x$ has density $f_x$, differentiation of the conditional pinball risk gives
\begin{gather*}
    \phi_x'(t)=F_x(t)-(1-\alpha)\,.
\end{gather*}
Let $Q_\alpha(x)=Q(1-\alpha;F_x)$. Since $F_x(Q_\alpha(x))=1-\alpha$, we have
\begin{gather*}
    \phi_x(t)-\phi_x(Q_\alpha(x))
    =\int_{Q_\alpha(x)}^t\{F_x(z)-F_x(Q_\alpha(x))\}\,dz\,.
\end{gather*}

The lower density bound converts this excess conditional risk into a quadratic error term. If $t\geq Q_\alpha(x)$, then for every $z\in[Q_\alpha(x),t]$,
\begin{gather*}
    F_x(z)-F_x(Q_\alpha(x))=\int_{Q_\alpha(x)}^z f_x(u)\,du\geq \underline{L}(x)\{z-Q_\alpha(x)\}\,.
\end{gather*}
Substituting this bound into the integral representation of $\phi_x(t)-\phi_x(Q_\alpha(x))$ gives
\begin{gather*}
    \phi_x(t)-\phi_x(Q_\alpha(x))
    \geq \underline{L}(x)\int_{Q_\alpha(x)}^t\{z-Q_\alpha(x)\}\,dz
    =\frac{\underline{L}(x)}{2}\{t-Q_\alpha(x)\}^2\,.
\end{gather*}
The same argument for $t\leq Q_\alpha(x)$ yields
\begin{gather*}
    \phi_x(t)-\phi_x(Q_\alpha(x))
    \geq \underline{L}(x)\int_t^{Q_\alpha(x)}\{Q_\alpha(x)-z\}\,dz
    =\frac{\underline{L}(x)}{2}\{Q_\alpha(x)-t\}^2\,.
\end{gather*}
Taking expectation with respect to $X$ yields
\begin{align*}
    &\mathbb{E}\left\{ \ell_{1-\alpha}(Y,\widehat{Q}(X))-\ell_{1-\alpha}(Y,Q_\alpha(X)) \right\}\\
    =&\mathbb{E}\left\{ \phi_X(\widehat{Q}(X))-\phi_X(Q_\alpha(X)) \right\}\geq 2^{-1}\mathbb{E}\left\{ \underline{L}(X)(Q_\alpha(X)-\widehat{Q}(X))^2 \right\}\,.
\end{align*}

It remains to convert the weighted $L_2$ term on the right-hand side into the $L_p$ norm in the lemma statement. Let
\begin{gather*}
    r=\frac{2}{p},\qquad r'=\frac{2}{2-p}\,,
\end{gather*}
with the usual convention that $r'=\infty$ when $p=2$. Then $1/r+1/r'=1$. Applying H\"older's inequality to
\begin{gather*}
    |Q_\alpha(X)-\widehat Q(X)|^p
    =\left\{\underline L(X)(Q_\alpha(X)-\widehat Q(X))^2\right\}^{p/2}\left\{1/\underline L(X)\right\}^{p/2}
\end{gather*}
gives
\begin{align*}
    &\mathbb{E}\left|Q_\alpha(X)-\widehat Q(X)\right|^p\\
    \leq&\left[\mathbb{E}\left\{\underline L(X)(Q_\alpha(X)-\widehat Q(X))^2\right\}\right]^{1/r}
    \left[\mathbb{E}\left\{\underline L(X)^{-pr'/2}\right\}\right]^{1/r'}\\
    =&\left[\mathbb{E}\left\{\underline L(X)(Q_\alpha(X)-\widehat Q(X))^2\right\}\right]^{p/2}
    \left[\mathbb{E}\left\{\underline L(X)^{-p/(2-p)}\right\}\right]^{(2-p)/2}\,.
\end{align*}
Raising both sides to the power $2/p$ yields
\begin{gather*}
    \|Q_\alpha-\widehat Q\|_{P_X,p}^2
    \leq \|1/\underline L\|_{P_X,p/(2-p)}\mathbb{E}\left\{\underline L(X)(Q_\alpha(X)-\widehat Q(X))^2\right\}.
\end{gather*}
Combining this bound with the excess-risk lower bound gives the claimed inequality, and thus completes the proof of this lemma.
\end{proof}

\subsection{Proof of Lemma~\ref{lem:pdim_Sprime}}
\begin{proof}
Recall that a function class $\mathcal{G}$ pseudo-shatters a set
$\{z_1,\dots,z_n\}$ if there exist real numbers $r_1,\dots,r_n$ such that,
for every $\sigma\in\{0,1\}^n$, there exists $g\in\mathcal{G}$ satisfying
\begin{gather*}
    \mathbbm{1}\{g(z_i)\ge r_i\}=\sigma_i,\quad i=1,\dots,n\,.
\end{gather*}
Suppose $\mathcal{S}$ pseudo-shatters a set
$\{(x_i,y_i)\}_{i=1}^n$.
Then there exist thresholds $r_1,\dots,r_n$ such that, for every
$\sigma\in\{0,1\}^n$, there exists $f_\sigma\in\mathcal{F}$ satisfying
\begin{gather*}
    \mathbbm{1}\{v(x_i,y_i)-f_\sigma(x_i)\ge r_i\}=\sigma_i,\quad i=1,\dots,n\,.
\end{gather*}

We first show that $x_i\neq x_j$ for any $i\neq j$. Suppose, to the contrary, that there exist $i\neq j$ such that $x_i=x_j$. Then the two constraints reduce to thresholding the same value $f(x_i)=f(x_j)$ against two constants. Assume without loss of generality that
$v(x_i,y_i)-r_i\geq v(x_j,y_j)-r_j$. Then any label vector with $\sigma_i=0$ and $\sigma_j=1$ would require
\begin{gather*}
    f_\sigma(x_i)>v(x_i,y_i)-r_i\qquad\text{and}\qquad f_\sigma(x_j)\leq v(x_j,y_j)-r_j\,,
\end{gather*}
which is impossible because $x_i=x_j$. Hence, $x_i\neq x_j$ for any $i\neq j$.
Equivalently,
\begin{gather*}
    \mathbbm{1}\{-f_\sigma(x_i)\geq r_i-v(x_i,y_i)\}=\sigma_i,\quad i=1,\dots,n\,.
\end{gather*}

By the assumption that $-f\in\mathcal{F}$ for every $f\in\mathcal{F}$, we have $-f_\sigma\in\mathcal{F}$ for every $\sigma\in\{0,1\}^n$. Define new thresholds $t_i=r_i-v(x_i,y_i)$. Thus, for every $\sigma\in\{0,1\}^n$, there exists $f_\sigma^\prime=-f_\sigma\in\mathcal{F}$ such that
\begin{gather*}
    \mathbbm{1}\{f_\sigma^\prime(x_i)\geq t_i\}=\sigma_i\,.
\end{gather*}
This shows that $\mathcal{F}$ pseudo-shatters the set $\{x_i\}_{i=1}^n$ with thresholds $\{t_i\}_{i=1}^n$. Therefore, $\mathrm{Pdim}(\mathcal{F})\geq n$. Since $n$ was arbitrary among sets pseudo-shattered by $\mathcal S$, it follows that $\mathrm{Pdim}(\mathcal{S})\leq \mathrm{Pdim}(\mathcal{F})$. This completes the proof of this lemma.
\end{proof}

\subsection{Proof of Lemma~\ref{theo: weighted quantile error}}
\begin{proof}
    Define 
    \begin{align*}
        A_n^a(u)=n^{-1}\sum_{i=1}^n a(X_i^{(0)})\mathbbm{1}(S_i^{(0)}\leq u), \quad A_n^{a,+}(u)=n^{-1}\sum_{i=1}^{n-1}a(X_i^{(0)})\mathbbm{1}(S_i^{(0)}\leq u)\,,
    \end{align*}
    and $B_n^a=n^{-1}\sum_{i=1}^n a(X_i^{(0)})$, so $F_n^a(u)=A_n^a(u)/B_n^a$ and $F_n^{a,+}(u)=A_n^{a,+}(u)/B_n^a$. Also let $A^a(u)=\mathbb{E}\{a(X^{(0)})\mathbbm{1}(S^{(0)}\leq u)\}$ and $B_a=\mathbb{E}\{a(X^{(0)})\}$, so $F^a(u)=A^a(u)/B_a$.

    The proof follows a Bahadur-type decomposition. We first express the quantile error through the difference between the weighted empirical c.d.f.\ $F_n^a$ and its population counterpart $F^a$. This reduces the problem to three components: fluctuation of the numerator process $A_n^a-A^a$, fluctuation of the normalizing denominator $B_n^a-B_a$, and a second-order remainder caused by the ratio structure. We then bound these three terms separately.

    By the definition of $A_n^{a,+}(\cdot)$ and $A_n^a(\cdot)$ and the assumption on $a(\cdot)$, we obtain
    \begin{align*}
        \sup_{u\in\mathbb{R}}\left|A_n^{a,+}(u)-A_n^a(u)\right|=n^{-1}a(X_n^{(0)})\leq n^{-1}M_a\,.
    \end{align*}
    
    Given Assumption~\ref{ass: weight function} that $0<a(X^{(0)})\leq M_a$ holds almost surely, we have $\sup_x a(x)-\inf_x a(x)\leq M_a$ almost surely and $\sigma_a^2=\mathbb{E}\{a^2(X^{(0)})\}$. By Bernstein's inequality,
    \begin{gather}
        \mathrm{Pr}\left(\left|B_n^a-B_a\right|\geq\varepsilon\right)\leq 2\exp\left(\frac{-n\varepsilon^2}{2\sigma_a^2+2M_a\varepsilon/3}\right)\leq 2\exp\left(\frac{-n\varepsilon^2}{2\sigma_a^2+M_a\varepsilon}\right). \label{eq: bahad rep B 3}
    \end{gather}
    Denote $D_B=\{B_n^a\geq B_a/2\}$, which has probability at least $1-\exp\{-nB_a^2/(8\sigma_a^2+2M_aB_a)\}$.

    Since $\widehat{\xi}_\alpha$ is the sample quantile, define $F_n^a(\widehat{\xi}_\alpha)=\alpha+\Delta_n$, where $\Delta_n$ satisfies
    \begin{gather*}
        0\leq \Delta_n\leq \max_{1\leq i\leq n}\frac{a(X_i^{(0)})}{\sum_{j=1}^na(X_j^{(0)})}\leq \frac{M_a}{\sum_{j=1}^na(X_j^{(0)})}=\frac{M_a}{nB_n^a}.
    \end{gather*}
    Based on this definition,
    \begin{gather}
        F_n^a(\widehat{\xi}_\alpha)-F^a(\widehat{\xi}_\alpha)=\alpha-F^a(\widehat{\xi}_\alpha)+\Delta_n. \label{eq: bahad rep 1}
    \end{gather}
    For any $u$,
    \begin{gather*}
        F_n^a(u)-F^a(u)=\frac{A_n^a(u)}{B_n^a}-\frac{A^a(u)}{B_a}=\frac{A_n^a(u)B_a-A^a(u)B_n^a}{B_n^aB_a},
    \end{gather*}
    and the numerator equals $B_a\{A_n^a(u)-A^a(u)\}-A^a(u)(B_n^a-B_a)$. Thus,
    \begin{align}
        \notag F_n^a(u)-F^a(u)
        =&\frac{1}{B_n^a}\{A_n^a(u)-A^a(u)\}-\frac{A^a(u)}{B_aB_n^a}(B_n^a-B_a)\\
        \notag
        =&\frac{A_n^a(u)-A^a(u)}{B_a}-\frac{A^a(u)(B_n^a-B_a)}{B_a^2}\\
        &+\{(B_n^a)^{-1}-B_a^{-1}\}\left[\{A_n^a(u)-A^a(u)\}-A^a(u)B_a^{-1}(B_n^a-B_a)\right]. \label{eq: bahad rep 2}
    \end{align}

    As $u\mapsto\mathbb{E}\{a(X^{(0)})\mathbbm{1}(S^{(0)}\leq u)\}$ is continuously differentiable at $\xi_\alpha$ with positive derivative, Lemma~\ref{lemma: sub Lipschitz weighted} shows that there exists a constant $\underline{L}_a(\xi_\alpha)$ such that $A^a(\cdot)$ is sub-Lipschitz with constant $B_a\underline{L}_a(\xi_\alpha)$ at $\xi_\alpha$. Since $F^a(u)=A^a(u)/B_a$, the sub-Lipschitz condition of $A^a(\cdot)$ at $\xi_\alpha$ yields
    \begin{gather*}
        \left|F^a(\widehat{\xi}_\alpha)-F^a(\xi_\alpha)\right|\geq \underline{L}_a(\xi_\alpha)\left|\widehat{\xi}_\alpha-\xi_\alpha\right|.
    \end{gather*}
    Since $F^a(\xi_\alpha)=\alpha$, we have $\bigl|\alpha-F^a(\widehat{\xi}_\alpha)\bigr|\geq \underline{L}_a(\xi_\alpha)\bigl|\widehat{\xi}_\alpha-\xi_\alpha\bigr|$. Applying \eqref{eq: bahad rep 2} at $u=\widehat{\xi}_\alpha$, noting that $A^a(u)\leq B_a$ for any $u$, and substituting into \eqref{eq: bahad rep 1}, we obtain
    \begin{align*}
        &\underline{L}_a(\xi_\alpha)\left|\widehat{\xi}_\alpha-\xi_\alpha\right|
        \leq \left|\{\alpha-F^a(\widehat{\xi}_\alpha)\}+\Delta_n\right|+\Delta_n\\
        \leq&\frac{\left|A_n^a(\widehat{\xi}_\alpha)-A^a(\widehat{\xi}_\alpha)\right|+\left|B_n^a-B_a\right|}{B_a}\\
        &+\frac{\left|B_n^a-B_a\right|\left|\{A_n^a(\widehat{\xi}_\alpha)-A^a(\widehat{\xi}_\alpha)\}-A^a(\widehat{\xi}_\alpha)B_a^{-1}(B_n^a-B_a)\right|}{B_n^aB_a}+\Delta_n.
    \end{align*}

    Because the denominator $B_n^a$ can be small on a rare event, we work separately on $D_B=\{B_n^a\geq B_a/2\}$. On its complement $D_B^c$, the quantile error is controlled trivially by the support bound $-M_S\leq S^{(0)}\leq M_S$:
    \begin{gather*}
        \mathbb{E}\left\{\left|\widehat{\xi}_\alpha-\xi_\alpha\right|\mathbbm{1}(D_B^c)\right\}\leq 2M_S\mathrm{Pr}(D_B^c)\leq 2M_S\exp\left\{-\frac{nB_a^2}{8\sigma_a^2+2M_aB_a}\right\}.
    \end{gather*}
    Combining the previous inequality with the bound on $D_B^c$, we obtain
    \begin{align*}
        &\underline{L}_a(\xi_\alpha) \mathbb{E}\left\{ \left|\hat{\xi}_\alpha - \xi_\alpha\right| \right\}\\
        \leq&\underline{L}_a(\xi_\alpha)\left[ \mathbb{E}\left\{ \left|\hat{\xi}_\alpha - \xi_\alpha\right|\mathbbm{1}(D_B^c) \right\}+\mathbb{E}\left\{ \left|\hat{\xi}_\alpha - \xi_\alpha\right|\mathbbm{1}(D_B) \right\} \right]\\
        \leq&2M_S\underline{L}_a(\xi_\alpha)\exp\left(-\dfrac{nB_a^2}{8\sigma_a^2+2M_aB_a}\right)\\
        & +\dfrac{\mathbb{E}\left\{ \left|A_n^a(\hat{\xi}_\alpha)-A^a(\hat{\xi}_\alpha)\right|+\left|B_n^a-B_a\right| \right\}}{B_a}+\mathbb{E}\left\{ \Delta_n\mathbbm{1}(D_B) \right\}\\
        & +\mathbb{E}\left\{ \dfrac{\left|B_n^a-B_a\right|\left|\{A_n^a(\hat{\xi}_\alpha) - A^a(\hat{\xi}_\alpha)\}-A^a(\hat{\xi}_\alpha)B_a^{-1}(B_n^a - B_a)\right|}{B_n^aB_a}\mathbbm{1}(D_B) \right\}\\
        \leq&2M_S\underline{L}_a(\xi_\alpha)\exp\left(-\frac{nB_a^2}{8\sigma_a^2+2M_aB_a}\right)\\
        &+\frac{\mathbb{E}\left\{\left|A_n^a(\widehat{\xi}_\alpha)-A^a(\widehat{\xi}_\alpha)\right|+\left|B_n^a-B_a\right|\right\}}{B_a}+\frac{2M_a}{nB_a}\\
        &+\mathbb{E}\left\{\frac{2\left|B_n^a-B_a\right|\left|\{A_n^a(\widehat{\xi}_\alpha)-A^a(\widehat{\xi}_\alpha)\}-A^a(\widehat{\xi}_\alpha)B_a^{-1}(B_n^a-B_a)\right|}{B_a^2}\right\}.
    \end{align*}
    The same decomposition applies to $\widehat{\xi}_\alpha^+$ after truncating $\bigl|\widehat{\xi}_\alpha^+-\xi_\alpha\bigr|$ at $2M_S$ and replacing $A_n^a(\widehat{\xi}_\alpha)$ by $A_n^{a,+}(\widehat{\xi}_\alpha^+)$:
    \begin{align*}
        &\underline{L}_a(\xi_\alpha)\mathbb{E}\left\{\left|\widehat{\xi}_\alpha^+-\xi_\alpha\right|\wedge(2M_S)\right\}\\
        \leq&2M_S\underline{L}_a(\xi_\alpha)\exp\left(-\frac{nB_a^2}{8\sigma_a^2+2M_aB_a}\right)\\
        &+\frac{\mathbb{E}\left\{\left|A_n^{a,+}(\widehat{\xi}_\alpha^+)-A^a(\widehat{\xi}_\alpha^+)\right|+\left|B_n^a-B_a\right|\right\}}{B_a}+\frac{2M_a}{nB_a}\\
        &+\mathbb{E}\left\{\frac{2\left|B_n^a-B_a\right|\left|\{A_n^{a,+}(\widehat{\xi}_\alpha^+)-A^a(\widehat{\xi}_\alpha^+)\}-A^a(\widehat{\xi}_\alpha^+)B_a^{-1}(B_n^a-B_a)\right|}{B_a^2}\right\}.
    \end{align*}
    Also, by the fact that $\sup_{u\in\mathbb{R}}|A_n^{a,+}(u)-A_n^a(u)|\leq n^{-1}M_a$, we can further simplify the above inequality as
    \begin{align*}
        &\underline{L}_a(\xi_\alpha)\mathbb{E}\left\{\left|\widehat{\xi}_\alpha^+-\xi_\alpha\right|\wedge(2M_S)\right\}\\
        \leq&2M_S\underline{L}_a(\xi_\alpha)\exp\left(-\frac{nB_a^2}{8\sigma_a^2+2M_aB_a}\right)\\
        &+\frac{\mathbb{E}\left\{\left|A_n^a(\widehat{\xi}_\alpha^+)-A^a(\widehat{\xi}_\alpha^+)\right|+\left|B_n^a-B_a\right|\right\}}{B_a}+\frac{3M_a}{nB_a}+\frac{2M_a\mathbb{E}|B_n^a-B_a|}{nB_a^2}\\
        &+\mathbb{E}\left\{\frac{2\left|B_n^a-B_a\right|\left|\{A_n^a(\widehat{\xi}_\alpha^+)-A^a(\widehat{\xi}_\alpha^+)\}-A^a(\widehat{\xi}_\alpha^+)B_a^{-1}(B_n^a-B_a)\right|}{B_a^2}\right\}.
    \end{align*}
    
    The remainder of the proof is organized accordingly: Part (i) controls the numerator fluctuation, Part (ii) controls the denominator fluctuation, and Part (iii) controls the second-order remainder generated by the ratio expansion.

    \noindent\textbf{Part (i): control the numerator fluctuation term.} We analyze $\mathbb{E}\{|A_n^a(\widehat{\xi}_\alpha)-A^a(\widehat{\xi}_\alpha)|\}$ by controlling it through the supremum discrepancy of the weighted empirical numerator process:
    \begin{gather*}
        \mathbb{E}\left\{\left|A_n^a(\widehat{\xi}_\alpha)-A^a(\widehat{\xi}_\alpha)\right|\right\}\leq \mathbb{E}\left\{\sup_{u\in\mathbb{R}}\left|A_n^a(u)-A^a(u)\right|\right\}.
    \end{gather*}
    According to Lemma~\ref{lemma: weighted dkw},
    \begin{gather}
        \mathrm{Pr}\left(\sup_{u\in\mathbb{R}}\left|A_n^a(u)-A^a(u)\right|>16\sigma_a\frac{\log^{1/2}(n)}{n^{1/2}}+\varepsilon\right)\leq
        \left\{\begin{aligned}
            &\exp(-n\varepsilon^2/(8\sigma_a^2)),&\varepsilon<2\sigma_a^2,\\
            &\exp(-n\varepsilon/4),&\varepsilon\geq 2\sigma_a^2.
        \end{aligned}\right.
        \label{eq: bahad rep A 4}
    \end{gather}
    Therefore,
    \begin{align*}
        &\mathbb{E}\left\{ \sup_{u\in\mathbb{R}}\left|A_n^a(u)-A^a(u)\right|-16\sigma_a\dfrac{\log^{1/2}(n)}{n^{1/2}} \right\}\\
        \leq&\mathbb{E}\left[ \left\{ \sup_{u\in\mathbb{R}}\left|A_n^a(u)-A^a(u)\right|-16\sigma_a\dfrac{\log^{1/2}(n)}{n^{1/2}} \right\}\vee 0 \right]\\
        =&\int_{0}^{\infty}\mathrm{Pr}\left( \sup_{u\in\mathbb{R}}\left|A_n^a(u)-A^a(u)\right|-16\sigma_a\dfrac{\log^{1/2}(n)}{n^{1/2}}>\varepsilon \right)d\varepsilon\\
        \leq&\int_{0}^{2\sigma_a^2}\exp(-n\varepsilon^2/(8\sigma_a^2))d\varepsilon+\int_{2\sigma_a^2}^{\infty}\exp(-n\varepsilon/4)d\varepsilon\\
        \leq&\int_{0}^{\infty}\exp(-n\varepsilon^2/(8\sigma_a^2))d\varepsilon+\int_{2\sigma_a^2}^{\infty}\exp(-n\varepsilon/4)d\varepsilon\\
        =&2\sqrt{\pi}n^{-1/2}\sigma_a+4n^{-1}\exp(-n\sigma_a^2/2)\\
        \leq&2\sqrt{\pi}n^{-1/2}\sigma_a+8n^{-2}\sigma_a^{-2}\,.
    \end{align*}
    Thus,
    \begin{gather*}
        \mathbb{E}\left\{\frac{\left|A_n^a(\widehat{\xi}_\alpha)-A^a(\widehat{\xi}_\alpha)\right|}{B_a}\right\}\leq \frac{20\sigma_a\log^{1/2}(n)}{B_an^{1/2}}+\frac{8}{B_an^2\sigma_a^2}.
    \end{gather*}

    \noindent\textbf{Part (ii): control the denominator fluctuation term.} Using \eqref{eq: bahad rep B 3}, we can bound the expectation of $|B_n^a-B_a|$ as
    \begin{align*}
        &\mathbb{E}\left( \left|B_n^a-B_a\right| \right)\\
        =&\int_{0}^{\infty}\mathrm{Pr}\left( \left|B_n^a-B_a\right|>\varepsilon \right)d\varepsilon\\
        \leq&2\int_{0}^{\infty}\exp\left( \dfrac{-n\varepsilon^2}{2\sigma_a^2+M_a\varepsilon} \right)d\varepsilon\\
        =&2\int_{0}^{2\sigma_a^2/M_a}\exp\left( -\dfrac{n\varepsilon^2}{4\sigma_a^2} \right)d\varepsilon+2\int_{2\sigma_a^2/M_a}^{\infty}\exp\left( -\dfrac{n\varepsilon}{2M_a} \right)d\varepsilon\\
        \leq&2\int_{0}^{\infty}\exp\left( -\dfrac{n\varepsilon^2}{4\sigma_a^2} \right)d\varepsilon+2\int_{2\sigma_a^2/M_a}^{\infty}\exp\left( -\dfrac{n\varepsilon}{2M_a} \right)d\varepsilon\\
        =&2\sqrt{2\pi}\sigma_an^{-1/2}+4M_an^{-1}\exp(-n\sigma_a^2/M_a^2)\\
        \leq&2\sqrt{2\pi}\sigma_an^{-1/2}+4M_a^3n^{-2}\sigma_a^{-2}\,.
    \end{align*}
    Thus,
    \begin{gather*}
        \mathbb{E}\left(\frac{|B_n^a-B_a|}{B_a}\right)\leq \frac{2\sqrt{2\pi}\sigma_a}{B_an^{1/2}}+\frac{4M_a^3}{B_an^2\sigma_a^2}.
    \end{gather*}

    \noindent\textbf{Part (iii): control the second-order remainder term.} We decompose the expectation as
    \begin{align*}
        &2\mathbb{E}\left[\frac{|B_n^a-B_a|\left|\{A_n^a(\widehat{\xi}_\alpha)-A^a(\widehat{\xi}_\alpha)\}-A^a(\widehat{\xi}_\alpha)B_a^{-1}(B_n^a-B_a)\right|}{B_a^2}\right]\\
        \leq&2\mathbb{E}\left\{\frac{|B_n^a-B_a|\sup_{u\in\mathbb{R}}|A_n^a(u)-A^a(u)|}{B_a^2}\right\}+2\mathbb{E}\left\{\frac{A^a(\widehat{\xi}_\alpha)|B_n^a-B_a|^2}{B_a^3}\right\}.
    \end{align*}
    Using the Cauchy--Schwarz inequality, we can bound the first term as
    \begin{align*}
        & 2\mathbb{E}\left\{\frac{|B_n^a-B_a|\sup_{u\in\mathbb{R}}|A_n^a(u)-A^a(u)|}{B_a^2}\right\} \\
        \leq & \frac{2}{B_a^2}\left[\mathbb{E}|B_n^a-B_a|^2\mathbb{E}\left\{\sup_{u\in\mathbb{R}}|A_n^a(u)-A^a(u)|^2\right\}\right]^{1/2}.
    \end{align*}
    Since $A^a(u)\leq B_a$ for any $u$, the second term is bounded by
    \begin{gather*}
        2\mathbb{E}\left\{\frac{A^a(\widehat{\xi}_\alpha)|B_n^a-B_a|^2}{B_a^3}\right\}\leq \frac{2}{B_a^2}\mathbb{E}|B_n^a-B_a|^2.
    \end{gather*}
    Using \eqref{eq: bahad rep B 3}, we obtain
    \begin{align*}
        \mathbb{E}|B_n^a-B_a|^2
        &=2\int_0^\infty u\,\mathrm{Pr}(|B_n^a-B_a|>u)du\\
        &\leq 2\int_0^\infty u\exp\left(-\frac{nu^2}{2\sigma_a^2+M_au}\right)du\,.
    \end{align*}
    Further decompose the previous integral as two parts similar as in part (ii), we get 
    \begin{align*}
        &2\int_{0}^{\infty}u\exp\left( -\dfrac{nu^2}{2\sigma_a^2+M_au} \right)du\\
        \leq&2\int_{0}^{2\sigma_a^2/M_a}u\exp\left( -\dfrac{nu^2}{4\sigma_a^2} \right)du+2\int_{2\sigma_a^2/M_a}^{\infty}u\exp\left( -\dfrac{nu}{2M_a} \right)du\\
        \leq&\dfrac{4\sigma_a^2}{n}\int_{0}^{\infty}\exp\left( -u \right)du+\dfrac{8M_a^2}{n^2}\int_{n\sigma_a^2/M_a^2}^\infty u\exp(-u)du\\
        \leq&\dfrac{4\sigma_a^2}{n}+\dfrac{16\sigma_a^2}{n}\exp(-n\sigma_a^2/M_a^2)\leq\dfrac{12\sigma_a^2}{n}\,,
    \end{align*}
    where the last two inequalities use the fact that $n\sigma_a^2/M_a^2>1$.
    This indicates
    \begin{gather*}
        \dfrac{2}{B_a^2}\mathbb{E} \left( \left|B_n^a-B_a\right|^2 \right)\leq \dfrac{24\sigma_a^2}{B_a^2n}\,.
    \end{gather*}

    Finally, we calculate $\mathbb{E}\{ \sup_{u\in\mathbb{R}}|A_n^a(u)-A^a(u)|^2 \}$. Repeating the procedure from the previous derivation, we transform the expectation into an integral of the tail probability and then use the change of variable for the squared tail:
    \begin{align*}
        &\mathbb{E}\left\{ \sup_{u\in\mathbb{R}}\left|A_n^a(u)-A^a(u)\right|^2 \right\}\\
        =&\int_{0}^{\infty}\mathrm{Pr}\left( \sup_{u\in\mathbb{R}}\left|A_n^a(u)-A^a(u)\right|^2>\varepsilon \right)d\varepsilon\\
        =&2\int_{0}^{\infty}\mathrm{Pr}\left( \sup_{u\in\mathbb{R}}\left|A_n^a(u)-A^a(u)\right|>\varepsilon \right)\varepsilon d\varepsilon\\
        =&2\int_{0}^{16\sigma_a\{\log(n)/n\}^{1/2}}\mathrm{Pr}\left( \sup_{u\in\mathbb{R}}\left|A_n^a(u)-A^a(u)\right|>\varepsilon \right)\varepsilon d\varepsilon\\
        & ~~~~~~~~~ +2\int_{16\sigma_a\{\log(n)/n\}^{1/2}}^{\infty}\mathrm{Pr}\left( \sup_{u\in\mathbb{R}}\left|A_n^a(u)-A^a(u)\right|>\varepsilon \right)\varepsilon d\varepsilon\\
        \leq&\dfrac{256\sigma_a^2\log(n)}{n}+2\int_{16\sigma_a\{\log(n)/n\}^{1/2}}^{\infty}\mathrm{Pr}\left( \sup_{u\in\mathbb{R}}\left|A_n^a(u)-A^a(u)\right|>\varepsilon \right)\varepsilon d\varepsilon\,.
    \end{align*}
    Combining this with inequality \eqref{eq: bahad rep A 4}, and replacing $\varepsilon$ by $\varepsilon+16\sigma_a\{\log(n)/n\}^{1/2}$ for $\varepsilon>0$, we further obtain
    \begin{align}
        \notag &2\int_{16\sigma_a\{\log(n)/n\}^{1/2}}^{\infty}\mathrm{Pr}\left( \sup_{u\in\mathbb{R}}\left|A_n^a(u)-A^a(u)\right|>\varepsilon \right)\varepsilon d\varepsilon\\
        \notag =&2\int_{0}^{\infty}\mathrm{Pr}\left( \sup_{u\in\mathbb{R}}\left|A_n^a(u)-A^a(u)\right|>16\sigma_a\dfrac{\log^{1/2}(n)}{n^{1/2}}+\varepsilon \right)\left[ \varepsilon+16\sigma_a\{\log(n)/n\}^{1/2} \right]d\varepsilon\\
        \notag=&2\int_{0}^{\infty}\mathrm{Pr}\left( \sup_{u\in\mathbb{R}}\left|A_n^a(u)-A^a(u)\right|>16\sigma_a\dfrac{\log^{1/2}(n)}{n^{1/2}}+\varepsilon \right)\varepsilon d\varepsilon\\
        \notag & ~~~~~~ +32\sigma_a\{\log(n)/n\}^{1/2}\int_{0}^{\infty}\mathrm{Pr}\left( \sup_{u\in\mathbb{R}}\left|A_n^a(u)-A^a(u)\right|>16\sigma_a\dfrac{\log^{1/2}(n)}{n^{1/2}}+\varepsilon \right)d\varepsilon\,.
    \end{align}
    
    In Part (i), we have already shown that
    \begin{gather*}
        \int_{0}^{\infty}\mathrm{Pr}\left( \sup_{u\in\mathbb{R}}\left|A_n^a(u)-A^a(u)\right|>16\sigma_a\dfrac{\log^{1/2}(n)}{n^{1/2}}+\varepsilon \right)d\varepsilon
        \leq 2\sqrt{\pi}n^{-1/2}\sigma_a+8n^{-2}\sigma_a^{-2}\,.
    \end{gather*}
    Thus, the previous inequality gives
    \begin{align*}
        &2\int_{16\sigma_a\{\log(n)/n\}^{1/2}}^{\infty}\mathrm{Pr}\left( \sup_{u\in\mathbb{R}}\left|A_n^a(u)-A^a(u)\right|>\varepsilon \right)\varepsilon d\varepsilon\\
        \leq&2\int_{0}^{\infty}\mathrm{Pr}\left( \sup_{u\in\mathbb{R}}\left|A_n^a(u)-A^a(u)\right|>16\sigma_a\dfrac{\log^{1/2}(n)}{n^{1/2}}+\varepsilon \right)\varepsilon d\varepsilon\\
        & ~~~~~~ +\dfrac{64\sqrt{\pi}\sigma_a^2\log^{1/2}(n)}{n}+\dfrac{256\log^{1/2}(n)}{n^{5/2}\sigma_a}\\
        \leq&2\int_{0}^{\infty}\mathrm{Pr}\left( \sup_{u\in\mathbb{R}}\left|A_n^a(u)-A^a(u)\right|>16\sigma_a\dfrac{\log^{1/2}(n)}{n^{1/2}}+\varepsilon \right)\varepsilon d\varepsilon+\dfrac{256\sigma_a^2\log^{1/2}(n)}{n}\,,
    \end{align*}
    where the last inequality follows from the condition $n\sigma_a^2/2>1$.
    Next, we calculate
    \begin{align*}
        &2\int_{0}^{\infty}\mathrm{Pr}\left( \sup_{u\in\mathbb{R}}\left|A_n^a(u)-A^a(u)\right|>16\sigma_a\dfrac{\log^{1/2}(n)}{n^{1/2}}+\varepsilon \right)\varepsilon d\varepsilon\\
        =&\int_{0}^{2\sigma_a^2}\exp(-n\varepsilon^2/(8\sigma_a^2))d\varepsilon^2+2\int_{2\sigma_a^2}^{\infty}\exp(-n\varepsilon/4)\varepsilon d\varepsilon\\
        \leq&\dfrac{8\sigma_a^2}{n}\int_{0}^{\infty}\exp(-v)dv+\dfrac{32}{n^2}\int_{\sigma_a^2n/2}^{\infty}\exp(-v)vdv\\
        \leq&\dfrac{24\sigma_a^2}{n}\,,
    \end{align*}
    where the last inequality follows from $\sigma_a^2n/2>1$, and the second last inequality follows by the changes of variables $v=n\varepsilon^2/(8\sigma_a^2)$ and $v=n\varepsilon/4$. Thus,
    \begin{gather*}
        2\int_{16\sigma_a\{\log(n)/n\}^{1/2}}^{\infty}\mathrm{Pr}\left( \sup_{u\in\mathbb{R}}\left|A_n^a(u)-A^a(u)\right|>\varepsilon \right)\varepsilon d\varepsilon\leq\dfrac{304\sigma_a^2\log^{1/2}(n)}{n}\,,
    \end{gather*}
    and hence
    \begin{gather*}
        \mathbb{E}\left\{ \sup_{u\in\mathbb{R}}\left|A_n^a(u)-A^a(u)\right|^2 \right\}\leq\dfrac{560\sigma_a^2\log(n)}{n}\,.
    \end{gather*}
    Together, we obtain
    \begin{gather*}
        \dfrac{2}{B_a^2}\left[ \mathbb{E}\left( \left|B_n^a-B_a\right|^2 \right)\mathbb{E}\left\{ \sup_{u\in\mathbb{R}}\left|A_n^a(u)-A^a(u)\right|^2 \right\} \right]^{1/2}\leq\dfrac{164\sigma_a^2\log^{1/2}(n)}{B_a^2n}\,,
    \end{gather*}
    and therefore
    \begin{gather*}
        2\mathbb{E}\left[ \dfrac{\left|B_n^a-B_a\right|\left|\{A_n^a(\widehat{\xi}_\alpha)-A^a(\widehat{\xi}_\alpha)\}-A^a(\widehat{\xi}_\alpha)B_a^{-1}(B_n^a-B_a)\right|}{B_a^2} \right]\leq\dfrac{188\sigma_a^2\log^{1/2}(n)}{B_a^2n}\,.
    \end{gather*}
    
    Combining Parts (i)--(iii) and the condition $\sigma_a^2n/(2M_a^2)>1$, we conclude that
    \begin{align*}
        \underline{L}_a(\xi_\alpha) \mathbb{E}\left\{ \left|\widehat{\xi}_\alpha-\xi_\alpha\right| \right\}
        \leq&2M_S\underline{L}_a(\xi_\alpha)\exp\{-nB_a^2/(8\sigma_a^2+2M_aB_a)\}\\
        &+\dfrac{20\sigma_a\log^{1/2}(n)}{B_an^{1/2}}+\dfrac{8}{B_an^{2}\sigma_a^{2}}\\
        &+\dfrac{2\sqrt{2\pi}\sigma_a}{B_an^{1/2}}+\dfrac{4M_a^3}{B_an^{2}\sigma_a^{2}}+\dfrac{2M_a}{B_an}+\dfrac{188\sigma_a^2\log^{1/2}(n)}{B_a^2n}\\
        \leq&2M_S\underline{L}_a(\xi_\alpha)\exp\{-nB_a^2/(8\sigma_a^2+2M_aB_a)\}\\
        &+\dfrac{24\sigma_a\log^{1/2}(n)}{B_an^{1/2}}+\dfrac{196(\sigma_a^2+M_aB_a)\log^{1/2}(n)}{B_a^2n}\,.
    \end{align*}
    Preserving $n,B_a,\sigma_a,M_a$ and treating the remaining terms as constants, there exists a constant $C$ such that $\mathbb{E}\{|\widehat{\xi}_\alpha-\xi_\alpha|\}$ is upper bounded by
    \begin{gather*}
        C \left\{ \exp\left( -\dfrac{nB_a^2}{8\sigma_a^2+2M_aB_a} \right)+\dfrac{\sigma_a\log^{1/2}(n)}{B_an^{1/2}}+\dfrac{(\sigma_a^2+M_aB_a)\log^{1/2}(n)}{B_a^2n} \right\}\,.
    \end{gather*}
    Also, by the result of Part (ii), since
    \begin{align*}
        \dfrac{2M_a\mathbb{E}\left|B_n^a-B_a\right|}{nB_a^2}\leq \dfrac{4\sqrt{2\pi}M_a\sigma_a}{B_a^2n^{3/2}}+\dfrac{8M_a^4}{B_a^2n^{3}\sigma_a^{2}}\leq\dfrac{32M_a\sigma_a}{B_a^2n^{3/2}}\,,
    \end{align*}
    we obtain that $\mathbb{E}\{|\widehat{\xi}_\alpha^+-\xi_\alpha|\wedge(2M_S)\}$ is bounded by
    \begin{gather*}
        C \left\{ \exp\left( -\dfrac{nB_a^2}{8\sigma_a^2+2M_aB_a} \right)+\dfrac{\sigma_a\log^{1/2}(n)}{B_an^{1/2}}+\dfrac{(\sigma_a^2+M_aB_a)\log^{1/2}(n)}{B_a^2n}+\dfrac{M_a\sigma_a}{B_a^2n^{3/2}} \right\}\,.       
    \end{gather*}
    Together, both $\mathbb{E}\{|\widehat{\xi}_\alpha-\xi_\alpha|\}$ and $\mathbb{E}\{|\widehat{\xi}_\alpha^+-\xi_\alpha|\wedge(2M_S)\}$ are bounded by the same order of convergence, which completes the proof.
\end{proof}

\subsection{Proof of Lemma \ref{theo: sup quantile error}}

\begin{proof}
    Denote the abstract slope function of $m$ at $\xi_{\alpha,m}$ by
    \begin{gather*}
        k_m(u)=\dfrac{\left|F_m^a(u)-F_m^a(\xi_{\alpha,m})\right|}{\left|u-\xi_{\alpha,m}\right|}\,,\qquad u\neq \xi_{\alpha,m}\,,
    \end{gather*}
    and let $k_m(\xi_{\alpha,m})$ be the derivative of $F_m^a$ at $\xi_{\alpha,m}$. Thus, for any $u\in[-M_{\mathcal M},M_{\mathcal M}]$, we have $k_m(u)>0$. Since $F_m^a$ is differentiable in an $\epsilon_0$-neighborhood of $\xi_{\alpha,m}$ with derivative lower bounded by $\underline L$, for $u\in[\xi_{\alpha,m}-\epsilon_0,\xi_{\alpha,m}+\epsilon_0]$, we have $k_m(u)\geq\underline L$. Also, since $|m(X^{(0)})|\leq M_{\mathcal M}$ holds almost surely, $k_m(u)$ is lower bounded by $\underline L\wedge \epsilon_0\underline L/M_{\mathcal M}$ on $[-M_{\mathcal M},M_{\mathcal M}]$. Thus, define
    \begin{gather*}
        \underline L_{\mathcal M}=\inf_{m\in\mathcal M,\,u\in[-M_{\mathcal M},M_{\mathcal M}]}k_m(u)>0\,.
    \end{gather*}

    Define $A_{n,m}^a(u)=n^{-1}\sum_{i=1}^n a(X_i^{(0)})\mathbbm{1}\{m(X_i^{(0)})\leq u\}$ and $B_n^a=n^{-1}\sum_{i=1}^n a(X_i^{(0)})$, so $F_{n,m}^a(u)=A_{n,m}^a(u)/B_n^a$. Let $A_m^a(u)=\mathbb{E}\{a(X^{(0)})\mathbbm{1}(m(X^{(0)})\leq u)\}$ and $B_a=\mathbb{E}\{a(X^{(0)})\}$, so $F_m^a(u)=A_m^a(u)/B_a$. Given Assumption~\ref{ass: weight function} that $0<a(X^{(0)})\leq M_a$ holds almost surely, we have $\sup_x a(x)-\inf_x a(x)\leq M_a$ almost surely. By Hoeffding's inequality,
    \begin{gather}
        \mathrm{Pr}\left( \left|B_n^a-B_a\right|>\varepsilon \right)\leq 2\exp\left( -\dfrac{2n\varepsilon^2}{M_a^2} \right)\,.\label{eq: sup bahad rep 1}
    \end{gather}
    Denote $D_B=\{B_n^a\geq B_a/2\}$, which has probability at least $1-\exp\{-nB_a^2/(2M_a^2)\}$.

    Since $\widehat{\xi}_{\alpha,m}$ is the sample quantile, define $F_{n,m}^a(\widehat{\xi}_{\alpha,m})=\alpha+\Delta_{n,m}$, where $\Delta_{n,m}$ satisfies
    \begin{gather*}
        0\leq \Delta_{n,m}\leq\max_{1\leq i\leq n}\dfrac{a(X_i^{(0)})}{\sum_{j=1}^{n}a(X_j^{(0)})}\leq\dfrac{M_a}{\sum_{j=1}^{n}a(X_j^{(0)})}=\dfrac{M_a}{nB_n^a}\,.
    \end{gather*}
    By the definition of $\Delta_{n,m}$,
    \begin{align*}
        \alpha-F_m^a(\widehat{\xi}_{\alpha,m})+\Delta_{n,m}
        =&F_{n,m}^a(\widehat{\xi}_{\alpha,m})-F_m^a(\widehat{\xi}_{\alpha,m})\\
        =&\frac{A_{n,m}^a(\widehat{\xi}_{\alpha,m})}{B_n^a}-\frac{A_m^a(\widehat{\xi}_{\alpha,m})}{B_a}\\
        =&\frac{A_{n,m}^a(\widehat{\xi}_{\alpha,m})B_a-A_m^a(\widehat{\xi}_{\alpha,m})B_n^a}{B_n^aB_a}\,.
    \end{align*}
    The numerator in the last formula equals $B_a\{A_{n,m}^a(\widehat{\xi}_{\alpha,m})-A_m^a(\widehat{\xi}_{\alpha,m})\}-A_m^a(\widehat{\xi}_{\alpha,m})(B_n^a-B_a)$. Thus,
    \begin{align*}
        & \alpha-F_m^a(\widehat{\xi}_{\alpha,m}) \\
        =&\dfrac{A_{n,m}^a(\widehat{\xi}_{\alpha,m})-A_m^a(\widehat{\xi}_{\alpha,m})}{B_n^a}
        -\dfrac{A_m^a(\widehat{\xi}_{\alpha,m})(B_n^a-B_a)}{B_n^aB_a}
        -\Delta_{n,m}\\
        =&\dfrac{A_{n,m}^a(\widehat{\xi}_{\alpha,m})-A_m^a(\widehat{\xi}_{\alpha,m})}{B_a}
        -\dfrac{A_m^a(\widehat{\xi}_{\alpha,m})(B_n^a-B_a)}{B_a^2}
        -\Delta_{n,m}\\
        &+\left\{(B_n^a)^{-1}-B_a^{-1}\right\}\left[\{A_{n,m}^a(\widehat{\xi}_{\alpha,m})-A_m^a(\widehat{\xi}_{\alpha,m})\}-A_m^a(\widehat{\xi}_{\alpha,m})B_a^{-1}(B_n^a-B_a)\right]\,.
    \end{align*}
    Since $k_m(u)\geq\underline L_{\mathcal M}$ for any $m\in\mathcal M$ and $u\in[-M_{\mathcal M},M_{\mathcal M}]$, and $A_m^a(u)\leq B_a$, we obtain
    \begin{align*}
        \left|\widehat{\xi}_{\alpha,m}-\xi_{\alpha,m}\right|
        \leq&\dfrac{\left|A_{n,m}^a(\widehat{\xi}_{\alpha,m})-A_m^a(\widehat{\xi}_{\alpha,m})\right|}{\underline L_{\mathcal M}B_a}
        +\dfrac{\left|B_n^a-B_a\right|}{\underline L_{\mathcal M}B_a}
        +\dfrac{M_a}{n\underline L_{\mathcal M}B_n^a}\\
        &+\dfrac{\left|A_{n,m}^a(\widehat{\xi}_{\alpha,m})-A_m^a(\widehat{\xi}_{\alpha,m})\right|\left|B_n^a-B_a\right|}{\underline L_{\mathcal M}B_aB_n^a}
        +\dfrac{(B_n^a-B_a)^2}{B_aB_n^a}\,.
    \end{align*}
    Taking the supremum over $m\in\mathcal M$ gives
    \begin{align*}
        \sup_{m\in\mathcal M}\left|\widehat{\xi}_{\alpha,m}-\xi_{\alpha,m}\right|
        \leq&\dfrac{\sup_{u\in\mathbb{R},\,m\in\mathcal M}\left|A_{n,m}^a(u)-A_m^a(u)\right|}{\underline L_{\mathcal M}B_a}
        +\dfrac{\left|B_n^a-B_a\right|}{\underline L_{\mathcal M}B_a}
        +\dfrac{M_a}{n\underline L_{\mathcal M}B_n^a}\\
        &+\dfrac{\sup_{u\in\mathbb{R},\,m\in\mathcal M}\left|A_{n,m}^a(u)-A_m^a(u)\right|\left|B_n^a-B_a\right|}{\underline L_{\mathcal M}B_aB_n^a}
        +\dfrac{(B_n^a-B_a)^2}{B_aB_n^a}\,.
    \end{align*}

    Since $B_n^a$ may be small, the tail term is difficult to control directly. We therefore decompose the expectation of $\sup_{m\in\mathcal M}|\widehat{\xi}_{\alpha,m}-\xi_{\alpha,m}|$ according to $D_B$ and $D_B^c$. Since $m(X^{(0)})$ is uniformly bounded by $M_{\mathcal M}$, on $D_B^c$ we have
    \begin{gather*}
        \sup_{m\in\mathcal M}\left|\widehat{\xi}_{\alpha,m}-\xi_{\alpha,m}\right|\leq 2M_{\mathcal M}\,.
    \end{gather*}
    Thus,
    \begin{align*}
        &\mathbb{E}\left( \sup_{m\in\mathcal M}\left|\widehat{\xi}_{\alpha,m}-\xi_{\alpha,m}\right| \right)\\
        =&\mathbb{E}\left\{\mathbbm{1}(D_B)\sup_{m\in\mathcal M}\left|\widehat{\xi}_{\alpha,m}-\xi_{\alpha,m}\right|\right\}
        +\mathbb{E}\left\{\mathbbm{1}(D_B^c)\sup_{m\in\mathcal M}\left|\widehat{\xi}_{\alpha,m}-\xi_{\alpha,m}\right|\right\}\\
        \leq&\mathbb{E}\left\{\mathbbm{1}(D_B)\sup_{m\in\mathcal M}\left|\widehat{\xi}_{\alpha,m}-\xi_{\alpha,m}\right|\right\}
        +2M_{\mathcal M}\mathrm{Pr}(D_B^c)\\
        \leq&\mathbb{E}\left\{\mathbbm{1}(D_B)\sup_{m\in\mathcal M}\left|\widehat{\xi}_{\alpha,m}-\xi_{\alpha,m}\right|\right\}
        +2M_{\mathcal M}\exp\left(-\dfrac{nB_a^2}{2M_a^2}\right)\,.
    \end{align*}
    On the event $D_B$, the previous inequality and Cauchy--Schwarz inequality imply
    \begin{align*}
        & \mathbbm{1}(D_B)\sup_{m\in\mathcal M}\left|\widehat{\xi}_{\alpha,m}-\xi_{\alpha,m}\right| \\
        \leq&\dfrac{\sup_{u\in\mathbb{R},\,m\in\mathcal M}\left|A_{n,m}^a(u)-A_m^a(u)\right|}{\underline L_{\mathcal M}B_a}
        +\dfrac{\left|B_n^a-B_a\right|}{\underline L_{\mathcal M}B_a}
        +\dfrac{2M_a}{n\underline L_{\mathcal M}B_a}\\
        &+\dfrac{2\sup_{u\in\mathbb{R},\,m\in\mathcal M}\left|A_{n,m}^a(u)-A_m^a(u)\right|\left|B_n^a-B_a\right|}{\underline L_{\mathcal M}B_a^2}
        +\dfrac{2(B_n^a-B_a)^2}{B_a^2}\,,
    \end{align*}
    and hence
    \begin{align*}
        &\mathbb{E}\left\{\mathbbm{1}(D_B)\sup_{m\in\mathcal M}\left|\widehat{\xi}_{\alpha,m}-\xi_{\alpha,m}\right|\right\}\\
        \leq&\dfrac{\mathbb{E}\left\{\sup_{u\in\mathbb{R},\,m\in\mathcal M}\left|A_{n,m}^a(u)-A_m^a(u)\right|\right\}}{\underline L_{\mathcal M}B_a}
        +\dfrac{2M_a}{n\underline L_{\mathcal M}B_a}\\
        &+\dfrac{\mathbb{E}|B_n^a-B_a|}{\underline L_{\mathcal M}B_a}
        +\dfrac{2\mathbb{E}|B_n^a-B_a|^2}{B_a^2}\\
        &+\dfrac{2\left[\mathbb{E}\left\{\sup_{u\in\mathbb{R},\,m\in\mathcal M}\left|A_{n,m}^a(u)-A_m^a(u)\right|^2\right\}\mathbb{E}|B_n^a-B_a|^2\right]^{1/2}}{\underline L_{\mathcal M}B_a^2}\,.
    \end{align*}

    We calculate these expectations one by one. First, using \eqref{eq: sup bahad rep 1},
    \begin{align*}
        \mathbb{E}|B_n^a-B_a|
        &=\int_{0}^{\infty}\mathrm{Pr}\left( |B_n^a-B_a|>\varepsilon \right)d\varepsilon
        \leq 2\int_{0}^{\infty}\exp\left(-\dfrac{2n\varepsilon^2}{M_a^2}\right)d\varepsilon\\
        &=M_a\sqrt{\pi}n^{-1/2}\leq 2M_an^{-1/2}\,.
    \end{align*}
    Similarly,
    \begin{align*}
        \mathbb{E}|B_n^a-B_a|^2
        &=\int_{0}^{\infty}\mathrm{Pr}\left( |B_n^a-B_a|^2>\varepsilon \right)d\varepsilon
        =2\int_{0}^{\infty}\mathrm{Pr}\left( |B_n^a-B_a|>\varepsilon \right)\varepsilon d\varepsilon\\
        &\leq 2\int_{0}^{\infty}\exp\left(-\dfrac{2n\varepsilon^2}{M_a^2}\right)\varepsilon d\varepsilon
        =\int_{0}^{\infty}\exp\left(-\dfrac{2n\varepsilon}{M_a^2}\right)d\varepsilon\\
        &=M_a^2(2n)^{-1}\,.
    \end{align*}
    Thus,
    \begin{gather*}
        \dfrac{\mathbb{E}|B_n^a-B_a|}{\underline L_{\mathcal M}B_a}
        +\dfrac{2\mathbb{E}|B_n^a-B_a|^2}{B_a^2}
        \leq\dfrac{2M_an^{-1/2}}{\underline L_{\mathcal M}B_a}
        +\dfrac{M_a^2n^{-1}}{B_a^2}\,.
    \end{gather*}
    Lemma~\ref{lemma: sup weighted dkw} yields
    \begin{align*}
        & \mathrm{Pr}\left( \sup_{m\in\mathcal M,\,u\in\mathbb{R}}\left|A_{n,m}^a(u)-A_m^a(u)\right|\leq 8M_a\sqrt{\dfrac{\mathrm{Pdim}(\mathcal M)\log(n)}{n}}+\varepsilon \right) \\ \geq & 1-\exp\left(-\dfrac{n\varepsilon^2}{2M_a^2}\right)\,.
    \end{align*}
    Therefore, similarly to the calculation of $\mathbb{E}|B_n^a-B_a|$, we obtain
    \begin{align*}
        \mathbb{E}\left\{ \sup_{m\in\mathcal M,\,u\in\mathbb{R}}\left|A_{n,m}^a(u)-A_m^a(u)\right|-8M_a\sqrt{\dfrac{\mathrm{Pdim}(\mathcal M)\log(n)}{n}} \right\}\leq 2M_an^{-1/2}\,,
    \end{align*}
    which implies
    \begin{gather*}
        \mathbb{E}\left\{ \sup_{m\in\mathcal M,\,u\in\mathbb{R}}\left|A_{n,m}^a(u)-A_m^a(u)\right| \right\}\leq 10M_a\sqrt{\dfrac{\mathrm{Pdim}(\mathcal M)\log(n)}{n}}\,.
    \end{gather*}
    Following the same calculation used for $\mathbb{E}|B_n^a-B_a|^2$, we also obtain
    \begin{align*}
        \mathbb{E}\left[ \left\{ \sup_{m\in\mathcal M,\,u\in\mathbb{R}}\left|A_{n,m}^a(u)-A_m^a(u)\right|-8M_a\sqrt{\dfrac{\mathrm{Pdim}(\mathcal M)\log(n)}{n}} \right\}^2 \right]\leq M_a^2n^{-1}\,,
    \end{align*}
    and hence
    \begin{gather*}
        \mathbb{E}\left\{ \sup_{m\in\mathcal M,\,u\in\mathbb{R}}\left|A_{n,m}^a(u)-A_m^a(u)\right|^2 \right\}\leq 97M_a^2\dfrac{\mathrm{Pdim}(\mathcal M)\log(n)}{n}\,.
    \end{gather*}
    Therefore,
    \begin{gather*}
        \dfrac{\mathbb{E}\left\{\sup_{u\in\mathbb{R},\,m\in\mathcal M}\left|A_{n,m}^a(u)-A_m^a(u)\right|\right\}}{\underline L_{\mathcal M}B_a}
        \leq\dfrac{10M_a}{\underline L_{\mathcal M}B_a}\sqrt{\dfrac{\mathrm{Pdim}(\mathcal M)\log(n)}{n}}\,,
    \end{gather*}
    and
    \begin{align*}
        & \dfrac{2\left[\mathbb{E}\left\{\sup_{u\in\mathbb{R},\,m\in\mathcal M}\left|A_{n,m}^a(u)-A_m^a(u)\right|^2\right\}\mathbb{E}|B_n^a-B_a|^2\right]^{1/2}}{\underline L_{\mathcal M}B_a^2} \\
        \leq & \dfrac{14M_a^2}{\underline L_{\mathcal M}B_a^2}\dfrac{\sqrt{\mathrm{Pdim}(\mathcal M)\log(n)}}{n}\,.
    \end{align*}
    Combining the preceding bounds, we get
    \begin{align*}
        &\mathbb{E}\left\{ \mathbbm{1}(D_B)\sup_{m\in\mathcal M}\left|\widehat{\xi}_{\alpha,m}-\xi_{\alpha,m}\right| \right\}\\
        \leq&\dfrac{2M_an^{-1/2}}{\underline L_{\mathcal M}B_a}
        +\dfrac{M_a^2n^{-1}}{B_a^2}
        +\dfrac{2M_a}{n\underline L_{\mathcal M}B_a} \\
        & 
        +\dfrac{14M_a^2}{\underline L_{\mathcal M}B_a^2}\dfrac{\sqrt{\mathrm{Pdim}(\mathcal M)\log(n)}}{n}
        +\dfrac{10M_a}{\underline L_{\mathcal M}B_a}\sqrt{\dfrac{\mathrm{Pdim}(\mathcal M)\log(n)}{n}}\,.
    \end{align*}
    Thus there exists a constant $C>0$ such that
    \begin{gather*}
        \mathbb{E}\left( \sup_{m\in\mathcal M}\left|\widehat{\xi}_{\alpha,m}-\xi_{\alpha,m}\right| \right)
        \leq \dfrac{CM_a}{B_a}\sqrt{\dfrac{\mathrm{Pdim}(\mathcal M)\log(n)}{n}}\,.
    \end{gather*}
    Similarly, the same bound holds for $\mathbb{E}\left(2M_{\mathcal M}\wedge\sup_{m\in\mathcal M}|\widehat{\xi}_{\alpha,m}^+-\xi_{\alpha,m}|\right)$. Taking the same constant $C$ for both bounds yields
    \begin{align*}
        & \mathbb{E}\left( \sup_{m\in\mathcal M}\left|\widehat{\xi}_{\alpha,m}-\xi_{\alpha,m}\right| \right)
        \vee
        \mathbb{E}\left(2M_{\mathcal M}\wedge\sup_{m\in\mathcal M}\left|\widehat{\xi}_{\alpha,m}^+-\xi_{\alpha,m}\right|\right) \\
        \leq & \dfrac{CM_a}{B_a}\sqrt{\dfrac{\mathrm{Pdim}(\mathcal M)\log(n)}{n}}\,.
    \end{align*}
\end{proof}

\section{Additional Numerical Results}\label{sec: additional numerical}

\subsection{Implementation Details of GLCP and CQR-based Methods}\label{sec: method_details}
We describe the implementation of GLCP and the CQR-based method used in the numerical experiments. Let $Z = \{(X_1, Y_1), \ldots, (X_n, Y_n),(X_{n+1}, Y_{n+1})\}$ denote the calibration and test data, where $(X_{n+1},Y_{n+1})$ is the test point. Let $V_i=v(X_i,Y_i)$ be the base score for $i\in[n+1]$.

Using the training subset $Z_{\mathrm{tr},2}$, we estimate the conditional c.d.f.~of the base score $v(x,y)$ given $X=x$, denoted by $\widehat F_{v\mid X}(v\mid x)$, and the conditional $(1-\alpha)$-quantile of the base score $v(x,y)$ given $X=x$, denoted by $\widehat Q_\alpha(x)$. The GLCP score is defined as
\begin{gather*}
    s_{\text{glcp}}(x, y) = \varphi\left(\widehat{F}_{v|X}\big(v(x, y) \mid x\big)\right)\,,
\end{gather*}
where $\varphi:[0,1]\to[0,1]$ is a fixed deterministic transformation. The CQR-based score is defined by
\begin{gather*}
    s_{\text{cqr}}(x, y) = v(x, y) - \widehat{Q}_\alpha(x)\,.
\end{gather*}
Thus, the implemented CQR-based method is a quantile-centered version for the base score, rather than the standard two-sided max-type CQR score.

For either choice $s\in\{s_{\rm GLCP},s_{\rm CQR}\}$, the conformal prediction set is
\begin{gather*}
    \widehat{C}(X_{n+1}) = \left\{ y : s(X_{n+1}, y) \leq Q\left(1-\alpha; \frac{1}{n+1}\left\{\sum_{i=1}^{n} \delta_{s(X_i, Y_i)} + \delta_{\infty}\right\}\right) \right\}\,.
\end{gather*}

\paragraph*{Kernel Bandwidth Selection}
For kernel-based methods, the bandwidth $h$ is chosen through the effective sample size criterion of \cite{hore2025conformal}. Specifically, define
$n_{\text{eff}}(h)=n \mathbb{E}\big[\mathbb{E}\{K(X,X^\prime;h)\mid X\}^2\big]/\mathbb{E}\{K^2(X,X^\prime;h)\}$ 
where $X$ and $X^\prime$ are independent draws from the covariate distribution. Using the training subsets $Z_{\mathrm{tr},1}$ and $Z_{\mathrm{tr},2}$, we compute the empirical analogue $\widehat n_{\rm eff}(h)$. For a prescribed target effective sample size, we choose $\widehat h_{n_{\rm eff}}$ by solving
$\hat{n}_{\text{eff}}(\hat{h}_{n_{\text{eff}}}) \approx n_{\text{eff}}$,
and use this bandwidth to form the kernel weights in the corresponding method.

\subsection{Additional Results for Conditional Conformal Methods under Covariate Shift}\label{sec: app cs experiments}

In Section~\ref{sec: covariate shift} of the main text, we study conditional conformal methods under covariate shift. The calibration covariates are drawn from $N_d(0,I_d)$, while the test covariates are drawn from $N_d(0,\sigma^2 I_d)$ for a prescribed $\sigma>0$. The conditional distribution $Y\mid X$ follows DGP1--3. This subsection provides additional implementation details for the five conditional methods fitted on $Z_{\rm tr,2}$, and reports supplementary results on their test-conditional coverage performance under covariate shift.

The five conditional methods are LCP, GLCP with Engression as the conditional distribution estimator, CQR-LR, CQR-RF, and CQR-LGB. Their weighted counterparts are obtained by density-ratio reweighting. We also include standard WCP based on the base score as a benchmark. The method-specific configurations used on $Z_{\rm tr,2}$ are as follows:
\begin{itemize}
    \item \textbf{GLCP \citep{min2025personalized}:} Engression is used to estimate the conditional distribution of the base score. The network has two hidden layers with 50 units per layer and is trained for 500 epochs.

    \item \textbf{LCP:} Kernel localization is applied to the base score, with the target effective sample size fixed at $n_{\rm eff}=40$.

    \item \textbf{CQR-LR:} Linear quantile regression is fitted with an $L_2$ penalty. The target quantile level is $0.9$, and the regularization coefficient is $0.01$.

    \item \textbf{CQR-RF \citep{meinshausen2006quantile}:} Quantile random forest is fitted with minimum split size 2 and maximum tree depth 10. The target quantile level is $0.9$.

    \item \textbf{CQR-LGB \citep{ke2017lightgbm}:} LightGBM quantile regression is fitted with 21 leaves, learning rate $0.05$, and feature and bagging fractions both equal to $0.6$. The target quantile level is $0.9$, and no early stopping is used.
\end{itemize}

Tables~\ref{tab: app cs miscov}--\ref{tab: app cs mar} report conditional miscoverage and marginal coverage for sample sizes $n\in\{500,1000\}$ and dimensions $d\in\{5,10,15,20\}$. Overall, the conditional conformal methods tend to have smaller test-conditional miscoverage than WCP, with CQR-LR being less competitive in more complex settings. The density-ratio-weighted versions generally attain marginal coverage close to the nominal level. By contrast, the unweighted versions can deviate substantially from nominal marginal coverage when $\sigma\neq1$, as expected under covariate shift.

{\fontsize{8}{7}\selectfont{\setlength{\tabcolsep}{2.5pt}
\begin{longtable}{cccccccccccccccccc}
\caption{Conditional miscoverage results under covariate shift. The method outperforming WCP is highlighted with \textbf{bold} and $*$. } \label{tab: app cs miscov} \\
\toprule
& & \multicolumn{5}{c}{No covariate-shift adjustment} & & \multicolumn{5}{c}{Density-ratio weighted}\\
\cmidrule(lr){3-7} \cmidrule(lr){9-13} 
& $\sigma$ & GLCP & LCP & CQR-LR & CQR-RF & CQR-LGB & & GLCP & LCP & CQR-LR & CQR-RF & CQR-LGB & & WCP \\
\midrule
\endfirsthead
\multicolumn{15}{c}{{\tablename\ \thetable{} -- continued from previous page}} \\
\toprule
& & \multicolumn{5}{c}{No covariate-shift adjustment} & & \multicolumn{5}{c}{Density-ratio weighted}\\
\cmidrule(lr){3-7} \cmidrule(lr){9-13} 
& $\sigma$ & GLCP & LCP & CQR-LR & CQR-RF & CQR-LGB & & GLCP & LCP & CQR-LR & CQR-RF & CQR-LGB & & WCP \\
\midrule
\endhead
\midrule
\multicolumn{15}{r}{{\small Continued on next page}} \\
\endfoot
\bottomrule
\endlastfoot
& & \multicolumn{11}{c}{$n=500\quad\quad d=5$}\\
DGP1 & $0.8$ & $0.110$ & $\textbf{0.094}^*$ & $0.110$ & $\textbf{0.100}^*$ & $\textbf{0.099}^*$ &  & $\textbf{0.104}^*$ & $\textbf{0.100}^*$ & $0.107$ & $\textbf{0.092}^*$ & $\textbf{0.091}^*$ &  & $0.108$\\
 & $1.0$ & $\textbf{0.110}^*$ & $\textbf{0.108}^*$ & $0.113$ & $\textbf{0.097}^*$ & $\textbf{0.101}^*$ &  & $\textbf{0.109}^*$ & $\textbf{0.108}^*$ & $0.113$ & $\textbf{0.097}^*$ & $\textbf{0.101}^*$ &  & $0.114$\\
 & $1.2$ & $0.139$ & $0.150$ & $0.145$ & $0.110$ & $0.124$ &  & $\textbf{0.105}^*$ & $0.109$ & $0.108$ & $\textbf{0.091}^*$ & $\textbf{0.096}^*$ &  & $0.108$\\
DGP2 & $0.8$ & $\textbf{0.036}^*$ & $\textbf{0.058}^*$ & $0.074$ & $\textbf{0.039}^*$ & $\textbf{0.032}^*$ &  & $\textbf{0.021}^*$ & $\textbf{0.049}^*$ & $0.067$ & $\textbf{0.032}^*$ & $\textbf{0.016}^*$ &  & $0.067$\\
 & $1.0$ & $\textbf{0.041}^*$ & $\textbf{0.059}^*$ & $0.081$ & $\textbf{0.036}^*$ & $\textbf{0.039}^*$ &  & $\textbf{0.041}^*$ & $\textbf{0.059}^*$ & $0.081$ & $\textbf{0.036}^*$ & $\textbf{0.039}^*$ &  & $0.081$\\
 & $1.2$ & $\textbf{0.066}^*$ & $\textbf{0.087}^*$ & $0.120$ & $\textbf{0.050}^*$ & $\textbf{0.063}^*$ &  & $\textbf{0.063}^*$ & $\textbf{0.076}^*$ & $0.089$ & $\textbf{0.051}^*$ & $\textbf{0.069}^*$ &  & $0.090$\\
DGP3 & $0.8$ & $\textbf{0.036}^*$ & $\textbf{0.042}^*$ & $0.050$ & $\textbf{0.023}^*$ & $\textbf{0.016}^*$ &  & $\textbf{0.036}^*$ & $\textbf{0.041}^*$ & $0.049$ & $\textbf{0.024}^*$ & $\textbf{0.016}^*$ &  & $0.049$\\
 & $1.0$ & $\textbf{0.035}^*$ & $\textbf{0.039}^*$ & $0.049$ & $\textbf{0.021}^*$ & $\textbf{0.016}^*$ &  & $\textbf{0.035}^*$ & $\textbf{0.039}^*$ & $0.049$ & $\textbf{0.021}^*$ & $\textbf{0.016}^*$ &  & $0.049$\\
 & $1.2$ & $\textbf{0.040}^*$ & $0.046$ & $0.059$ & $\textbf{0.020}^*$ & $\textbf{0.020}^*$ &  & $\textbf{0.032}^*$ & $\textbf{0.038}^*$ & $0.046$ & $\textbf{0.018}^*$ & $\textbf{0.017}^*$ &  & $0.046$\\
\midrule
& & \multicolumn{11}{c}{$n=500\quad\quad d=15$}\\
DGP1 & $0.8$ & $0.135$ & $0.136$ & $0.143$ & $0.176$ & $0.157$ &  & $\textbf{0.086}^*$ & $\textbf{0.092}^*$ & $0.095$ & $0.096$ & $\textbf{0.094}^*$ &  & $0.096$\\
 & $1.0$ & $\textbf{0.100}^*$ & $\textbf{0.106}^*$ & $\textbf{0.105}^*$ & $\textbf{0.105}^*$ & $\textbf{0.103}^*$ &  & $\textbf{0.099}^*$ & $\textbf{0.106}^*$ & $\textbf{0.105}^*$ & $\textbf{0.105}^*$ & $\textbf{0.103}^*$ &  & $0.107$\\
 & $1.2$ & $0.172$ & $0.177$ & $0.168$ & $0.122$ & $0.151$ &  & $0.102$ & $0.099$ & $\textbf{0.097}^*$ & $\textbf{0.088}^*$ & $\textbf{0.093}^*$ &  & $0.099$\\
DGP2 & $0.8$ & $\textbf{0.028}^*$ & $0.063$ & $0.066$ & $\textbf{0.036}^*$ & $\textbf{0.037}^*$ &  & $\textbf{0.023}^*$ & $\textbf{0.039}^*$ & $0.041$ & $\textbf{0.023}^*$ & $\textbf{0.019}^*$ &  & $0.041$\\
 & $1.0$ & $\textbf{0.034}^*$ & $\textbf{0.048}^*$ & $0.051$ & $\textbf{0.025}^*$ & $\textbf{0.035}^*$ &  & $\textbf{0.029}^*$ & $\textbf{0.048}^*$ & $0.051$ & $\textbf{0.025}^*$ & $\textbf{0.035}^*$ &  & $0.051$\\
 & $1.2$ & $0.075$ & $0.096$ & $0.102$ & $\textbf{0.048}^*$ & $0.069$ &  & $\textbf{0.040}^*$ & $\textbf{0.054}^*$ & $0.057$ & $\textbf{0.033}^*$ & $\textbf{0.048}^*$ &  & $0.057$\\
DGP3 & $0.8$ & $\textbf{0.023}^*$ & $0.037$ & $0.039$ & $\textbf{0.016}^*$ & $\textbf{0.014}^*$ &  & $\textbf{0.026}^*$ & $0.029$ & $0.030$ & $\textbf{0.018}^*$ & $\textbf{0.014}^*$ &  & $0.030$\\
 & $1.0$ & $0.028$ & $\textbf{0.027}^*$ & $0.028$ & $\textbf{0.015}^*$ & $\textbf{0.013}^*$ &  & $\textbf{0.022}^*$ & $\textbf{0.027}^*$ & $0.028$ & $\textbf{0.015}^*$ & $\textbf{0.013}^*$ &  & $0.028$\\
 & $1.2$ & $0.044$ & $0.039$ & $0.041$ & $\textbf{0.013}^*$ & $\textbf{0.017}^*$ &  & $\textbf{0.019}^*$ & $0.023$ & $0.025$ & $\textbf{0.015}^*$ & $\textbf{0.011}^*$ &  & $0.024$\\
\midrule
& & \multicolumn{11}{c}{$n=500\quad\quad d=20$}\\
DGP1 & $0.8$ & $0.173$ & $0.159$ & $0.161$ & $0.203$ & $0.182$ &  & $0.105$ & $\textbf{0.084}^*$ & $0.085$ & $0.085$ & $0.085$ &  & $0.086$\\
 & $1.0$ & $\textbf{0.104}^*$ & $0.105$ & $\textbf{0.104}^*$ & $0.105$ & $\textbf{0.104}^*$ &  & $\textbf{0.081}^*$ & $0.105$ & $\textbf{0.104}^*$ & $0.105$ & $\textbf{0.104}^*$ &  & $0.106$\\
 & $1.2$ & $0.207$ & $0.187$ & $0.181$ & $0.126$ & $0.158$ &  & $0.123$ & $0.088$ & $0.089$ & $\textbf{0.082}^*$ & $\textbf{0.085}^*$ &  & $0.089$\\
DGP2 & $0.8$ & $\textbf{0.022}^*$ & $0.063$ & $0.065$ & $\textbf{0.035}^*$ & $\textbf{0.037}^*$ &  & $\textbf{0.023}^*$ & $0.037$ & $0.038$ & $\textbf{0.022}^*$ & $\textbf{0.021}^*$ &  & $0.038$\\
 & $1.0$ & $0.056$ & $\textbf{0.044}^*$ & $0.045$ & $\textbf{0.021}^*$ & $\textbf{0.031}^*$ &  & $\textbf{0.026}^*$ & $\textbf{0.044}^*$ & $0.045$ & $\textbf{0.021}^*$ & $\textbf{0.031}^*$ &  & $0.045$\\
 & $1.2$ & $0.114$ & $0.097$ & $0.102$ & $0.049$ & $0.067$ &  & $0.051$ & $0.044$ & $0.046$ & $\textbf{0.026}^*$ & $\textbf{0.038}^*$ &  & $0.045$\\
DGP3 & $0.8$ & $0.033$ & $0.037$ & $0.038$ & $\textbf{0.015}^*$ & $\textbf{0.014}^*$ &  & $\textbf{0.022}^*$ & $0.027$ & $0.028$ & $\textbf{0.018}^*$ & $\textbf{0.016}^*$ &  & $0.028$\\
 & $1.0$ & $0.058$ & $0.025$ & $0.026$ & $\textbf{0.014}^*$ & $\textbf{0.011}^*$ &  & $\textbf{0.022}^*$ & $0.025$ & $0.026$ & $\textbf{0.014}^*$ & $\textbf{0.011}^*$ &  & $0.025$\\
 & $1.2$ & $0.082$ & $0.038$ & $0.040$ & $\textbf{0.013}^*$ & $\textbf{0.015}^*$ &  & $0.030$ & $0.021$ & $0.023$ & $\textbf{0.016}^*$ & $\textbf{0.012}^*$ &  & $0.022$\\
\midrule
& & \multicolumn{11}{c}{$n=1000\quad\quad d=5$}\\
DGP1 & $0.8$ & $0.117$ & $\textbf{0.089}^*$ & $0.114$ & $\textbf{0.105}^*$ & $\textbf{0.093}^*$ &  & $\textbf{0.107}^*$ & $\textbf{0.098}^*$ & $0.110$ & $\textbf{0.095}^*$ & $\textbf{0.085}^*$ &  & $0.110$\\
 & $1.0$ & $\textbf{0.112}^*$ & $\textbf{0.106}^*$ & $0.115$ & $\textbf{0.100}^*$ & $\textbf{0.095}^*$ &  & $\textbf{0.112}^*$ & $\textbf{0.106}^*$ & $0.115$ & $\textbf{0.100}^*$ & $\textbf{0.095}^*$ &  & $0.116$\\
 & $1.2$ & $0.134$ & $0.149$ & $0.149$ & $0.113$ & $0.115$ &  & $\textbf{0.109}^*$ & $0.112$ & $0.111$ & $\textbf{0.095}^*$ & $\textbf{0.094}^*$ &  & $0.111$\\
DGP2 & $0.8$ & $\textbf{0.025}^*$ & $\textbf{0.054}^*$ & $0.075$ & $\textbf{0.044}^*$ & $\textbf{0.024}^*$ &  & $\textbf{0.014}^*$ & $\textbf{0.043}^*$ & $0.067$ & $\textbf{0.038}^*$ & $\textbf{0.015}^*$ &  & $0.066$\\
 & $1.0$ & $\textbf{0.030}^*$ & $\textbf{0.053}^*$ & $0.081$ & $\textbf{0.038}^*$ & $\textbf{0.029}^*$ &  & $\textbf{0.030}^*$ & $\textbf{0.053}^*$ & $0.081$ & $\textbf{0.038}^*$ & $\textbf{0.029}^*$ &  & $0.081$\\
 & $1.2$ & $\textbf{0.050}^*$ & $\textbf{0.078}^*$ & $0.119$ & $\textbf{0.049}^*$ & $\textbf{0.047}^*$ &  & $\textbf{0.053}^*$ & $\textbf{0.071}^*$ & $0.092$ & $\textbf{0.046}^*$ & $\textbf{0.060}^*$ &  & $0.092$\\
DGP3 & $0.8$ & $\textbf{0.034}^*$ & $\textbf{0.039}^*$ & $0.051$ & $\textbf{0.024}^*$ & $\textbf{0.017}^*$ &  & $\textbf{0.034}^*$ & $\textbf{0.038}^*$ & $0.049$ & $\textbf{0.024}^*$ & $\textbf{0.016}^*$ &  & $0.049$\\
 & $1.0$ & $\textbf{0.032}^*$ & $\textbf{0.035}^*$ & $0.049$ & $\textbf{0.020}^*$ & $\textbf{0.016}^*$ &  & $\textbf{0.032}^*$ & $\textbf{0.035}^*$ & $0.049$ & $\textbf{0.020}^*$ & $\textbf{0.016}^*$ &  & $0.048$\\
 & $1.2$ & $\textbf{0.034}^*$ & $\textbf{0.041}^*$ & $0.058$ & $\textbf{0.018}^*$ & $\textbf{0.017}^*$ &  & $\textbf{0.030}^*$ & $\textbf{0.036}^*$ & $0.047$ & $\textbf{0.018}^*$ & $\textbf{0.016}^*$ &  & $0.047$\\
\midrule
& & \multicolumn{11}{c}{$n=1000\quad\quad d=10$}\\
DGP1 & $0.8$ & $0.137$ & $0.113$ & $0.132$ & $0.139$ & $0.128$ &  & $0.106$ & $\textbf{0.100}^*$ & $0.106$ & $\textbf{0.102}^*$ & $\textbf{0.098}^*$ &  & $0.106$\\
 & $1.0$ & $\textbf{0.103}^*$ & $\textbf{0.103}^*$ & $0.105$ & $\textbf{0.100}^*$ & $\textbf{0.097}^*$ &  & $\textbf{0.103}^*$ & $\textbf{0.103}^*$ & $0.105$ & $\textbf{0.100}^*$ & $\textbf{0.097}^*$ &  & $0.105$\\
 & $1.2$ & $0.150$ & $0.177$ & $0.161$ & $0.121$ & $0.138$ &  & $0.104$ & $0.105$ & $0.104$ & $\textbf{0.094}^*$ & $\textbf{0.097}^*$ &  & $0.104$\\
DGP2 & $0.8$ & $\textbf{0.030}^*$ & $0.062$ & $0.069$ & $\textbf{0.044}^*$ & $\textbf{0.032}^*$ &  & $\textbf{0.016}^*$ & $\textbf{0.043}^*$ & $0.049$ & $\textbf{0.032}^*$ & $\textbf{0.014}^*$ &  & $0.049$\\
 & $1.0$ & $\textbf{0.030}^*$ & $\textbf{0.054}^*$ & $0.062$ & $\textbf{0.034}^*$ & $\textbf{0.032}^*$ &  & $\textbf{0.030}^*$ & $\textbf{0.054}^*$ & $0.062$ & $\textbf{0.034}^*$ & $\textbf{0.032}^*$ &  & $0.063$\\
 & $1.2$ & $\textbf{0.057}^*$ & $0.092$ & $0.108$ & $\textbf{0.052}^*$ & $\textbf{0.058}^*$ &  & $\textbf{0.047}^*$ & $\textbf{0.066}^*$ & $0.073$ & $\textbf{0.043}^*$ & $\textbf{0.057}^*$ &  & $0.073$\\
DGP3 & $0.8$ & $\textbf{0.027}^*$ & $0.038$ & $0.041$ & $\textbf{0.021}^*$ & $\textbf{0.016}^*$ &  & $\textbf{0.026}^*$ & $\textbf{0.032}^*$ & $0.035$ & $\textbf{0.021}^*$ & $\textbf{0.015}^*$ &  & $0.034$\\
 & $1.0$ & $\textbf{0.024}^*$ & $\textbf{0.031}^*$ & $0.034$ & $\textbf{0.017}^*$ & $\textbf{0.014}^*$ &  & $\textbf{0.024}^*$ & $\textbf{0.031}^*$ & $0.034$ & $\textbf{0.017}^*$ & $\textbf{0.014}^*$ &  & $0.034$\\
 & $1.2$ & $\textbf{0.029}^*$ & $0.040$ & $0.045$ & $\textbf{0.015}^*$ & $\textbf{0.017}^*$ &  & $\textbf{0.023}^*$ & $\textbf{0.030}^*$ & $0.032$ & $\textbf{0.015}^*$ & $\textbf{0.014}^*$ &  & $0.032$\\
\midrule
& & \multicolumn{11}{c}{$n=1000\quad\quad d=15$}\\
DGP1 & $0.8$ & $0.144$ & $0.132$ & $0.143$ & $0.166$ & $0.149$ &  & $0.093$ & $\textbf{0.091}^*$ & $0.093$ & $0.092$ & $\textbf{0.090}^*$ &  & $0.093$\\
 & $1.0$ & $\textbf{0.102}^*$ & $0.105$ & $0.105$ & $\textbf{0.104}^*$ & $\textbf{0.102}^*$ &  & $\textbf{0.101}^*$ & $0.105$ & $0.105$ & $\textbf{0.104}^*$ & $\textbf{0.102}^*$ &  & $0.106$\\
 & $1.2$ & $0.156$ & $0.177$ & $0.167$ & $0.125$ & $0.145$ &  & $\textbf{0.093}^*$ & $0.100$ & $0.100$ & $\textbf{0.093}^*$ & $\textbf{0.096}^*$ &  & $0.100$\\
DGP2 & $0.8$ & $\textbf{0.029}^*$ & $0.062$ & $0.066$ & $0.042$ & $\textbf{0.032}^*$ &  & $\textbf{0.016}^*$ & $\textbf{0.037}^*$ & $0.040$ & $\textbf{0.026}^*$ & $\textbf{0.014}^*$ &  & $0.040$\\
 & $1.0$ & $\textbf{0.026}^*$ & $\textbf{0.047}^*$ & $0.051$ & $\textbf{0.028}^*$ & $\textbf{0.030}^*$ &  & $\textbf{0.026}^*$ & $\textbf{0.047}^*$ & $0.051$ & $\textbf{0.028}^*$ & $\textbf{0.030}^*$ &  & $0.051$\\
 & $1.2$ & $\textbf{0.057}^*$ & $0.095$ & $0.104$ & $\textbf{0.051}^*$ & $0.061$ &  & $\textbf{0.037}^*$ & $\textbf{0.056}^*$ & $0.060$ & $\textbf{0.035}^*$ & $\textbf{0.048}^*$ &  & $0.060$\\
DGP3 & $0.8$ & $\textbf{0.022}^*$ & $0.036$ & $0.038$ & $\textbf{0.018}^*$ & $\textbf{0.013}^*$ &  & $\textbf{0.022}^*$ & $\textbf{0.027}^*$ & $0.029$ & $\textbf{0.018}^*$ & $\textbf{0.013}^*$ &  & $0.029$\\
 & $1.0$ & $\textbf{0.020}^*$ & $\textbf{0.026}^*$ & $0.028$ & $\textbf{0.015}^*$ & $\textbf{0.012}^*$ &  & $\textbf{0.021}^*$ & $\textbf{0.026}^*$ & $0.028$ & $\textbf{0.015}^*$ & $\textbf{0.012}^*$ &  & $0.028$\\
 & $1.2$ & $0.028$ & $0.037$ & $0.040$ & $\textbf{0.014}^*$ & $\textbf{0.013}^*$ &  & $\textbf{0.020}^*$ & $\textbf{0.024}^*$ & $0.026$ & $\textbf{0.015}^*$ & $\textbf{0.011}^*$ &  & $0.026$\\
\midrule
& & \multicolumn{11}{c}{$n=1000\quad\quad d=20$}\\
DGP1 & $0.8$ & $0.164$ & $0.162$ & $0.170$ & $0.199$ & $0.181$ &  & $\textbf{0.085}^*$ & $\textbf{0.092}^*$ & $0.093$ & $0.093$ & $\textbf{0.090}^*$ &  & $0.093$\\
 & $1.0$ & $\textbf{0.101}^*$ & $0.108$ & $0.108$ & $\textbf{0.108}^*$ & $\textbf{0.105}^*$ &  & $\textbf{0.101}^*$ & $0.108$ & $0.108$ & $\textbf{0.108}^*$ & $\textbf{0.105}^*$ &  & $0.109$\\
 & $1.2$ & $0.176$ & $0.196$ & $0.187$ & $0.135$ & $0.162$ &  & $0.099$ & $0.097$ & $0.094$ & $\textbf{0.089}^*$ & $\textbf{0.092}^*$ &  & $0.094$\\
DGP2 & $0.8$ & $\textbf{0.029}^*$ & $0.064$ & $0.066$ & $0.044$ & $\textbf{0.035}^*$ &  & $\textbf{0.020}^*$ & $\textbf{0.035}^*$ & $0.036$ & $\textbf{0.024}^*$ & $\textbf{0.015}^*$ &  & $0.037$\\
 & $1.0$ & $\textbf{0.029}^*$ & $\textbf{0.043}^*$ & $0.045$ & $\textbf{0.025}^*$ & $\textbf{0.028}^*$ &  & $\textbf{0.027}^*$ & $\textbf{0.043}^*$ & $0.045$ & $\textbf{0.025}^*$ & $\textbf{0.028}^*$ &  & $0.045$\\
 & $1.2$ & $0.068$ & $0.095$ & $0.101$ & $\textbf{0.050}^*$ & $0.060$ &  & $\textbf{0.034}^*$ & $\textbf{0.049}^*$ & $0.051$ & $\textbf{0.030}^*$ & $\textbf{0.042}^*$ &  & $0.051$\\
DGP3 & $0.8$ & $\textbf{0.021}^*$ & $0.037$ & $0.038$ & $\textbf{0.017}^*$ & $\textbf{0.013}^*$ &  & $\textbf{0.021}^*$ & $0.026$ & $0.027$ & $\textbf{0.018}^*$ & $\textbf{0.012}^*$ &  & $0.027$\\
 & $1.0$ & $\textbf{0.022}^*$ & $0.024$ & $0.025$ & $\textbf{0.015}^*$ & $\textbf{0.011}^*$ &  & $\textbf{0.019}^*$ & $0.024$ & $0.025$ & $\textbf{0.015}^*$ & $\textbf{0.011}^*$ &  & $0.025$\\
 & $1.2$ & $0.033$ & $0.036$ & $0.038$ & $\textbf{0.013}^*$ & $\textbf{0.012}^*$ &  & $\textbf{0.017}^*$ & $0.022$ & $0.023$ & $\textbf{0.015}^*$ & $\textbf{0.010}^*$ &  & $0.022$\\
\end{longtable}
}}

{\fontsize{8}{6}\selectfont{\setlength{\tabcolsep}{2.5pt}
\begin{longtable}{cccccccccccccccccc}
\caption{Marginal coverage results under covariate shift. The method with marginal coverage deviated from $1-\alpha$ by more than $0.02$ is highlighted with \textit{italic} and $\times$.} \label{tab: app cs mar} \\
\toprule
& & \multicolumn{5}{c}{No covariate-shift adjustment} & & \multicolumn{5}{c}{Density-ratio weighted}\\
\cmidrule(lr){3-7} \cmidrule(lr){9-13} 
& $\sigma$ & GLCP & LCP & CQR-LR & CQR-RF & CQR-LGB & & GLCP & LCP & CQR-LR & CQR-RF & CQR-LGB & & WCP \\
\midrule
\endfirsthead
\multicolumn{15}{c}{{\tablename\ \thetable{} -- continued from previous page}} \\
\toprule
& & \multicolumn{5}{c}{No covariate-shift adjustment} & & \multicolumn{5}{c}{Density-ratio weighted}\\
\cmidrule(lr){3-7} \cmidrule(lr){9-13} 
& $\sigma$ & GLCP & LCP & CQR-LR & CQR-RF & CQR-LGB & & GLCP & LCP & CQR-LR & CQR-RF & CQR-LGB & & WCP \\
\midrule
\endhead
\midrule
\multicolumn{15}{r}{{\small Continued on next page}} \\
\endfoot
\bottomrule
\endlastfoot
& & \multicolumn{11}{c}{$n=500\quad\quad d=5$}\\
DGP1 & $0.8$ & $\textit{0.879}^\times$ & $0.904$ & $0.886$ & $\textit{0.877}^\times$ & $0.880$ &  & $0.887$ & $0.893$ & $0.890$ & $0.890$ & $0.892$ &  & $0.890$\\
 & $1.0$ & $0.892$ & $0.899$ & $0.897$ & $0.896$ & $0.896$ &  & $0.893$ & $0.899$ & $0.897$ & $0.896$ & $0.896$ &  & $0.897$\\
 & $1.2$ & $\textit{0.849}^\times$ & $\textit{0.840}^\times$ & $\textit{0.848}^\times$ & $\textit{0.872}^\times$ & $\textit{0.861}^\times$ &  & $0.906$ & $0.910$ & $0.908$ & $0.907$ & $0.908$ &  & $0.910$\\
DGP2 & $0.8$ & $\textit{0.927}^\times$ & $\textit{0.947}^\times$ & $\textit{0.960}^\times$ & $\textit{0.933}^\times$ & $\textit{0.927}^\times$ &  & $0.896$ & $0.906$ & $0.907$ & $0.904$ & $0.902$ &  & $0.907$\\
 & $1.0$ & $0.899$ & $0.907$ & $0.908$ & $0.906$ & $0.907$ &  & $0.900$ & $0.907$ & $0.908$ & $0.906$ & $0.907$ &  & $0.908$\\
 & $1.2$ & $\textit{0.858}^\times$ & $\textit{0.852}^\times$ & $\textit{0.835}^\times$ & $\textit{0.874}^\times$ & $\textit{0.867}^\times$ &  & $0.909$ & $0.918$ & $0.916$ & $0.914$ & $0.915$ &  & $0.915$\\
DGP3 & $0.8$ & $0.916$ & $\textit{0.926}^\times$ & $\textit{0.934}^\times$ & $0.909$ & $0.909$ &  & $0.897$ & $0.904$ & $0.904$ & $0.901$ & $0.902$ &  & $0.904$\\
 & $1.0$ & $0.896$ & $0.903$ & $0.903$ & $0.900$ & $0.900$ &  & $0.897$ & $0.903$ & $0.903$ & $0.900$ & $0.900$ &  & $0.903$\\
 & $1.2$ & $0.880$ & $0.880$ & $\textit{0.873}^\times$ & $0.895$ & $0.890$ &  & $0.900$ & $0.908$ & $0.908$ & $0.905$ & $0.902$ &  & $0.908$\\
\midrule
& & \multicolumn{11}{c}{$n=500\quad\quad d=15$}\\
DGP1 & $0.8$ & $\textit{0.822}^\times$ & $\textit{0.831}^\times$ & $\textit{0.822}^\times$ & $\textit{0.779}^\times$ & $\textit{0.800}^\times$ &  & $0.893$ & $0.899$ & $0.897$ & $0.894$ & $0.897$ &  & $0.898$\\
 & $1.0$ & $0.899$ & $0.913$ & $0.911$ & $0.906$ & $0.907$ &  & $0.902$ & $0.913$ & $0.911$ & $0.906$ & $0.907$ &  & $0.911$\\
 & $1.2$ & $\textit{0.790}^\times$ & $\textit{0.798}^\times$ & $\textit{0.806}^\times$ & $\textit{0.855}^\times$ & $\textit{0.822}^\times$ &  & $0.885$ & $0.907$ & $0.909$ & $0.909$ & $0.910$ &  & $0.907$\\
DGP2 & $0.8$ & $\textit{0.922}^\times$ & $\textit{0.962}^\times$ & $\textit{0.965}^\times$ & $\textit{0.935}^\times$ & $\textit{0.934}^\times$ &  & $0.894$ & $0.903$ & $0.904$ & $0.899$ & $0.900$ &  & $0.904$\\
 & $1.0$ & $0.884$ & $0.906$ & $0.907$ & $0.903$ & $0.902$ &  & $0.905$ & $0.906$ & $0.907$ & $0.903$ & $0.902$ &  & $0.906$\\
 & $1.2$ & $\textit{0.829}^\times$ & $\textit{0.821}^\times$ & $\textit{0.816}^\times$ & $\textit{0.860}^\times$ & $\textit{0.844}^\times$ &  & $0.887$ & $0.903$ & $0.905$ & $0.900$ & $0.901$ &  & $0.902$\\
DGP3 & $0.8$ & $0.908$ & $\textit{0.935}^\times$ & $\textit{0.937}^\times$ & $0.903$ & $0.909$ &  & $0.897$ & $0.903$ & $0.903$ & $0.901$ & $0.900$ &  & $0.903$\\
 & $1.0$ & $0.882$ & $0.903$ & $0.903$ & $0.899$ & $0.900$ &  & $0.903$ & $0.903$ & $0.903$ & $0.899$ & $0.900$ &  & $0.902$\\
 & $1.2$ & $\textit{0.859}^\times$ & $\textit{0.872}^\times$ & $\textit{0.870}^\times$ & $0.899$ & $0.888$ &  & $0.901$ & $0.904$ & $0.904$ & $0.901$ & $0.902$ &  & $0.903$\\
\midrule
& & \multicolumn{11}{c}{$n=500\quad\quad d=20$}\\
DGP1 & $0.8$ & $\textit{0.762}^\times$ & $\textit{0.796}^\times$ & $\textit{0.793}^\times$ & $\textit{0.741}^\times$ & $\textit{0.766}^\times$ &  & $\textit{0.844}^\times$ & $0.911$ & $0.909$ & $0.910$ & $0.910$ &  & $0.909$\\
 & $1.0$ & $\textit{0.874}^\times$ & $0.911$ & $0.910$ & $0.905$ & $0.908$ &  & $0.915$ & $0.911$ & $0.910$ & $0.905$ & $0.908$ &  & $0.911$\\
 & $1.2$ & $\textit{0.736}^\times$ & $\textit{0.777}^\times$ & $\textit{0.782}^\times$ & $\textit{0.842}^\times$ & $\textit{0.807}^\times$ &  & $\textit{0.831}^\times$ & $0.910$ & $0.909$ & $0.904$ & $0.909$ &  & $0.910$\\
DGP2 & $0.8$ & $0.895$ & $\textit{0.962}^\times$ & $\textit{0.964}^\times$ & $\textit{0.934}^\times$ & $\textit{0.935}^\times$ &  & $0.904$ & $0.908$ & $0.909$ & $0.908$ & $0.910$ &  & $0.908$\\
 & $1.0$ & $\textit{0.848}^\times$ & $0.905$ & $0.904$ & $0.901$ & $0.902$ &  & $0.899$ & $0.905$ & $0.904$ & $0.901$ & $0.902$ &  & $0.905$\\
 & $1.2$ & $\textit{0.787}^\times$ & $\textit{0.815}^\times$ & $\textit{0.811}^\times$ & $\textit{0.856}^\times$ & $\textit{0.842}^\times$ &  & $\textit{0.856}^\times$ & $0.907$ & $0.907$ & $0.903$ & $0.907$ &  & $0.907$\\
DGP3 & $0.8$ & $\textit{0.875}^\times$ & $\textit{0.935}^\times$ & $\textit{0.936}^\times$ & $0.904$ & $0.909$ &  & $0.905$ & $0.910$ & $0.910$ & $0.910$ & $0.909$ &  & $0.910$\\
 & $1.0$ & $\textit{0.844}^\times$ & $0.902$ & $0.902$ & $0.901$ & $0.901$ &  & $0.896$ & $0.902$ & $0.902$ & $0.901$ & $0.901$ &  & $0.902$\\
 & $1.2$ & $\textit{0.818}^\times$ & $\textit{0.870}^\times$ & $\textit{0.868}^\times$ & $0.902$ & $0.890$ &  & $\textit{0.878}^\times$ & $0.901$ & $0.903$ & $0.900$ & $0.900$ &  & $0.900$\\
\midrule
& & \multicolumn{11}{c}{$n=1000\quad\quad d=5$}\\
DGP1 & $0.8$ & $\textit{0.868}^\times$ & $0.910$ & $0.882$ & $\textit{0.873}^\times$ & $\textit{0.878}^\times$ &  & $0.884$ & $0.891$ & $0.888$ & $0.891$ & $0.892$ &  & $0.888$\\
 & $1.0$ & $0.889$ & $0.899$ & $0.894$ & $0.893$ & $0.894$ &  & $0.889$ & $0.899$ & $0.894$ & $0.893$ & $0.894$ &  & $0.895$\\
 & $1.2$ & $\textit{0.852}^\times$ & $\textit{0.839}^\times$ & $\textit{0.845}^\times$ & $\textit{0.871}^\times$ & $\textit{0.865}^\times$ &  & $0.896$ & $0.904$ & $0.905$ & $0.900$ & $0.903$ &  & $0.905$\\
DGP2 & $0.8$ & $0.919$ & $\textit{0.945}^\times$ & $\textit{0.962}^\times$ & $\textit{0.934}^\times$ & $\textit{0.922}^\times$ &  & $0.899$ & $0.906$ & $0.908$ & $0.905$ & $0.903$ &  & $0.908$\\
 & $1.0$ & $0.900$ & $0.908$ & $0.910$ & $0.907$ & $0.907$ &  & $0.901$ & $0.908$ & $0.910$ & $0.907$ & $0.907$ &  & $0.910$\\
 & $1.2$ & $\textit{0.868}^\times$ & $\textit{0.859}^\times$ & $\textit{0.837}^\times$ & $\textit{0.877}^\times$ & $\textit{0.878}^\times$ &  & $0.906$ & $0.911$ & $0.914$ & $0.910$ & $0.913$ &  & $0.914$\\
DGP3 & $0.8$ & $0.916$ & $\textit{0.925}^\times$ & $\textit{0.935}^\times$ & $0.913$ & $0.911$ &  & $0.899$ & $0.903$ & $0.905$ & $0.903$ & $0.902$ &  & $0.905$\\
 & $1.0$ & $0.899$ & $0.904$ & $0.905$ & $0.904$ & $0.903$ &  & $0.900$ & $0.904$ & $0.905$ & $0.904$ & $0.903$ &  & $0.905$\\
 & $1.2$ & $0.886$ & $0.884$ & $\textit{0.875}^\times$ & $0.899$ & $0.895$ &  & $0.901$ & $0.904$ & $0.907$ & $0.906$ & $0.905$ &  & $0.907$\\
\midrule
& & \multicolumn{11}{c}{$n=1000\quad\quad d=10$}\\
DGP1 & $0.8$ & $\textit{0.841}^\times$ & $\textit{0.877}^\times$ & $\textit{0.854}^\times$ & $\textit{0.837}^\times$ & $\textit{0.847}^\times$ &  & $0.887$ & $0.896$ & $0.894$ & $0.893$ & $0.894$ &  & $0.894$\\
 & $1.0$ & $0.898$ & $0.907$ & $0.906$ & $0.903$ & $0.904$ &  & $0.898$ & $0.907$ & $0.906$ & $0.903$ & $0.904$ &  & $0.906$\\
 & $1.2$ & $\textit{0.827}^\times$ & $\textit{0.805}^\times$ & $\textit{0.824}^\times$ & $\textit{0.862}^\times$ & $\textit{0.842}^\times$ &  & $0.900$ & $0.913$ & $0.911$ & $0.906$ & $0.910$ &  & $0.912$\\
DGP2 & $0.8$ & $\textit{0.926}^\times$ & $\textit{0.957}^\times$ & $\textit{0.964}^\times$ & $\textit{0.940}^\times$ & $\textit{0.929}^\times$ &  & $0.895$ & $0.903$ & $0.905$ & $0.902$ & $0.900$ &  & $0.904$\\
 & $1.0$ & $0.899$ & $0.907$ & $0.908$ & $0.904$ & $0.904$ &  & $0.899$ & $0.907$ & $0.908$ & $0.904$ & $0.904$ &  & $0.908$\\
 & $1.2$ & $\textit{0.856}^\times$ & $\textit{0.835}^\times$ & $\textit{0.823}^\times$ & $\textit{0.862}^\times$ & $\textit{0.861}^\times$ &  & $0.902$ & $0.912$ & $0.911$ & $0.909$ & $0.911$ &  & $0.912$\\
DGP3 & $0.8$ & $0.917$ & $\textit{0.932}^\times$ & $\textit{0.936}^\times$ & $0.910$ & $0.911$ &  & $0.896$ & $0.901$ & $0.901$ & $0.899$ & $0.899$ &  & $0.902$\\
 & $1.0$ & $0.897$ & $0.904$ & $0.904$ & $0.901$ & $0.902$ &  & $0.898$ & $0.904$ & $0.904$ & $0.901$ & $0.902$ &  & $0.904$\\
 & $1.2$ & $0.882$ & $\textit{0.877}^\times$ & $\textit{0.872}^\times$ & $0.899$ & $0.892$ &  & $0.899$ & $0.907$ & $0.905$ & $0.902$ & $0.903$ &  & $0.905$\\
\midrule
& & \multicolumn{11}{c}{$n=1000\quad\quad d=15$}\\
DGP1 & $0.8$ & $\textit{0.816}^\times$ & $\textit{0.838}^\times$ & $\textit{0.825}^\times$ & $\textit{0.793}^\times$ & $\textit{0.811}^\times$ &  & $0.894$ & $0.906$ & $0.906$ & $0.905$ & $0.907$ &  & $0.906$\\
 & $1.0$ & $0.908$ & $0.916$ & $0.916$ & $0.910$ & $0.913$ &  & $0.909$ & $0.916$ & $0.916$ & $0.910$ & $0.913$ &  & $0.916$\\
 & $1.2$ & $\textit{0.815}^\times$ & $\textit{0.800}^\times$ & $\textit{0.812}^\times$ & $\textit{0.856}^\times$ & $\textit{0.832}^\times$ &  & $0.910$ & $0.915$ & $0.915$ & $0.906$ & $0.912$ &  & $0.915$\\
DGP2 & $0.8$ & $\textit{0.927}^\times$ & $\textit{0.960}^\times$ & $\textit{0.964}^\times$ & $\textit{0.940}^\times$ & $\textit{0.930}^\times$ &  & $0.900$ & $0.905$ & $0.905$ & $0.902$ & $0.904$ &  & $0.905$\\
 & $1.0$ & $0.897$ & $0.905$ & $0.906$ & $0.902$ & $0.903$ &  & $0.897$ & $0.905$ & $0.906$ & $0.902$ & $0.903$ &  & $0.905$\\
 & $1.2$ & $\textit{0.850}^\times$ & $\textit{0.821}^\times$ & $\textit{0.814}^\times$ & $\textit{0.858}^\times$ & $\textit{0.851}^\times$ &  & $0.902$ & $0.907$ & $0.910$ & $0.904$ & $0.906$ &  & $0.909$\\
DGP3 & $0.8$ & $0.915$ & $\textit{0.934}^\times$ & $\textit{0.936}^\times$ & $0.908$ & $0.909$ &  & $0.897$ & $0.904$ & $0.905$ & $0.903$ & $0.905$ &  & $0.904$\\
 & $1.0$ & $0.895$ & $0.903$ & $0.903$ & $0.901$ & $0.902$ &  & $0.895$ & $0.903$ & $0.903$ & $0.901$ & $0.902$ &  & $0.902$\\
 & $1.2$ & $\textit{0.879}^\times$ & $\textit{0.874}^\times$ & $\textit{0.870}^\times$ & $0.902$ & $0.894$ &  & $0.893$ & $0.901$ & $0.902$ & $0.900$ & $0.902$ &  & $0.902$\\
\midrule
& & \multicolumn{11}{c}{$n=1000\quad\quad d=20$}\\
DGP1 & $0.8$ & $\textit{0.782}^\times$ & $\textit{0.793}^\times$ & $\textit{0.784}^\times$ & $\textit{0.746}^\times$ & $\textit{0.765}^\times$ &  & $0.891$ & $0.900$ & $0.900$ & $0.900$ & $0.902$ &  & $0.900$\\
 & $1.0$ & $0.899$ & $0.908$ & $0.907$ & $0.903$ & $0.906$ &  & $0.901$ & $0.908$ & $0.907$ & $0.903$ & $0.906$ &  & $0.907$\\
 & $1.2$ & $\textit{0.779}^\times$ & $\textit{0.769}^\times$ & $\textit{0.778}^\times$ & $\textit{0.832}^\times$ & $\textit{0.804}^\times$ &  & $0.882$ & $0.908$ & $0.911$ & $0.902$ & $0.907$ &  & $0.911$\\
DGP2 & $0.8$ & $\textit{0.926}^\times$ & $\textit{0.963}^\times$ & $\textit{0.965}^\times$ & $\textit{0.943}^\times$ & $\textit{0.934}^\times$ &  & $0.899$ & $0.904$ & $0.905$ & $0.902$ & $0.905$ &  & $0.904$\\
 & $1.0$ & $0.889$ & $0.906$ & $0.906$ & $0.904$ & $0.904$ &  & $0.897$ & $0.906$ & $0.906$ & $0.904$ & $0.904$ &  & $0.906$\\
 & $1.2$ & $\textit{0.836}^\times$ & $\textit{0.817}^\times$ & $\textit{0.812}^\times$ & $\textit{0.857}^\times$ & $\textit{0.850}^\times$ &  & $0.898$ & $0.907$ & $0.907$ & $0.903$ & $0.903$ &  & $0.909$\\
DGP3 & $0.8$ & $0.912$ & $\textit{0.935}^\times$ & $\textit{0.937}^\times$ & $0.908$ & $0.909$ &  & $0.899$ & $0.903$ & $0.903$ & $0.903$ & $0.903$ &  & $0.903$\\
 & $1.0$ & $0.890$ & $0.903$ & $0.903$ & $0.902$ & $0.902$ &  & $0.901$ & $0.903$ & $0.903$ & $0.902$ & $0.902$ &  & $0.902$\\
 & $1.2$ & $\textit{0.871}^\times$ & $\textit{0.873}^\times$ & $\textit{0.870}^\times$ & $0.903$ & $0.893$ &  & $0.904$ & $0.901$ & $0.902$ & $0.900$ & $0.900$ &  & $0.903$\\
\end{longtable}
}}

\subsection{Details for Conditional-Coverage-Oriented Model Selection Experiments}\label{sec: detail selection experiment}

This section gives additional details for the model-selection experiment in Section~\ref{sec: model selection} of the main text. The experiment compares several rules for selecting one method from a pool of candidate conformal procedures.
We consider the following selection rules:
\begin{itemize}
    \item \textbf{AvgLoss:} selects the candidate with the smallest average loss, computed over bandwidths corresponding to the three target effective sample sizes $30$, $40$, and $50$.
    \item \textbf{AvgRankLoss:} ranks the candidates separately for the three losses, averages the ranks, and selects the candidate with the smallest average rank.
    \item \textbf{EffSize baseline:} selects the candidate with the smallest average prediction-set size on the calibration set.
    \item \textbf{Rand baseline:} randomly selects a candidate from the pool.
\end{itemize}

The candidate conformal sets differ in the underlying conditional quantile estimator. We consider $20$ candidates, grouped as follows:
\begin{itemize}
    \item \textbf{RKHS-based quantile regression:} four candidates obtained from penalty coefficient $\lambda_0\in\{0.01,0.05\}$ and kernel bandwidth $h\in\{1.0,1.5\}$.
    
    \item \textbf{Engression-based quantile regression \citep{shen2025engression}:} four candidates with network architectures $$(\text{number of layers},\text{hidden dimension})\in\{(3,50),(3,100),(2,50),(2,100)\}\,.$$ 
    All are trained for $500$ epochs with learning rate $0.001$. Quantile predictions are obtained by inverting the estimated conditional distribution at the target quantile level.
    
    \item \textbf{Linear quantile regression:} two candidates fitted by penalized quantile regression with $L_2$ penalty coefficients $\lambda_0\in\{0.01,0.05\}$.
    
    \item \textbf{Random-forest quantile regression \citep{meinshausen2006quantile}:} two candidates, each using $100$ trees, a fixed random seed, and a minimum split size of $2$. The maximum tree depth is set to $5$ or $10$.
    
    \item \textbf{LightGBM quantile regression \citep{NIPS2017_6449f44a}:} eight candidates, each trained for $200$ boosting rounds without early stopping. The candidates vary by number of leaves $\{11,21\}$, learning rate $\{0.01,0.05\}$, and feature and bagging fractions $\{0.6,0.8\}$, with bagging frequency fixed at $5$.
\end{itemize}

\subsection{Detailed Setup for the Toloker Graph Application}
\label{sec: graph details}

This section gives the experimental details for the Toloker graph application in Section~\ref{sec: graph real data}.

\paragraph*{Dataset and preprocessing} 
We use the Toloker crowdsourcing dataset, downloaded from \texttt{Kaggle.com}. The graph contains 11,758 nodes, corresponding to annotators, and 519,000 undirected edges, corresponding to shared task annotations. Each node has four numerical features, \texttt{approved\_rate}, \texttt{skipped\_rate}, \texttt{expired\_rate}, and \texttt{rejected\_rate}; three categorical features, \texttt{english\_profile}, \texttt{english\_tested}, and \texttt{education}; and a binary label indicating whether the annotator was banned. The goal is to construct prediction sets for the ban status with reliable performance across latent graph communities.

\paragraph*{Data splits and graph construction} 
The dataset provides 10 predefined random splits. For each split, nodes are allocated to training, calibration, and test sets in a $2:1:1$ ratio. To maintain the inductive setting, we construct the graph as follows: (i) for the training set, we keep only edges whose endpoints are both in the training set; (ii) for the calibration set, we keep only edges whose endpoints are both in the calibration set; (iii) for a test node, we keep only its edges to nodes in the calibration set. This guarantees that predictions for a test node rely solely on the calibration set and the node's own features.

\paragraph*{Feature engineering and model training} 
We augment each node with features derived from its neighborhood. For a node and its neighbors (according to the edges retained as described above), we compute: (a) seven mean statistics (one per original numerical/categorical feature) of the neighbor's numerical/categorical features, and (b) the average proportion of banned neighbors. For test nodes, the latter is computed using only their known neighbors in the calibration set. This yields eight new features. Adding a binary indicator for whether the node has any neighbor gives a total of fifteen features (twelve numeric, three categorical). A LightGBM classifier is trained on these features (100 estimators, learning rate 0.1, max depth 5), achieving an average accuracy of about 79\% over the 10 splits.

\paragraph*{Community detection and evaluation}
To evaluate community-conditional coverage, we obtain a reference community partition by applying the Louvain algorithm to the full graph. Communities with fewer than 10 nodes are merged into a single ``outlier'' community, resulting in 8 final communities. These community labels are not used by the conformal procedures and are used only for evaluation. We compare the proposed GraphCP method with standard conformal prediction (StdCP), using the same base score derived from the LightGBM predicted probabilities. Table~\ref{table: tolokerTS} in the main text reports the results aggregated over the 10 splits, including the community proportions, conditional miscoverage, and average prediction set sizes for each method, together with the overall aggregate performance.

\end{document}